\documentclass[a4paper,11pt]{article}
\usepackage[utf8]{inputenc}
\usepackage[T1]{fontenc}
\usepackage{amsfonts}

\usepackage{tikz}
\usetikzlibrary{shapes,calc,decorations.markings, decorations.pathmorphing, 3d}
\usepackage{ifthen}

\tikzset{
  equation/.style={
    baseline={([yshift=-1.5ex]current bounding box.center)}
  },
  torus horizontal/.style = {
  decoration={
    markings,
    mark=at position 0.5 with {
            \draw (-2pt,-2pt) -- (2pt,2pt);
            \draw (2pt,-2pt) -- (-2pt,2pt);
    }}, decorate},
  torus vertical/.style = {
  decoration={
    markings,
    mark=at position 0.5 with {
            \draw (-2pt,-2pt) -- (2pt,2pt);
            \draw (-3pt,-2pt) -- (1pt,2pt);
    }}, decorate}
}

\usepackage{caption}
\usepackage{subcaption}

\usepackage[english]{babel}
\usepackage{amssymb,amsmath}
\usepackage{xcolor}
\usepackage{hyperref}
\usepackage{bbm}
\usepackage{mathtools}

\usepackage{physics}

\usepackage{amsthm}

\newtheorem{Theo}{Theorem}[section]
\newtheorem{Prop}[Theo]{Proposition}
\newtheorem{Coro}[Theo]{Corollary}
\newtheorem{Lemm}[Theo]{Lemma}
\newtheorem{Defi}[Theo]{Definition}
\newtheorem{Assumption}[Theo]{Assumption}

\newtheorem{Rema}[Theo]{Remark}

\newcommand{\EE}{\mathcal{E}}
\newcommand{\RR}{\mathcal{R}}
\newcommand{\VV}{\mathcal{V}}
\newcommand{\SSS}{\mathcal{S}}
\newcommand{\XXX}{\mathcal{X}}

\def\gap{\operatorname{gap}}
\def\R{\mathbb{ R}}

\usepackage{authblk}

\title{Thermalization in Kitaev's quantum double models via Tensor Network techniques}
\author[1,2,3]{Angelo Lucia\thanks{ {\tt anglucia@ucm.es} , \, \url{https://orcid.org/0000-0003-1709-1220}}}
\author[2,3]{David Pérez-García\thanks{ {\tt dperezga@ucm.es} , \, \url{https://orcid.org/0000-0003-2990-791X}}}
\author[3,4]{Antonio Pérez-Hernández\thanks{ {\tt antperez@ind.uned.es} , \, \url{https://orcid.org/0000-0001-8600-7083}  }   }

\affil[1]{
    Walter Burke Institute for Theoretical Physics
    and Institute for Quantum Information and Matter, 
    Caltech, Pasadena, CA 91125, USA
}

\affil[2]{
    Instituto de Ciencias Matemáticas,
    28049 Madrid, Spain
    }
    
\affil[3]{
    Departamento de Análisis y Matemática Aplicada,
    Universidad Complutense de Madrid, 28040 Madrid, Spain
}

\affil[4]{
    Departamento de Matemática Aplicada I,
    Universidad Nacional de Educación a Distancia, 28040 Madrid, Spain
}

\begin{document}

\maketitle

{\small \noindent \textbf{Acknowledgements:} \footnotesize \\
We acknowledge financial support from MCIN/AEI/10.13039/501100011033 (grants MTM2017-88385-P, PID2020-113523GB-I00, SEV-2015-0554 and CEX2019-000904-S), from Comunidad de Madrid (grant QUITEMAD-CM, ref. S2018/TCS-4342), and the European Research Council (ERC) under the European Union’s Horizon 2020 research and innovation programme (grant agreement No 648913).
A.\,L.~acknowledges support from the Walter Burke Institute for Theoretical Physics in the form of the Sherman Fairchild Fellowship, and  by grant RYC2019-026475-I funded by MCIN/AEI/10.13039/501100011033 and ``ESF Investing in your future''. 
A.\,P.\,H.~ acknowledges support from the grants ``Juan de la Cierva Formación'' (FJC2018-036519-I), and of the ETSI Industriales (UNED) of Spain, projects 2023-ETSII-UNED-01 and 2021-MAT11 .\\
}

{\small\noindent \textbf{Abstract}:
We show that every ergodic Davies generator associated to \textit{any} 2D Kitaev's quantum double model has a non-vanishing spectral gap in the thermodynamic limit. This validates rigorously the extended belief that those models are useless as self-correcting quantum memories, even in the non-abelian case. The proof uses recent ideas and results regarding the characterization of the spectral gap for parent Hamiltonians associated to Projected Entangled Pair States in terms of a bulk-boundary correspondence.\\


}

\newpage

\tableofcontents


\def\weight{\mathcal{G}}

\def\weightS{
\mathcal{G}_{
\begin{tikzpicture}[equation, scale=0.2]
\draw[thick, gray] (-0.5,0) -- (0.5,0);
\draw[thick, gray] (0,-0.5) -- (0,0.5);
\end{tikzpicture}}
}

\def\weightP{
\mathcal{G}_{\,
\begin{tikzpicture}[equation, scale=0.2]
\draw[thick, gray] (-0.5,-0.5) rectangle (0.5,0.5);
\end{tikzpicture}}
}


\def\centerarc[#1](#2)(#3:#4:#5)
    { \draw[#1] ($(#2)+({#5*cos(#3)},{#5*sin(#3)})$) arc (#3:#4:#5); }


\def\edgePEPO[#1]{


\begin{scope}[ scale=1]

\draw[thick, gray]  (0,-0.5) -- (0,0.5) ;

\def\r{0.8};

\begin{scope}[xshift=1cm, thick]
  \draw[postaction={decorate}, blue] (-0.5*\r,0.5*\r)  -- (0-\r,0); 
\draw[postaction={decorate}, blue] 
(0-\r,0)  -- (-0.5*\r,-0.5*\r); 

     \ifthenelse{#1=1}{ 
\filldraw[fill=white, draw=blue, thick] (-0.75*\r,0.25*\r) circle (0.05);
\filldraw[fill=white, draw=blue, thick] (-0.75*\r,-0.25*\r) circle (0.05);
}{}
\end{scope}

\begin{scope}[xshift=-1cm, thick]    
    \draw[postaction={decorate}, blue]   (\r,0) -- (0.5*\r,0.5*\r); 
    \draw[postaction={decorate}, blue]  (0.5*\r,-0.5*\r) -- (\r,0);

     \ifthenelse{#1=1}{ 
\filldraw[fill=white, draw=blue, thick] (0.75*\r,0.25*\r) circle (0.05);
\filldraw[fill=white, draw=blue, thick] (0.75*\r,-0.25*\r) circle (0.05);
}{}
\end{scope}

\begin{scope}[yshift=1cm, thick]
  \draw[postaction={decorate}, red] (-0.5*\r,-0.5*\r)  -- (0,-\r); 
    \draw[postaction={decorate}, red] (0,-\r) -- (0.5*\r,-0.5*\r); 
    
 \ifthenelse{#1=1}{ 
\filldraw[fill=white, draw=red, thick] (-0.25*\r,-0.75*\r) circle (0.05);
\filldraw[fill=white, draw=red, thick] (0.25*\r,-0.75*\r) circle (0.05);
}{}

\end{scope}

\begin{scope}[yshift=-1cm, thick]
  \draw[postaction={decorate}, red] (0.5*\r,0.5*\r)  -- (0,\r); 
    \draw[postaction={decorate}, red] (0,\r) -- (-0.5*\r,0.5*\r); 
    
 \ifthenelse{#1=1}{   
\filldraw[fill=white, draw=red, thick] (0.25*\r,0.75*\r) circle (0.05);
\filldraw[fill=white, draw=red, thick] (-0.25*\r,0.75*\r) circle (0.05);
}{}

\end{scope}

\end{scope}
}


\def\TransferOperator[#1,#2]{


\begin{scope}[decoration={markings, mark=at position 0.55 with {\arrow{stealth}}}]

\ifthenelse{#1=1}{  

\draw[very thin] (0,-0.5) -- (0,0.5);    

\draw[thick, blue, postaction={decorate}] + (0,0.5) arc (0:50:1.5cm and 0.2cm);

\draw[thick, blue, postaction={decorate}] + (0,-0.5) arc (0:50:1.5cm and 0.2cm);

\ifthenelse{#2=1}{

\filldraw[fill=white, draw=blue, thin] ({-1.5+1.5*cos(43)},{0.5+0.2*sin(43)}) circle (0.05);

\filldraw[fill=white, draw=blue, thin] ({-1.5+1.5*cos(43)},{-0.5+0.2*sin(43)}) circle (0.05);
}{}

}{}

\ifthenelse{#1=2}{  

\draw[very thin] (0,-0.5) -- (0,0.5);

\draw[thick, red,postaction={decorate}] + (0,0.5) arc (90:120:1cm and 1cm);

\draw[thick, red, postaction={decorate}] + (0,-0.5) arc (90:120:1cm and 1cm);

\ifthenelse{#2=1}{

\filldraw[fill=white, draw=red, thin] ({0+1*cos(113)},{-0.5+1*sin(113)}) circle (0.05);

\filldraw[fill=white, draw=red, thin] ({0+1*cos(113)},{-1.5+1*sin(113)}) circle (0.05);
}{}

}{}
\end{scope}

\begin{scope}[decoration={markings, mark=at position 0.55 with {\arrow{stealth reversed}}}]

\ifthenelse{#1=1}{  

\draw[very thin] (0,-0.5) -- (0,0.5);    

\draw[thick, blue, postaction={decorate}] + (0,0.5) arc (0:-50:1.5cm and 0.2cm);

\draw[thick, blue,postaction={decorate}] + (0,-0.5) arc (0:-50:1.5cm and 0.2cm);

\ifthenelse{#2=1}{

\filldraw[fill=white, draw=blue, thin] ({-1.5+1.5*cos(-43)},{0.5+0.2*sin(-43)}) circle (0.05);

\filldraw[fill=white, draw=blue, thin] ({-1.5+1.5*cos(-43)},{-0.5+0.2*sin(-43)}) circle (0.05);
}{}

}{}

\ifthenelse{#1=2}{  

\draw[very thin] (0,-0.5) -- (0,0.5);    

\draw[thick, red, postaction={decorate}] + (0,0.5) arc (90:60:1cm and 1cm);

\draw[thick, red, postaction={decorate}] + (0,-0.5) arc (90:60:1cm and 1cm);

\ifthenelse{#2=1}{

\filldraw[fill=white, draw=red, thin] ({0+1*cos(67)},{-0.5+1*sin(67)}) circle (0.05);

\filldraw[fill=white, draw=red, thin] ({0+1*cos(67)},{-1.5+1*sin(67)}) circle (0.05);
}{}

}{}

\end{scope}

}


\def\edgefour[#1,#2]{


\begin{scope}[scale=1]

\draw[thick, black!20!white]  (0,-1) -- (0,1) ;

\ifthenelse{#2=1}{\filldraw[fill=black] (-0.06,-0.06) rectangle (0.06,0.06) ;}{}

\def\r{0.5};
\def\R{0.8};

\begin{scope}[xshift=1cm, thick]
\centerarc[blue,thick](0,0)(135:225:\R);
\draw [blue,-{latex}] (-\R,0) -- (-\R,-0.1);

\ifthenelse{#1=1}{\filldraw[fill=white, draw=blue, thin] ($(0,0)+({\R*cos(150)},{\R*sin(150)})$) circle (0.05) ;}{}

\ifthenelse{#1=1}{\filldraw[fill=white, draw=blue, thin] ($(0,0)+({\R*cos(210)},{\R*sin(210)})$) circle (0.05) ;}{}
\end{scope}

\begin{scope}[xshift=-1cm, thick]   
 \centerarc[blue,thick](0,0)(315:405:\R);
 \draw [blue,-{latex}] (\R,0) -- (\R,0.1);
 
 \ifthenelse{#1=1}{\filldraw[fill=white, draw=blue, thin] ($(0,0)+({\R*cos(330)},{\R*sin(330)})$) circle (0.05) ;}{}

\ifthenelse{#1=1}{\filldraw[fill=white, draw=blue, thin] ($(0,0)+({\R*cos(390)},{\R*sin(390)})$) circle (0.05) ;}{}
\end{scope}

\begin{scope}[yshift=1cm, thick]
\centerarc[red,thick](0,0)(225:315:\r);

\ifthenelse{#1=1}{\filldraw[fill=white, draw=red, thin] ($(0,0)+({\r*cos(240)},{\r*sin(240)})$) circle (0.05) ;}{}

\ifthenelse{#1=1}{\filldraw[fill=white, draw=red, thin] ($(0,0)+({\r*cos(300)},{\r*sin(300)})$) circle (0.05) ;}{}
\end{scope}

\begin{scope}[yshift=-1cm, thick]
\centerarc[postaction={decorate},red,thick](0,0)(45:135:\r);

\ifthenelse{#1=1}{\filldraw[fill=white, draw=red, thin] ($(0,0)+({\r*cos(60)},{\r*sin(60)})$) circle (0.05) ;}{}

\ifthenelse{#1=1}{\filldraw[fill=white, draw=red, thin] ($(0,0)+({\r*cos(120)},{\r*sin(120)})$) circle (0.05) ;}{}
\end{scope}

\end{scope}
}


\def\plaquettefive[#1,#2](#3,#4,#5,#6){






\def\r{0.4};
\def\R{0.7};

\begin{scope}[xshift=-0.5 cm, yshift=-0.5 cm]

\draw[step=1cm,black,xshift=0.5cm,yshift=0.5cm] (0,0) grid (#1,#2);

\foreach \x in {1,...,#1} 
\foreach \y in {1,...,#2}{

\begin{scope}[xshift=\x cm, yshift=\y cm, scale=0.5]


\ifthenelse{#3=3 \OR #3=5}{
\filldraw[fill=black,xshift=1cm,yshift=0cm] (-0.07,-0.07) rectangle (0.07,0.07);
\filldraw[fill=black,xshift=0cm,yshift=1cm] (-0.07,-0.07) rectangle (0.07,0.07);
\filldraw[fill=black,xshift=-1cm,yshift=0cm] (-0.07,-0.07) rectangle (0.07,0.07);
\filldraw[fill=black,xshift=0cm,yshift=-1cm] (-0.07,-0.07) rectangle (0.07,0.07);
}{}


\ifthenelse{#3=1 \OR #3=4 \OR #3=5} { 


\ifthenelse{\x=1}{  

\centerarc[blue,thick](-2,0)(-45:45:\R);


\ifthenelse{#5=1}{  
\draw [blue,-{latex}] (-2+\R,0.1) -- (-2+\R,0.2);
}{}


\ifthenelse{#6=1}{
\filldraw[fill=white, draw=blue, thin] ($(-2,0)+({\R*cos(-30)},{\R*sin(-30)})$) circle (0.05) ;
\filldraw[fill=white, draw=blue, thin] ($(-2,0)+({\R*cos(30)},{\R*sin(30)})$) circle (0.05) ;
}{}
}{}

\ifthenelse{\x=#1}{  

\centerarc[blue,thick](2,0)(135:225:\R);


\ifthenelse{#5=1}{  
\draw [blue,-{latex}] (2-\R,-0.1) -- (2-\R,-0.2);
}{}


\ifthenelse{#6=1}{
\filldraw[fill=white, draw=blue, thin] ($(2,0)+({\R*cos(150)},{\R*sin(150)})$) circle (0.05) ;
\filldraw[fill=white, draw=blue, thin] ($(2,0)+({\R*cos(210)},{\R*sin(210)})$) circle (0.05) ;
}{}
}{}

\ifthenelse{\y=1}{  
\centerarc[blue,thick](0,-2)(45:135:\R);


\ifthenelse{#5=1}{  
\draw [blue,-{latex}] (-0.1,\R-2) -- (-0.2,\R-2);
}{}


\ifthenelse{#6=1}{
\filldraw[fill=white, draw=blue, thin] ($(0,-2)+({\R*cos(60)},{\R*sin(60)})$) circle (0.05) ;
\filldraw[fill=white, draw=blue, thin] ($(0,-2)+({\R*cos(120)},{\R*sin(120)})$) circle (0.05) ;
}{}
}{}

\ifthenelse{\y=#2}{  

\centerarc[blue,thick](0,2)(225:315:\R);


\ifthenelse{#5=1}{  
\draw [blue,-{latex}] (0.1,2-\R) -- (0.2,2-\R);
}{}


\ifthenelse{#6=1}{
\filldraw[fill=white, draw=blue, thin] ($(0,2)+({\R*cos(240)},{\R*sin(240)})$) circle (0.05) ;
\filldraw[fill=white, draw=blue, thin] ($(0,2)+({\R*cos(300)},{\R*sin(300)})$) circle (0.05) ;
}{}
}{}


\ifthenelse{#4=1}{   
\centerarc[blue,thick](0,0)(0:360:\R);


\ifthenelse{#5=1}{  
\draw [blue,-{latex}] (\R,0.1) -- (\R,0.2);
\draw [blue,-{latex}] (-0.1,\R) -- (-0.2,\R);
\draw [blue,-{latex}] (-\R,-0.1) -- (-\R,-0.2);
\draw [blue,-{latex}] (0.1,-\R) -- (0.2,-\R);
}{}

}{}

}{}


\ifthenelse{#3=2 \OR #3=4 \OR #3=5} { 


\ifthenelse{\y < #2 \AND \x < #1 \AND #4=1}{  \centerarc[red,thick](1,1)(0:360:\r);}{}


\ifthenelse{\y=1 \AND \x < #1}{  \centerarc[red,thick](1,-1)(-45:225:\r);

\ifthenelse{#6=1}{
\filldraw[fill=white, draw=red, thin] ($(1,-1)+({\r*cos(-30)},{\r*sin(-30)})$) circle (0.05) ;
\filldraw[fill=white, draw=red, thin] ($(1,-1)+({\r*cos(210)},{\r*sin(210)})$) circle (0.05) ;
}{}

}{}

\ifthenelse{\y=#2 \AND \x < #1}{  \centerarc[red,thick](1,1)(135:405:\r);

\ifthenelse{#6=1}{
\filldraw[fill=white, draw=red, thin] ($(1,1)+({\r*cos(-150)},{\r*sin(150)})$) circle (0.05) ;
\filldraw[fill=white, draw=red, thin] ($(1,1)+({\r*cos(390)},{\r*sin(390)})$) circle (0.05) ;
}{}

}{}

\ifthenelse{\x=1 \AND \y < #2}{  \centerarc[red,thick](-1,1)(-135:135:\r);

\ifthenelse{#6=1}{
\filldraw[fill=white, draw=red, thin] ($(-1,1)+({\r*cos(-120)},{\r*sin(-120)})$) circle (0.05) ;
\filldraw[fill=white, draw=red, thin] ($(-1,1)+({\r*cos(120)},{\r*sin(120)})$) circle (0.05) ;
}{}

}{}

\ifthenelse{\x=#1 \AND \y < #2}{ 
 \centerarc[red,thick](1,1)(45:315:\r);

\ifthenelse{#6=1}{
\filldraw[fill=white, draw=red, thin] ($(1,1)+({\r*cos(60)},{\r*sin(60)})$) circle (0.05) ;
\filldraw[fill=white, draw=red, thin] ($(1,1)+({\r*cos(300)},{\r*sin(300)})$) circle (0.05) ;
}{}
 
 }{}

\ifthenelse{\y=1 \AND \x =1}{  \centerarc[red,thick](-1,-1)(-45:135:\r);

\ifthenelse{#6=1}{
\filldraw[fill=white, draw=red, thin] ($(-1,-1)+({\r*cos(-30)},{\r*sin(-30)})$) circle (0.05) ;
\filldraw[fill=white, draw=red, thin] ($(-1,-1)+({\r*cos(120)},{\r*sin(120)})$) circle (0.05) ;
}{}

}{}

\ifthenelse{\y=1 \AND \x = #1}{  \centerarc[red,thick](1,-1)(45:225:\r);

\ifthenelse{#6=1}{
\filldraw[fill=white, draw=red, thin] ($(1,-1)+({\r*cos(60)},{\r*sin(60)})$) circle (0.05) ;
\filldraw[fill=white, draw=red, thin] ($(1,-1)+({\r*cos(210)},{\r*sin(210)})$) circle (0.05) ;
}{}

}{}

\ifthenelse{\y=#2 \AND \x = 1}{  \centerarc[red,thick](-1,1)(45:-135:\r);

\ifthenelse{#6=1}{
\filldraw[fill=white, draw=red, thin] ($(-1,1)+({\r*cos(30)},{\r*sin(30)})$) circle (0.05) ;
\filldraw[fill=white, draw=red, thin] ($(-1,1)+({\r*cos(-120)},{\r*sin(-120)})$) circle (0.05) ;
}{}

}{}

\ifthenelse{\y=#2 \AND \x = #1}{  \centerarc[red,thick](1,1)(135:315:\r);

\ifthenelse{#6=1}{
\filldraw[fill=white, draw=red, thin] ($(1,1)+({\r*cos(150)},{\r*sin(150)})$) circle (0.05) ;
\filldraw[fill=white, draw=red, thin] ($(1,1)+({\r*cos(300)},{\r*sin(300)})$) circle (0.05) ;
}{}

}{}

}{}

\end{scope}

}

\end{scope}

}


\def\cylinderopen[#1,#2]{

\def\r{0.2};
\def\R{0.35};

\begin{scope}

\draw[step=1.0,black,thin] (0.1,0) grid (#1-0.1,#2);

\foreach \y in {0,...,#2}{
\draw (0,\y) -- (#1,\y);
}


\foreach \y in {1,...,#2}{
    \centerarc[blue, thick](0.5,\y-0.5)(-135:135:\R);
    \centerarc[blue, thick](#1-0.5,\y-0.5)(45:315:\R);
}

\foreach \y in {0,...,#2}{
    \centerarc[red, thick](0,\y)(-45:45:\r);
     \centerarc[red, thick](#1,\y)(135:225:\r);
}


\foreach \x in {1,...,#1}{
    \centerarc[blue, thick](\x-0.5,-0.5)(45:135:\R);
    \centerarc[blue, thick](\x-0.5,#2+0.5)(-135:-45:\R);
}


\foreach \x in {2,...,#1}{
    \centerarc[red, thick](\x-1,0)(-45:225:\r);
    \centerarc[red, thick](\x-1,#2)(135:405:\r);
}

\end{scope}

}


\newpage

\section{Introduction}

In the seminal works \cite{Dennis02,Alicki4D} it is shown that the 4D Toric Code is a self-correcting quantum memory, that is, it allows to keep quantum information protected against thermal errors (for all temperatures below a threshold) without the need for active error correction, for times that grow exponentially with the system size $N$. As interactions become highly non-local after mapping the 4D Toric Code to a 2D or 3D geometry, it has been a major open question whether similar self-correction is possible in 2D or 3D, where the information is encoded in the degenerate ground space of a locally interacting Hamiltonian in a 2D or 3D geometry. We refer to the review \cite{review-memories} for a very detailed discussion of the many different contributions to the problem, that still remains open up to date. 

Before focusing on the 2D case, which is the main goal of this work, let us briefly comment that in 3D, this question motivated the discovery of Haah's cubic code \cite{Haah-code, Haah-code-1}, which was the opening door to a family of new ultra-exotic quantum phases of matter, currently known as {\it fractons} \cite{fractons}. 

In  2D,  it is a general belief that self-correction is not possible. There is indeed compelling evidence for that. For instance, Landon-Cardinal and Poulin \cite{Landon}, extending a result of Bravyi and Terhal \cite{Bravyi-Terhal}, showed that  commuting frustration free models in 2D display only a constant energy barrier. That is, it is possible to implement a sequence of ${\rm poly}(N)$ local operations that maps one ground state into an orthogonal one and, at the same time, the energy of all intermediate states is bounded by a constant independent of $N$. This seemed to rule out the existence of self-correction in 2D. 

However, it was later shown in \cite{Brown} that having a bounded energy barrier does not exclude self-correction, since it could happen that the paths implementing changes in the ground space are highly non-typical, and hence, the system could be {\it entropically} protected. Indeed, an example is shown in \cite{Brown} where, in a very particular regime of temperatures though, entropic protection occurs. 

Therefore, in order to solve the problem in a definite manner, one needs to consider directly the mixing time of the thermal evolution operator which, in the weak coupling limit, is given by the Davies master equation \cite{Davies}. Self-correction will not be possible if the noise operator relaxes fast to the Gibbs ensemble, where all information is lost. As detailed in \cite{Alicki2D} or \cite{Komar} using standard arguments on Markovian semigroups, the key quantity that controls this relaxation time is the spectral gap of the Davies Lindbladian generator. Self-correction in 2D would be excluded if one is able to show that such a gap is uniformly lower bounded independently of the system size. This is precisely the result proven for the Toric Code by Alicki et al, already in 2002, in the pioneer work \cite{Alicki2D}.  The result was extended for the case of all abelian quantum double models by Komar et al in 2016 \cite{Komar}. Indeed, up to now, these were the only cases for which the belief that self-correction does not exist in 2D have been rigorously proven. In particular, it remained an open question (as  highlighted in the review \cite{review-memories}) whether the same result would hold for the case of {\it non-abelian} quantum double models. 

In this work, we address and solve this problem showing that non-abelian quantum double models behave as their abelian counterparts. The main result of this work is the summarized as follow: for any finite group $G$, we consider Kitaev's quantum double Hamiltonian $H$ of group $G$ defined on $\mathbb{Z}_N \times \mathbb{Z}_N$. We consider a thermal bath at inverse temperature $\beta < \infty$, acting independently on each site of $\mathbb{Z}_N \times \mathbb{Z}_N$ in a translation invariant way, described by a Davies semigroup. This is given by a family of single-site jump operators $\{S_\alpha\}_\alpha$ and positive coupling functions $\widehat{g}_\alpha$ satisfying detailed balance. The resulting generator $\mathcal{L}$ is then given by
\begin{equation}
    \mathcal{L}(Q) = \sum_{e \in \mathbb{Z}_N \times \mathbb{Z}_N} \sum_{\alpha,\omega} \frac{\widehat{g}_\alpha(\omega)}{2}
    \qty( S_{e,\alpha}^{\dagger}(\omega) [Q, S_{e,\alpha}(\omega)] +  [S_{e,\alpha}^{\dagger}(\omega), Q] \, S_{e,\alpha}(\omega) )
\end{equation}
where $\omega$ runs over the Bohr frequencies of $H$, i.e., the differences between eigenvalues of $H$, and $S_{e,\alpha}(\omega)$ are the Fourier coefficients of $S_\alpha$ acting on edge $e$, with respect to the evolution by $H$. As these vanish for all values of $\omega$ outside of a finite set $\Omega$, we can without loss of generality restrict the sum to $\omega \in \Omega$. See Section~\ref{sec:davies-generators} for a more complete explanation of the construction. 
We state and prove the theorem for the case of a translation invariant Lindbladian for simplicity: with minor adaptations, our proof could be extended to the non translation invariant case, as long as it is possible to obtain uniform estimates on the behavior of the local generators (see Remark~\ref{rema:translation-invariance}).

\vspace{0.2cm}

\begin{Theo}\label{Theo:MainResultIntroduction} 
Suppose that the jump operators satisfy $\{ S_\alpha\}_\alpha' = \mathbb{C}\mathbbm{1}$, where $'$ denotes the commutator. Then the Davies generator $\mathcal{L}$ defined above is ergodic, and its spectral gap  has a lower bound which is independent of the system size $N$.
Specifically, there exist positive constants $C$ and $\lambda$, independent of $\beta$ and the system size, such that
\begin{equation}
    \operatorname{gap}(\mathcal{L}) \ge\, \widehat{g}_{\min} \, e^{-C\, e^{\beta}} \, \lambda \qc \widehat{g}_{\min} = \min_{\alpha} \min_{\omega \in \Omega} \widehat{g}_\alpha(\omega)
\end{equation}
\end{Theo}
The constant $\lambda$ will depend both on the group $G$ and on the choice of the jump operators $\{S_\alpha\}_\alpha$. 
Note that while in principle $\widehat{g}_{\min}$ could also scale with $\beta$, there are examples where it can be lower bounded by a strictly positive constant independent of the temperature. The dependence of our bound on $\beta$ is \emph{worse} than the ones obtained in the previous works for the case of an abelian group $G$ \cite{Alicki2D, Komar} (double exponential instead of exponential): we believe this dependence is an artifact of our proof and therefore is probably not optimal.

The tools used to address the main theorem are completely different from those used in the  abelian case in \cite{Komar}. There, following ideas of \cite{Temme}, the authors bound the spectral gap of $\mathcal{L}$ via a quantum version of the canonical-paths method in classical Markov chains. 
Instead, we go back to the original idea of Alicki et al. for the Toric Code \cite{Alicki2D}: construct an artificial Hamiltonian from the Davies generator $\mathcal{L}$ so that the spectral gap of $\mathcal{L}$ coincides with the spectral gap above the ground state of that Hamiltonian, and then use techniques to bound spectral gaps of many body Hamiltonians. 
This trick has already found other interesting implications in quantum information, especially in problems related to thermalization such as the behavior of random quantum circuits \cite{random-circuits} or the convergence of Gibbs sampling protocols \cite{Kastoryano-Brandao}.

In particular, we will follow closely the implementation of the idea used in \cite{VWPGC06}, and reason as follows. We purify the Gibbs state $\rho_\beta$ and consider the (pure) thermofield double state $|\rho_\beta^{1/2}\rangle$ (i.e.\ the cyclic vector of the GNS representation of the algebra of observables with state $\rho_\beta$). The commutativity of the terms in the quantum double Hamiltonian $H$ makes $|\rho_\beta^{1/2}\rangle$ a Projected Entangled Pair State (PEPS). We will show then (see Proposition~\ref{Prop:davies-to-parent-hamiltonian}) that the gap of $\mathcal{L}$ can be lower bounded by the gap of the parent Hamiltonian of $|\rho_\beta^{1/2}\rangle$ in the PEPS formalism. 

This opens the door to exploit the extensive knowledge gained in the area of Tensor Networks during the last decades. Tensor Networks, and in particular PEPS, have revealed themselves as an invaluable tool to understand, classify and simulate strongly correlated quantum systems (see e.g. the reviews \cite{Romanreview,Murg-review,Cirac20}). The key reason is that they approximate well the ground and thermal states of short-range Hamiltonians  and, at the same time, display a local structure that allows to describe and manipulate them efficiently \cite{Cirac20}.

Such a local structure manifests itself in a bulk-boundary correspondence that was first uncovered in \cite{Cirac11}, where one can associate to each patch of the 2D PEPS a 1D mixed state that lives on the boundary of the patch. It is conjectured in \cite{Cirac11}, and verified numerically for some examples, that the gap of the parent Hamiltonian in the bulk corresponds to a form of locality in the associated boundary state. 

This bulk-boundary correspondence was made rigorous for the first time in \cite{KaLuPeGa19} (see also the subsequent contribution \cite{Antonio}). In particular, it is shown in \cite{KaLuPeGa19} that if the boundary state displays a locality property called {\it approximate factorization}, then the bulk parent Hamiltonian has a non-vanishing spectral gap in the thermodynamic limit.

Roughly speaking, approximate factorization can be defined as follows. Consider a 1D chain of $N$ sites that we divide in 3 regions: left (L), middle (M) and right (R).  A mixed state $\rho_{LMR}$ is said to approximately factorize if if can be written as
$$\rho_{LMR} \approx \left( \Omega_{LM} \otimes \mathbbm{1}_R \right) \left( \mathbbm{1}_L\otimes \Delta_{MR} \right) \,, $$
where, for a particular notion of distance, the error in the approximation decays fast with the size of $M$.

It is one of the main contributions of \cite{KaLuPeGa19, Antonio} to show that Gibbs states of 1D Hamiltonians with sufficiently fast decaying interactions fulfill the approximate factorization property. Indeed, this idea has been used in \cite{Kuwahara} to give algorithms that provide efficiently Matrix Product Operator (MPO) descriptions of 1D Gibbs states.

We will precisely show (Theorem~\ref{Theo:leadingTerm}) that the boundary states associated to the thermofield double PEPS $|\rho_\beta^{1/2}\rangle$ approximately factorize. In order to finish the proof of our main theorem,  we will also need to extend the validity of the results in \cite{KaLuPeGa19} beyond the  cases considered there (injective and MPO-injective PEPS), so that it applies to $|\rho_\beta^{1/2}\rangle$. Indeed, it has been a technical challenge in the paper to deal with a PEPS which is neither injective nor MPO-injective, the classes for which essentially all the analytical results for PEPS have been proven \cite{Cirac20}. 

Let us finish this introduction by commenting that the results presented in this work can be seen as a clear illustration of the power of the bulk-boundary correspondence in PEPS, and in particular, the power of the ideas and techniques developed in \cite{KaLuPeGa19}. 

We are very confident that the result presented here can be extended, using similar techniques, to cover all possible 2D models that are renormalization fixed points, like string net models \cite{LevinWen}. The reason is that all those models have shown to be very naturally described and analyzed in the language of PEPS \cite{Cirac20}. We leave such extension for future work. 

This paper is structured as follow. 
In Section~\ref{sec:quantum-spin-systems}, we recall some elementary properties of quantum spin system, and explain the strategy we will use to estimate spectral gaps of 2D quantum Hamiltonians. Moreover, we will explain under which assumptions we can estimate the spectral gap of a model on a 2D torus with the spectral gap of the same model with open boundary conditions.
In Section~\ref{sec:PEPS}, we give a general introduction to the tensor networks and PEPS formalism, and explain the graphical notation we will use to represent tensors. We will then recall the results from \cite{KaLuPeGa19} that connect the quasi-factorization property with the spectral gap of a parent Hamiltonian of a PEPS, and present the necessary modifications of this results that we will need in this paper.
In Section~\ref{sec:peps-quantum-double} we introduce the Quantum Double Models, and present the PEPS representation of the thermofield double state $\ket*{\rho_\beta^{1/2}}$. From this construction, we will compute the corresponding boundary state, prove the approximate factorization condition, construct a parent Hamiltonian and estimate its spectral gap.
In Section~\ref{sec:davies-generators}, we recall the definition and elementary properties of Davies generators. We will then show that we can lower bound the spectral gap of an ergodic Davies generator by the spectral gap of a parent Hamiltonian for $\ket*{\rho_\beta^{1/2}}$, which will imply our main result.


\section{Quantum spin systems}\label{sec:quantum-spin-systems}

In this section, we are presenting some of the concepts and auxiliary results that we will use for the main result of the paper. Since we expect that they are useful in other contexts, we decided to present them in a more general setting. 

\subsection{Notation and elementary properties}\label{sec:notation}
We use Dirac's \emph{bra-ket} notation. Vectors in a Hilbert space $\mathcal{H}$ will be represented as ``\emph{kets}'' $\ket{\phi}$, and the scalar product between $\ket{\phi}$ and $\ket{\psi}$ is written as $\braket{\phi}{\psi}$ (which is anti-linear in the first argument). The linear functional $\ket{\psi} \mapsto \braket{\phi}{\psi}$ is then denoted as a ``\emph{bra}'' $\bra{\phi}$. Rank-one linear maps will be written as $\ketbra{\psi}{\phi}$. 

Let us consider an arbitrary set $\Lambda$, and associate to every site $x \in \Lambda$ a finite-dimensional Hilbert space $\mathcal{H}_{x} \equiv \mathbb{C}^{d}$ for a prefixed  $d \in \mathbb{N}$. As usual, for a finite subset $X \subset \Lambda$ we define the corresponding space of states $\mathcal{H}_{X} := \otimes_{x \in X}\mathcal{H}_{x}$ and the space of bounded linear operators (observables) $\mathcal{B}_{X} := \mathcal{B}(\mathcal{H}_{X})$ endowed with the usual operator norm. We denote by $\mathbbm{1}_{X} \in \mathcal{B}_{X}$ the identity. We identify for $X \subset X' \subset \Lambda$ observables $\mathcal{B}_{X} \hookrightarrow \mathcal{B}_{X'}$ via the isometric embedding $Q \mapsto Q \otimes \mathbbm{1}_{X' \setminus X}$. Given an operator $Q$, the minimal region $X\subset \Lambda$ such that $Q \in \mathcal{B}_X$ is called the \emph{support} of $Q$.
Let us observe that if $Q \in \mathcal{B}_{X}$ is a self-adjoint element, then $Q \otimes \mathbbm{1}_{X' \setminus X}$ and $Q$ have the same eigenvalues $\lambda$, and the corresponding eigenspaces are related via $V_{\lambda} \mapsto V_{\lambda} \otimes \mathcal{H}_{X' \setminus X}$. In particular, $\ker{(Q \otimes \mathbbm{1}_{X' \setminus X})} = \ker{(Q)} \otimes \mathcal{H}_{X' \otimes X}$. We define the \emph{spectral gap} of such $Q$, denoted $\gap(Q)$, as the difference of the two lowest (unequal) eigenvalues of $Q$. If $Q$ has only one eigenvalue, we set $\gap(Q)=0$.

A local Hamiltonian is defined in terms of a family of \emph{local interactions}, that is, a map $\Phi$ that associates to each finite subset $Z \subset \Lambda$ a self-adjoint observable $\Phi_{Z} = \Phi_{Z}^{\dagger} \in \mathcal{B}_{Z}$.
For each finite $X \subset \Lambda$, the corresponding \emph{Hamiltonian} is the self-adjoint operator $H_{X} \in \mathcal{B}_{X}$ given by
\[ H_{X} = \sum_{Z \subset X}{\Phi_{Z}}\,. \]
Next, let us introduce some further conditions on the type of interactions we are going to deal with. 

First, we are going to assume the lowest eigenvalue of $\Phi_{Z}$ is zero for each $Z$. This means that $\Phi_{Z} \geq 0$, and thus $H_{X} \geq 0$, for each finite subset $X \subset \Lambda$. In general, from an arbitrary local interaction $\Phi_{Z}$ we can always construct a new interaction satisfying this property by shifting each local term $\Phi_{Z}$ to $\Phi_{Z} - c_{Z} \mathbbm{1}_{Z}$ where $c_{Z}$ is the lowest eigenvalue of $\Phi_{Z}$. As a consequence, each local Hamiltonian $H_{X}$ is shifted to $H_{X} - (\sum_{Z \subset X} c_{Z}) \mathbbm{1}_{X}$, and its eigenvalues $\lambda$ are then shifted to $\lambda - (\sum_{Z \subset X} c_{Z}) $, although the corresponding eigenspaces and the spectral gap are preserved. It should be mentioned that this shifting procedure introduces an energy constraint that can significantly impact certain physical properties of the original system. However, since our sole focus is on studying the spectral gap properties in relation to the orthogonal projectors onto the ground spaces, this argument appears reasonable for reducing the overall problem to this particular setting.

On the other hand, we are also going to assume that the local interaction $\Phi$ is \emph{frustration-free}, namely that for every finite subset $X \subset \Lambda$ it holds that
\[ W_{X}:=\bigcap_{Z \subset X}{\ker(\Phi_{Z})} \neq \{0\} \,. \]
Let  $P_{X} \in \mathcal{B}_{X}$ be the orthogonal projector onto $W_{X}$. As a consequence of 
\begin{equation}\label{eq:frustration-free-proof}
\expval{H_X}{\phi} = \sum_{Z\subset X} \expval{\Phi_Z}{\phi} \ge 0\,,
\end{equation}
we immediately get that $\ker{(H_{X})}= W_{X}$, and that the frustration-free condition is equivalent to each $H_{X}$ having zero as the lowest eigenvalue. In this case, $W_{X}$ is the ground space of $H_{X}$, $P_{X}$ is the ground state projector, and the spectral gap of $H_{X}$, in case the latter is non-zero, can be described as the largest positive constant satisfying that for every $\ket{\phi} \in \mathcal{H}_{X}$
\begin{equation}\label{eq:frustration-free-gap-formula}
\expval{P_X^{\perp}}{\phi} \cdot \gap(H_{X}) \leq \expval{H_X}{\phi} \,, 
\end{equation}
where $P_X^{\perp} := \mathbbm{1}_{X} - P_{X}$. 

We conclude by remarking that for every $X \subset Y \subset \Lambda$ we have $W_{Y} \subset W_{X}$ where we are identifying $W_{X} \equiv \mathcal{H}_{Y \setminus X} \otimes W_{X}$, and therefore  $P_{Y} P_{X} = P_{Y}$.

\subsection{Spectral gap of local Hamiltonians}\label{sec:toolsSpectralGap}

We will now recall a recursive strategy to obtain lower bounds to the spectral gap of frustration-free Hamiltonians described in \cite{KaLu18}, which we will also slightly improve over the original formulation.
The main tool will be the following lemma, which is an improved version of \cite[Lemma 14]{KaLu18} in which the constant $(1-2c)$ has been improved to $(1-c)$. The argument here is different and inspired by \cite[Lemma 14.4]{KSVV02}.

\begin{Lemm}\label{Lemm:MartingaleProjector}
Let $U, V, W$ be subspaces of a finite-dimensional Hilbert space $\mathcal{H}$ with corresponding orthogonal projectors $\Pi_{U}, \Pi_{V}, \Pi_{W}$ and assume that $W \subset U \cap V$. Then,
\[ \Pi_{U}^{\perp} + \Pi_{V}^{\perp} \geq (1 - c) \, \Pi_{W}^{\perp} \quad \mbox{ where }  \quad c:=\| \Pi_{U}\Pi_{V} - \Pi_{W}\|\,. \]
Moreover, $c \in [0,1]$ always holds, and $c \in [0,1)$ if and only if $U \cap V =W$.
\end{Lemm}

\begin{proof}
Let us start by observing that, since $W\subset U$, it holds that $\Pi_W \Pi_U = \Pi_U \Pi_W = \Pi_W$, and similarly for $V$. Thus, the constant $c$ can be rewritten as
\begin{equation}\label{equa:MartingaleProjectorAux0}
\begin{split}
c & = \| (\Pi_{U} - \Pi_{W}) \, (\Pi_{V} - \Pi_{W}) \|\\ 
& = \| \Pi_{W}^{\perp}  \, \Pi_{U} \, \Pi_{V} \, \Pi_{W}^{\perp}\|\\
& = \sup{\{ \, |\bra{a} \Pi_{W}^{\perp}  \, \Pi_{U} \, \Pi_{V} \, \Pi_{W}^{\perp} \ket{b}| \, \colon \,  \| a\| \leq 1, \, \| b\| \leq 1  \rangle  \}}\\
& = \sup{\{ \, | \bra{a}  \Pi_{U} \, \Pi_{V} \ket{b}| \, \colon \, a, b \in W^{\perp}, \, \| a\| \leq 1, \, \| b\| \leq 1  \rangle  \}}\\ 
& = \sup{\{ \, \abs{\braket{a}{b}} \, \colon \, a \in U \cap W^{\perp}, \, b\in V \cap W^{\perp}, \, \| a\| \leq 1, \, \| b\| \leq 1  \rangle  \}}.
\end{split}
\end{equation}
From here it immediately follows that $c \in [0,1]$. Moreover, since $\mathcal{H}$ is finite-dimensional, the 
set over which the supremum is taken is compact, meaning that the supremum is always attained. Therefore, $c=1$ if and only if there exists $a\in U \cap W^{\perp}$ and $b\in V \cap W^{\perp}$ such that $\abs{\braket{a}{b}} = \| a\| \| b\|$. As the Cauchy–Schwarz inequality is only saturated by vectors which are proportional to each other, $c=1$ is equivalent to the fact that there exists an $a \in (U \cap V \cap W^{\perp}) \setminus \{ 0\}$ , or equivalently, that $W \subsetneq U \cap V$\,.

The first observation also implies that $\Pi_{U} = \Pi_{W} + \Pi_{W}^{\perp} \, \Pi_{U} \, \Pi_{W}^{\perp}$ and $\Pi_{V} = \Pi_{W} + \Pi_{W}^{\perp} \, \Pi_{V} \, \Pi_{W}^{\perp}$, and therefore
\[ \Pi_{U}^\perp + \Pi_{V}^\perp = 2\Pi_W^\perp - \Pi_{W}^{\perp} \, (\Pi_{U} + \Pi_{V}) \Pi_{W}^{\perp}.\]
This allows us to reformulate the original inequality we aim to prove as
\begin{equation}\label{equa:MartingaleProjectorAux1} 
\Pi^{\perp}_{U} + \Pi_{V}^{\perp} \geq (1-c) \, \Pi_{W}^{\perp} \,\, \Leftrightarrow \,\, (1+c) \Pi_{W}^{\perp} \geq \Pi_{W}^{\perp} \, (\Pi_{U} + \Pi_{V}) \Pi_{W}^{\perp}\,. 
\end{equation}
Let $\ket{x}$ be a norm-one eigenvector of $\Pi_{W}^{\perp} \, (\Pi_{U} + \Pi_{V}) \Pi_{W}^{\perp}$ with corresponding eigenvalue $\lambda >0$\,. Note that $\Pi_{W}^{\perp} \ket{x} = \ket{x}$ necessarily, since eigenvectors with different eigenvalues are orthogonal, and $W$ is contained in the kernel. Thus, to show that the right-hand side inequality of \eqref{equa:MartingaleProjectorAux1} holds, it is enough to check that $\lambda \leq 1+c$ necessarily. For that, let us write
\begin{align*}
\Pi_{U}\ket{x} & = \lambda_{U} \, \ket{x_{U}} \quad \mbox{ for some } \quad \ket{x_{U}} \in U \cap W^{\perp}\,, \,  \braket{x_{U}} = 1\,, \, \lambda_{U} \geq 0\,,\\
\Pi_{V}\ket{x} & = \lambda_{V} \, \ket{x_{V}} \quad \mbox{ for some } \quad  \ket{x_{V}} \in V \cap W^{\perp}\,, \,  \braket{x_{V}}= 1\,, \, \lambda_{V} \geq 0\,.
\end{align*}
On the one hand, we have
\begin{align*}
\lambda = \bra{x} \Pi_{U} + \Pi_{V} \ket{x} & = \bra{x} \Pi_{U} \ket{x} + \bra{x} \Pi_{V} \ket{x}\\ 
& \, = \bra{x} \Pi_{U}^{2} \ket{x} + \bra{x} \Pi_{V}^{2} \ket{x} \\
& \, = \, \lambda_{U}^{2} + \lambda_{V}^{2},
\end{align*}
and, on the other hand
\begin{align*} 
\lambda^{2} & = \bra{x} (\Pi_{U} + \Pi_{V}) \, (\Pi_{U} + \Pi_{V}) \ket{x}\\ 
& = \lambda_{U}^{2} + \lambda_{V}^{2} +  \bra{x}\Pi_{U} \Pi_{V}\ket{x} + \bra{x}\Pi_{V} \Pi_{U}\ket{x}\\
& =   \lambda_{U}^{2} + \lambda_{V}^{2} + 2 \, \lambda_{U} \, \lambda_{V} \, \operatorname{Re}\braket{x_{U}}{x_{V}} \,.
\end{align*}
Combining both equalities we get, denoting $c_{x}:=|\operatorname{Re}\braket{x_{U}}{x_{V}}|$,
\begin{align*}
(1+c_{x}) \lambda - \lambda^{2} & = c_{x} \left( \lambda_{U}^{2} + \lambda_{V}^{2} \right) - 2  \lambda_{U}  \lambda_{V} \operatorname{Re}\braket{x_{U}}{x_{V}}\\
& \geq c_{x}(\lambda_{U}^{2} + \lambda_{V}^{2} - 2 \lambda_{U} \lambda_{V}) = c_{x}(\lambda_{U} - \lambda_{V})^{2} \geq 0\,.
\end{align*}
Since $\lambda>0$, we conclude that
\[ \lambda \leq 1 + c_{x} \leq 1+c, \]
where the last inequality follows from \eqref{equa:MartingaleProjectorAux0}. This concludes the argument.
\end{proof}

This lemma has important implications for frustration-free local Hamiltonians, when applied to the ground state subspaces $W_X$ and their associated orthogonal projections $P_X$. The frustration-free condition yields that, for $X, Y \subset \Lambda$, the ground state subspaces satisfy $ W_{X \cup Y} \subset W_{X} \cap W_{Y}$. 
As a consequence of Lemma~\ref{Lemm:MartingaleProjector}, 
\[ \| P_{X \cup Y} - P_{X} P_{Y}\| = \| (P_{X \cup Y} - P_{X}) (P_{X \cup Y} - P_{Y})\| \in [0, 1]\,. \]
Moreover, $W_{X \cup Y} = W_{X} \cap W_{Y}\,$
if and only if $\| P_{X \cup Y} - P_{X} P_{Y}\| \in [0, 1)$.
 This happens whenever $\Lambda$ is a metric space, $\Phi$ has finite range $r>0$, namely that $\Phi_{X}=0$ if the diameter of $X$ is larger than $r$, and the distance $d(X \setminus Y, Y \setminus X)$ is greater than $r$, since in this case every subset $Z \subset X \cup Y$ with $\Phi_{Z} \neq 0$ is either contained in $X$ or $Y$, so that
 \[ W_{X \cup Y} = \bigcap_{Z \subset X \cup Y} \ker{\Phi_{Z}} = \bigcap_{Z \subset X} \ker{\Phi_{X}} \cap \bigcap_{Z \subset  Y} \ker{\Phi_{Z}} = W_{X} \cap W_{Y} \,. \]

\begin{Defi}[Spectral gap]\label{defi:spectralgap}
For each finite subset $Y \subset \Lambda$, let us denote by $\gap(H_{Y})$, or simply $\gap(Y)$, the \emph{spectral gap} of $H_{Y}$, namely the difference between the two lowest unequal eigenvalues of $H_{Y}$. If it has only one eigenvalue, then we define $\gap(Y) = 0$.
Given a family $\mathcal{F}$ of finite subsets of $\Lambda$, we say that the system of Hamiltonians $(H_{Y})_{Y \in \mathcal{F}}$ is \emph{gapped} whenever
\[ \gap(\mathcal{F}):=\inf{\{ \gap(Y) \colon Y \in \mathcal{F} \}} > 0\,. \]
Otherwise, it is said to be \emph{gapless}\,.
\end{Defi}

The following result allows to relate the gap of two families. It adapts a result from \cite[Section 4.2]{KaLu18}.

\begin{Theo}\label{Theo:RecursiveGapEstimate}
Let $\mathcal{F}$ and $\mathcal{F}'$ be two families of finite subsets of $\Lambda$. Suppose that there are $s \in \mathbb{N}$ and $\delta \in [0,1]$ satisfying the following property: for each $Y \in \mathcal{F}' \setminus \mathcal{F}$ there exist $(A_{i}, B_{i})_{i=1}^{s}$ pairs  of elements in $\mathcal{F}$ such that:
\begin{enumerate}
    \item[(i)] $Y = A_{i} \cup B_{i}$ for each $i=1, \ldots, s$,
    \item[(ii)] $(A_{i} \cap B_{i}) \cap (A_{j} \cap B_{j}) = \emptyset$ whenever $i \neq j$,
    \item[(iii)] $\|P_{A_{i}}P_{B_{i}} - P_{Y}\| \leq \delta$ for every $i=1, \ldots, s$.
\end{enumerate}
Then,
\[ \gap(\mathcal{F}') \, \geq \, \frac{1-\delta}{1+ \frac{1}{s}} \, \gap(\mathcal{F}) \, \underset{(\text{if }\delta < 1)}{\geq} \, \exp\left[ -\frac{\delta}{1- \delta} - \frac{1}{s}\right] \, \gap(\mathcal{F}) \,. \]
\end{Theo}

\begin{proof}[Proof of Theorem \ref{Theo:RecursiveGapEstimate}]
Let $Y \in \mathcal{F}' \setminus \mathcal{F}$ and let $(A_{i}, B_{i})_{i=1}^{s}$ the family of pairs satisfying $(i)$-$(iii)$ provided by the hypothesis. To prove the first inequality, we  can assume that $\delta < 1$ and $\gap(\mathcal{F})>0$, since otherwise the inequality is obvious. Applying Lemma \ref{Lemm:MartingaleProjector}, and using \eqref{eq:frustration-free-gap-formula}, we can estimate
\begin{align*}
\bra{x} P_{Y}^{\perp} \ket{x} \,  = \, \frac{1}{s} \sum_{i=1}^{s}{\bra{x} P_{Y}^{\perp} \ket{x}} \,\, & \leq \,\, \frac{1}{s} \sum_{i=1}^{s} \frac{1}{1-\delta} \left(\bra{x} P_{A_{i}}^{\perp} \ket{x} + \bra{x} P_{B_{i}}^{\perp} \ket{x} \right) \\
& \leq \,\, \frac{1}{1-\delta}  \, \frac{1}{s}  \sum_{i=1}^{s} \frac{1}{\gap(\mathcal{F})}  \left(\bra{x} H_{A_{i}} \ket{x} + \bra{x} H_{B_{i}} \ket{x} \right)\\
& \leq  \,\, \frac{1}{1-\delta}  \, \frac{1}{\gap(\mathcal{F})} \, \frac{1}{s} \sum_{i=1}^{s} \bra{x} H_{Y} + H_{A_{i} \cap B_{i}} \ket{x}\\
& \leq \,\, \frac{1}{1- \delta} \, \frac{1}{\gap(\mathcal{F})} \,  \, \bra{x} H_{Y} + \frac{1}{s} \sum_{i=1}^{s} H_{A_{i} \cap B_{i}} \ket{x}\\
& \leq \,\, \frac{1}{1- \delta} \, \frac{1}{\gap(\mathcal{F})} \,  \, \left( 1+ \frac{1}{s} \right) \, \bra{x} H_{Y}  \ket{x}\,.
\end{align*}
Notice that in the third line we have used that $H_{A_{i}} + H_{B_{i}} \leq H_{Y} + H_{A_{i} \cap B_{i}}$, which holds since all the local interactions are positive semidefinite. Therefore, again by~\eqref{eq:frustration-free-gap-formula}, it holds that
\[ \gap(Y) \geq \frac{1- \delta}{1+\frac{1}{s}} \, \gap(\mathcal{F})\quad \text{whenever} \quad Y \in \mathcal{F}' \setminus \mathcal{F}\,. \]
\noindent On the other hand, if $Y \in \mathcal{F}' \cap \mathcal{F}$, then  $\gap(Y) \geq \gap(\mathcal{F}) $ by definition. Hence, we conclude the that the first inequality holds. To obtain the second inequality, we simply use twice that $(1+x)^{-1} \geq e^{-x}$ for every $x \geq 0$.
\end{proof}




In the next section, we will apply Theorem \ref{Theo:RecursiveGapEstimate} in order to bound the spectral gap of a quantum spin Hamiltonian on an arbitrarily large torus (with periodic boundary conditions) in terms of  the spectral gap of the same model on a finite family of rectangles with open boundary conditions. This is reminiscent to the bounds on the spectral gap based on the local gap thresholds \cite{Knabe88, Gosset2016, Lemm2019, Lemm2020, Anshu2020}. The reason for which we are following a different approach is that the constants appearing in the spectral gap thresholds are highly dependent on the specific shape and range of the interactions of the Hamiltonian, and in our case we will have to consider $\beta$-dependent interaction length. Therefore it will be unfeasible to verify the local gap threshold conditions for our models.
The connection between the control of quantities of the type $\norm{P_{X\cup Y} - P_X P_Y}$ and spectral gap estimates originated in the seminal work on Finitely Correlated States~\cite{FCS}, later extended to more general spin models~\cite{Nachtergaele}. The main difference between that approach and the one of \cite{KaLu18}, which we are following here, is due to the way in which we are growing the lattice: while the methods of \cite{FCS, Nachtergaele} consider a single increasing sequence of regions, Theorem~\ref{Theo:RecursiveGapEstimate} permits more rich families of subsets, thus allowing us to keep the shape of their intersections more well behaved (i.e.\ they will always be rectangles or cylinders).

\subsection{Periodic boundary conditions on a torus}\label{subsec:torusSettingAndGap}

To describe the periodic boundary conditions case, we have to introduce further notation. For each natural $N$, let us denote by $\mathbb{S}_{N}$ the quotient $\mathbb{R}/\sim$ where we relate $x \sim x+N$ for every $x \in \mathbb{R}$. Note that we can identify $\mathbb{S}_{N} \equiv [0, N)$. 

We will take as $\Lambda_{N}$, the set where the \emph{spins} of the system are located, the set of midpoints of the edges $\EE_{N}$ of the square lattice on the torus $\mathbb{S}_{N} \times \mathbb{S}_{N}$ (as this is the setting in which the Quantum Double Models are defined). We will identify each point of $\Lambda_{N}$ with the corresponding edge from $\EE_{N}$, so that we will indistinctly use $\Lambda_{N}$ or $\EE_{N}$ (see Figure \ref{Fig:squareLatticeFirstRep}).

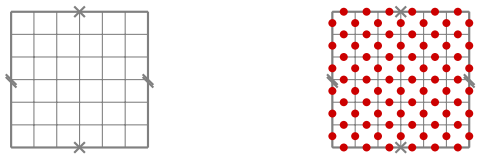
\begin{figure}[ht] 
\centering
  \begin{tikzpicture}[equation,scale=0.3]
    \draw[gray, thin] (0,0) grid (6,6);
    \draw[gray,thick, postaction=torus horizontal] (0,0) -- (6,0);
    \draw[gray,thick, postaction=torus horizontal] (0,6) -- (6,6);
    \draw[gray,thick, postaction=torus vertical] (0,0) -- (0,6);
    \draw[gray,thick, postaction=torus vertical] (6,0) -- (6,6);
    \end{tikzpicture}
    \hspace{2cm}
    \begin{tikzpicture}[equation,scale=0.3]
    \draw[gray, thin] (0,0) grid (6,6);
    \draw[gray,thick, postaction=torus horizontal] (0,0) -- (6,0);
    \draw[gray,thick, postaction=torus horizontal] (0,6) -- (6,6);
    \draw[gray,thick, postaction=torus vertical] (0,0) -- (0,6);
    \draw[gray,thick, postaction=torus vertical] (6,0) -- (6,6);
    \foreach \x in {0,...,5}
        \foreach \y in {0,...,6}
            \node[circle,fill=red!80!black, inner sep=0pt,minimum size=3pt] at (0.5+\x,\y) {};
             \foreach \x in {0,...,6}
        \foreach \y in {0,...,5}
            \node[circle,fill=red!80!black, inner sep=0pt,minimum size=3pt] at (\x,0.5+\y) {};
    \end{tikzpicture}\\[2mm]
\caption{The square lattice on the torus (left), and the quantum spin system with spins  on the midpoints of the edges (right). The marks on the borders of both squares represent the pairwise identification of edges, a standard depiction of the torus in topology. A similar notation will be used for the cylinder, where only one pair of edges is identified.}
\label{Fig:squareLatticeFirstRep}
\end{figure}

Let us recall the notion of a (closed) \emph{interval} in $\mathbb{S}_{N}$. Given $x,y \in \mathbb{S}_{N}$ we denote by $d_{+}(x,y)$ the unique $0 \leq c < N$ such that $x+c \sim y$.  Then, we define the \emph{interval} $[a,b]$ as the set $\{ x \in \mathbb{S}_{N} \colon d_{+}(a, x) \leq d_{+}(a,b) \}$.  We are only going to consider intervals with integer endpoints, that is, $a,b \in \mathbb{Z}_{N}$. 

A \emph{proper rectangle} $\RR$ in $\mathbb{R}^{2}$ or $\mathbb{S}_{N} \times \mathbb{S}_{N}$ is a Cartesian product of intervals $\RR=[a_{1}, b_{1}] \times [a_{2}, b_{2}]$ (with integer endpoints). Its number of plaquettes per row is then $d_{+}(a_{1},b_1)$ and per column is $d_{+}(a_{2},b_2)$. Shortly, we say that $\RR$ has \emph{dimensions} $d_{+}(a_{1},b_1)$ and $d_{+}(a_{2},b_2)$. A \emph{cylinder} is a Cartesian product of the form $\mathbb{S}_{N} \times [a,b]$ or $[a,b] \times \mathbb{S}_{N}$. We will refer simply as \emph{rectangles} to  proper rectangles, cylinders and the whole torus $\mathbb{S}_{N} \times \mathbb{S}_{N}$. In an abuse of notation, we will identify $\RR$ with $\RR \cap \Lambda$ and often write $\RR \subset \Lambda$ and $\mathcal{H}_{\RR}$ to denote the associated Hilbert space (see Figure \ref{Figure:rectanglesTorus}). 
\begin{figure}[ht]
\centering
\begin{tikzpicture}[equation,scale=0.3]
    \draw[gray, thin] (0,0) grid (6,6);
    \draw[gray,thick, postaction=torus horizontal] (0,0) -- (6,0);
    \draw[gray,thick, postaction=torus horizontal] (0,6) -- (6,6);
    \draw[gray,thick, postaction=torus vertical] (0,0) -- (0,6);
    \draw[gray,thick, postaction=torus vertical] (6,0) -- (6,6);
    \draw[black, ultra thick] (2,2) grid (5,4);
    \end{tikzpicture}
    \hspace{1cm} 
    \begin{tikzpicture}[equation,scale=0.3]
    \draw[gray, thin] (0,0) grid (6,6);
    \draw[gray,thick, postaction=torus horizontal] (0,0) -- (6,0);
    \draw[gray,thick, postaction=torus horizontal] (0,6) -- (6,6);
    \draw[gray,thick, postaction=torus vertical] (0,0) -- (0,6);
    \draw[gray,thick, postaction=torus vertical] (6,0) -- (6,6);
     \foreach \x in {2,...,5}
        \foreach \y in {2,...,3}
            \node[circle,fill=red!80!black, inner sep=0pt,minimum size=3pt] at (\x,\y+0.5) {};
    \foreach \x in {2,...,4}
        \foreach \y in {2,...,4}
        \node[circle,fill=red!80!black, inner sep=0pt,minimum size=3pt] at (\x+0.5,\y) {};
    \end{tikzpicture}
    \hspace{1.5cm} 
    \begin{tikzpicture}[equation,scale=0.3]
    \draw[gray, thin] (0,0) grid (6,6);
    \draw[gray,thick, postaction=torus horizontal] (0,0) -- (6,0);
    \draw[gray,thick, postaction=torus horizontal] (0,6) -- (6,6);
    \draw[gray,thick, postaction=torus vertical] (0,0) -- (0,6);
    \draw[gray,thick, postaction=torus vertical] (6,0) -- (6,6);
    \draw[black, ultra thick] (2,4) grid (5,6);
    \draw[black, ultra thick] (2,0) grid (5,1);
    \end{tikzpicture}
    \hspace{1cm} 
    \begin{tikzpicture}[equation,scale=0.3]
    \draw[gray, thin] (0,0) grid (6,6);
    \draw[gray,thick, postaction=torus horizontal] (0,0) -- (6,0);
    \draw[gray,thick, postaction=torus horizontal] (0,6) -- (6,6);
    \draw[gray,thick, postaction=torus vertical] (0,0) -- (0,6);
    \draw[gray,thick, postaction=torus vertical] (6,0) -- (6,6);
     \foreach \x in {2,...,5}
        \foreach \y in {4,...,5}
            \node[circle,fill=red!80!black, inner sep=0pt,minimum size=3pt] at (\x,\y+0.5) {};
    \foreach \x in {2,...,4}
        \foreach \y in {4,...,6}
        \node[circle,fill=red!80!black, inner sep=0pt,minimum size=3pt] at (\x+0.5,\y) {};
    \foreach \x in {2,...,5}
        \foreach \y in {0}
            \node[circle,fill=red!80!black, inner sep=0pt,minimum size=3pt] at (\x,\y+0.5) {};
    \foreach \x in {2,...,4}
        \foreach \y in {0,...,1}
        \node[circle,fill=red!80!black, inner sep=0pt,minimum size=3pt] at (\x+0.5,\y) {};    
    \end{tikzpicture}\\[10mm]

    \begin{tikzpicture}[equation,scale=0.3]
    \draw[gray, thin] (0,0) grid (6,6);
    \draw[gray,thick, postaction=torus horizontal] (0,0) -- (6,0);
    \draw[gray,thick, postaction=torus horizontal] (0,6) -- (6,6);
    \draw[gray,thick, postaction=torus vertical] (0,0) -- (0,6);
    \draw[gray,thick, postaction=torus vertical] (6,0) -- (6,6);
    \draw[black, ultra thick] (0,1) grid (6,3);
    \end{tikzpicture}
    \hspace{1cm}
    \begin{tikzpicture}[equation,scale=0.3]
    \draw[gray, thin] (0,0) grid (6,6);
    \draw[gray,thick, postaction=torus horizontal] (0,0) -- (6,0);
    \draw[gray,thick, postaction=torus horizontal] (0,6) -- (6,6);
    \draw[gray,thick, postaction=torus vertical] (0,0) -- (0,6);
    \draw[gray,thick, postaction=torus vertical] (6,0) -- (6,6);
    \foreach \x in {0,...,6}
        \foreach \y in {1,...,2}
            \node[circle,fill=red!80!black, inner sep=0pt,minimum size=3pt] at (\x,\y+0.5) {};
    \foreach \x in {0,...,5}
        \foreach \y in {1,...,3}
        \node[circle,fill=red!80!black, inner sep=0pt,minimum size=3pt] at (\x+0.5,\y) {}; 
    \end{tikzpicture}
    \hspace{1.5cm}
    \begin{tikzpicture}[equation,scale=0.3]
    \draw[gray] (0,0) grid (6,6);
    \draw[gray,thick, postaction=torus horizontal] (0,0) -- (6,0);
    \draw[gray,thick, postaction=torus horizontal] (0,6) -- (6,6);
    \draw[gray,thick, postaction=torus vertical] (0,0) -- (0,6);
    \draw[gray,thick, postaction=torus vertical] (6,0) -- (6,6);
    \draw[black, ultra thick] (3,0) grid (5,6);
    \end{tikzpicture}
    \hspace{1cm}
    \begin{tikzpicture}[equation,scale=0.3]
    \draw[gray] (0,0) grid (6,6);
    \draw[gray,thick, postaction=torus horizontal] (0,0) -- (6,0);
    \draw[gray,thick, postaction=torus horizontal] (0,6) -- (6,6);
    \draw[gray,thick, postaction=torus vertical] (0,0) -- (0,6);
    \draw[gray,thick, postaction=torus vertical] (6,0) -- (6,6);
    \foreach \x in {3,...,5}
        \foreach \y in {0,...,5}
            \node[circle,fill=red!80!black, inner sep=0pt,minimum size=3pt] at (\x,\y+0.5) {};
    \foreach \x in {3,...,4}
        \foreach \y in {0,...,6}
        \node[circle,fill=red!80!black, inner sep=0pt,minimum size=3pt] at (\x+0.5,\y) {}; 
    \end{tikzpicture}\\[2mm]
    
\caption{Examples of rectangles as a subset of the square lattice $\Lambda_{N} \equiv \EE_{N}$. Above, two examples of proper rectangles, whereas below, two examples of cylinders. In each case we present two pictures: on the right, we highlight the spins belonging to the region, while on the left we highlight the edges. We will use this latter representation in the forthcoming pictures.}
\label{Figure:rectanglesTorus}
\end{figure}
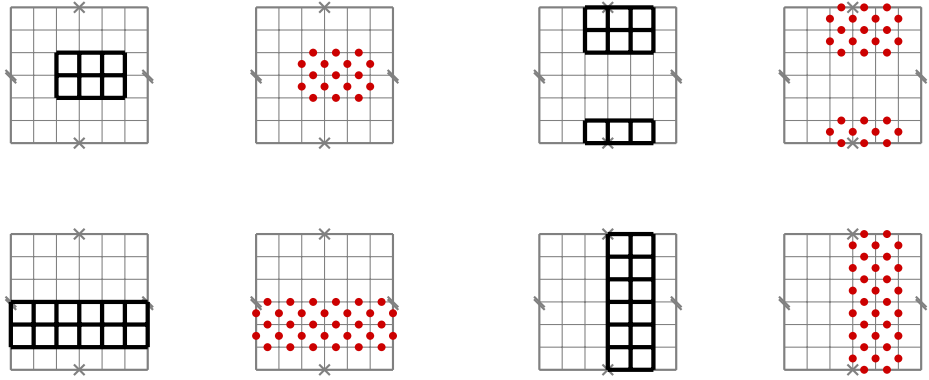


 We are going to define for every $N,r \in \mathbb{N}$ with $N \geq r \geq 2$ the following sets of rectangular regions in $\EE_{N}$:
 
\begin{itemize}
    \item[$\triangleright$] $\mathcal{F}_{N}$ is the set of all rectangular regions having at least two plaquettes per row and per column.
    \item[$\triangleright$] $\mathcal{F}_{N}^{torus}$ is the family consisting of only one element, the whole torus. 
    \item[$\triangleright$] $\mathcal{F}_{N}^{cylin}$ is the family of all cylinders having at least two plaquettes per row and per column.
    \item[$\triangleright$] $\mathcal{F}_{N, r}^{rect}$ is the family of all proper rectangles having at least two and at most $r$ plaquettes per row and per column. 
\end{itemize}
 
\subsection*{Martingale condition on local projectors}

 For each $X \in \mathcal{F}_{N}$ let $\Pi_{X}$ be an orthogonal projector onto a subspace of $\mathcal{H}_{X}$ such that $(\Pi_{X})_{X}$ satisfies the \emph{frustration-free condition}: $\Pi_{X} \Pi_{Y}=\Pi_{Y}\Pi_{X} =\Pi_{Y}$ for every pair of rectangular regions $X \subset Y$. Notice that the family of projectors $(P_{X})_{X}$ associated to a frustration-free Hamiltonian, as it was defined in Section \ref{sec:notation}, satisfies this condition. 

\begin{Defi}[\textbf{Martingale condition}] \label{Defi:MartingaleCondition}
We say that  $(\Pi_{X})_{X}$ as above satisfies the \emph{Martingale Condition} if there is a non-increasing function \mbox{$\delta: (0,\infty)  \longrightarrow [0, 1]$} such that $\lim_{\ell \rightarrow \infty} \delta(\ell) = 0$, called the \emph{decay function}, satisfying for each $N \geq 2$ the following properties:
\begin{enumerate}
\item[(i)] For every proper rectangle $\RR$ split along the rows (resp. columns) into three disjoint parts $\RR=ABC$ as in the next picture 
\begin{equation}\label{equa:decompPropRect}
\begin{tikzpicture}[equation,scale=0.35]
\draw[thin, gray] (0,0) grid (7,3);
\draw[black, thick] (2,0) rectangle (5,3);
\draw (1,1.5) node{\small $A$};
\draw (3.5,1.5) node{\small $B$};
\draw (6,1.5) node{\small $C$};
\begin{scope}[xshift=10cm]
\draw[gray] (0,0) grid (1.9,3);
\draw[gray,xshift=1cm] (2,0) grid (5,3);
\draw[gray, xshift=2cm] (5.1,0) grid (7,3);
\draw[black, thick, xshift=1cm] (2,0) rectangle (5,3);
\draw[|-|] (3,-0.5) -- (4.5, -0.5) 
        node[below, black] {\small $\ell$} -- (6,-0.5);
\draw (1,1.5) node{\small $A$};
\draw (4.5,1.5) node{\small $B$};
\draw (8,1.5) node{\small $C$};
\end{scope}
\end{tikzpicture}
\end{equation}
so that $\RR_{1}=AB$ and $\RR_{2}=BC$ are proper rectangles and $\RR_{1}\cap \RR_{1}=B$ is a rectangle containing at least $\ell$ plaquettes along the splitting direction, it holds that
\[ \| \Pi_{ABC} - \Pi_{AB}\Pi_{BC}\| \leq \delta(\ell)\,. \]
\item[(ii)] For every cylinder $\RR$ split along the wrapping direction into  four disjoint parts $\RR=ABCB'$ as in the next picture
\begin{equation}\label{equa:decompCylinder}
\begin{tikzpicture}[equation, scale=0.35]

\begin{scope}[yshift=1cm]
\node [cylinder, draw, minimum height=1.4cm, minimum width=1cm, fill=white, rotate=90, thick, black, fill=white] {};
\end{scope}

\begin{scope}[xshift=5cm]

\draw[thin, gray] (0,0) grid (8,3);
\draw[thick, gray, postaction=torus vertical] (0,0) -- (0,3); 
\draw[thick, gray, postaction=torus vertical] (8,0) -- (8,3); 

\draw[black, thick] (2,0) rectangle (4,3);
\draw[black, thick] (6,0) rectangle (8,3);

\draw (1,1.5) node{\small $A$};
\draw (3,1.5) node{\small $B$};
\draw (5,1.5) node{\small $C$};
\draw (7,1.5) node{\small $B'$};
\end{scope}

\begin{scope}[xshift=17cm]
\draw[step=1,white!50!black,thin] (0.1,0) grid (1.9,3);
\draw[step=1,white!50!black,thin,xshift=1cm] (2,0) grid (4,3);
\draw[step=1,white!50!black,thin,xshift=2cm] (4.1,0) grid (5.9,3);
\draw[step=1,white!50!black,thin,xshift=3cm] (6,0) grid (8,3);

\draw (1,1.5) node{\small $A$};
\draw[black,thick,xshift=1cm] (2,0) rectangle (4,3);
\draw[|-|] (3,-0.5) -- (4, -0.5) 
        node[below, black] {\small $\ell$} -- (5,-0.5);
\draw (4,1.5) node{\small $B$};
\draw (7,1.5) node{\small $C$};
\draw[black,thick,xshift=3cm] (6,0) rectangle (8,3);
\draw[|-|] (9,-0.5) -- (10, -0.5) 
        node[below, black] {\small $\ell$} -- (11,-0.5);
\draw (10,1.5) node{\small $B'$};
\end{scope}
\end{tikzpicture}
\end{equation}
so that $\RR_{1}=B'AB$ and $\RR_{2}=BCB'$ are proper rectangles, and $B$ and $B'$ are rectangles containing at least $\ell$ plaquettes along the wrapping direction, it holds that
\[ \| \Pi_{ABCB'} - \Pi_{B'AB}\Pi_{BCB'}\| \leq \delta(\ell)\,. \]

\item[(iii)] For every torus split along any of the two wrapping directions into four disjoint parts $\EE_{N}=ABCB'$ as in the next picture
\begin{equation}\label{equa:decompTorus}
\begin{tikzpicture}[equation, scale=0.35]
\begin{scope}[yshift=1cm,scale=2.2]
\path[rounded corners=24pt, blue] (-.9,0)--(0,.6)--(.9,0) (-.9,0)--(0,-.56)--    (.9,0);
\draw[rounded corners=28pt, black] (-1.1,.1)--(0,-.6)--(1.1,.1);
\draw[rounded corners=24pt, black] (-.9,0)--(0,.6)--(.9,0);
\draw[black] (0,0) ellipse (1.5 and 0.8);
\end{scope}   
  
\begin{scope}[xshift=6cm]
\draw[gray,thin] (0,0) grid (8,3);
\draw[thick, gray, postaction=torus horizontal] (0,0) -- (8,0); 
\draw[thick, gray, postaction=torus horizontal] (0,3) -- (8,3);
\draw[thick, gray, postaction=torus vertical] (0,0) -- (0,3); 
\draw[thick, gray, postaction=torus vertical] (8,0) -- (8,3); 
\draw (1,1.5) node{\small $A$};
\draw[black,thick] (2,0) rectangle (4,3);
\draw (3,1.5) node{\small $B$};
\draw (5,1.5) node{\small $C$};
\draw[black,thick] (6,0) rectangle (8,3);
\draw (7,1.5) node{\small $B'$};

\end{scope}
\hspace{1cm}  
\begin{scope}[xshift=14cm]

\draw[step=1,white!50!black,thin] (0.1,0) grid (1.9,3);
\draw[thick, gray, postaction=torus horizontal] (0.1,0) -- (1.9,0); 
\draw[thick, gray, postaction=torus horizontal] (0.1,3) -- (1.9,3); 

\draw[step=1,white!50!black,thin,xshift=1cm] (2,0) grid (4,3);
\draw[thick, gray, postaction=torus horizontal,xshift=1cm] (2,0) -- (4,0); 
\draw[thick, gray, postaction=torus horizontal,xshift=1cm] (2,3) -- (4,3); 

\draw[step=1,white!50!black,thin,xshift=2cm] (4.1,0) grid (5.9,3);
\draw[thick, gray, postaction=torus horizontal,xshift=2cm] (4.1,0) -- (5.9,0); 
\draw[thick, gray, postaction=torus horizontal,xshift=2cm] (4.1,3) -- (5.9,3); 
\draw[step=1,white!50!black,thin,xshift=3cm] (6,0) grid (8,3);
\draw[thick, gray, postaction=torus horizontal,xshift=3cm] (6,0) -- (8,0); 
\draw[thick, gray, postaction=torus horizontal,xshift=3cm] (6,3) -- (8,3); 
\draw (1,1.5) node{\small $A$};
\draw[black,thick,xshift=1cm] (2,0) rectangle (4,3);
\draw[|-|] (3,-0.5) -- (4, -0.5)   node[below, black] {\small $\ell$} -- (5,-0.5);
\draw (4,1.5) node{\small $B$};
\draw (7,1.5) node{\small $C$};
\draw[black,thick,xshift=3cm] (6,0) rectangle (8,3);
\draw[|-|] (9,-0.5) -- (10, -0.5) node[below, black] {\small $\ell$} -- (11,-0.5);
\draw (10,1.5) node{\small $B'$};
\end{scope}
\end{tikzpicture}
\hspace{1cm} 
\end{equation}

so that $\RR_{1}=B'AB$ and $\RR_{2}=BCB'$ are cylinders whose intersection consists of two cylinders $B$ and $B'$ containing at least $\ell$ plaquettes along the splitting direction, it holds that
\[ \| \Pi_{ABCB'} - \Pi_{B'AB}\Pi_{BCB'}\| \leq \delta(\ell)\,. \]

\end{enumerate}
\end{Defi}

\subsection*{Estimating the gap from below}

Let us fix a local interaction $\Phi$ on the torus $\EE_{N}$ defining a frustration-free Hamiltonian.  For each rectangular region $X \subset \EE_{N}$ let us denote by $P_{X}$ the orthogonal projector onto the ground space of $H_{X}$. Let us moreover assume that the family of projectors $(P_{X})_{X}$ satisfies the Martingale Condition for a decay function $\delta(\ell)$ as in Definition \ref{Defi:MartingaleCondition}.

\begin{Theo}\label{Theo:fromPeriodictoOpenBoundary}
If $\delta(\lfloor N/2-1 \rfloor)<1/2$, then
\[ \gap(\mathcal{F}_{N}^{torus}) \, \geq \,  \frac{1}{4} \gap(\mathcal{F}_{N}^{cylin}) \, \geq \, \frac{1}{16} \gap(\mathcal{F}_{N,N}^{rect})\,.\]
\end{Theo}

\begin{proof}
 Let us first compare $\gap(\mathcal{F}_{N}^{torus})$ and $\gap(\mathcal{F}_{N}^{cylin})$. We consider a decomposition of the torus $\Lambda_{N}$ into four regions $A,B,C,B'$ as in \eqref{equa:torusDecomposition}, so that $\mathcal{C}_{1}:=B'AB$ and $\mathcal{C}_{2}:=BCB'$ are cylinders having dimensions $N-1$ and $N$ belonging to $\mathcal{F}_{N}^{cylin}$ and whose intersection consists of two cylinders $B$ and $B'$, each having $N$ plaquettes per column and at least $\lfloor N/2-1\rfloor$ plaquettes per row.
 
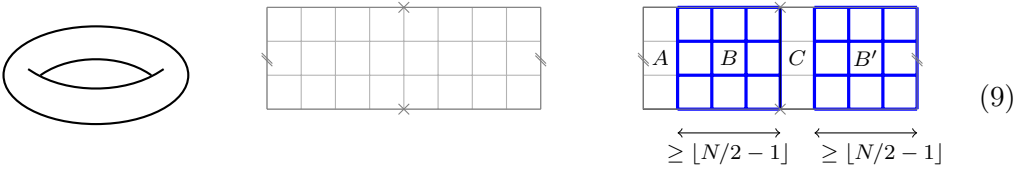
\begin{equation}\label{equa:torusDecomposition}
\begin{tikzpicture}[equation, scale=0.9]

\begin{scope}[scale=0.9, thick]

\path[rounded corners=24pt, blue] (-.9,0)--(0,.6)--(.9,0) (-.9,0)--(0,-.56)--    (.9,0);
\draw[rounded corners=28pt, black] (-1.1,.1)--(0,-.6)--(1.1,.1);
\draw[rounded corners=24pt, black] (-.9,0)--(0,.6)--(.9,0);
\draw[black] (0,0) ellipse (1.5 and 0.8);
\end{scope}   
    
   
\begin{scope}[xshift=2.5cm, scale=0.5, yshift=-1cm]

\draw[step=1,white!70!black] (0,0) grid (8,3);

\draw[gray, postaction=torus horizontal] (0,0) -- (8,0); 
\draw[gray, postaction=torus horizontal] (0,3) -- (8, 3);
\draw[gray, postaction=torus vertical] (0,0) -- (0, 3);
\draw[gray, postaction=torus vertical] (8,0) -- (8, 3);

\end{scope}   
   
    
\begin{scope}[xshift=8cm, scale=0.5, yshift=-1cm]

\draw[step=1,white!70!black] (0,0) grid (8,3);

\draw[black] (0,0) rectangle (1,3);
\draw (0.5,1.5) node{\scriptsize  $A$};
\draw[very thick, blue] (1,0) grid (4,3);
\draw (2.5,1.5) node{\scriptsize  $B$};
\draw[black] (4,0) rectangle (5,3);
\draw (4.5,1.5) node{\scriptsize  $C$};
\draw[very thick, blue] (5,0) grid (8,3);
\draw (6.5,1.5) node{\scriptsize  $B'$};

\draw[gray, postaction=torus horizontal] (0,0) -- (8,0); 
\draw[gray, postaction=torus horizontal] (0,3) -- (8, 3);
\draw[gray, postaction=torus vertical] (0,0) -- (0, 3);
\draw[gray, postaction=torus vertical] (8,0) -- (8, 3);


\draw[<->] (1,-0.7) -- (4,-0.7);
\draw (2.5,-1.3) node{\scriptsize $\geq \lfloor N / 2-1 \rfloor$};
\draw[<->] (5,-0.7) -- (8,-0.7);
\draw (7,-1.3) node{\scriptsize $\geq \lfloor N / 2-1 \rfloor$};

\end{scope}

\end{tikzpicture} 
\end{equation}
 Applying the Martingale Condition from Definition \ref{Defi:MartingaleCondition}.(iii) we deduce that
\begin{equation}\label{equa:comparingGapTorusCylin} 
\| P_{\mathcal{C}_{1}} P_{\mathcal{C}_{2}} - P_{\Lambda_{N}} \| \, \leq \, \delta(\lfloor N/2-1 \rfloor)  < \frac{1}{2}\,, 
\end{equation}
and so, by Theorem~\ref{Theo:RecursiveGapEstimate} with $s=1$ and $\delta=1/2$, we can estimate
\[ \gap(\mathcal{F}_{N}^{torus}) \, \geq \, \frac{1}{4} \, \gap(\mathcal{F}_{N}^{cylin}) \,. \]
Next, we compare $\mathcal{F}_{N}^{cylin}$ and $\mathcal{F}_{N,N}^{rect}$. Given a cylinder $\mathcal{C} \in \mathcal{F}_{N}^{cylin}$, we can split it along the wrapping direction into four regions $A,B,C,B'$ as in \eqref{equa:cylinderDecomposition}, so that $\mathcal{R}_{1}:=B'AB$ and $\mathcal{R}_{2}:=BCB'$ are proper rectangles having $N$ plaquettes along the splitting direction, belonging to $\mathcal{F}_{N}^{rect}$, and whose intersection consists of two proper rectangles $B$ and $B'$ having at least $\lfloor N/2-1 \rfloor$ plaquettes along the splitting direction.
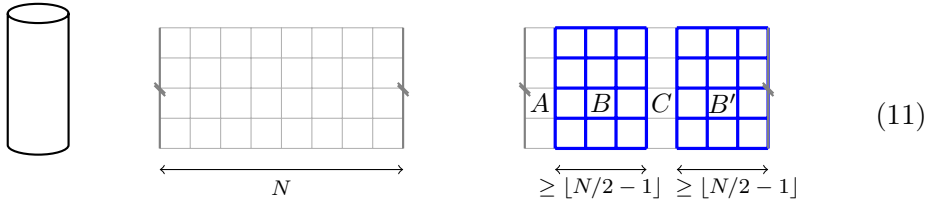
\begin{equation}\label{equa:cylinderDecomposition}
\begin{tikzpicture}[equation, scale=0.8]

\begin{scope}[xshift=-2cm, yshift=1cm]
\node [cylinder, draw, minimum height=2cm, minimum width=0.8cm, fill=white, rotate=90, thick, black, fill=white] {};
\end{scope}

\begin{scope}[ scale=0.5]


\draw[white!70!black] (0,0) grid (8,4);


\draw[thick, gray, postaction=torus vertical] (0,0) -- (0,4); 
\draw[thick, gray, postaction=torus vertical] (8,0) -- (8,4); 


\draw[<->] (0,-0.7) -- (8,-0.7);
\draw (4,-1.3) node{\scriptsize $N$};

\end{scope}


\begin{scope}[xshift=6cm, scale=0.5]


\draw[white!70!black] (0,0) grid (8,4);

\draw[blue, very thick] (1,0) grid (4,4);
\draw[blue, very thick] (5,0) grid (8,4);

\draw (0.5,1.5) node{\small $A$};
\draw (2.5,1.5) node{\small $B$};
\draw (4.5,1.5) node{\small $C$};
\draw (6.5,1.5) node{\small $B'$};


\draw[thick, gray, postaction=torus vertical] (0,0) -- (0,4); 
\draw[thick, gray, postaction=torus vertical] (8,0) -- (8,4); 


\draw[<->] (1,-0.7) -- (4,-0.7);
\draw (2.5,-1.3) node{\scriptsize $\geq \lfloor N / 2-1 \rfloor$};
\draw[<->] (5,-0.7) -- (8,-0.7);
\draw (7,-1.3) node{\scriptsize $\geq \lfloor N / 2-1 \rfloor$};

\end{scope}

\end{tikzpicture}
\end{equation}
Thus, applying the Martingale Condition from Definition \ref{Defi:MartingaleCondition}.(ii) we deduce that
\[ \| P_{\mathcal{R}_{1}} P_{\mathcal{R}_{2}} - P_{\mathcal{C}} \| \, \leq \, \delta(\lfloor N/2-1 \rfloor) < \frac{1}{2}\,, \]
and so, by Theorem~\ref{Theo:RecursiveGapEstimate} with $s=1$ and $\delta=1/2$, we can estimate
\begin{equation}\label{equa:comparingGapRectCylin} 
\gap(\mathcal{F}_{N}^{cylin}) \, \geq \, \frac{1}{4} \, \gap(\mathcal{F}_{N,N}^{rect})\,. 
\end{equation}
Combining \eqref{equa:comparingGapTorusCylin} and \eqref{equa:comparingGapRectCylin} we conclude the result.
\end{proof}

\begin{Theo} \label{Theo:spectralGapOpenBoundary}
For fixed integers $N \geq r \geq 16$, let us denote  $\delta_{k}:= \delta(\lfloor \tfrac{r}{4}(\sqrt{9/8}\,)^{k}\rfloor\,)$ and $s_k:=\lfloor (\sqrt{4/3}\,)^{k} \rfloor$ for each integer $k \geq 0$. Then, we can estimate from below
\[ \gap(\mathcal{F}_{N,N}^{rect}) \, \geq \, \left(\, \prod_{k=0}^{\infty} \frac{1-\delta_{k}}{1+\frac{1}{s_{k}}}\right) \, \gap(\mathcal{F}^{rect}_{N,r})\,. \]
\end{Theo}

Notice that the infinite product that appears in the previous expression is convergent to a positive value whenever the decay function $\delta(\ell)$ decays polynomially fast, namely $\delta(\ell) = O(\ell^{-\alpha})$ for some $\alpha >0$. 

\begin{proof}
Fix $N,r$ as above, for each $k \geq 0$ let $\mathcal{G}_{k}$ be the family of all proper rectangles in $\mathcal{F}_{N,N}^{rect}$ of dimensions $a$ and $b$ satisfying 
\[ a,b \geq 2 \quad  \mbox{ and } \quad a \cdot b \leq r (3/2)^{k}\,.\]
Observe that $\mathcal{G}_{k} \subset \mathcal{G}_{k+1}$ for every $k \geq 0$ and that $\mathcal{G}_{k} = \mathcal{F}_{N,N}^{rect}$ for $k$ large enough. Next, let $Y \in \mathcal{G}_{k+1} \setminus \mathcal{G}_{k}$ of dimensions $a$ and $b$. Then
\[ r (3/2)^{k} < a \cdot b \leq r (3/2)^{k+1}\,. \]
We can assume without loss of generality that $a \leq b$, so using the previous inequalities we can estimate from below (recall that $r \geq 16$)
\begin{equation}\label{equa:auxdimensionsGap1} 
4 \leq \sqrt{r} (\sqrt{3/2}\,)^{k} \leq b\,.  
\end{equation}
Let us define
\[ s_{k}:= \lfloor (\sqrt{4/3}\,)^{k} \rfloor \quad \mbox{and} \quad  \ell_{k} := \left\lfloor  b/(4 s_{k}) \right\rfloor \,. \]
Observe that by \eqref{equa:auxdimensionsGap1}
\begin{equation}\label{equa:auxdimensionsGap2}
\ell_{k} \geq \left\lfloor\frac{b}{4s_{k}}\right\rfloor \geq \left\lfloor \frac{\sqrt{r}}{4} (\sqrt{3/2}\,)^{k} (\sqrt{3/4}\,)^{k} \right\rfloor =  \left\lfloor\frac{\sqrt{r}}{4} (\sqrt{9/8}\,)^{k}\right\rfloor \geq 1\,, 
\end{equation}
and for every $j = 0, \ldots, s_{k}-1$ 
\begin{equation}\label{equa:auxdimensionsGap3}
(2j+1) \ell_{k} \leq 2 s_{k} \ell_{k} -\ell_{k} \leq \frac{b}{4}-1\,.  
\end{equation}
Next, let us identify our rectangle $Y$ with $[0,b] \times [0,a]$, and consider the subrectangles (see Figure \ref{Fig:rectangularSplit})
\begin{align*}
A_{j} := \left[0, \lceil b/3 \rceil + (2j+1) \ell_k \, \right] \times [0,a] \quad ,\quad B_{j}  := \left[\lceil b/3 \rceil + 2j \ell_k, b \, \right] \times [0,a] \,.
\end{align*}
Note that, applying \eqref{equa:auxdimensionsGap3}, we deduce that $A_{j}$ is contained in a rectangle of dimensions $a$ and $\lceil b/3\rceil +b/4-1 \leq 2b/3$. Similarly, $B_{j}$ is contained in a rectangle of dimensions $a$ and $b-\lceil b/3 \rceil \leq 2b/3$. Since $a \geq 2$, $2b/3 \geq 2$, and $a \cdot (2b/3) \leq r(3/2)^{k}$, we conclude that $A_{j}, B_{j} \in \mathcal{G}_{k}$ by definition. Moreover, the intersections
\[ 
A_{j} \cap B_{j} = \left[ \, \lceil b/3 \rceil + 2j \ell_k, \lceil b/3 \rceil + (2j+1) \ell_k \, \right] \times [0,a] \,.
\]
are disjoint, i.e. $A_{j} \cap B_{j} \cap A_{j'} \cap B_{j'} = \emptyset$ whenever $j \neq j'$,  and have $\ell_{k}$ plaquettes per row and $a$ plaquettes per column. Using the Martingale Condition from Definition \ref{Defi:MartingaleCondition}.(i) and \eqref{equa:auxdimensionsGap2}
\[ \| P_{A_{j}} P_{B_{j}} - P_{Y} \| \leq \delta(\ell_{k})  \leq \delta(\lfloor\tfrac{\sqrt{r}}{4}(\sqrt{9/8})^{k}\rfloor)\,.\]

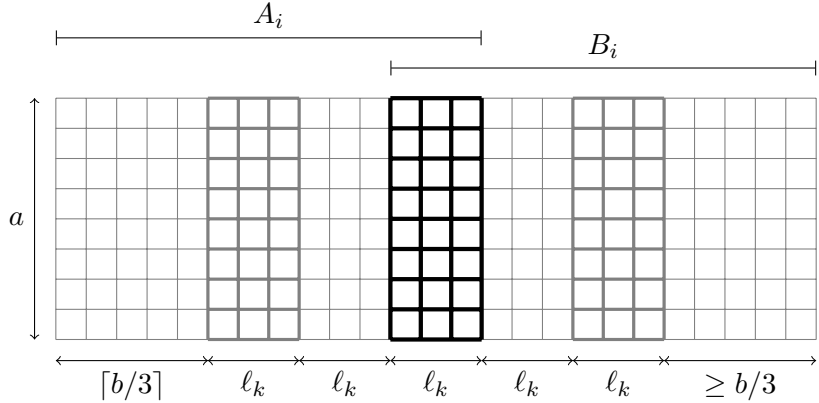
\begin{figure}[ht]
\begin{center}
\begin{tikzpicture}[scale=0.4, decoration={markings, mark=at position 0.55 with {\arrow[very thick, scale=1.5]{>}}}]
    \draw[gray, thin] (0,0) grid (25,8);
    \draw[black, very thick, gray] (5,0) grid (8,8);
    \draw[black, ultra thick] (11,0) grid (14,8);
    \draw[black, very thick, gray] (17,0) grid (20,8);


    \draw[<->] (0,-0.7) -- (5, -0.7);
    \draw[<->] (5,-0.7) -- (8, -0.7);
    \draw[<->] (8,-0.7) -- (11, -0.7);
    \draw[<->] (11,-0.7) -- (14, -0.7);
    \draw[<->] (14,-0.7) -- (17, -0.7);
    \draw[<->] (17,-0.7) -- (20, -0.7);
    \draw[<->] (20,-0.7) -- (25, -0.7);

    \draw (2.5,-0.7) node[below] {$\lceil b/3 \rceil$};
    \draw (6.5,-0.7) node[below] {$\ell_{k}$};
    \draw (9.5,-0.7) node[below] {$\ell_{k}$};
    \draw (12.5,-0.7) node[below] {$\ell_{k}$};
    \draw (15.5,-0.7) node[below] {$\ell_{k}$};
    \draw (18.5,-0.7) node[below] {$\ell_{k}$};
    \draw (22.5,-0.7) node[below] {$\geq b/3$};
    
    \draw[|-|, color = black] (0,10) -- (14, 10);
    \draw (7,10) node[above] {$A_{i}$};
    \draw[|-|, color = black] (11,9) -- (25, 9);
    \draw (18,9) node[above] {$B_{i}$};


\draw[<->] (-0.7,0) -- (-0.7, 8);
\draw (-0.7,4) node[left] {$a$};
    
    \end{tikzpicture}
    \end{center}
    \caption{Split of the rectangle $Y$ in the proof of Theorem \ref{Theo:spectralGapOpenBoundary}. The thick subrectangles correspond to the intersections $A_{i} \cap B_{i}$, being disjoint and having width equal to $\ell_{k}$.}
\label{Fig:rectangularSplit}
\end{figure}

\noindent Applying now Theorem~\ref{Theo:RecursiveGapEstimate} we deduce that
\[ \gap(\mathcal{G}_{k+1}) \geq \frac{1-\delta_{k}}{1+\frac{1}{s_{k}}} \gap(\mathcal{G}_{k}) \geq \ldots \geq \, \left(\prod_{j=0}^{k} \frac{1-\delta_{j}}{1+\frac{1}{s_{j}}}\right) \,\, \gap(\mathcal{G}_{0})\,. \]
Finally, noticing that $\mathcal{G}_{0} \subset \mathcal{F}_{r,N}^{rect}$ we conclude the result.
\end{proof}


\section{Projected Entangled-Pair States (PEPS)}
\label{sec:PEPS}

In Section~\ref{sec:quantum-spin-systems}, we introduced some general results that allow us to estimate the spectral gap of a quantum spin Hamiltonian. In particular, in order to apply Theorem~\ref{Theo:RecursiveGapEstimate}, we need to verify condition $(iii)$ for a specific family of local ground state projections. Projected Entangled-Pair States (PEPS for short) constitute a class of quantum spin models for which there exists tools that allow to control the bound in condition $(iii)$.
In the current section, we will briefly recall their definition, and how to evaluate condition $(iii)$ of Theorem~\ref{Theo:RecursiveGapEstimate} in the case of a specific kind of local Hamiltonian, known as a \emph{PEPS parent Hamiltonian}.

\subsection{Tensor notation}\label{sec:tensornotation}

Let us denote $[d]:=\{ 0,1,\ldots, d-1\}$ for each $d \in \mathbb{N}$. Recall that a \emph{tensor} with $n$ indices is simply an element $T \in \mathbb{C}^{[d_{1}] \times \ldots \times [d_{n}]} $ where each $d_{j} \in \mathbb{N}$ is called the \emph{dimension} of the $j$-th index. We will employ the usual notation
\[ T_\alpha =  T_{\alpha_{1} \ldots \alpha_{n}} \quad , \quad \alpha = (\alpha_{1}, \ldots, \alpha_{n}) \in [d_{1}] \times \ldots \times [d_{n}]\,. \]
By definition, a tensor with zero indices will be a scalar $T \in  \mathbb{C}$. We can represent tensors in the form of a ball with a leg for each index:

\begin{equation*}
    \begin{tikzpicture}[equation,very thick,decoration={
    markings,
    mark=at position 0.5 with {\arrow{>}}}
    ]
    
    \node (B1) at (30:1) {};
    \node  at (30:1.1) {\small $\alpha_{1}$};
    \node (B2) at (120:1) {};
    \node  at (120:1.1) {\small $\alpha_{2}$};
    \node (B3) at (185:1) {};
    \node  at (185:1.1) {\small $\alpha_{3}$};
    \node (B4) at (260:1) {};
    \node  at (260:1.1) {\small $\alpha_{4}$};
    \node (B5) at (-45:1) {};
    \node  at (-45:1.1) {\small $\alpha_{5}$};

    \node[fill=black, circle, scale=0.5] (A) at (0,0) {};
    \draw  (B1) --  (A);
    \draw (A) -- (B2);
    \draw (B3) -- (A);
    \draw (A) -- (B4);
    \draw (A) -- (B5);
    \end{tikzpicture}
     \quad \quad T = (T_{\alpha_{1} \ldots \alpha_{n}})
\end{equation*}

Let us describe the basic operations that we will perform with tensors. Given a tensor with (at least) two indices, say $\alpha_{1}$ and $\alpha_{2}$ with dimensions $d_{1}$ and $d_{2}$ respectively, we can combine them into one index $\gamma$ of dimension $d_{1}\cdot d_{2}$, so that the resulting tensor has one index less. This process can be iterated to combine several indices into one index. Graphically, we have the following example:
\begin{equation*}
    \begin{tikzpicture}[equation,very thick,decoration={
    markings,
    mark=at position 0.5 with {\arrow{>}}}
    ]
    \node (B1) at (30:1) {};
    \node  at (30:1.1) {\small $\alpha_{1}$};
    \node (B2) at (120:1) {};
    \node  at (120:1.1) {\small $\alpha_{2}$};
    \node (B3) at (185:1) {};
    \node  at (185:1.1) {\small $\alpha_{3}$};
    \node (B4) at (260:1) {};
    \node  at (260:1.1) {\small $\alpha_{4}$};
    \node (B5) at (-45:1) {};
    \node  at (-45:1.1) {\small $\alpha_{5}$};
    \node  at (-10:2.1) {\small $(T_{\alpha_{1} \alpha_{2} \alpha_{3}\alpha_{4}\alpha_{5}})$};    
        
    \node[fill=black, circle, scale=0.5] (A) at (0,0) {};
    \draw[blue]  (B1) --  (A);
    \draw[blue] (A) -- (B2);
    \draw (B3) -- (A);
    \draw (A) -- (B4);
    \draw (A) -- (B5);
    \end{tikzpicture}
     \quad  \equiv \quad
     \begin{tikzpicture}[equation,very thick,decoration={
    markings,
    mark=at position 0.5 with {\arrow{>}}}
    ]
    \node (B1) at (75:1) {};
    \node  at (75:1.1) {\small $\gamma$};
    \node (B3) at (185:1) {};
    \node  at (185:1.1) {\small $\alpha_{3}$};
    \node (B4) at (260:1) {};
    \node  at (260:1.1) {\small $\alpha_{4}$};
    \node (B5) at (-45:1) {};
    \node  at (-45:1.1) {\small $\alpha_{5}$};
   \node  at (0:2) {\small $(T_{\gamma \alpha_{3} \alpha_{4} \alpha_{5}})$};
   
    \node[fill=black, circle, scale=0.5] (A) at (0,0) {};
    \draw[very thick, blue]  (B1) --  (A);
    \draw (B3) -- (A);
    \draw (A) -- (B4);
    \draw (A) -- (B5);
    \end{tikzpicture}
\end{equation*}
Given two tensors $T=(T_{\alpha})$ and $S=(S_{\beta})$ with $m$ and $n$ indices, respectively, we define its \emph{tensor product} as the unique tensor with $n+m$ indices given by $T \otimes S :=(T_{\alpha} \cdot S_{\beta})$. For instance:
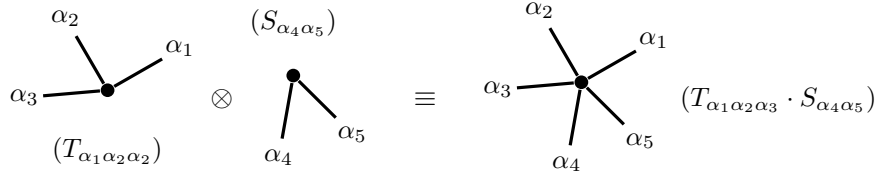
\begin{equation*}
    \begin{tikzpicture}[equation,very thick,decoration={
    markings,
    mark=at position 0.5 with {\arrow{>}}}
    ]
    \node (B1) at (30:1) {};
    \node  at (30:1.1) {\small $\alpha_{1}$};
    \node (B2) at (120:1) {};
    \node  at (120:1.1) {\small $\alpha_{2}$};
    \node (B3) at (185:1) {};
    \node  at (185:1.1) {\small $\alpha_{3}$};
      \node  at (270:0.8) {\small $(T_{\alpha_{1} \alpha_{2} \alpha_{2}})$};  
    \node[fill=black, circle, scale=0.5] (A) at (0,0) {};
    \draw  (B1) --  (A);
    \draw (A) -- (B2);
    \draw (B3) -- (A);
    \end{tikzpicture}
      \otimes 
     \begin{tikzpicture}[equation,very thick,decoration={
    markings,
    mark=at position 0.5 with {\arrow{>}}}
    ]
    \node (B4) at (260:1) {};
    \node  at (260:1.1) {\small $\alpha_{4}$};
    \node (B5) at (-45:1) {};
    \node  at (-45:1.1) {\small $\alpha_{5}$};
    \node  at (90:0.7) {\small $(S_{\alpha_{4} \alpha_{5}})$};
    
    \node[fill=black, circle, scale=0.5] (A) at (0,0) {};
    \draw (A) -- (B4);
    \draw (A) -- (B5);
    \end{tikzpicture}
    \quad \equiv  \quad
    \begin{tikzpicture}[equation,very thick,decoration={
    markings,
    mark=at position 0.5 with {\arrow{>}}}
    ]
    \node (B1) at (30:1) {};
    \node  at (30:1.1) {\small $\alpha_{1}$};
    \node (B2) at (120:1) {};
    \node  at (120:1.1) {\small $\alpha_{2}$};
    \node (B3) at (185:1) {};
    \node  at (185:1.1) {\small $\alpha_{3}$};
    \node (B4) at (260:1) {};
    \node  at (260:1.1) {\small $\alpha_{4}$};
    \node (B5) at (-45:1) {};
    \node  at (-45:1.1) {\small $\alpha_{5}$};
    \node  at (-5:2.6) {\small $(T_{\alpha_{1}\alpha_{2}\alpha_{3}} \cdot S_{\alpha_{4} \alpha_{5}})$};    
        
    \node[fill=black, circle, scale=0.5] (A) at (0,0) {};
    \draw (B1) --  (A);
    \draw (A) -- (B2);
    \draw (B3) -- (A);
    \draw (A) -- (B4);
    \draw (A) -- (B5);
    \end{tikzpicture}
\end{equation*}

If a tensor $(T_{\alpha})$ has two indices with the same dimension, say $\alpha_{1}$ and $\alpha_{2}$  then we can \emph{contract} them, resulting in a tensor with $n-2$ indices $S_{\alpha_{3} \ldots \alpha_{n}} = \sum_{j} T_{jj\alpha_{3}\ldots \alpha_{n}}$. Combining this operation with the tensor product of tensors, we define the \emph{contraction} of tensors: given two tensors $(T_\alpha)$ and $(S_\beta)$ having both an index with the same dimension, say $\alpha_{1}$ and $\beta_{1}$, we can contract these indices to generate a new tensor $(\sum_{j}T_{j\alpha_{2}\ldots \alpha_{n}} \cdot S_{j\beta_{2}\ldots \beta_{m}})$. Graphically, we have the following example:

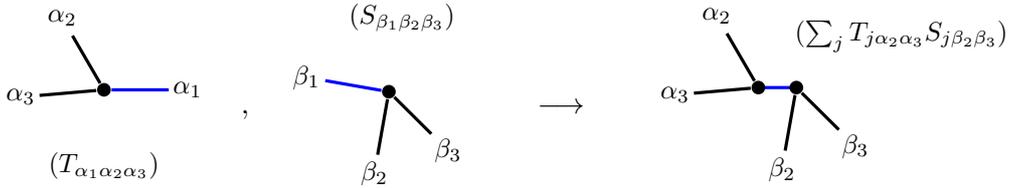
\begin{equation*}
    \begin{tikzpicture}[equation,very thick,decoration={
    markings,
    mark=at position 0.5 with {\arrow{>}}}
    ]
    
    \node (B1) at (0:1) {};
    \node  at (0:1.1) {\small $\alpha_{1}$};
    \node (B2) at (120:1) {};
    \node  at (120:1.1) {\small $\alpha_{2}$};
    \node (B3) at (185:1) {};
    \node  at (185:1.1) {\small $\alpha_{3}$};
        \node  at (-90:1) {\small $(T_{\alpha_{1}\alpha_{2}\alpha_{3}})$};  
        
    \node[fill=black, circle, scale=0.5] (A) at (0,0) {};
    \draw[blue]  (B1) --  (A);
    \draw (A) -- (B2);
    \draw (A) -- (B3);
    \end{tikzpicture}
     \quad , \quad
      \begin{tikzpicture}[equation,very thick,decoration={
    markings,
    mark=at position 0.5 with {\arrow{>}}}
    ]
    
    \node (D1) at (170:1) {};
    \node  at (170:1.1) {\small $\beta_{1}$};
    \node (D2) at (260:1) {};
    \node  at (260:1.1) {\small $\beta_{2}$};
    \node (D3) at (-45:1) {};
    \node  at (-45:1.1) {\small $\beta_{3}$};
    \node  at (80:1) {\small $(S_{\beta_{1}\beta_{2}\beta_{3}})$};      
        
    \node[fill=black, circle, scale=0.5] (A) at (0,0) {};
    \draw[blue] (D1) -- (A);
    \draw (A) -- (D2);
    \draw (A) -- (D3);
    \end{tikzpicture}
     \quad\quad \longrightarrow \quad\quad 
         \begin{tikzpicture}[equation,very thick,decoration={
    markings,
    mark=at position 0.5 with {\arrow{>}}}
    ]
    
    \node (B2) at (120:1) {};
    \node  at (120:1.1) {\small $\alpha_{2}$};
    \node (B3) at (185:1) {};
    \node  at (185:1.1) {\small $\alpha_{3}$};

    \node  at (20:2) {\small $(\sum_{j}T_{j\alpha_{2}\alpha_{3}} S_{j \beta_2 \beta_3})$}; 
    
    \draw[blue, very thick] (0,0) -- (0.5,0);
    
    \node[fill=black, circle, scale=0.5] (A) at (0,0) {};
    \draw (A) -- (B2);
    \draw (A) -- (B3);
    
    \begin{scope}[xshift=0.5cm]
    \node[fill=black, circle, scale=0.5] (A) at (0,0) {};
        \node (D2) at (260:1) {};
    \node  at (260:1.1) {\small $\beta_{2}$};
    \node (D3) at (-45:1) {};
    \node  at (-45:1.1) {\small $\beta_{3}$};
    \draw (A) -- (D2);
    \draw (A) -- (D3);
    \end{scope}

    \end{tikzpicture}
\end{equation*}

If two tensors have the same indices with the same dimensions $(T_{\alpha})_\alpha$ and $(S_{\alpha})_\alpha$, we can define their tensor sum $T+S:=(T_{\alpha}+S_{\alpha})_\alpha$. Given a tensor $(T_{\alpha})_\alpha$ and a scalar $\lambda \in \mathbb{C}$ we define $\lambda T := (\lambda T_{\alpha})_{\alpha}$.

If we identify each index $j$ with a Hilbert space $\mathbb{C}^{d_{j}}$, we can interpret a tensor $T$ as the coefficients of a ket $\ket{\psi}$ in the computational basis
\[ \ket{\psi} \in \otimes_{j=1}^{n} \mathbb{C}^{d_{1}} \quad , \quad \ket{\psi} = \sum_{\alpha_{1}, \ldots, \alpha_{n}} T_{\alpha_{1} \ldots \alpha_{n}} \ket{\alpha_{1} \ldots \alpha_{n}}\,.  \]
More generally, if we split the set of indices into two subsets $A$ and $B$ called \emph{input} and \emph{output} indices, respectively, we can then associate to $T$ the operator
\[ T: \otimes_{j \in A} \mathbb{C}^{d_{j}} \longrightarrow \otimes_{j \in B}\mathbb{C}^{d_{j}} \quad , \quad \sum_{\alpha} T_{\alpha} \ket{\alpha_{B}} \bra{\alpha_{A}} \,. \]
where $\ket{\alpha_{J}} = \otimes_{j \in J} \ket{\alpha_{j}}$ for any subset of indices $J$. Following the graphical description, we  will represent input (resp. output) indices with arrows that will point at (away from) the ball. For instance, in the case of a tensor with five legs $(T_{\alpha_{1} \ldots \alpha_{5}})$ we can consider
\begin{equation*}
    \begin{tikzpicture}[equation,very thick,decoration={
    markings,
    mark=at position 0.5 with {\arrow{>}}}
    ]
    
    \node (B1) at (30:1) {};
    \node  at (30:1.1) {\small $\alpha_{1}$};
    \node (B2) at (120:1) {};
    \node  at (120:1.1) {\small $\alpha_{2}$};
    \node (B3) at (185:1) {};
    \node  at (185:1.1) {\small $\alpha_{3}$};
    \node (B4) at (260:1) {};
    \node  at (260:1.1) {\small $\alpha_{4}$};
    \node (B5) at (-45:1) {};
    \node  at (-45:1.1) {\small $\alpha_{5}$};

    \node[fill=black, circle, scale=0.5] (A) at (0,0) {};
    \draw[postaction={decorate}]  (B1) --  (A);
    \draw[postaction={decorate}] (A) -- (B2);
    \draw[postaction={decorate}] (B3) -- (A);
    \draw[postaction={decorate}] (A) -- (B4);
    \draw[postaction={decorate}] (A) -- (B5);

    \node at (5,0) {$\sum\limits_{\alpha_{1}, \alpha_{2}, \alpha_{3}, \alpha_{4}, \alpha_{5}} T_{\alpha_{1} \ldots \alpha_{5}} \ket{\alpha_{2} \alpha_{4} \alpha_{5}} \bra{\alpha_{1} \alpha_{3}}$};
        
    \end{tikzpicture}\,.
\end{equation*}
We will take advantage of these multiple interpretations to find easy descriptions of tensors. For instance, the tensor $T$ with four indices of dimension two given by $T_{0 0 0 0} = T_{1 1 1 1} =1, T_{0 1 1 0} = T_{1 0 0 1}=-1$ and $T_{\alpha_{1} \alpha_{2} \alpha_{3} \alpha_{4}}=0$ otherwise, can be represented as
\[
\begin{tikzpicture}[equation, very thick ,decoration={
    markings,
    mark=at position 0.5 with {\arrow{>}}}
   ]

        \node (B1) at (60:1) {$\alpha_{1}$};
        \node (B2) at (120:1) {$\alpha_{2}$};
        \node (B3) at (240:1) {$\alpha_{3}$};
        \node (B4) at (300:1) {$\alpha_{4}$};

        \node[fill=black, circle, scale=0.5] (A) at (0,0) {};
        \draw[postaction={decorate}] (A) -- (B1);
        \draw[postaction={decorate}] (B2) -- (A);
        \draw[postaction={decorate}] (A) -- (B3);
        \draw[postaction={decorate}] (B4) -- (A);
        
        \node at (0:0.5) {$T$};
        
    \end{tikzpicture}
    \quad = \quad
\begin{tikzpicture}[equation, very thick ,decoration={
    markings,
    mark=at position 0.5 with {\arrow{>}}}
   ]

        \node (B2) at (120:1) {};
        \node (B3) at (240:1) {};

        \node[fill=black, circle, scale=0.5] (A) at (0,0) {};
        \draw[postaction={decorate}] (B2) -- (A);
        \draw[postaction={decorate}] (A) -- (B3);
        
        \node at (180:0.5) {$Z$};
        
    \end{tikzpicture}   
    \,\, \otimes \,\,
    \begin{tikzpicture}[equation, very thick ,decoration={
    markings,
    mark=at position 0.5 with {\arrow{>}}}
   ]

        \node (B1) at (60:1) {};
        \node (B4) at (300:1) {};

        \node[fill=black, circle, scale=0.5] (A) at (0,0) {};
        \draw[postaction={decorate}] (A) -- (B1);
        \draw[postaction={decorate}] (B4) -- (A);
        
        \node at (0:0.5) {$Z$};
        
    \end{tikzpicture}   
\]
Simple descriptions of tensor can also be obtained by taking linear combinations of tensors having the same number of indexes and the same dimensions. For instance, the tensor $T$ with four indexes of dimension two given by $T_{0 0 0 0} = T_{1 1 1 1} =1$ and $T_{\alpha_{1} \alpha_{2} \alpha_{3} \alpha_{4}}=0$ otherwise, can be represented as
\begin{equation*}
    \begin{tikzpicture}[equation, very thick ,decoration={
    markings,
    mark=at position 0.5 with {\arrow{>}}}
   ]

        \node (B1) at (60:1) {$\alpha_{1}$};
        \node (B2) at (120:1) {$\alpha_{2}$};
        \node (B3) at (240:1) {$\alpha_{3}$};
        \node (B4) at (300:1) {$\alpha_{4}$};

        \node[fill=black, circle, scale=0.5] (A) at (0,0) {};
        \draw[postaction={decorate}] (A) -- (B1);
        \draw[postaction={decorate}] (B2) -- (A);
        \draw[postaction={decorate}] (A) -- (B3);
        \draw[postaction={decorate}] (B4) -- (A);
        \node at (0:0.5) {$T$};
    \end{tikzpicture}
     \quad = \quad \frac{1}{2} \,
\begin{tikzpicture}[equation, very thick ,decoration={
    markings,
    mark=at position 0.5 with {\arrow{>}}}
   ]

        \node (B2) at (120:1) {};
        \node (B3) at (240:1) {};

        \node[fill=black, circle, scale=0.5] (A) at (0,0) {};
        \draw[postaction={decorate}] (B2) -- (A);
        \draw[postaction={decorate}] (A) -- (B3);
        
        \node at (180:0.5) {$\mathbbm{1}$};
        
    \end{tikzpicture}   
    \,\, \otimes \,\,
    \begin{tikzpicture}[equation, very thick ,decoration={
    markings,
    mark=at position 0.5 with {\arrow{>}}}
   ]

        \node (B1) at (60:1) {};
        \node (B4) at (300:1) {};

        \node[fill=black, circle, scale=0.5] (A) at (0,0) {};
        \draw[postaction={decorate}] (A) -- (B1);
        \draw[postaction={decorate}] (B4) -- (A);
        
        \node at (0:0.5) {$\mathbbm{1}$};
        
    \end{tikzpicture} 
    \quad + \quad \frac{1}{2}  \,
   \begin{tikzpicture}[equation, very thick ,decoration={
    markings,
    mark=at position 0.5 with {\arrow{>}}}
   ]

        \node (B2) at (120:1) {};
        \node (B3) at (240:1) {};

        \node[fill=black, circle, scale=0.5] (A) at (0,0) {};
        \draw[postaction={decorate}] (B2) -- (A);
        \draw[postaction={decorate}] (A) -- (B3);
        
        \node at (180:0.5) {$Z$};
        
    \end{tikzpicture}   
    \,\, \otimes \,\,
    \begin{tikzpicture}[equation, very thick ,decoration={
    markings,
    mark=at position 0.5 with {\arrow{>}}}
   ]

        \node (B1) at (60:1) {};
        \node (B4) at (300:1) {};

        \node[fill=black, circle, scale=0.5] (A) at (0,0) {};
        \draw[postaction={decorate}] (A) -- (B1);
        \draw[postaction={decorate}] (B4) -- (A);
        
        \node at (0:0.5) {$Z$};
    \end{tikzpicture} 
\end{equation*}
Finally, we can also represent the previous tensor as the contraction of two tensors, namely
\begin{equation}\label{equa:axampleTensorNotation}
    \begin{tikzpicture}[equation, very thick ,decoration={
    markings,
    mark=at position 0.5 with {\arrow{>}}}
   ]

        \node (B1) at (60:1) {$\alpha_{1}$};
        \node (B2) at (120:1) {$\alpha_{2}$};
        \node (B3) at (240:1) {$\alpha_{3}$};
        \node (B4) at (300:1) {$\alpha_{4}$};

        \node[fill=black, circle, scale=0.5] (A) at (0,0) {};
        \draw[postaction={decorate}] (A) -- (B1);
        \draw[postaction={decorate}] (B2) -- (A);
        \draw[postaction={decorate}] (A) -- (B3);
        \draw[postaction={decorate}] (B4) -- (A);
        \node at (0:0.5) {$T$};
    \end{tikzpicture}
= \begin{tikzpicture}[equation, very thick ,decoration={
    markings,
    mark=at position 0.5 with {\arrow{>}}}
   ]
        \draw[very thick, blue] (0,0) -- (0.5,0);
        \begin{scope}[xshift=0.5cm]
        \node (B1) at (60:1) {};
        \node (B4) at (300:1) {};
        \node[fill=black, circle, scale=0.5] (A) at (0,0) {};
        \draw[postaction={decorate}] (A) -- (B1);
        \draw[postaction={decorate}] (B4) -- (A);
        \end{scope}
        
        \node (B2) at (120:1) {};
        \node (B3) at (240:1) {};
        \node[fill=black, circle, scale=0.5] (A) at (0,0) {};
        \draw[postaction={decorate}] (B2) -- (A);
        \draw[postaction={decorate}] (A) -- (B3);
        
    \end{tikzpicture}
\quad \text{where} \quad
    \begin{tikzpicture}[equation, very thick ,decoration={
    markings,
    mark=at position 0.5 with {\arrow{>}}}
   ]
        \draw[thick, blue, postaction={decorate}] (1,0) -- (0,0);
        \node (B2) at (120:1) {};
        \node (B3) at (240:1) {};
        \node[fill=black, circle, scale=0.5] (A) at (0,0) {};
        \draw[postaction={decorate}] (B2) -- (A);
        \draw[postaction={decorate}] (A) -- (B3);
        
    \end{tikzpicture}
     = \frac{1}{\sqrt{2}} \mathbbm{1} \otimes \bra{0} + \frac{1}{\sqrt{2}} Z \otimes \bra{1} \,\,.
\end{equation}

\subsection{PEPS}

Let us recall the notation and main concepts for Projected Entangled Pair States, or PEPS for short~\cite{Cirac2019, Cirac2021}. Let us consider a finite graph consisting of a finite set of vertices $\Lambda$ and a set of edges $E$. At each vertex $x \in \Lambda$ consider a tensor in the form of an operator

\[
\begin{tikzpicture}[scale=0.3, equation]
\draw (-6.1,0) node{$V_{x} \,  :  (\mathbb{C}^{D})^{\otimes \partial x}  \,\,  \longrightarrow \,\,  \mathbb{C}^{d}$};

\draw (0,-4) node{$V_{x} \, = \, \sum\limits_{k=1}^{d} \,\, \sum\limits_{j_{1}, j_{2}, j_{3}, j_{4}=1}^{D} T_{j_{1}, j_{2}, j_{3}, j_{4}}^{k} \ket{k} \bra{j_{1} j_{2} j_{3} j_{4}}$}; 

\begin{scope}[xshift=-22cm, yshift=-2.5cm]
\draw[ thick] (0,-3) -- (0,3) ;  
\draw[thick] (-3,0) -- (3,0);  
\draw[fill=black!20] (-0.9,-0.9) rectangle (0.9,0.9);  

\draw (0,0) node{$k$}; 
\draw (2.5,1) node{$j_1$}; 
\draw (-1,2.5) node{$j_2$}; 
\draw (-2.5,-1) node{$j_3$}; 
\draw (1,-2.5) node{$j_4$}; 
\end{scope}
\end{tikzpicture}
\]

\noindent Here, $\mathbb{C}^{d}$ is the \emph{physical} space associated with $x$ and each $\mathbb{C}^{D}$ is the \emph{virtual} space corresponding to an edge $e \in \partial x$, the set of edges incident to $x$ (in the figure, we have shown the case $\abs{\partial x} =4$).


For a finite region $X \subset \Lambda$, we can assign the tensors $V_{x}$ to each site $x \in X$ and perform contractions between pairs of sites $x, y \in X$ that share an edge $e$. Specifically, we contract the two virtual indices from $V_{x}$ and $V_{y}$ that are associated with the same edge $e$. The contraction of the PEPS tensors gives a linear map from the virtual edges $\partial X$ connecting $X$ with its complement to the bulk physical Hilbert space, which we denote as
\[ 
V_X : \mathcal{H}_{\partial X} \longrightarrow \mathcal{H}_{X}\,,
\]
where we can identify  $\mathcal{H}_{\partial X} \equiv (\mathbb{C}^{D})^{\otimes \partial X}$ and $\mathcal{H}_{X} \equiv (\mathbb{C}^{d})^{\otimes X}$. The image of $V_X$ is the space of physical states which can be represented by the PEPS with an appropriate choice of boundary condition. 
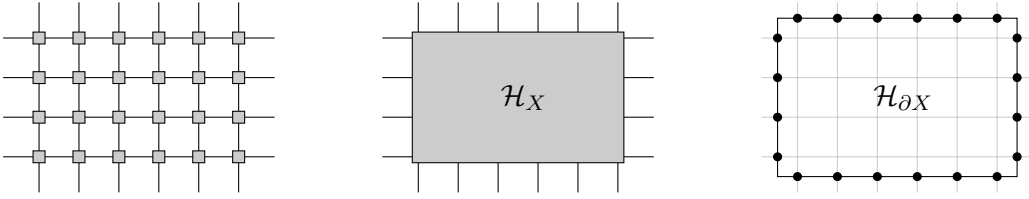
\[
\begin{tikzpicture}[scale=0.35, equation]
    \begin{scope}[scale=1.5, xshift=6cm,yshift=-0.5cm]
        \draw[step=1cm,black] (-1.9,-1.9) grid (4.9,2.9);
        \draw[step=1cm,blue, very thick] (-1,-1) grid (4,2);
    \foreach \x in {-1,...,4} \foreach \y in {-1,...,2}
    \draw[fill=black!20] (\x-0.15,\y-0.15) rectangle     (\x+0.15,\y+0.15);  
    \end{scope}

    \begin{scope}[scale=1.5, xshift=15.5cm,yshift=-0.5cm]
    \draw[step=1cm,black] (-1.9,-1.9) grid (4.9,2.9);
    \draw[fill=black!20] (-1.15,-1.15) rectangle (4.15,2.15);  
    \draw (1.65,0.5) node{$\mathcal{H}_{X}$}; 
    \end{scope}

    \begin{scope}[scale=1.5, xshift=25cm, yshift=-0.5cm]
    \draw[step=1cm,black!20] (-1.9,-1.9) grid (4.9,2.9);
    \draw (-1.5,-1.5) --  ++(6,0) --  ++(0,4) --  ++(-6,0) --  cycle  {};
    \foreach \x in {-1.5, 4.5} \foreach \y in {-1,...,2}
    \draw[fill=black] (\x,\y) circle [radius = 0.1];   
    \foreach \x in {-1,...,4} \foreach \y in {-1.5,2.5}
    \draw[fill=black] (\x,\y) circle [radius = 0.1];   
    \draw (1.65,0.5) node{$\mathcal{H}_{\partial X}$}; 
    \end{scope}
\end{tikzpicture}
\]
\noindent The \emph{boundary state} on region $X$ is defined as
\begin{equation}
    \rho_{\partial X} = V_X^\dag V_X \in \mathcal{B}(\mathcal{H}_{\partial X}).
\end{equation}
If $\rho_{\partial X}$ has full rank, we say that the PEPS is \emph{injective} on region $X$. If the PEPS is injective on every sufficiently large region $X$, we will simply say that it is injective. We will denote by $J_{\partial X}$ the orthogonal projector onto the support of $\rho_{\partial X}$.

If we replace the physical space $\mathbb{C}^d$ with a space of operators $\mathcal{B}(\mathbb{C}^d)$, then we talk of Projected Entangled-Pair Operators (PEPO) instead. If we fix a basis of matrix units $\{ \dyad{i}{j} \}_{i,j=1}^{d}$ for $\mathcal{B}(\mathbb{C}^d)$ (although we will find convenient to sometimes work with different bases), then the single-site tensor $V_x$ takes the form:
\[
\begin{tikzpicture}[scale=0.3, equation]
\draw (-6,0) node{$V_{x} \,  :  (\mathbb{C}^{D})^{\otimes \partial x}  \,\,  \longrightarrow \,\,  \mathcal{B}(\mathbb{C}^{d})$};

\draw (1,-4) node{$V_{x} \, = \, \sum\limits_{i,j=1}^{d} \dyad{i}{j} \, \otimes \, \sum\limits_{k_{1}, k_{2}, k_{3}, k_{4}=1}^{D} T_{k_{1}, k_{2}, k_{3}, k_{4}}^{i,j} \bra{k_{1} k_{2} k_{3} k_{4}}$}; 

\begin{scope}[xshift=-22cm, yshift=-2.5cm]
\end{scope}
\end{tikzpicture}
\]
In the formula, we have separated the physical indices, given by the matrix units $\dyad{i}{j}$, from the virtual indices. In the figures, we add arrows on the physical indices to indicate which spaces correspond to the ``ket'' and ``bra'' part of $\dyad{i}{j}$, see e.g. equation \eqref{equa:axampleTensorNotation}. A particular simple case of a PEPO is when the lattice is simply a 1D ring of $n$ sites, in which case it is also called a Matrix Product Operator (MPO). These will be the building blocks for our PEPO constructions. A MPO is any operator that can be written as
\[
    \sum_{i_1,\dots,i_n=1}^{d^2} \Tr[ M^{(1)}_{i_1} \cdots M^{(n)}_{i_n} ] B^{(1)}_{i_1} \otimes \cdots \otimes B^{(n)}_{i_n},
\]
where for each site $k=1,\dots, n$ we have fixed a basis $\{ B^{(k)}_i \}_{i=1}^{d^2}$ of $\mathcal{B}(\mathbb{C}^d)$ and a set of $d^2$ matrices $\{ M^{(k)}_i \}_{i=1}^{d^2}$ of dimension $D\times D$. To pass from a MPO to a PEPO representation, it is sufficient to check that the single site tensor
\[ 
    V_{k} = \sum_{i=1}^{d^2}  B^{(k)}_i \otimes \sum_{\alpha, \beta = 1}^D (M^{(k)}_{i})_{\alpha,\beta} \bra{\alpha, \beta},
\]
where $(M^{(k)}_{i})_{\alpha,\beta}$ denotes the matrix units of $M^{(k)}_{i}$, represents the same operator. It will be convenient to have a ``hybrid'' representation of a PEPO, in which the virtual level is still represented as a matrix, i.e. by writing 
\[ 
    V_{k} = \sum_{i=1}^{d^2}  B^{(k)}_i \otimes  M^{(k)}_{i},
\]
as this allows a more direct calculation of the tensor contractions. 

Note that we can always represent a PEPO as a PEPS with physical space $\mathbb{C}^d \otimes \mathbb{C}^d$, by choosing an orthonormal basis of $\mathcal{B}(\mathbb{C}^d)$ in order to identify that space with $\mathbb{C}^d \otimes \mathbb{C}^d$. If we choose the basis given by matrix units $\{ \dyad{i}{j} \}_{i,j=1}^d$, then this identification can be done as follows. Let $\ket{\Psi} = (\sum_{i=1}^d \ket{i,i})^{\otimes \abs{\RR}}$ be a maximally entangled state on $\mathcal{H}_{\RR} \otimes \mathcal{H}_{\RR}$. Each operator $Q \in \mathcal{B}(\mathcal{H}_{\RR})$ can be represented in ``vector form'' by 
\[\ket{Q} := (Q\otimes \mathbbm{1}) \ket{\Psi} \in \mathcal{H}_{\RR} \otimes \mathcal{H}_{\RR}\,.\] 
The map $Q \mapsto \ket{Q}$ is an isometry between $\mathcal{B}(\mathcal{H}_{\RR})$ with the Hilbert-Schmidt scalar product and $\mathcal{H}_{\RR} \otimes \mathcal{H}_{\RR}$.

We should mention that, if $\rho \in \mathcal{B}(\mathcal{H}_\RR)$ is a positive PEPO (in the sense that it is a positive operator in the range of the map $V_{\RR}$ for some region $\RR$), it is not always true that it admits a \emph{local purification}, in the sense that there exists a PEPS on a doubled physical space $\mathcal{H}_{\RR} \otimes \mathcal{H}_{\RR}$ such that we recover $\rho$ when we trace out one of the copies of $\mathcal{H}_{\RR}$. 
On the other hand, in the case when $\rho^{1/2}$ is a PEPO, then $\ket{\rho^{1/2}}$ is a PEPS and a purification of $\rho$. Therefore, when studying the case in which $\rho_\beta$ is the Gibbs state of a local, commuting Hamiltonian at inverse temperature $\beta$, we will write a PEPO representation for $\rho_\beta^{1/2}$ (which is proportional to $\rho_{\beta/2}$ up to normalization), and from it we will obtain the PEPS representation for the thermofield double state $\ket*{\rho_{\beta}^{1/2}}$.

Let us briefly discuss some characteristics that the PEPS description of the thermofield double state will have.
It will be convenient for us to consider a more general definition of PEPS in which the underlying graph $(\Lambda, E)$ can be a multigraph. This means that there might be multiple edges joining the same pair of different of vertices of $\Lambda$. Of course, we might combine virtual indices joining the same pair of sites into only one virtual index,  adhering to the original definition of PEPS on graphs, as we represent in the next picture: 
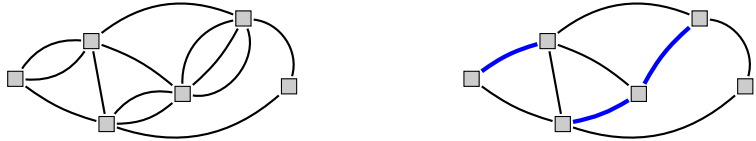
\[
\begin{tikzpicture}[equation,scale=1, decoration={markings, mark=at position 0.65 with {\arrow[very thick, scale=1.5]{>}}}]

\begin{scope}

    \node  (A) at (0,0) {};
    \node  (B) at (1,0.5) {};
    \node  (C) at (1.2,-0.6) {};
    \node  (D) at (2.2,-0.2) {};
    \node  (E) at (3,0.8) {};
    \node  (F) at (3.6,-0.1) {};

    \draw[fill=black!20] ($(A)+(-0.1,-0.1)$) rectangle ($(A)+(0.1,0.1)$);
    \draw[fill=black!20] ($(B)+(-0.1,-0.1)$) rectangle ($(B)+(0.1,0.1)$);
    \draw[fill=black!20] ($(C)+(-0.1,-0.1)$) rectangle ($(C)+(0.1,0.1)$);
    \draw[fill=black!20] ($(D)+(-0.1,-0.1)$) rectangle ($(D)+(0.1,0.1)$);
    \draw[fill=black!20] ($(E)+(-0.1,-0.1)$) rectangle ($(E)+(0.1,0.1)$);
    \draw[fill=black!20] ($(F)+(-0.1,-0.1)$) rectangle ($(F)+(0.1,0.1)$);
    
    \path [-, thick] (A) edge[bend right=30] node {} (B);
    \path [-, thick] (A) edge[bend left=30] node {} (B);
    \path [-, thick] (A) edge[bend right=10] node {} (C);
    \path [-, thick] (B) edge[bend left=0] node {} (C);
    \path [-, thick] (B) edge[bend left=10] node {} (D);
    \path [-, thick] (B) edge[bend left=30] node {} (E);
    \path [-, thick] (C) edge[bend left=30] node {} (D);
    \path [-, thick] (C) edge[bend right=20] node {} (D);
    \path [-, thick] (C) edge[bend right=30] node {} (F);
    \path [-, thick] (D) edge[bend right=60] node {} (E);
    \path [-, thick] (D) edge[bend right=10] node {} (E);
    \path [-, thick] (D) edge[bend left=40] node {} (E);
    \path [-, thick] (E) edge[bend left=50] node {} (F);
    
\end{scope}

\begin{scope}[xshift=6cm]
    
    \node  (A) at (0,0) {};
    \node  (B) at (1,0.5) {};
    \node  (C) at (1.2,-0.6) {};
    \node  (D) at (2.2,-0.2) {};
    \node  (E) at (3,0.8) {};
    \node  (F) at (3.6,-0.1) {};
    
    \draw[fill=black!20] ($(A)+(-0.1,-0.1)$) rectangle ($(A)+(0.1,0.1)$);
    \draw[fill=black!20] ($(B)+(-0.1,-0.1)$) rectangle ($(B)+(0.1,0.1)$);
    \draw[fill=black!20] ($(C)+(-0.1,-0.1)$) rectangle ($(C)+(0.1,0.1)$);
    \draw[fill=black!20] ($(D)+(-0.1,-0.1)$) rectangle ($(D)+(0.1,0.1)$);
    \draw[fill=black!20] ($(E)+(-0.1,-0.1)$) rectangle ($(E)+(0.1,0.1)$);
    \draw[fill=black!20] ($(F)+(-0.1,-0.1)$) rectangle ($(F)+(0.1,0.1)$);
    
    \path [-, ultra  thick, blue] (A) edge[bend left=10] node {} (B);
    \path [-, thick] (A) edge[bend right=10] node {} (C);
    \path [-, thick] (B) edge[bend left=0] node {} (C);
    \path [-, thick] (B) edge[bend left=10] node {} (D);
    \path [-, thick] (B) edge[bend left=30] node {} (E);
    \path [-, ultra  thick, blue] (C) edge[bend right=10] node {} (D);
    \path [-, thick] (C) edge[bend right=30] node {} (F);
    \path [-, ultra thick, blue] (D) edge[bend left=10] node {} (E);
    \path [-, thick] (E) edge[bend left=50] node {} (F);
    
\end{scope}

\end{tikzpicture}
\]
However, as we will see in Section \ref{sec:peps-quantum-double} when finding the PEPS description of the thermofield double state, it is more \emph{natural} and useful using the representation with the multigraph, especially when applying results such as Theorem \ref{theo:non-injective-approx-fact}, that require a suitable arrangement of the set of all virtual indices of $V_{X}$ for a region $X$ into several subsets.  In the setting of the Quantum Double Model, recall that $\Lambda_{N}$ is the set of edges $\EE_{N}$ of the squared lattice on the torus. It is important not to confuse this set of edges with the set of edges $E_{N}$ joining the sites in $\Lambda_{N}$ and define the PEPS. In this case, every site $x \in \Lambda_{N}$ will be connected to other four sites via two edges of $E_{N}$ as in the next picture:

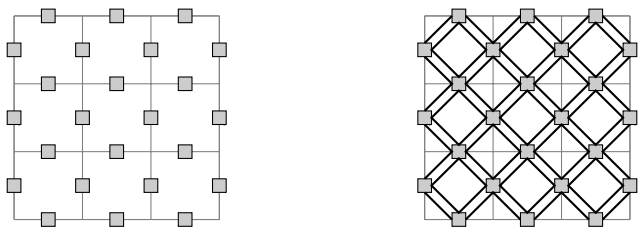
\[
\begin{tikzpicture}[equation,scale=0.9]
    
    \def\r{0.4}
    
\begin{scope}[xshift=-6cm]
 \draw[gray, thin] (0,0) grid (3,3);
    \draw[gray, postaction=torus horizontal] (0,0) -- (3,0);
    \draw[gray, postaction=torus horizontal] (0,3) -- (3,3);
    \draw[gray, postaction=torus vertical] (0,0) -- (0,3);
    \draw[gray, postaction=torus vertical] (3,0) -- (3,3);
    
    
    \foreach \x in {0.5,1.5,2.5} {
         \foreach \y in {0,1,2,3} {
         \draw[fill=black!20,xshift=\x cm,yshift=\y cm] (-0.1,-0.1) rectangle (0.1,0.1);
     }}
     
     \foreach \x in {0,1,2,3} {
         \foreach \y in {0.5,1.5,2.5} {
         \draw[fill=black!20,xshift=\x cm,yshift=\y cm] (-0.1,-0.1) rectangle (0.1,0.1);
     }}
     
\end{scope}
    
\begin{scope}
 \draw[gray, thin] (0,0) grid (3,3);
    \draw[gray, postaction=torus horizontal] (0,0) -- (3,0);
    \draw[gray, postaction=torus horizontal] (0,3) -- (3,3);
    \draw[gray, postaction=torus vertical] (0,0) -- (0,3);
    \draw[gray, postaction=torus vertical] (3,0) -- (3,3);
    
    
    \foreach \x in {0.5,1.5,2.5} {
         \foreach \y in {0,1,2,3} {
         \draw[fill=black!20,xshift=\x cm,yshift=\y cm] (-0.1,-0.1) rectangle (0.1,0.1);
     }}
     
     \foreach \x in {0,1,2,3} {
         \foreach \y in {0.5,1.5,2.5} {
         \draw[fill=black!20,xshift=\x cm,yshift=\y cm] (-0.1,-0.1) rectangle (0.1,0.1);
     }}
    
    
    \draw[black, thick,xshift=0 cm,yshift=0 cm] ({\r},{0}) -- ({0},{\r});
    
    \draw[black, thick,xshift=3 cm,yshift=0 cm] ({0},{\r}) -- ({-\r},{0});
    
    \draw[black, thick,xshift=3 cm,yshift=3 cm] ({-\r},{0}) -- ({0},{-\r});
    
    \draw[black, thick,xshift=0 cm,yshift=3 cm] ({0},{-\r}) -- ({\r},{0});
    
    
     \foreach \x in {1,2} {
     
        \draw[black, thick,xshift=\x cm,yshift=0 cm] ({\r},{0}) -- ({0},{\r}) -- ({-\r},{0}) ;
        
        \draw[black, thick,xshift=\x cm,yshift=3 cm] ({\r},{0}) -- ({0},{-\r}) -- ({-\r},{0});
     }
     
     \foreach \y in {1,2} {
        \draw[black, thick,xshift=0 cm,yshift=\y cm] ({0},{\r}) -- ({\r},{0}) -- ({0},{-\r}) ;
        
        \draw[black, thick,xshift=3 cm,yshift=\y cm] ({0},{\r}) -- ({-\r},{0}) -- ({0},{-\r});
     }
     
     \foreach \x in {1,2} {
         \foreach \y in {1,2} {
         \draw[black, thick,xshift=\x cm,yshift=\y cm] ({\r},{0}) -- ({0},{-\r}) -- ({-\r},{0}) -- ({0},{\r}) -- cycle;
     }}
    
     \foreach \x in {0.5,1.5,2.5} {
         \foreach \y in {0.5,1.5,2.5} {
         \draw[black, thick,xshift=\x cm,yshift=\y cm] ({\r},{0}) -- ({0},{-\r}) -- ({-\r},{0}) -- ({0},{\r}) -- cycle;
     }}
    
\end{scope}

\end{tikzpicture}
\]
If $G$ is the finite group from which the Quantum Double Model will be constructed, then the individual tensor $V_{x}$ will have as physical space $\mathbb{C}^{\abs{G}}$, and each virtual space will correspond to $\mathbb{C}^{\abs{G}}$:
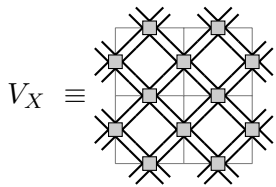
\[
\begin{tikzpicture}[equation,scale=0.9, decoration={markings, mark=at position 0.65 with {\arrow[very thick, scale=1.5]{>}}}]
    
    \def\r{0.4}
    
    \begin{scope}[xshift=-2cm, yshift=1cm]
    
    \draw (-0.5,1) node {$V_{x}:(\mathbb{C}^{\abs{G}})^{\otimes 8} \longrightarrow \mathbb{C}^{\abs{G}}$};
    
     \draw (-1,0) node {$V_{x} \,\, \equiv$}; 
    
    \draw[black, thick,xshift=-0.5cm] ({\r/2},{-\r/2}) -- ({\r},{0}) -- ({\r/2},{\r/2});
    
    \draw[black, thick,yshift=-0.5cm] ({\r/2},{\r/2}) -- ({0},{\r}) -- ({-\r/2},{\r/2});
    
    \draw[black, thick,xshift=0.5cm] ({-\r/2},{\r/2}) -- ({-\r},{0}) -- ({-\r/2},{-\r/2});
    
    \draw[black, thick,yshift=0.5cm] ({-\r/2},{-\r/2}) -- ({0},{-\r}) -- ({\r/2},{-\r/2});
    
    \draw[fill=black!20] (-0.1,-0.1) rectangle (0.1,0.1);
    \end{scope}


\begin{scope}[xshift=3cm, yshift=0.5cm]

\draw (-1,1) node {$V_{X} \,\, \equiv$}; 

 \draw[gray, thin] (0,0) grid (2,2);

    
    \foreach \x in {0.5,1.5} {
         \foreach \y in {0,1,2} {
         
         \begin{scope}[xshift=\x cm, yshift=\y cm]
          \draw[black, thick,xshift=-0.5cm] ({\r/2},{-\r/2}) -- ({\r},{0}) -- ({\r/2},{\r/2});
    
        \draw[black, thick,yshift=-0.5cm] ({\r/2},{\r/2}) -- ({0},{\r}) -- ({-\r/2},{\r/2});
    
        \draw[black, thick,xshift=0.5cm] ({-\r/2},{\r/2}) -- ({-\r},{0}) -- ({-\r/2},{-\r/2});
    
        \draw[black, thick,yshift=0.5cm] ({-\r/2},{-\r/2}) -- ({0},{-\r}) -- ({\r/2},{-\r/2});
    
        \draw[fill=black!20] (-0.1,-0.1) rectangle (0.1,0.1);
        \end{scope}
    
     }}
     
     \foreach \x in {0,1,2} {
         \foreach \y in {0.5,1.5} {
         
         \begin{scope}[xshift=\x cm, yshift=\y cm]
          \draw[black, thick,xshift=-0.5cm] ({\r/2},{-\r/2}) -- ({\r},{0}) -- ({\r/2},{\r/2});
    
        \draw[black, thick,yshift=-0.5cm] ({\r/2},{\r/2}) -- ({0},{\r}) -- ({-\r/2},{\r/2});
    
        \draw[black, thick,xshift=0.5cm] ({-\r/2},{\r/2}) -- ({-\r},{0}) -- ({-\r/2},{-\r/2});
    
        \draw[black, thick,yshift=0.5cm] ({-\r/2},{-\r/2}) -- ({0},{-\r}) -- ({\r/2},{-\r/2});
    
        \draw[fill=black!20] (-0.1,-0.1) rectangle (0.1,0.1);
        \end{scope}
        
     }}
 
\end{scope}

\end{tikzpicture}
\]

\subsection{Parent Hamiltonian}\label{subsec:parentHam}
A \emph{parent Hamiltonian} of a PEPS is a local, frustration-free Hamiltonian whose local ground state spaces $W_{\RR}$ coincide with the range of $V_{\RR}$ for all sufficiently large regions $\RR$.

There are well-known conditions that can be imposed on the local PEPS tensors $V_x$ that ensure that a parent Hamiltonian exists (in which case, it will not be unique), one of which is injectivity, and which can be generalized to \emph{$G$-injectivity} and \emph{MPO-injectivity}. Since the PEPS we will consider will not satisfy any of them, we omit further details in this direction, and we will  prove directly the existence of a parent Hamiltonian in Section~\ref{sec:parent-hamiltonian}.

For a given PEPS, there is a canonical construction  of a local and frustration-free Hamiltonian whose local ground spaces contain the range of $V_{\mathcal{R}}$. For every finite subset $X$ of $\Lambda$, let $P_{X}:\mathcal{H}_{X} \longrightarrow \mathcal{H}_{X}$ be the orthogonal projector onto $\Im(V_{X})$. Note that if $X \subset Y$ then \mbox{$\Im(V_{Y}) \subset \Im(V_{X}) \otimes \mathcal{H}_{Y \setminus X}$}. Thus, the projectors satisfy the frustration-free condition $P_{X}P_{Y}=P_{Y}P_{X}=P_{Y}$, or also, $P_{X}^{\perp} \geq P_{Y}^{\perp}$.

Next, let us fix a certain family $\XXX$ of subsets of $\Lambda$ having small range, e.g. rectangular regions in $\Lambda = \mathbb{Z}^{2}$ of dimensions $r \times r$ for a fixed value $r$. Then, consider the local interaction defined by the operators  $P_{X}^{\perp}:=\mathbbm{1}_{X} - P_{X} \, , \, X \in \mathcal{X}$. Note that for each finite region $R$, the local Hamiltonian
\[ H_{\RR} = \sum_{X \in \XXX \, , \, X \subset \RR} P_{X}^{\perp}\]
satisfies $\Im(V_{R}) \subset \ker{H_{R}}$, so that the PEPS is in the ground space of this Hamiltonian. For this to be a parent Hamiltonian, we need a condition ensuring that this is indeed an equality for large enough regions. We will now show that such condition can be obtained starting from the Martingale Condition (see Definition \ref{Defi:MartingaleCondition}). As the result is not specific to PEPS, we will state it here in its full generality.

Let us consider the setting of a quantum spin system over a finite set $\Lambda$ that we presented at the beginning of Section \ref{subsec:torusSettingAndGap}. Let us assume that for each finite subset $X$ we have an orthogonal projector $P_{X}$ onto a subspace of $\mathcal{H}_{X}$ satisfying the frustration-free condition $P_{X}P_{Y} = P_{Y}P_{X}=P_{Y}$ for every pair of subsets $X \subset Y$. In the case of a PEPS, these projections will be the ones onto $\Im(V_{X})$. 

\begin{Lemm}\label{Prop:KernelImageEqualityRises}
Under the setting just described, fix a family $\mathcal{X}$ of finite subsets of $\Lambda$, and consider the local interactions $P_{X}^{\perp} = \mathbbm{1} - P_{X}$, $X \in \mathcal{X}$ defining for each finite $\RR \subset \Lambda$ the local Hamiltonian  
\[ H_{\RR} = \sum_{\substack{X \in \mathcal{X}\, , \,  X \subset \RR}} P_{X}^{\perp}\,. \]
If $\RR_{1}, \RR_{2} \subset \Lambda$ satisfy the following properties: 
\begin{enumerate}
\item[(i)] $\ker(H_{\RR_{i}}) = \Im(P_{\RR_{i}})$ for $i=1,2$,
    \item[(ii)]  $\|P_{\RR_{1} \cup \RR_{2}} - P_{\RR_{1}} P_{\RR_{2}}\| < 1$,
    \item[(iii)] for every $X \in \mathcal{X}$ with $X \subset \RR_{1} \cup \RR_{2}$ we have that $X \subset \RR_{1}$ or $X \subset \RR_{2}$,
\end{enumerate}
then $\ker(H_{\RR_{1} \cup \RR_{2}}) = \Im(P_{\RR_{1} \cup \RR_{2}})$.
\end{Lemm}

\begin{proof}
Let us denote $\RR:= \RR_{1} \cup \RR_{2}$. Applying $(iii)$ we get
\begin{align*}
\ker{(H_{\RR})} & = \bigcap_{X \in \XXX \, , \, X \subset \RR} \ker{(P_{X}^{\perp})} \\[2mm]
& = \bigcap_{X \in \XXX \, , \, X \subset \RR_{1}} \ker{(P_{X}^{\perp})} \,\, \cap \,\, \bigcap_{X \in \XXX \, , \, X \subset \RR_{2}} \ker{(P_{X}^{\perp})}  \,\, = \,\, \ker (H_{\RR_{1}})  \cap \ker (H_{\RR_{2}})\,.
\end{align*}
Combining the previous equality with the frustration-free condition $P_{\RR} P_{\RR_{i}} = P_{\RR}$ and with $(i)$, we obtain that
\begin{equation*}
  \operatorname{Im}(P_{\RR}) \subset  \ker (H_{\RR}) = \operatorname{Im}(P_{\RR_{1}}) \cap \operatorname{Im}(P_{\RR_{2}})\,.
\end{equation*}
It remains to prove that this is indeed a chain of equalities, which is a consequence of the last statement of Lemma \ref{Lemm:MartingaleProjector} together with condition $(ii)$.
\end{proof}

Let us now restrict to the case in which $\Lambda$ is the torus $\Lambda_{N}$ and the family of projectors $P_X$ satisfies the Martingale Condition from Definition \ref{Defi:MartingaleCondition} with decay function $\delta(\ell)$.
Let $3 \leq r \in \mathbb{N}$ with $N \geq 2(1+r)$ and such that $\delta(\ell) <1/2$ for every $\ell \geq r-2$. Then, let us consider the family $\XXX = \mathcal{F}_{N,r}^{rect}$ of all proper rectangles in $\mathcal{F}_{N}$ having at most $r$ plaquettes per row and per column
\[
    \begin{tikzpicture}[equation, scale=0.3, baseline={([yshift=0.5ex]current bounding box.center)}]
    \draw[black, very thick] (0,0) grid (4,3);
    \draw[<->] (-0.6,0) -- (-0.6,1.5) node[left, black] {\small a} -- (-0.6, 3); 
    \draw[<->] (0,-0.6) -- (2, -0.6) 
        node[below, black] {\small $b$} -- (4,-0.6);
    \end{tikzpicture}
    \hspace{2cm} 1 \leq a,b \leq r\,,
\]
and the set of local interactions $(P_{X}^{\perp})_{X \in \XXX}$ where $P_{X}^{\perp} = \mathbbm{1} - P_{X}$.

\begin{Prop}\label{Prop:parentHamiltonianProperty}
Under the previous hypothesis, for every rectangular region $\RR \in \mathcal{F}_{N}$ containing at least $r$ plaquettes per row and per column we have that the associated Hamiltonian  
\[ H_{\RR} = \sum_{X \in \XXX , X \subset \RR} P_{X}^{\perp} \quad \quad \mbox{satisfies} \quad \quad \ker(H_{\RR}) = \operatorname{Im}(P_{\RR})\,. \]
In other words, $P_{\RR}$ is the orthogonal projector onto the ground state space of $H_{\RR}$.
\end{Prop}

\begin{proof}
We are going to prove that every rectangle $\RR$ having $a$ plaquettes per row and $b$ per column with $a,b \geq r$ satisfies $\ker{(H_{\RR})} = \operatorname{Im}(P_{\RR})$ arguing by induction on $a+b$. The first case is $a+b=2r$, for which we necessarily have $a=b=r$ and so $\RR \in \XXX$. In this case, the frustration-free condition of the projectors and the fact that $P_{\RR}^{\perp}$ is one of the summands of $H_{\RR}$ immediately yields the equality.

Let us assume that $a+b>2r$ and that the claim holds for all rectangular regions $\mathcal{R'}$ with dimensions $a',b'$ satisfying $a'+b' < a+b$. We claim that there exist rectangular subregions $\mathcal{R}_{1}, \mathcal{R}_{2} \subset \mathcal{R}$ such that:
\begin{itemize}
    \item[(i)] $\mathcal{R} = \mathcal{R}_{1} \cup \mathcal{R}_{2}$,
    \item[(ii)]  $\mathcal{R}_{j}$ has dimensions $a_{j}, b_{j}$ satisfying $a_{j} + b_{j} < a+b$ for each $j=1,2$,
    \item[(iii)] If $X \in \XXX$ is contained in $\mathcal{R}$, then $X \subset \mathcal{R}_{1}$ or $X \subset \mathcal{R}_{2}$,
    \item[(iv)] $\| P_{\mathcal{R}_{1}} P_{\mathcal{R}_{2}} - P_{\mathcal{R}} \| < 1$.
\end{itemize}
If this claim holds, then by Lemma \ref{Prop:KernelImageEqualityRises} we immediately conclude that $\ker(H_{\RR}) = \Im(P_{\RR})$, and so the proof is finished. Let us the show the validity of the claim distinguishing three possible cases according to whether the region $\RR$ is a proper rectangle, a cylinder or the whole torus.

\emph{Case 1}: If $\RR$ is a proper rectangle, we can assume without loss of generality that $b>r$ where recall that $b$ is the number of plaquettes of each row. Then, we split $\RR$ along the rows into  three disjoint parts $A,B,C$ as in the next picture

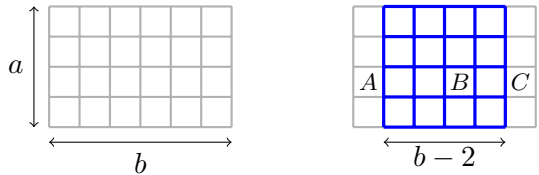
\[
\begin{tikzpicture}[equation, scale=0.8]

\begin{scope}[scale=0.5]
\draw[thick, white!70!black] (0,0) grid (6,4);

\draw[<->] (0,-0.5) -- (3,-0.5) node[below]{$b$} -- (6,-0.5);

\draw[<->] (-0.5,0) -- (-0.5,2) node[left]{$a$} -- (-0.5,4);

\end{scope}

\begin{scope}[scale=0.5, xshift=10cm]
\draw[thick, white!70!black] (0,0) grid (6,4);
\draw[very thick, blue] (1,0) grid (5,4);
\draw (0.5,1.5) node{\footnotesize $A$};
\draw (3.5,1.5) node{\footnotesize $B$};
\draw (5.5,1.5) node{\footnotesize $C$};

\draw[<->] (1,-0.5) -- (5,-0.5);
\draw (3,-1) node{$b-2$};
\end{scope}

\end{tikzpicture}
\]
so that $\mathcal{R}_{1}=AB$ and $\mathcal{R}_{2} = BC$ are proper rectangles of dimensions $a$ and $b-1$, and thus they satisfy (i)-(ii). They also satisfy (iii), since if $X \subset \RR$ is a rectangle contained neither in $\RR_{1}$ nor in $\RR_{2}$, it must have $b$ plaquettes per row, but $b>r$, so $X$ cannot belong to $\XXX$.  Since  $\mathcal{R}_{1} \cap \mathcal{R}_{2} = B$ is again a rectangle of dimensions $a$ and $b-2$, we deduce from the Martingale Condition that
\[ \| P_{\RR_{1}} P_{\RR_{2}} - P_{\RR_{1} \cup \RR_{2}}  \| \, \leq \, \delta(b-2) < 1  \,. \]

\emph{Case 2:} If $\mathcal{R}$ is a cylinder, we can assume without loss of generality that its border lies on the horizontal sides, so that it contains $N$ plaquettes per row. We then split $\mathcal{R}$ along the rows into four disjoint regions $A, B, C, B'$ as in the next picture, where $A$ and $C$ correspond to columns of horizontal edges, and $B, B'$ have horizontal dimension greater than $\lfloor N/2-1 \rfloor$ 

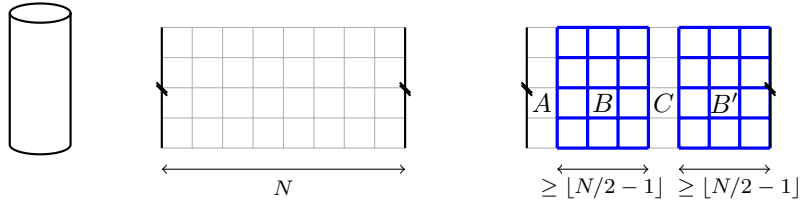
\begin{equation*}
\begin{tikzpicture}[equation, scale=0.8]

\begin{scope}[xshift=-2cm, yshift=1cm]
\node [cylinder, draw, minimum height=2cm, minimum width=0.8cm, fill=white, rotate=90, thick, black, fill=white] {};
\end{scope}

\begin{scope}[scale=0.5]


\draw[white!70!black] (0,0) grid (8,4);


\draw[thick, black, postaction=torus vertical] (0,0) -- (0,4); 
\draw[thick, black, postaction=torus vertical] (8,0) -- (8,4); 


\draw[<->] (0,-0.7) -- (8,-0.7);
\draw (4,-1.3) node{\scriptsize $N$};

\end{scope}


\begin{scope}[xshift=6cm, scale=0.5]


\draw[white!70!black] (0,0) grid (8,4);

\draw[blue, very thick] (1,0) grid (4,4);
\draw[blue, very thick] (5,0) grid (8,4);

\draw (0.5,1.5) node{\small $A$};
\draw (2.5,1.5) node{\small $B$};
\draw (4.5,1.5) node{\small $C$};
\draw (6.5,1.5) node{\small $B'$};


\draw[thick, black, postaction=torus vertical] (0,0) -- (0,4); 
\draw[thick, black, postaction=torus vertical] (8,0) -- (8,4); 


\draw[<->] (1,-0.7) -- (4,-0.7);
\draw (2.5,-1.3) node{\scriptsize $\geq \lfloor N / 2-1 \rfloor$};
\draw[<->] (5,-0.7) -- (8,-0.7);
\draw (7,-1.3) node{\scriptsize $\geq \lfloor N / 2-1 \rfloor$};

\end{scope}

\end{tikzpicture}
\end{equation*}
Taking $\mathcal{R}_{1}=B'AB$ and $\mathcal{R}_{2}=BCB'$ we immediately get that these are proper rectangles satisfying (i) and (ii). They also satisfy (iii) since $\lfloor N/2-1\rfloor  \geq r$ by the hypothesis. Property (iv) follows from the Martingale Condition since it yields that
\[ \| P_{B'AB} P_{BCB'} - P_{ABCB'} \| \, \leq \, \delta(\lfloor N / 2-1 \rfloor) < 1\,. \]

\emph{Case 3}: Assume $\mathcal{R}$ is the whole torus $\Lambda_{N}$. Then, we can split it into four regions $A,B,C,B'$ as in the next picture, where $B$ and $B'$ have dimensions $a = N$ and $b \geq \lfloor N/2-1\rfloor$  
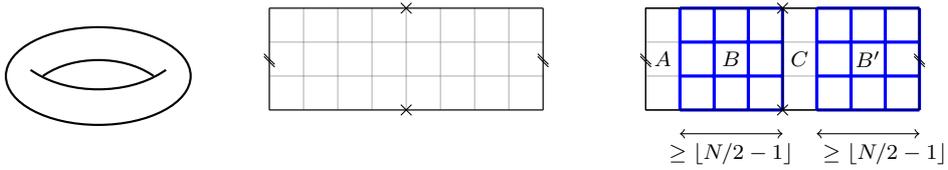
\begin{equation*}
\begin{tikzpicture}[equation, scale=0.9]

\begin{scope}[scale=0.9, thick]

\path[rounded corners=24pt, blue] (-.9,0)--(0,.6)--(.9,0) (-.9,0)--(0,-.56)--    (.9,0);
\draw[rounded corners=28pt, black] (-1.1,.1)--(0,-.6)--(1.1,.1);
\draw[rounded corners=24pt, black] (-.9,0)--(0,.6)--(.9,0);
\draw[black] (0,0) ellipse (1.5 and 0.8);
\end{scope}   
    
   
\begin{scope}[xshift=2.5cm, scale=0.5, yshift=-1cm]

\draw[step=1,white!70!black] (0,0) grid (8,3);

\draw[black, postaction=torus horizontal] (0,0) -- (8,0); 
\draw[black, postaction=torus horizontal] (0,3) -- (8, 3);
\draw[black, postaction=torus vertical] (0,0) -- (0, 3);
\draw[black, postaction=torus vertical] (8,0) -- (8, 3);

\end{scope}   
   
    
\begin{scope}[xshift=8cm, scale=0.5, yshift=-1cm]

\draw[step=1,white!70!black] (0,0) grid (8,3);

\draw[black] (0,0) rectangle (1,3);
\draw (0.5,1.5) node{\scriptsize  $A$};
\draw[very thick, blue] (1,0) grid (4,3);
\draw (2.5,1.5) node{\scriptsize  $B$};
\draw[black] (4,0) rectangle (5,3);
\draw (4.5,1.5) node{\scriptsize  $C$};
\draw[very thick, blue] (5,0) grid (8,3);
\draw (6.5,1.5) node{\scriptsize  $B'$};

\draw[black, postaction=torus horizontal] (0,0) -- (8,0); 
\draw[black, postaction=torus horizontal] (0,3) -- (8, 3);
\draw[black, postaction=torus vertical] (0,0) -- (0, 3);
\draw[black, postaction=torus vertical] (8,0) -- (8, 3);


\draw[<->] (1,-0.7) -- (4,-0.7);
\draw (2.5,-1.3) node{\scriptsize $\geq \lfloor N / 2-1 \rfloor$};
\draw[<->] (5,-0.7) -- (8,-0.7);
\draw (7,-1.3) node{\scriptsize $\geq \lfloor N / 2-1 \rfloor$};

\end{scope}

\end{tikzpicture} 
\end{equation*}

\noindent Taking $\mathcal{R}_{1} = B'AB$ and $\mathcal{R}_{2}=BCB'$ as rectangular subregions, we can argue analogously to the previous cases to deduce that they satisfy (i)-(iii) and
\begin{align*}
\| P_{\mathcal{R}_{1}} P_{\mathcal{R}_{2}} - P_{\Lambda_{N}} \| \, &\leq \, \delta(\lfloor N / 2-1 \rfloor) < 1\,. 
\end{align*}
This concludes the proof of the claim.
\end{proof}

\subsection{Spectral gap of a parent Hamiltonian}
We will now recall the relationship established in \cite{KaLuPeGa19} between boundary states and the spectral gap of a parent Hamiltonian of a PEPS.

\subsubsection{Boundary states and approximate factorization}
Let us consider three (connected) regions $A,B,C \subset \Lambda$ and assume that $B$ shields $A$ from $C$, so that there is no edge joining vertices from $A$ and $C$ (see Figure~\ref{Fig:three-regions}). Let us consider the boundary states $\rho_{\partial ABC}, \rho_{\partial AB}, \rho_{\partial BC}$ and $\rho_{\partial B}$. In the case where they are all full rank, the approximate factorization condition is defined as follows.

\begin{figure}[hbt]
\begin{center}
\begin{tikzpicture}[scale=0.75]
    \draw[step=1cm,black] (-1.9,-1.9) grid (5.9,2.9);
    \foreach \x in {-1,...,5} \foreach \y in {-1,...,2}
    \draw[fill=black!20] (\x-0.15,\y-0.15) rectangle (\x+0.15,\y+0.15); 
    
    \foreach \x in {0, 1} \foreach \y in {0,1}
        \draw[fill=orange!80] (\x-0.15,\y-0.15) rectangle (\x+0.15,\y+0.15); 
    \foreach \x in {2} \foreach \y in {0,1}
        \draw[fill=magenta!80] (\x-0.15,\y-0.15) rectangle (\x+0.15,\y+0.15); 
    \foreach \x in {3,4} \foreach \y in {0,1}
        \draw[fill=cyan!80] (\x-0.15,\y-0.15) rectangle (\x+0.15,\y+0.15); 
    
    \draw[orange!90,dashed] (-0.25,-0.25) rectangle (1.25, 1.25);
    \node at (0.5, 0.5) {$A$};
    \draw[magenta!90,dashed] (1.75,-0.25) rectangle (2.25, 1.25);
    \node at (2, 0.5) {$B$};
    \draw[cyan!90,dashed] (2.75,-0.25) rectangle (4.25, 1.25);
    \node at (3.5, 0.5) {$C$};
    
\end{tikzpicture}
\caption{An example of three regions $A$, $B$, $C$.}
\label{Fig:three-regions}
\end{center}
\end{figure}
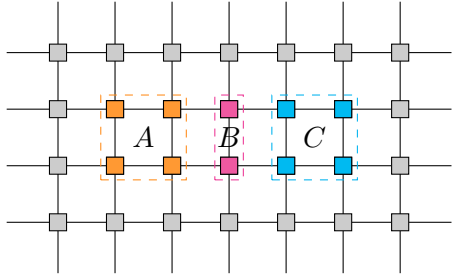

\begin{Defi}[Approximate factorization for injective PEPS \cite{KaLuPeGa19}]
\label{def:approx-fact-injective}
Let $\varepsilon>0$. We will say that the boundary states are $\varepsilon$-approximately factorizable, if we can divide the regions

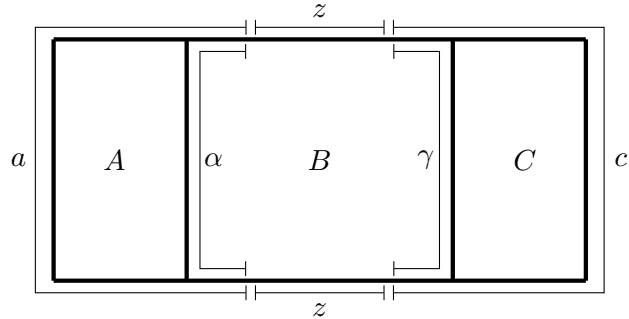
\[
\begin{tikzpicture}[equation]
\pgfmathsetmacro{\base}{3.5}
\pgfmathsetmacro{\altu}{1.6}

\begin{scope}[scale=1,xshift=6cm,yshift=-4cm]

\begin{scope}[ultra thick]

\draw (-\base,\altu) -- ({\base},\altu);   
\draw (-\base,-\altu) -- ({\base},-\altu);   
\draw (-\base,\altu) -- (-\base,-\altu);   
\draw (\base,-\altu) -- (\base,\altu); 

\draw (-2.7,0) node{$A$}; 
\draw (0,0) node{$B$}; 
\draw (2.7,0) node{$C$}; 

\draw (-{\base/2},\altu) -- (-{\base/2},-\altu);   
\draw ({\base/2},\altu) -- ({\base/2},-\altu);   

\end{scope}



\begin{scope}[scale=1.1]
\draw [|-|] (-{\base/4},\altu) -- (-{\base}+0.1,\altu) -- (-{\base}+0.1,-\altu) -- (-{\base/4},-\altu);
\draw (-{\base-0.1},0) node{$a$}; 
\end{scope}

\begin{scope}[scale=1.1]
\draw [|-|] ({\base/4},\altu) -- ({\base-0.1},\altu) -- ({\base-0.1},-\altu) -- ({\base/4},-\altu);
\draw ({\base+0.1},0) node{$c$}; 
\end{scope}

\begin{scope}[scale=1.1]
\draw [|-|] (-{\base/4+0.1},\altu) -- ({\base/4-0.1},\altu);
\draw (0,{\altu+0.2}) node{$z$}; 
\end{scope}

\begin{scope}[scale=1.1]
\draw [|-|] (-{\base/4+0.1},-\altu) -- ({\base/4-0.1},-\altu);
\draw (0,{-\altu-0.2}) node{$z$}; 
\end{scope}

\begin{scope}[scale=0.9]
\draw [|-|] ({-\base/4-0.2},\altu) -- ({-\base/2},\altu) -- ({-\base/2},-\altu) -- ({-\base/4-0.2},-\altu);
\draw ({-\base/2+0.2},0) node{$\alpha$}; 
\end{scope}

\begin{scope}[scale=0.9]
\draw [|-|] ({\base/4+0.2},\altu) -- ({\base/2},\altu) -- ({\base/2},-\altu) -- ({\base/4+0.2},-\altu);
\draw ({\base/2-0.2},0) node{$\gamma$}; 
\end{scope}



\end{scope}

\end{tikzpicture}    
\]

\noindent and find invertible matrices $\Delta_{az}$, $\Delta_{zc}$, $\Omega_{\alpha z}$, $\Omega_{z \gamma}$ with support in the regions indicated by the respective subindices such that the boundary observables
\begin{align*}
& \sigma_{\partial AB} = \Omega_{z \gamma} \Delta_{az} \quad \quad \sigma_{\partial BC} = \Delta_{z c} \Omega_{\alpha z}\\[2mm]
& \sigma_{\partial ABC} = \Delta_{z c} \Delta_{az} \quad \quad \sigma_{\partial B} = \Omega_{z \gamma} \Omega_{\alpha z}
\end{align*}
approximate the boundary states 
\begin{align*}
& \| \rho_{\partial \RR}^{1/2} \sigma_{\partial \RR}^{-1} \rho_{\partial \RR}^{1/2} - \mathbbm{1} \|  \leq  \varepsilon \quad \mbox{for each } \RR \in \{ ABC, AB, BC \},\\
& \| \rho_{\partial B}^{-1/2} \sigma_{\partial B} \rho_{\partial B}^{-1/2} - \mathbbm{1} \|  \leq  \varepsilon.
\end{align*}
\end{Defi}

The approximate factorization of the boundary states implies a small norm of the overlaps of ground space projections.

\begin{Theo}[{\cite[Theorem 10]{KaLuPeGa19}}]\label{Theo:approximateFactInjectivePEPS}
If the boundary states are $\varepsilon$-approximately factorizable for some $\varepsilon \leq 1$, then 
\[ \| P_{AB}P_{BC} - P_{ABC} \| < 8 \varepsilon\,. \]
\end{Theo}

\subsubsection{Approximate factorization for locally non injective PEPS}
\label{sec:non-injective-peps}
In \cite{KaLuPeGa19}, the approximate factorization condition was extended to non-injective PEPS satisfying what is known as the \emph{pulling through condition}, which holds in the case of $G$-injective and MPO-injective PEPS. Unfortunately the PEPS representing the thermofield double state $\ket*{\rho_\beta^{1/2}}$ will neither be injective nor satisfy such a condition. At the same time, it will turn out to have some stronger property which will make up for the lack of it: it can be well approximated by a tensor product operator.
We will now present the necessary modifications to the results of \cite{KaLuPeGa19} required to treat this case.


There are three geometrical cases we need to consider in our decomposition of the torus $\mathbb{Z}_{N}\times \mathbb{Z}_{N}$ into sub-regions: two cylinders to cover the torus, two rectangles to cover a cylinder, and two rectangles to cover a rectangle. The following theorem is an adaptation of  \cite[Theorem 10]{KaLuPeGa19} that covers each of these three cases.


\begin{Theo}\label{theo:non-injective-approx-fact}
Let $A,B,C$ be three disjoint regions of $\Lambda$ such that $A$ and $C$ do not share mutually contractible boundary indices. Let us moreover assume that the (boundary) virtual indices of $ABC$, $AB$, $BC$ and $B$ can be arranged into four sets $a ,c, \alpha, \gamma$ as in the next picture so that
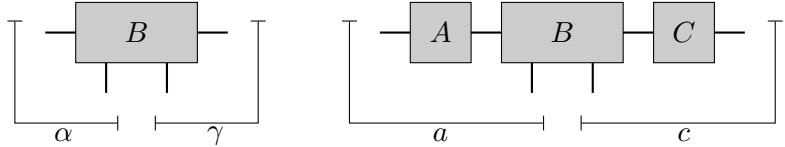
\[
\begin{tikzpicture}[equation,scale=0.8, decoration={markings, mark=at position 0.65 with {\arrow[very thick, scale=1.5]{>}}}]

\begin{scope}

\path [-, thick] (0,0.5) edge node {} (6,0.5);
\path [-, thick] (2.5,0.5) edge node {} (2.5,-0.5);
\path [-, thick] (3.5,0.5) edge node {} (3.5,-0.5);

\draw[fill=black!20] (0.5,0) rectangle (1.5,1);
\draw[fill=black!20] (2,0) rectangle (4,1);
\draw[fill=black!20] (4.5,0) rectangle (5.5,1);

\draw node at (1,0.5) {$A$};
\draw node at (3,0.5) {$B$};
\draw node at (5,0.5) {$C$};

\draw [|-|] (-0.5,0.7) -- (-0.5,-1) -- (2.7,-1);
\draw [|-|] (3.3,-1) -- (6.5,-1) -- (6.5,0.7);

\draw node at (1,-1.2) {$a$};
\draw node at (5,-1.2) {$c$};
\end{scope}

\begin{scope}[xshift=-5cm]
\path [-, thick] (-0.5,0.5) edge node {} (2.5,0.5);
\path [-, thick] (0.5,0.5) edge node {} (0.5,-0.5);
\path [-, thick] (1.5,0.5) edge node {} (1.5,-0.5);
\draw[fill=black!20] (0,0) rectangle (2,1);
\draw node at (1,0.5) {$B$};

\draw [|-|] (-1,0.7) -- (-1,-1)  -- (0.7,-1);
\draw [|-|] (1.3,-1) -- (3,-1)  -- (3,0.7);
\draw node at (-0.2,-1.2) {$\alpha$};
\draw node at (2.3,-1.2) {$\gamma$};
\end{scope}

\end{tikzpicture}
\]

\begin{align*}
    & \partial A \setminus \partial B \subset a \quad , \quad \partial C \setminus \partial B \subset c  \quad , \quad   \partial A \cap \partial B \subset \alpha \quad , \quad \partial C \cap \partial B \subset \gamma \\[2mm]
    & \quad  \partial ABC = a c  \quad , \quad \partial AB = a \gamma \quad , \quad \partial BC = \alpha c \quad , \quad \partial B = \alpha \gamma\,.
\end{align*}
Let us also assume that the orthogonal projections $J_{\partial \mathcal{R}}$ onto the support of $\rho_{\partial \mathcal{R}}$ admit a factorization in terms of projections $J_{a}, J_{c}, J_{\alpha}, J_{\gamma}$ (subindices indicate their corresponding support), namely
\[ J_{\partial ABC} = J_{a} \otimes J_{c} \quad , \quad J_{\partial AB} = J_{a} \otimes J_{\gamma} \quad , \quad J_{\partial BC} = J_{\alpha} \otimes J_{c} \quad , \quad J_{\partial B} = J_{\alpha} \otimes J_{\gamma}\,,\]
and there also exist positive semi-definite operators $\sigma_{a}, \sigma_{c}, \sigma_{\alpha}, \sigma_{\gamma}$ with full-rank on $J_{\alpha}, J_{a}, J_{\gamma}, J_{c}$ such that 
\[ \sigma_{\partial ABC}:= \sigma_{a} \otimes \sigma_{c} \quad , \quad \sigma_{\partial AB}:= \sigma_{a} \otimes \sigma_{\gamma} \quad , \quad \sigma_{\partial BC}:= \sigma_{\alpha} \otimes \sigma_{c} \quad , \quad \sigma_{\partial B}:= \sigma_{\alpha} \otimes \sigma_{\gamma} \]
satisfy for some $0 \le \varepsilon \leq 1$
\begin{align*} 
& \|  J_{\partial \mathcal{R}} - \rho_{\partial \mathcal{R}}^{1/2} \sigma_{\partial \mathcal{R}}^{-1}  \rho_{\partial \mathcal{R}}^{1/2}\| < \varepsilon \quad , \quad \mathcal{R} \in \{ ABC, AB, BC\}\,, \\[2mm]
& \| J_{\partial B}  - \rho_{\partial B}^{-1/2} \sigma_{\partial B} \rho_{\partial B}^{-1/2} \|\, < \varepsilon.
\end{align*}
Here, inverses are taken on the corresponding support. Then,
\[ \| P_{AB} P_{BC} - P_{ABC} \| \, \leq \, 8\varepsilon  \,.\]
\end{Theo}

In case the region $\mathcal{R}$ consists of the whole lattice (e.g. torus) then the sets $a$ and $c$ would be empty. Consequently, $\sigma_{a}$ and $\sigma_{c}$ would be simply scalars.

\begin{proof}
Let us define the approximate projections
\[ Q_{\RR} := V_{\RR} \sigma_{\partial \RR}^{-1} V_{\RR}^{\dagger} \quad , \quad \RR \in \{ BC, AB, ABC \}\,.  \]
Note that $V_{\RR}\rho_{\RR}^{-1/2}$ is a partial isometry from  the support of $\rho_{\partial \RR}$ to $\Im (V_{\RR})$, since
\[
(V_{\RR} \rho_{\partial \RR}^{-1/2} )^\dag V_{\RR} \rho_{\partial \RR}^{-1/2}  =
 \rho_{\partial \RR}^{-1/2} V_{\RR}^\dag V_{\RR} \rho_{\partial \RR}^{-1/2} =  
 \rho_{\partial \RR}^{-1/2} \, \rho_{\partial \RR}\, \rho_{\partial \RR}^{-1/2} = J_{\partial \RR},
\]
and 
\[
V_{\RR}\rho_{\partial \RR}^{-1/2} (V_{\RR} \rho_{\partial \RR}^{-1/2} )^\dag = V_{\RR} \rho_{\partial \RR}^{-1} V_{\RR}^\dag = P_{\RR},
\]
where the last equality is a consequence of the fact that $V_{\RR} \rho_{\partial \RR}^{-1} V_{\RR}^\dag$ is a self-adjoint projection whose image is exactly $\Im (V_\RR)$.
As a consequence
\begin{equation}\label{equa:MartingaleCylinderAux0}
\begin{split}
 \| P_{\RR} - Q_{\RR}\| & =  \|  V_{\RR} \rho_{\partial \RR}^{-1} V_{\RR}^\dag  -  V_{\RR} \sigma_{\partial \RR}^{-1} V_{\RR}^{\dagger}\|   \\[2mm]
 & =\|V_{\RR} \rho_{\partial \RR}^{-1/2} ( J_{\partial \RR} - \rho_{\partial \RR}^{1/2} \sigma_{\partial \RR}^{-1} \rho_{\partial \RR}^{1/2}) \rho_{\partial \RR}^{-1/2} V_{\RR}^\dag \|\\[2mm] 
 &= \| J_{\partial \RR} - \rho_{\partial \RR}^{1/2} \sigma_{\partial \RR}^{-1} \rho_{\partial \RR}^{1/2} \| < \varepsilon\,. 
\end{split} 
\end{equation}

\noindent We are going to denote by $V_{C \rightarrow B}$ the tensor obtained from $V_{C}$ by taking all input indices that connect with $B$ into output indices, so that
\[ V_{ABC} = V_{AB} V_{C \rightarrow B} \quad \text{and} \quad V_{BC} = V_{B} \, V_{C \rightarrow B} \,. \]
Analogously, we define $V_{A \rightarrow B}$ satisfying
\[ V_{ABC} = V_{BC} V_{A \rightarrow B} \quad \text{and} \quad V_{AB} = V_{B} V_{A \rightarrow B}\,. \]
Then, we can rewrite
\begin{align*}
Q_{ABC} & = V_{ABC} \sigma_{\partial ABC}^{-1} V_{ABC}^{\dagger} =    V_{ABC} \sigma_{a}^{-1}  \sigma_{c}^{-1}  V_{ABC}^{\dagger} \\[2mm]
& = V_{AB} V_{C \rightarrow B} \sigma_{a}^{-1}  \sigma_{c}^{-1} V_{A \rightarrow B}^{\dagger} V_{BC}^{\dagger}\\[2mm]
& = V_{AB} \sigma_{a}^{-1}   V_{C \rightarrow B}  \, V_{A \rightarrow B}^{\dagger} \, \sigma_{c}^{-1} V_{BC}^{\dagger}\,.
\end{align*}
At this point, we can use the local structure of the projections to write $V_{AB} = V_{AB} J_{\partial AB} = V_{AB} J_{\partial AB} J_{\gamma} = V_{AB} J_{\gamma} = V_{AB} \sigma_{\gamma} \sigma_{\gamma}^{-1}$. Analogously, $V_{BC}= V_{BC} \sigma_{\alpha} \sigma_{\alpha}^{-1}$. Inserting both identities above, we can rewrite
\begin{align*}
Q_{ABC} & = V_{AB} \sigma_{a}^{-1}  \sigma_{\gamma}^{-1}  \, V_{A \rightarrow B}^{\dagger} \sigma_{\gamma} \sigma_{\alpha}  \, V_{C \rightarrow B} \, \sigma_{\alpha}^{-1}  \sigma_{c}^{-1} V_{BC}^{\dagger}\\[2mm]
& = V_{AB} \, \sigma_{\partial AB}^{-1} \, V_{A \rightarrow B}^{\dagger} \, \sigma_{\partial B}  \, V_{C \rightarrow B} \, \sigma_{\partial BC}^{-1} \, V_{BC}^{\dagger}\,.
\end{align*}
Similarly, we handle
\begin{align*}
Q_{AB} Q_{BC} & = V_{AB} \sigma_{\partial AB}^{-1} V_{AB}^\dagger \, V_{BC} \sigma_{\partial BC}^{-1} V_{BC}^\dagger \\[2mm]
& = V_{AB} \sigma_{\partial AB}^{-1} V_{A \rightarrow B}^{\dagger} \, V_{B}^\dagger \, V_{B} \, V_{C \rightarrow B}\, \sigma_{\partial BC}^{-1} V_{BC}^\dagger \\[2mm]
& = V_{AB} \sigma_{\partial AB}^{-1} V_{A \rightarrow B}^{\dagger} \, \rho_{\partial B} \, V_{C \rightarrow B}\, \sigma_{\partial BC}^{-1} V_{BC}^\dagger \,.
\end{align*}

\noindent To compare the expressions for $Q_{ABC}$ and $Q_{AB} \, Q_{BC}$ we introduce
\begin{align*} 
\Delta_{AB} & := V_{AB} \, \sigma_{\partial AB}^{-1} \, V_{A \rightarrow B}^{\dagger}  \, \rho_{\partial B}^{1/2}  \,,\\[2mm] 
\Delta_{BC} & := V_{BC} \, \sigma_{\partial BC}^{-1} \, V_{C \rightarrow B}^{\dagger}  \, \rho_{\partial B}^{1/2} \,.
\end{align*}
It is easy to check that
\begin{equation}\label{equa:MartingaleCylinderAux1}
Q_{ABC}  - Q_{AB} Q_{BC}  = \Delta_{AB}\big( \, (\rho_{\partial B}^{-1/2} \sigma_{\partial B} \rho_{\partial B}^{-1/2} \, - J_{\partial B} \,)  \, \big) \Delta_{BC}^{\dagger}\,\,. 
\end{equation}

\noindent Since $Q_{BC}^2 = \Delta_{BC} \Delta_{BC}^{\dagger}$, we can apply \eqref{equa:MartingaleCylinderAux0} to estimate 
\[ \| \Delta_{BC} \|^{2}  = \| \Delta_{BC}\Delta_{BC}^\dag \| = \| Q_{BC}^2\| \leq \| Q_{BC}\|^2 \leq (1 + \varepsilon)^2\,.  \]
Analogously $Q_{AB}^2 = \Delta_{AB} \Delta_{AB}^{\dagger}$ and so $\| \Delta_{AB}\| \leq 1 + \varepsilon$. Combining these inequalities with \eqref{equa:MartingaleCylinderAux1}, we get
\begin{multline}\label{equa:MartingaleCylinderAux1l} 
\| Q_{ABC} - Q_{AB} Q_{BC} \| \leq (1+ \varepsilon)^2 \, \| J_{\partial B}  - (\rho_{\partial B}^{-1/2} \sigma_{\partial B} \rho_{\partial B}^{-1/2}) \, \|\,  \leq \varepsilon(1+\varepsilon)^2\,.
\end{multline}

\noindent Finally, we combine the previous inequality with \eqref{equa:MartingaleCylinderAux0} to conclude
\begin{align*} 
\| P_{ABC} - P_{AB} P_{BC}\| \, &  \leq \, \|P_{ABC} - Q_{ABC} \| + \| Q_{ABC} - Q_{AB} Q_{BC} \|\\
& \quad + \| P_{AB} P_{BC} - Q_{AB} Q_{BC} \|  \\[2mm]
& \leq \varepsilon + \varepsilon (1+\varepsilon)^2 + \| P_{AB} - Q_{AB}\| \, \| P_{BC}\|\\ 
& \quad \quad+ \|Q_{AB} \| \, \| P_{BC} - Q_{BC}\|\\[2mm]
& \leq \varepsilon + \varepsilon (1+\varepsilon)^2 + \varepsilon + (1+\varepsilon) \varepsilon  = \, (\varepsilon^2 +3\varepsilon+4) \varepsilon \leq 8 \varepsilon\,,
\end{align*}
which gives the result.
\end{proof}

\subsubsection{Gauge invariance of the approximate factorization condition}\label{sec:gauge-invariance}
An interesting observation omitted in \cite{KaLuPeGa19}, is that the property of $\varepsilon$-approximately factorization is gauge invariant if the transformation does not change the support of the boundary state. Indeed, for every vertex $x \in \Lambda$ and every edge $e \in E$ incident to $x$, let us fix an invertible matrix $\mathcal{G}(e,x) \in \mathbb{C}^{D} \otimes \mathbb{C}^{D}$. We assume that for every edge $e$ with vertices $x,y$ we have
\begin{equation}\label{equa:gaugeCancellation} 
\mathcal{G}(e,x) = \mathcal{G}(e,y)^{-1}\,. 
\end{equation}
\noindent We will simply write $\mathcal{G}(e)$ when the site is clear from the context. Let us assume that we have two PEPS related via this gauge, namely for every site $x \in \Lambda$ we have that the local tensors $\widetilde{V}_{x}$ and $V_{x}$ are related via (see Figure~\ref{fig:gauge})
\[ V_{x} =  \widetilde{V}_{x} \, \circ  \, \mathcal{G}_{x} \quad \quad \mbox{ where } \quad  \mathcal{G}_{x}:=\bigotimes_{e \in \partial x}{\mathcal{G}(e,x)} \,. \]

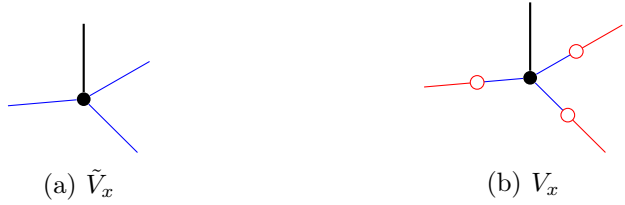
\begin{figure}[hbt]
    \centering
    \begin{subfigure}[t]{0.4\textwidth}
    \centering
        \begin{tikzpicture}
            \node[fill=black, circle, scale=0.5] (A) at (0,0) {};
            \draw[-, thick] (A) -- +(0,1);
            \draw[blue] (A) -- +(30:1);
            \draw[blue] (A) -- +(-45:1);
            \draw[blue] (A) -- +(185:1);
        \end{tikzpicture}
    \caption{$\tilde V_x$}
    \end{subfigure}
    ~
    \begin{subfigure}[t]{0.4\textwidth}
        \centering
        \begin{tikzpicture}
            \node[fill=black, circle, scale=0.5] (A) at (0,0) {};
            \draw[-, thick] (A) -- +(0,1);
            \draw[blue] (A) -- +(30:0.7) node[circle,fill=white, draw=red, scale=0.5] (B) {};
            \draw[red]  (B) -- +(30:0.7);
            \draw[blue] (A) -- +(-45:0.7)
                node[circle,fill=white, draw=red, scale=0.5] (C) {};
            \draw[red]  (C) -- +(-45:0.7);
            \draw[blue] (A) -- +(185:0.7)
                node[circle,fill=white, draw=red, scale=0.5] (D) {};
            \draw[red]  (D) -- +(185:0.7);
        \end{tikzpicture}
    \caption{$V_x$}
    \end{subfigure}
    \caption{Tensors with (right) and without (left) gauge.}
    \label{fig:gauge}
\end{figure}

For a region $\RR \subset \Lambda$, when contracting indices to construct $V_{\RR}$ we have, as a consequence of \eqref{equa:gaugeCancellation}, that contracting inner edges of $\mathcal{R}$ cancel the gauge matrices. Thus $\widetilde{V}_{\RR}$ and $V_{\RR}$ are related via (see Figure~\ref{fig:gauge-2}):
\begin{equation}\label{equa:GaugeRegion} 
V_{\RR} =  \widetilde{V}_{\RR} \, \circ  \, \mathcal{G}_{\partial \RR} \quad \quad \mbox{ with } \quad  \mathcal{G}_{\partial \RR}:=\bigotimes_{e \in \partial \RR}{\mathcal{G}(e)} \, 
\end{equation}

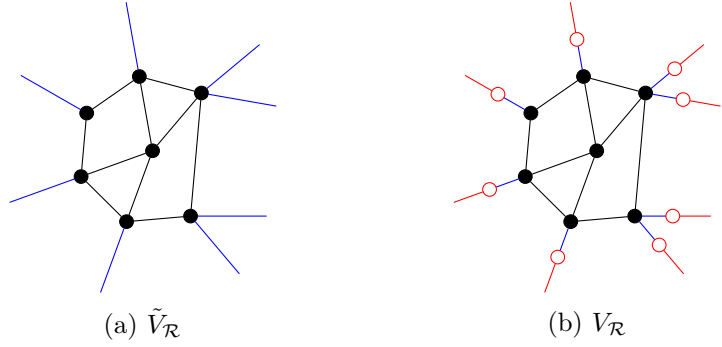
\begin{figure}[hbt]
    \centering
    \begin{subfigure}[t]{0.4\textwidth}
    \centering
        \begin{tikzpicture}[tensor/.style={circle, draw, fill=black, scale=0.5}]
        \node[tensor] (A) at (0,0) {};
        \foreach \c/\a in { B/50, C/100, D/150, E/200, F/250, G/300}
            \node[tensor] (\c) at (\a:1) {};
        \foreach \c in {B, C, E, F}
            \draw[-] (A) -- (\c);
        \draw[-] (B) -- (C) -- (D) -- (E) -- (F) -- (G) -- (B);
            
        \foreach \c/\a in {B/-10, B/40, C/100, D/150, E/200, F/250, G/0, G/-50}
            \draw[blue] (\c) -- +(\a:1);
        \end{tikzpicture}
    \caption{$\tilde V_{\RR}$}
    \end{subfigure}
    ~
    \begin{subfigure}[t]{0.4\textwidth}
        \centering
        \begin{tikzpicture}[tensor/.style={circle, draw, fill=black, scale=0.5}]
        \node[tensor] (A) at (0,0) {};
        \foreach \c/\a in { B/50, C/100, D/150, E/200, F/250, G/300}
            \node[tensor] (\c) at (\a:1) {};
        \foreach \c in {B, C, E, F}
            \draw[-] (A) -- (\c);
        \draw[-] (B) -- (C) -- (D) -- (E) -- (F) -- (G) -- (B);
            
        \foreach \c/\a in {B/-10, B/40, C/100, D/150, E/200, F/250, G/0, G/-50}{
            \draw[blue] (\c) -- +(\a:0.5)
            node[circle,fill=white, draw=red, scale=0.5] (\c w) {};  
            \draw[red]  (\c w) -- +(\a:0.5);
        }
        \end{tikzpicture}
    \caption{$V_{\RR}$}
    \end{subfigure}
    \caption{Tensor network with (right) and without (left) gauge (physical indices are not shown).}
    \label{fig:gauge-2}
\end{figure}

\noindent The boundary state after the change of gauge is transformed as
\[ \rho_{\partial \RR} = \mathcal{G}_{\partial \RR}^{\dagger} \widetilde{\rho}_{\partial \RR} \mathcal{G}_{\partial \RR}\,, \]
where $\widetilde{\rho}_{\partial \RR} = \widetilde{V}_{\RR}^{\dagger} \widetilde{V}_{\RR}$.

\begin{Prop}\label{prop:gauge-invariance}
    Assume that $\comm{J_{\partial \RR}}{\mathcal{G}_{\partial \RR}}=0$. 
    Let $\widetilde \sigma_{\partial \RR}$ supported on $J_{\partial \RR}$, and define
    \[
        {\sigma}_{\partial \RR} = \mathcal{G}_{\partial \RR}^{\dagger}  \, \widetilde{\sigma}_{\partial \RR} \, \mathcal{G}_{\partial \RR}.
    \]
    Then, ${\sigma}_{\partial \RR}$ is also supported on $J_{\partial \RR}$ and it holds that
    \[ \left\| \rho_{\partial \RR}^{1/2} \sigma_{\partial \RR}^{-1} \rho_{\partial \RR}^{1/2} - J_{\partial R} \right\| \, = \, \left\| \widetilde{\rho}_{\partial \RR}^{1/2} \widetilde{\sigma}_{\partial \RR}^{-1}  \widetilde{\rho}_{\partial \RR}^{1/2} - J_{\partial R} \right\|\,. \]
\end{Prop}
\begin{proof}
Since $J_{\partial \RR}$ and $\mathcal G_{\partial \RR}$ commute, we have that $\rho_{\partial \RR}$, $\widetilde{\rho}_{\partial \RR}$, $\sigma_{\partial \RR}$ and $\widetilde{\sigma}_{\partial \RR}$ all have the same support, namely $J_{\partial \RR}$, and so
\[ {\sigma}_{\partial \RR}^{-1} = \mathcal{G}^{-1}_{\partial \RR} \, \widetilde{\sigma}_{\partial \RR}^{-1} \, \mathcal{G}_{\partial \RR}^{\dagger \, -1}\,. \]
Therefore, if $P_{\RR}$ denotes the orthogonal projection onto $\operatorname{Im}(V_{\RR}) = \operatorname{Im}(\widetilde{V}_{\RR})$, then we can write
\[  P_{\RR}  \, = \, V_{\RR} \, \rho_{\partial \RR}^{-1} \, V_{\RR}^{\dagger}  \, = \, \widetilde{V}_{\RR} \, \widetilde{\rho}_{\partial \RR}^{\,\, -1} \, \widetilde{V}_{\RR}^{\dagger}\,. \]
From the definition of $\widetilde{\sigma}_{\partial \RR}$, we similarly see that
\[
V_{\RR} \sigma_{\partial \RR}^{-1} \, V_{\RR}^{\dagger}  \, = \, \widetilde{V}_{\RR} \, \widetilde{\sigma}_{\partial \RR}^{\,\, -1} \, \widetilde{V}_{\RR}^{\dagger}.
\]
The statement then follows from the fact that
\begin{align*}
 W_{\RR}(\rho_{\partial \RR}^{1/2} \sigma_{\partial \RR}^{-1} \rho_{\partial \RR}^{1/2} - J_{\partial R})W_{\RR}^\dag 
& = V_{\RR} \sigma_{\partial \RR}^{-1} V_{\RR}^{\dagger} - V_{\RR} \rho_{\partial \RR}^{-1} V_{\RR} \\[2mm]
& = \widetilde{V}_{\RR} \widetilde{\sigma}_{\partial \RR}^{-1} \widetilde{V}_{\RR}^{\dagger} - \widetilde{V}_{\RR} \widetilde{\rho}_{\partial \RR}^{-1} \widetilde{V}_{\RR} \\[2mm]
&=  \widetilde{W}_{\RR}(\widetilde{\rho}_{\partial \RR}^{1/2} \widetilde{\sigma}_{\partial \RR}^{-1} \widetilde{\rho}_{\partial \RR}^{1/2} - J_{\partial R})\widetilde{W}_{\RR}^{\dagger} \,,
\end{align*}
where $W_{\RR} = V_{\RR} \rho_{\partial \RR}^{-1/2}$ and $\widetilde{W}_{\RR} = \widetilde{V}_{\RR} \widetilde{\rho}_{\partial \RR}^{-1/2}$ are isometries.
\end{proof}

\begin{Coro}
Let $A, B,C$ be three regions of $\Lambda$ as in the definition of approximate factorization (Definition \ref{def:approx-fact-injective}), and assume that $\comm{J_{\partial \RR}}{G_{\partial \RR}}=0$ for $\RR \in \{ABC, AB, BC, B\}$. If the PEPS generated by $\widetilde{V}_{x}$ is $\varepsilon$-approximately factorizable, then so does the PEPS generated by $V_{x}$.
\end{Coro}

\begin{proof}
Let us assume then that the PEPS with local tensors $\widetilde{V}_{x}$ is $\varepsilon$-approximately factorizable. 
Because of Proposition~\ref{prop:gauge-invariance}, it is sufficient to verify that ${\sigma}_{\partial \RR}$ satisfies the necessary locality properties.
If $\widetilde J_{\partial \RR}$ and $\widetilde{\sigma}_{\partial \RR}$ are product operators (as in Theorem~\ref{theo:non-injective-approx-fact}), then so are $J_{\partial \RR}$ and $\sigma_{\partial \RR}$, and there is nothing to prove.

Let us now consider the case in which $\widetilde{\sigma}_{\partial \RR}$ is not in a tensor product form (as in Definition~\ref{def:approx-fact-injective}).
Let $a,\alpha, z, \gamma, c$ be the regions dividing the boundaries  $\partial ABC = a z c$, $\partial AB = az \gamma$, $\partial BC = \alpha z c$, $\partial B = \alpha z \gamma$ and  let   $\widetilde{\Delta}_{az}, \widetilde{\Delta}_{zc}, \widetilde{\Omega}_{\alpha z}, \widetilde{\Omega}_{z \gamma}$ be the corresponding matrices. Note that the gauge matrices $\mathcal{G}_{\partial \RR}$ can be rearranged according to the boundary subregions, e.g.
\[ \mathcal{G}_{\partial ABC} = \mathcal{G}_{a zc} = \mathcal{G}_{az} \mathcal{G}_{c} = \mathcal{G}_{a} \mathcal{G}_{zc}\,. \]
If we define
\begin{equation}\label{equa:Approx_Fact_Gauge_Aux_1}
\begin{split}
& \Delta_{az} := \mathcal{G}_{a}^{\dagger} \, \widetilde{\Delta}_{az} \, \mathcal{G}_{az} , \quad
\Delta_{zc} := \mathcal{G}_{zc}^{\dagger} \, \widetilde{\Delta}_{zc} \, \mathcal{G}_{c}\,,\\
&\Omega_{\alpha z} := \mathcal{G}_{\alpha}^{\dagger} \, \widetilde{\Omega}_{\alpha z} \mathcal{G}_{\alpha z} , \quad 
\Omega_{z\gamma} := \mathcal{G}_{z \gamma}^{\dagger} \,  \widetilde{\Omega}_{z\gamma} \, \mathcal{G}_{\gamma}\,,
\end{split}
\end{equation}
then, we can directly check that $\sigma_{\partial \RR}$ satisfies
\begin{align*}
& \sigma_{\partial AB} := \Omega_{z \gamma} \Delta_{az}, 
& \sigma_{\partial BC} := \Delta_{z c} \Omega_{\alpha z},\\
& \sigma_{\partial ABC} := \Delta_{z c} \Delta_{az},
& \sigma_{\partial B} := \Omega_{z \gamma} \Omega_{\alpha z}.
\end{align*}
This finishes the proof.
\end{proof}



\section{PEPS description of the thermofield double}
\label{sec:peps-quantum-double}
\subsection{Quantum Double Models}
Let us begin by recalling the definition of the Quantum Double Models. They are defined on the lattice $\Lambda_{N}$ consisting of midpoints of the edges of the square lattice $\mathbb{Z}_N \times \mathbb{Z}_{N}$ (see Section~\ref{subsec:torusSettingAndGap}). Let us denote by $\mathcal{V}=\mathcal{V}_{N}$ the set of vertices, and by $\mathcal{E}=\mathcal{E}_{N}$ the set of edges of $\mathbb{Z}_N \times \mathbb{Z}_{N}$. Each edge is given an orientation: for simplicity, we will assume that all horizontal edges point to the left, while vertical edges point downwards.

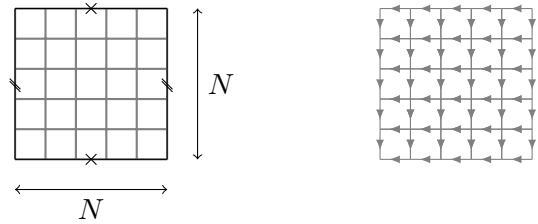
\[
\begin{tikzpicture}[equation, scale=0.4]

    \draw[step=1.0,gray,thick] (0,0) grid (5,5);
    \draw[postaction=torus horizontal] (0,0) -- (5,0);
    \draw[postaction=torus horizontal] (0,5) -- (5,5);
    \draw[postaction=torus vertical] (5,0) -- (5,5);
    \draw[postaction=torus vertical] (0,0) -- (0,5);

    \draw[<->] (0,-1) -- node[below]{$N$} (5,-1);
    \draw[<->] (6,0) --  node[right]{$N$} (6,5);

\begin{scope}[xshift=12cm, decoration={
    markings,
    mark=at position 0.7 with {\arrow{latex}}}]
\foreach \x in {1,...,5} 
\foreach \y in {1,...,5}{

\begin{scope}[xshift=\x cm, yshift=\y cm]
\draw[postaction={decorate}, gray] (0,0)  -- (-1,0); 
\draw[postaction={decorate}, gray] (-1,0)  -- (-1,-1); 
\draw[postaction={decorate}, gray] (0,0)  -- (0,-1); 
\draw[postaction={decorate}, gray] (0,-1)  -- (-1,-1); 
\end{scope}
}
\end{scope}

\end{tikzpicture}
\]

%
%
%

Let us fix an arbitrary finite group $G$ and denote let $\ell_2(G)$ be the complex finite dimensional Hilbert space with orthonormal basis given by $\{ \ket{g} \mid g \in G\}$. At each edge $e \in \mathcal{E}$ we have a local Hilbert space $\mathcal{H}_e$ and a space of observables $\mathcal{B}_e$ defined as
\[
\mathcal{H}_{e} = \ell_{2}(G) \quad \text{and} \quad  \mathcal{B}_{e} = \mathcal{B}(\mathcal{H}_{e}) = \mathcal{M}_{|G|}(\mathbb{C}).
\]
We will use the alternative notation $\mathcal{H}_{\Lambda} = \mathcal{H}_{\EE}$ and $\mathcal{B}_{\Lambda} = \mathcal{B}_{\EE}$. 

Given $g \in G$ we define operators on $\ell_2(G)$ by
\begin{equation}\label{eq:left-regular-repr} 
L^{g} := \sum_{h \in G} \dyad{gh}{h}.
\end{equation}
Then $g\mapsto L^g$ is a representation of the group $G$, known as the \emph{left regular} representation.

For each finite group $G$, the Quantum Double Model on $\Lambda$ is defined by a Hamiltonian $H_\Lambda^{\text{syst}}$ of the form
\[ H_\Lambda^{\text{syst}} = - \sum_{v \mbox{\tiny \,  vertex}}{A(v)} \,\, - \sum_{p \mbox{\tiny \,  plaquette }} B(p)\,;\]
where the terms $A(v)$ are \emph{star operators}, supported on the four incident edges of $v$, which we will denote as $\partial v$, while $B(p)$ are \emph{plaquette operators}, supported on the four edges forming the plaquette $p$. Both terms are projections and they commute, namely
\[ [A(v), A(v')] \,\, = \,\, [B(p), B(p')] \,\, = \,\, [A(v), B(p)] \,\, = \,\, 0  \]
for all vertices $v,v'$ and plaquettes $p,p'$.
We will now explicitly define these terms, and a straightforward calculation will show that they satisfy these properties.

Let $v$ be a vertex and $e$ an edge incident to $v$. For each $g \in G$ we define the operator $T^{g}(v,e)$ acting on $\mathcal{H}_{e}$ according to the orientation given to $e$ as 
\[
  \begin{tikzpicture}[equation, decoration={
    markings,
    mark=at position 0.65 with {\arrow{latex}}}]   
    \draw (-2,0.5) node {$T^{g}(v,e) \quad =$};  
    \draw[black, thick] (0,-0.1)  -- (0,0.1); 
    \draw[postaction={decorate}, black, thick] (-0.1,0)  -- (1,0); 
    \draw (0.5,0.7) node {\small $\displaystyle \sum_{h \in G} \dyad{gh}{h}$};  
    \draw (2.4,0.5) node {or};    
    \draw[black, thick] (4.1,-0.1)  -- (4.1,0.1);
    \draw[postaction={decorate}, black, thick] (5.1,0)  -- (4,0); 
    \draw (4.5,0.7) node {\small $\displaystyle \sum_{h \in G} \dyad{hg^{-1}}{h}$};  
\end{tikzpicture}
\,.
\]
In other words, the operator $T^{g}(v,e)$ acts on the basis vector of $\mathcal{H}_{e}$ by taking $h$ into $g h$ (resp. $hg^{-1}$) if the oriented edge $e$ points away from (resp. to) $v$. It is easily checked that
\begin{equation}\label{equa:TgRelations} 
T^{g}(v,e) \, T^{h}(v,e) = T^{gh}(v,e) \quad \mbox{ and }  \quad T^{g}(v,e)^{\dagger} = T^{g^{-1}}(v,e) 
\end{equation}
for every $g,h \in G$.

With this definition, the vertex operator $A(v)$ is given by
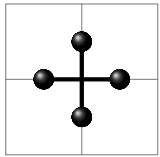
\[
  \begin{tikzpicture}[equation]   
  \draw[step=1.0,gray,thin] (-1,-1) grid (1,1);
  \draw[ultra thick] (-0.5,0) -- (0.5, 0);
  \draw[ultra thick] (0,-0.5) -- (0,0.5);
  \shade[ball color=black] (-0.5,0) circle (0.8ex);
  \shade[ball color=black] (0.5,0) circle (0.8ex);
  \shade[ball color=black] (0,-0.5) circle (0.8ex);
  \shade[ball color=black] (0,0.5) circle (0.8ex);
      \draw (6,0) node {$A(v) \, = \, \displaystyle \frac{1}{|G|} \, \sum\limits_{g \in G} \, \bigotimes\limits_{e \in \partial v} T^{g}(v,e)$};
\end{tikzpicture}
\,.
\]
Using  \eqref{equa:TgRelations}, it is easy to verify that $A(v)$ is a projection.

The plaquette operator $B(p)$ is defined as follows. Let us enumerate the four edges of $p$ as $e_{1}, e_{2}, e_{3}, e_{4}$ following counterclockwise order starting from the upper horizontal edge.  The plaquette operator on $p$ acts on $\otimes_{j=1}^{4}{\mathcal{H}_{e_{j}}}$ and is defined as the orthogonal projection $B(p)$ onto the subspace spanned by basis vectors of the form $\ket{g_{1} g_{2}g_{3}g_{4}}$ with $\sigma_{p}(g_{1})  \sigma_{p}(g_{2}) \sigma_{p}(g_{3}) \sigma_{p}(g_{4}) = 1$. Here, $\sigma_{p}(g)$ is equal to $g$ if the orientation of the corresponding edge agrees with the counter-clockwise labelling, otherwise it is equal to $g^{-1}$. 

For the orientation we have previously fixed, we can give an explicit expression in terms of the regular character $g \mapsto  \chi^{reg}(g) = \operatorname{Tr}(L^{g}) = |G| \, \delta_{g,1}\,$, namely 

\[
  \begin{tikzpicture}[equation]   
  \draw[step=1.0,gray,thin] (-0.5,-0.5) grid (1.5,1.5);
  \draw[ultra thick] (0,0) -- (1, 0) -- (1,1) -- (0,1) -- (0,0);
  \shade[ball color=black] (0.5,0) circle (0.8ex);
  \shade[ball color=black] (1,0.5) circle (0.8ex);
  \shade[ball color=black] (0.5,1) circle (0.8ex);
  \shade[ball color=black] (0,0.5) circle (0.8ex);
  
    \draw (7.5,0.5) node {$B(p) \, = \, \displaystyle\frac{1}{|G|} \, \sum_{g_{1}, g_{2}, g_{3}, g_{4} \in G} \chi^{reg}(g_{1} g_{2} g_{3}^{-1} g_{4}^{-1}) \,\, \bigotimes_{j=1}^{4} \dyad{g_{j}} $};
\end{tikzpicture}.
\]


\noindent Fixed $\beta < \infty$, the associated Gibbs state at inverse temperature $\beta$ is given by
\[ \rho_{\beta} = e^{- \beta H^{\text{syst}}_{\Lambda}}/\operatorname{Tr}(e^{- \beta H^{\text{syst}}_{\Lambda}})\,. \]
Since the star and plaquette operators commute, we can decompose
\begin{equation}\label{equa:GibbsQDMdecomposition1} 
e^{-\frac{\beta}{2} H^{\text{syst}}_{\Lambda}} =  \prod_{v \mbox{\tiny \,  vertex}}{e^{\frac{\beta}{2}A(v)}}  \prod_{p \mbox{\tiny \,  plaquette }} e^{\frac{\beta}{2} B(p)}\,. 
\end{equation}

\noindent Using the fact that $A(v)$ and $B(p)$ are projections, we can rewrite last expression as
\[
    e^{-\frac{\beta}{2} H^{\text{syst}}_{\Lambda}} = \prod_{v \mbox{\tiny \,  vertex}}
    \qty( \operatorname{Id} + (e^{\beta/2} - 1) \, A(v)) 
    \prod_{p \mbox{\tiny \,  plaquette }} 
    \qty( \operatorname{Id} + (e^{\beta/2} - 1) \, B(p)) .
\]

\subsection{PEPO elementary tensors}
We will now construct a PEPO representation of the interactions $A(v)$ and $B(p)$ (which will actually be a MPO representation). From this we will obtain a very similar PEPO representation for $\operatorname{Id} + (e^{\beta/2} - 1) \, A(v) $ and $ \operatorname{Id} + (e^{\beta/2} - 1) \, B(p)$. Combining the single-site tensors of each we will derive the PEPO representation of $ e^{-\frac{\beta}{2} H^{\text{syst}}_{\Lambda}} $ and the corresponding PEPS representation of $\ket*{\rho_\beta^{1/2}}$. We will use the notation
\[ \gamma_{\beta}:= \frac{e^{\beta}-1}{|G|} \]
along the section. It is also recommended to review the tensor notation that was introduced in Section \ref{sec:tensornotation}, as we will be using it extensively in the upcoming sections.
\subsubsection{Star  operator as a PEPO}
The star operator $A(v)$ admits an easy representation as a PEPO, namely
\[
A(v) \,\,\, = \,\,\, \frac{1}{|G|} \, \sum\limits_{g \in G} \, \bigotimes\limits_{e \in \partial v} T^{g}(e,v)  \,\,\, = \,\,\, \frac{1}{|G|}\, \cdot \,\, 
 \begin{tikzpicture}[equation, scale=0.8,decoration={
    markings,
    mark=at position 0.3 with {\arrow{latex}},
    mark=at position 0.9 with {\arrow{latex}}}, rotate around y=12]
    
       \begin{scope}[canvas is zx plane at y=0]
     \draw[red, thick] (0,0) circle (1);

        \end{scope}
        
        \draw[thick, postaction={decorate}] (1,-1,0) -- (1,0.8,0);
        \draw[thick, postaction={decorate}] (0,-1,1) -- (0,0.8,1);
        \draw[thick, postaction={decorate}] (-1,-1,0) -- (-1,0.8,0);
        \draw[thick,postaction={decorate}] (0,-1,-1) -- (0,0.8,-1);
        
     \shade[ball color=black] (1,0,0) circle (0.7ex);
     \shade[ball color=black] (0,0,1) circle (0.7ex);
     \shade[ball color=black] (-1,0,0) circle (0.7ex);
     \shade[ball color=black] (0,0,-1) circle (0.7ex);
    
    \draw[gray, thick]     (0,-1,-1.5)  -- (0,-1,1.5); 
    \draw[gray, thick]   (-1.5,-1,0)  -- (1.5,-1,0);
    
    
    \end{tikzpicture}
    \,,
\]

\noindent where each individual tensor consists of four indices: two physical indices colored in black and two virtual indices colored in red. All indices have the same dimension and are identified with $\ell_{2}(G)$. These tensors can be explicitly described using the notation introduced in Section \ref{sec:tensornotation} as follows:

\[
  \begin{tikzpicture}[equation, decoration={
    markings,
    mark=at position 0.2 with {\arrow{latex}},
    mark=at position 0.93 with {\arrow{latex}}}]
    \draw[postaction={decorate}, thick] (0,-0.8) -- (0,0.8);
  \draw [postaction={decorate}, red, thick,domain=45:135] plot ({cos(\x)}, {sin(\x)-1});
   \shade[ball color=black] (0,0) circle (1ex);
\draw[white] (0,1) circle (0.1);
    \end{tikzpicture}
    \,\,\, = \,\,\, \sum_{g \in G}  \quad T_{g} \,\,
     \begin{tikzpicture}[equation, decoration={
    markings,
    mark=at position 0.2 with {\arrow{latex}},
    mark=at position 0.93 with {\arrow{latex}}}]
\draw[white] (0,1) circle (0.1);
  \draw[postaction={decorate}, thick] (0,-0.8) -- (0,0.8);
  \shade[ball color=black] (0,0) circle (1ex);

\end{tikzpicture}
\,\, \otimes \,\,
    \begin{tikzpicture}[equation, decoration={
    markings,
    mark=at position 0.2 with {\arrow{latex}},
    mark=at position 0.93 with {\arrow{latex}}}]
\draw[postaction={decorate}, red, thick]
(2.5,-0.5) -- (2,0) -- (1.5,-0.5);
\draw (2,0.4) node {\small $\dyad{g}{g}$};
    \end{tikzpicture}\,.
\]

\noindent Since $A(v)$ is a projection,  
\[ e^{\frac{\beta}{2} A(v)} \,  = \, \operatorname{Id} + (e^{\beta/2} - 1) \, A(v) \, = \, \operatorname{Id} + \gamma_{\beta/2} \, |G| \, A(v) \,,  \] 
or equivalently
\begin{align*} 
e^{\frac{\beta}{2} A(v)} \, & = \,  \left( \, 1 + \gamma_{\beta/2} \right) \, \bigotimes\limits_{e \in \partial v} T^{1}(e,v)  + \, \left( \, \gamma_{\beta/2} \right) \, \sum\limits_{\substack{g \in G\\ g \neq 1}} \, \bigotimes\limits_{e \in \partial v} T^{g}(e,v) \,.
\end{align*}
Comparing with $A(v)$, we find a natural description 
as a PEPO
\[
e^{\frac{\beta}{2}A(v)} \, = \,\, 
 \begin{tikzpicture}[equation, scale=0.8,decoration={
    markings,
    mark=at position 0.4 with {\arrow{latex}},
    mark=at position 0.9 with {\arrow{latex}}}, rotate around y=12]
    
   \begin{scope}[canvas is zx plane at y=0]
     \draw[red, thick] (0,0) circle (1);

        \end{scope}
        
           \draw[red, fill=white] ({cos(-40)},0,{sin(-40}) circle (0.5ex);  
     \draw[red, fill=white] ({cos(40)},0,{sin(40)}) circle (0.5ex);   
     
         \draw[red, fill=white] ({cos(65)},0,{sin(65)}) circle (0.5ex);   
     \draw[red, fill=white] ({cos(110)},0,{sin(110)}) circle (0.5ex); 
     
     \draw[red, fill=white] ({cos(135)},0,{sin(135)}) circle (0.5ex);   
     \draw[red, fill=white] ({cos(215)},0,{sin(215)}) circle (0.5ex);

     \draw[red, fill=white] ({cos(240)},0,{sin(240)}) circle (0.5ex);   
     \draw[red, fill=white] ({cos(290)},0,{sin(290)}) circle (0.5ex);

        \draw[thick, postaction={decorate}] (1,-1,0) -- (1,0.8,0);
        \draw[thick, postaction={decorate}] (0,-1,1) -- (0,0.8,1);
        \draw[thick, postaction={decorate}] (-1,-1,0) -- (-1,0.8,0);
        \draw[thick,postaction={decorate}] (0,-1,-1) -- (0,0.8,-1);
        
     \shade[ball color=black] (1,0,0) circle (0.7ex);
     \shade[ball color=black] (0,0,1) circle (0.7ex);
     \shade[ball color=black] (-1,0,0) circle (0.7ex);
     \shade[ball color=black] (0,0,-1) circle (0.7ex);

    
    \draw[gray, thick]     (0,-1,-1.5)  -- (0,-1,1.5); 
    \draw[gray, thick]   (-1.5,-1,0)  -- (1.5,-1,0);
    
   \end{tikzpicture}
\]

\noindent where we are adding to the above representation for $A(v)$ suitable \emph{weights}
\begin{equation}\label{equa:redwhiteweight}
 \begin{tikzpicture}[equation, decoration={
    markings,
    mark=at position 0.2 with {\arrow{latex}},
    mark=at position 0.93 with {\arrow{latex}}}]

\draw[postaction={decorate}, red, thick] (0.8,0) -- (-0.8,0);
\draw[red, very thick, fill=white] (0,0) circle (1ex); 
\draw (0,0.5) node {$\weightS$};
\end{tikzpicture} 
\,\, = \,\, (1+ \gamma_{\beta/2})^{1/8} \,\,
\begin{tikzpicture}[equation, decoration={
    markings,
    mark=at position 0.2 with {\arrow{latex}},
    mark=at position 0.93 with {\arrow{latex}}}]
    
    \draw[postaction={decorate}, red, thick] (1.6,0) -- (0,0);
\draw (0.8,0.4) node {\small $\dyad{1}{1}$};

\end{tikzpicture}    
\,\, + \,\,  \left( \gamma_{\beta/2}\right)^{1/8} \,\,
\begin{tikzpicture}[equation, decoration={
    markings,
    mark=at position 0.2 with {\arrow{latex}},
    mark=at position 0.93 with {\arrow{latex}}}]
    
    \draw[postaction={decorate}, red, thick] (1.6,0) -- (0,0);
\draw (0.8,0.4) node {\small $\sum_{g \neq 1} \dyad{g}{g}$};

\end{tikzpicture}  \,.
\end{equation}

\noindent Therefore, we have a PEPO decomposition of $e^{\frac{\beta}{2} A(v)}$ into four identical tensors acting individually on each edge 

\[
  \begin{tikzpicture}[equation, decoration={
    markings,
    mark=at position 0.2 with {\arrow{latex}},
    mark=at position 0.93 with {\arrow{latex}}}]
    \draw[postaction={decorate}, thick] (0,-0.8) -- (0,0.8);
  \draw [postaction={decorate}, red, thick,domain=45:135] plot ({cos(\x)}, {sin(\x)-1});
   \shade[ball color=black] (0,0) circle (1ex);
\draw[red,  fill=white] ({cos(70)}, {sin(70)-1}) circle (0.5ex);
\draw[red,  fill=white] ({cos(110)}, {sin(110)-1}) circle (0.5ex); 
\draw[white] (0,1) circle (0.1);
    \end{tikzpicture}
    \,\,\, = \,\,\, \sum_{g \in G} \quad T^{g} \,\,
     \begin{tikzpicture}[equation, decoration={
    markings,
    mark=at position 0.2 with {\arrow{latex}},
    mark=at position 0.93 with {\arrow{latex}}}]
\draw[white] (0,1) circle (0.1);
  \draw[postaction={decorate}, thick] (0,-0.8) -- (0,0.8);
  \shade[ball color=black] (0,0) circle (1ex);

\end{tikzpicture}
\,\,\,\, \otimes \,\,
    \begin{tikzpicture}[equation, decoration={
    markings,
    mark=at position 0.2 with {\arrow{latex}},
    mark=at position 0.93 with {\arrow{latex}}}]
\draw[postaction={decorate}, red, thick]
(2.5,-0.5) -- (2,0) -- (1.5,-0.5);
\draw (2,0.4) node {\small $\weightS  \dyad{g}{g}  \weightS$};
    \end{tikzpicture}\,,
\]
where we can expand
\[ \weightS  \dyad{g}{g}  \weightS = \left( \delta_{g,1} + \gamma_{\beta/2} \right)^{1/4} \, \dyad{g}{g}\,. \]

\subsubsection{Plaquette operator as a PEPO}

The plaquette operator
\[  B(p) \,\, = \,\, \dfrac{1}{|G|} \, \displaystyle\sum_{g_{1}, g_{2}, g_{3}, g_{4} \in G} \chi^{reg}(g_{1} g_{2} g_{3}^{-1} g_{4}^{-1}) \,\, \bigotimes_{j=1}^{4} \dyad{g_{j}}  \]
admits an easy PEPO representation.  Using that
\[ \chi^{reg}(g_{1} g_{2} g_{3}^{-1} g_{4}^{-1}) =  \operatorname{Tr}(L^{g_{1}} L^{g_{2}}  L^{g_{3}^{-1}}  L^{g_{4}^{-1}})\,,
\]
we can decompose

\[
B(p) \, = \, \frac{1}{|G|}\, \cdot \,\, 
 \begin{tikzpicture}[equation, scale=0.8,decoration={
    markings,
    mark=at position 0.37 with {\arrow{latex}},
    mark=at position 0.9 with {\arrow{latex}}}, rotate around y=12]
    
       \begin{scope}[canvas is zx plane at y=0]
     \draw[blue, thick] (0,0) circle (1);

        \end{scope}
        
        \draw[thick, postaction={decorate}] (1,-1,0) -- (1,0.8,0);
        \draw[thick, postaction={decorate}] (0,-1,1) -- (0,0.8,1);
        \draw[thick, postaction={decorate}] (-1,-1,0) -- (-1,0.8,0);
        \draw[thick,postaction={decorate}] (0,-1,-1) -- (0,0.8,-1);
        
     \shade[ball color=black] (1,0,0) circle (0.7ex);
     \shade[ball color=black] (0,0,1) circle (0.7ex);
     \shade[ball color=black] (-1,0,0) circle (0.7ex);
     \shade[ball color=black] (0,0,-1) circle (0.7ex);
    
    \draw[gray, thick]     (-1,-1,-1)  -- (1,-1,-1)  -- (1,-1,1)  -- (-1,-1,1)  -- (-1,-1,-1);
    
    
    \end{tikzpicture}
\]

\noindent where
\[
  \begin{tikzpicture}[equation, decoration={
    markings,
    mark=at position 0.2 with {\arrow{latex}},
    mark=at position 0.93 with {\arrow{latex}}}]
    \draw[postaction={decorate}, thick] (0,-0.8) -- (0,0.8);
  \draw [postaction={decorate}, blue, thick,domain=45:135] plot ({cos(\x)}, {sin(\x)-1});
   \shade[ball color=black] (0,0) circle (1ex);
\draw[white] (0,1) circle (0.1);
    \end{tikzpicture}
    \,\,\, = \,\,\, \sum_{g \in G}  \quad \dyad{g}{g} \,\,
     \begin{tikzpicture}[equation, decoration={
    markings,
    mark=at position 0.2 with {\arrow{latex}},
    mark=at position 0.93 with {\arrow{latex}}}]
\draw[white] (0,1) circle (0.1);
  \draw[postaction={decorate}, thick] (0,-0.8) -- (0,0.8);
  \shade[ball color=black] (0,0) circle (1ex);

\end{tikzpicture}
\,\, \otimes \,\,
    \begin{tikzpicture}[equation, decoration={
    markings,
    mark=at position 0.2 with {\arrow{latex}},
    mark=at position 0.93 with {\arrow{latex}}}]
\draw[postaction={decorate}, blue, thick]
(2.5,-0.5) -- (2,0) -- (1.5,-0.5);
\draw (2.3,0.4) node {$L^{g^{\pm}}$};
    \end{tikzpicture}
\]
and $L^{g^{\pm}}$ is a shorthand for $L^{\sigma_p(g)}$: it is either $L^g$ or $L^{g^{-1}} = (L^g)^\dag$ depending on whether the orientation of the edge agrees with the counter-clockwise orientation of the plaquette ($L^g$), or is the opposite ($L^{g^{-1}}$).

\noindent As in the case of the star operator, $B(p)$ is a projection so 
\begin{align*} 
e^{\frac{\beta}{2} B(p)} & = \operatorname{Id} + (e^{\beta/2} - 1) \, B(p)\\[2mm]
& = \sum_{g_{1}, g_{2}, g_{3}, g_{4} \in G} \left( 1 + \, \gamma_{\beta/2} \,\,  \chi^{reg}(g_{1} g_{2} g_{3}^{-1} g_{4}^{-1}) \right)\,\, \bigotimes_{j=1}^{4} \dyad{g_{j}}
\end{align*}

\noindent Recall that left regular representation is in general not irreducible \cite{James}. Let us denote
\begin{equation}\label{eq:p1-projection}
    P_1 = \frac{1}{|G|} \sum_{g\in G} L_g.
\end{equation}
By direct calculation, one can verify that $P_1$ is an orthogonal projection on a subspace of dimension 1 (since $\Tr P_1 =1$). Moreover, since $L_g P_1 = P_1$ for every $g\in G$, the regular representation acts trivially on this subspace ($P_1$ is the orthogonal projection onto $V_{1}$ the unique trivial irreducible sub-representation of the regular representation).
As a consequence of this, we get that for every $g \in G$
\[ \Tr(P_1 L^g) = \Tr P_1 = 1 .\]

\noindent If we denote $P_{0}:= P_{1}^{\perp} = \mathbbm{1} - P_{1}$, then we have for every $g \in G$
\begin{equation}\label{equa:decompLgIrrepProjectors}
    L^{g} = P_{1}L^{g} P_{1} + P_{0} L^{g}P_{0} = P_{1} + P_{0}L^{g} P_{0}\,.
\end{equation}
and so
\begin{equation}\label{equa:RegularCharDecompProjectors}
\begin{split} 
1 + (\gamma_{\beta/2}) \, \chi^{reg}(g_{1} g_{2} g_{3}^{-1} g_{4}^{-1})& =1 + (\gamma_{\beta/2}) \, \chi^{reg}(g_{4} g_{3} g_{2}^{-1} g_{1}^{-1})\\[2mm]
&= 1 + (\gamma_{\beta/2}) \, \operatorname{Tr}(L^{g_{4}} L^{g_{3}}  L^{g_{2}^{-1}}  L^{g_{1}^{-1}})\\[2mm]
& = \left( 1 + \gamma_{\beta/2} \right) \operatorname{Tr}(P_{1} \, L^{g_{4}} L^{g_{3}}  L^{g_{2}^{-1}}  L^{g_{1}^{-1}})\\[2mm]
& \quad + \, (\gamma_{\beta/2}) \, \operatorname{Tr}( \, P_{0} \, L^{g_{4}} L^{g_{3}}  L^{g_{2}^{-1}}  L^{g_{1}^{-1}})\,.
\end{split}
\end{equation}
Comparing with $B(p)$, we have then the following decomposition

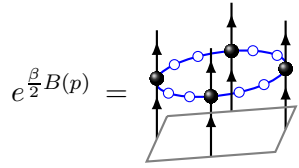
\[
e^{\frac{\beta}{2}B(p)} \, = \,\, 
 \begin{tikzpicture}[equation, scale=0.8,decoration={
    markings,
    mark=at position 0.4 with {\arrow{latex}},
    mark=at position 0.9 with {\arrow{latex}}}, rotate around y=12]
    
   \begin{scope}[canvas is zx plane at y=0]
     \draw[blue, thick] (0,0) circle (1);

        \end{scope}
        
           \draw[blue, fill=white] ({cos(-40)},0,{sin(-40}) circle (0.5ex);  
     \draw[blue, fill=white] ({cos(40)},0,{sin(40)}) circle (0.5ex);   
     
         \draw[blue, fill=white] ({cos(65)},0,{sin(65)}) circle (0.5ex);   
     \draw[blue, fill=white] ({cos(110)},0,{sin(110)}) circle (0.5ex); 
     
     \draw[blue, fill=white] ({cos(135)},0,{sin(135)}) circle (0.5ex);   
     \draw[blue, fill=white] ({cos(215)},0,{sin(215)}) circle (0.5ex);

     \draw[blue, fill=white] ({cos(240)},0,{sin(240)}) circle (0.5ex);   
     \draw[blue, fill=white] ({cos(290)},0,{sin(290)}) circle (0.5ex);

        \draw[thick, postaction={decorate}] (1,-1,0) -- (1,0.8,0);
        \draw[thick, postaction={decorate}] (0,-1,1) -- (0,0.8,1);
        \draw[thick, postaction={decorate}] (-1,-1,0) -- (-1,0.8,0);
        \draw[thick,postaction={decorate}] (0,-1,-1) -- (0,0.8,-1);
        
     \shade[ball color=black] (1,0,0) circle (0.7ex);
     \shade[ball color=black] (0,0,1) circle (0.7ex);
     \shade[ball color=black] (-1,0,0) circle (0.7ex);
     \shade[ball color=black] (0,0,-1) circle (0.7ex);

    
     \draw[gray, thick]     (-1,-1,-1)  -- (1,-1,-1)  -- (1,-1,1)  -- (-1,-1,1)  -- (-1,-1,-1);
    
   \end{tikzpicture}
\]
\noindent where we are adding to the above representation for $B(p)$ suitable weights

\begin{equation}\label{equa:bluewhiteweight}
 \begin{tikzpicture}[equation, decoration={
    markings,
    mark=at position 0.2 with {\arrow{latex}},
    mark=at position 0.93 with {\arrow{latex}}}]

\draw[postaction={decorate}, blue, thick] (0.8,0) -- (-0.8,0);
\draw[blue, very thick, fill=white] (0,0) circle (1ex); 
\draw (0,0.5) node {$\weightP$};
\end{tikzpicture} 
\,\, = \,\, (1+ \gamma_{\beta/2})^{1/8} \,\,
\begin{tikzpicture}[equation, decoration={
    markings,
    mark=at position 0.2 with {\arrow{latex}},
    mark=at position 0.93 with {\arrow{latex}}}]
    
    \draw[postaction={decorate}, blue, thick] (1.6,0) -- (0,0);
\draw (0.8,0.4) node {\small $P_{1}$};

\end{tikzpicture}    
\,\, + \,\,  \left( \gamma_{\beta/2}\right)^{1/8} \,\,
\begin{tikzpicture}[equation, decoration={
    markings,
    mark=at position 0.2 with {\arrow{latex}},
    mark=at position 0.93 with {\arrow{latex}}}]
    
    \draw[postaction={decorate}, blue, thick] (1.6,0) -- (0,0);
\draw (0.8,0.4) node {\small $P_{0}$};

\end{tikzpicture}  
\end{equation}

Therefore, we have a PEPO decomposition of $e^{\frac{\beta}{2} B(p)}$ into four identical tensors acting individually on each edge and given by

\[
  \begin{tikzpicture}[equation, decoration={
    markings,
    mark=at position 0.2 with {\arrow{latex}},
    mark=at position 0.93 with {\arrow{latex}}}]
    \draw[postaction={decorate}, thick] (0,-0.8) -- (0,0.8);
  \draw [postaction={decorate}, blue, thick,domain=45:135] plot ({cos(\x)}, {sin(\x)-1});
   \shade[ball color=black] (0,0) circle (1ex);
\draw[blue,  thick, fill=white] ({cos(70)}, {sin(70)-1}) circle (0.5ex);
\draw[blue, thick,  fill=white] ({cos(110)}, {sin(110)-1}) circle (0.5ex); 
\draw[white] (0,1) circle (0.1);
    \end{tikzpicture}
    \,\,\, = \,\,\, \sum_{g \in G}  \quad \dyad{g}{g} \,\,
     \begin{tikzpicture}[equation, decoration={
    markings,
    mark=at position 0.2 with {\arrow{latex}},
    mark=at position 0.93 with {\arrow{latex}}}]
\draw[white] (0,1) circle (0.1);
  \draw[postaction={decorate}, thick] (0,-0.8) -- (0,0.8);
  \shade[ball color=black] (0,0) circle (1ex);

\end{tikzpicture}
\,\, \otimes \,\,
    \begin{tikzpicture}[equation, decoration={
    markings,
    mark=at position 0.2 with {\arrow{latex}},
    mark=at position 0.93 with {\arrow{latex}}}]
\draw[postaction={decorate}, blue, thick]
(2.5,-0.5) -- (2,0) -- (1.5,-0.5);
\draw (2,0.4) node {\small $\weightP \, L^{g^{\pm}} \, \weightP$};
    \end{tikzpicture}
\]
where we can expand
\[ \weightP \, L^{g^{\pm}} \, \weightP =  (1+ \gamma_{\beta/2})^{1/4}  \, P_{1} \, + \, \left( \gamma_{\beta/2}\right)^{1/4} \, P_{0} \, L^{g^{\pm}} \, P_{0}\,.\]
The choice of the sign is given according to the next picture
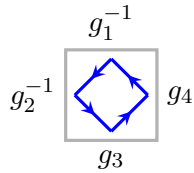
\[
\begin{tikzpicture}[equation, scale=0.6]
\draw [black!30!white,very thick] (-1,-1) -- (1,-1) -- (1,1) -- (-1,1) -- cycle;
\begin{scope}[scale=0.8, decoration={markings, mark=at position 0.55 with {\arrow{stealth}}}]
\draw [blue,very thick,postaction={decorate}] (0,-1) -- (1,0);
\draw [blue,very thick,postaction={decorate}] (1,0) -- (0,1);
\draw [blue,very thick,postaction={decorate}] (0,1) -- (-1,0);
\draw [blue,very thick,postaction={decorate}] (-1,0) -- (0,-1);
\end{scope}
\draw (0,1) node[above] {$g_{1}^{-1}$}; 
\draw (-1,0) node[left] {$g_{2}^{-1}$}; 
\draw (0,-1) node[below] {$g_{3}$}; 
\draw (1,0) node[right] {$g_{4}$}; 
\end{tikzpicture}
\]
so that the order in the composition of the maps fits with \eqref{equa:RegularCharDecompProjectors}.

\subsubsection{PEPS tensor on an edge}

We have decomposed each star  operator $e^{\frac{\beta}{2} A(s)}$, resp. plaquette operator $e^{\frac{\beta}{2} B(p)}$, into four tensors acting respectively on the incident, resp. surrounding, edges. Let us now fix an edge $e$ with orientation:
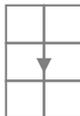
\begin{center}
\begin{tikzpicture}[scale=0.5]

\begin{scope}
\foreach \x in {0,1} 
\foreach \y in {0,1,2}
\draw[gray, thick] (\x,\y) -- (\x+1,\y) -- (\x+1,\y+1) -- (\x,\y+1) -- cycle;
\end{scope}

\fill[gray, thick] (0.8,1.6) -- (1,1.2) -- (1.2,1.6);

\end{tikzpicture}
\end{center} 

\noindent On this edge, we will have four tensors acting, two coming from the plaquettes and two coming from the stars which the edge belongs to. Each of these four tensors has a component acting on the physical space $\mathcal{H}_e$, which in the graphs has been denoted as a a solid black ball, and a component acting on some virtual space (which is also isomorphic to $\ell_2(G)$), which we denoted as colored lines (either red or blue). All four tensors act on the same physical space, but each of them has a separate virtual space distinct from the others. We want now to contract the indices of these tensors acting on the physical space. As the weights tensors (the white dots in our graphical notation) act only on the virtual indices, we can for the moment ignore them. The resulting tensor will be called \emph{slim} (as it is lacking the weights), and will be denoted by $\widetilde V_e \in \mathcal{B}(\mathcal{H}_e) \otimes \mathcal{B}(\ell_2(G))^{\otimes 4}$.

\begin{figure}[ht]
\begin{center}
\begin{tikzpicture}[decoration={
    markings,
    mark=at position 0.6 with {\arrow{latex}}}, scale=0.8]


\begin{scope}

\draw (-1,-0.4) -- (0,-0.4) -- (1,0.4) -- (0,0.4) -- (-1,-0.4);

\draw (0,-0.4) -- (1,-0.4) -- (2,0.4) -- (1,0.4) -- (0,-0.4);

\draw (1,-0.4) -- (2,-0.4) -- (3,0.4) -- (2,0.4) -- (1,-0.4);

\draw (-0.5,0) -- (2.5,0);

\fill[gray, thick] (1.3,0) -- (0.9,0.1) -- (0.7,-0.1);


\draw[very thick] (1,0) -- (1,4);

\draw[ultra thick,blue] ({0+0.25*1},{1+-0.4 + 0.25*0.8}) -- (1,1) -- ({1+0.25*1},{1 -0.4 + 0.25*0.8});
\shade[ball color=black] (1,1) circle (0.8ex);

\draw[ultra thick, blue] ({0+0.75*1},{1.75+-0.4 + 0.75*0.8}) -- (1,1.75) -- ({1+0.75*1},{1.75 -0.4 + 0.75*0.8});
\shade[ball color=black] (1,1.75) circle (0.8ex);

\draw[ultra thick, red] ({1+0.25*1},{2.5+-0.4 + 0.25*0.8}) -- (1,2.5) -- ({1+0.75*1},{2.5 -0.4 + 0.75*0.8});
\shade[ball color=black] (1,2.5) circle (0.8ex);

\draw[ultra thick, red] ({0+0.25*1},{3.25+-0.4 + 0.25*0.8}) -- (1,3.25) -- ({0+0.75*1},{3.25+-0.4 + 0.75*0.8});
\shade[ball color=black] (1,3.25) circle (0.8ex);

\end{scope}


\begin{scope}[xshift=7cm, yshift=2cm, scale=1.5]

\draw[thick, gray]  (0,-1.2) -- (0,1.2) ;
\draw[thick, gray]  (-0.4,-1) -- (0.4,-1) ;
\draw[thick, gray]  (-0.4,1) -- (0.4,1) ;

\fill[black] (-0.1,-0.1) rectangle (0.1,0.1);

\begin{scope}[xshift=-0.3cm, very thick]
  \draw[postaction={decorate}, blue] (-0.5,-0.5)  -- (0,0); 
    \draw[postaction={decorate}, blue] (0,0)  -- (-0.5,0.5); 
\end{scope}

\begin{scope}[xshift=0.3cm, very thick]    
      \draw[postaction={decorate}, blue] (0.5,0.5)  -- (0,0); 
    \draw[postaction={decorate}, blue] (0,0)  -- (0.5,-0.5); 
\end{scope}

\begin{scope}[yshift=0.3cm, very thick]
\draw[postaction={decorate}, red] (-0.5,0.5)  -- (0,0); 
    \draw[postaction={decorate}, red] (0,0)  -- (0.5,0.5); 
\end{scope}

\begin{scope}[yshift=-0.3cm, very thick]
\draw[postaction={decorate}, red] (0,0) -- (-0.5,-0.5); 
    \draw[postaction={decorate}, red] (0.5,-0.5) -- (0,0); 
\end{scope}

\end{scope}

\end{tikzpicture} 
\end{center}
\caption{The contraction defining the slim edge tensor. Note that the resulting tensor depends on the ordering of the product of the four original tensors. Here we have convened to set the two tensors coming from the plaquettes below the two tensors coming from the vertices.}
\label{figure:edge-contraction}
\end{figure}
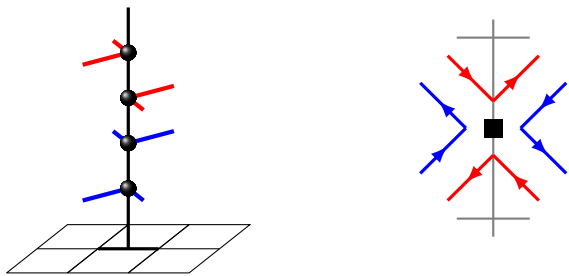

Graphically, the situation is represented in Figure~\ref{figure:edge-contraction}: on the left side, we have shown (in perspective) the four tensors acting on the given edge following the specific order for their contraction (see Remark \ref{Rema:orderContractionisImportant}). On the right-hand side, the graphical representation of a  tensor in $\mathcal{B}(\mathcal{H}_e) \otimes \mathcal{B}(\ell_2(G))^{\otimes 4}$ (seen from a top-down view): here the black square represents an element of $\mathcal{B}(\mathcal{H}_e)$, while each of the four colored lines is an element of $\mathcal{B}(\ell_2(G))$.
Performing the contraction, we obtain the following decomposition of the tensor $\widetilde V_e$.
\begin{equation}\label{equa:slimtensor1} 
\widetilde{V}_{e} \, \equiv \, \sum_{g,h,k \in G} \,\,\,\,  \dyad{hgk^{-1}}{g} \,\,
     \begin{tikzpicture}[equation, decoration={
    markings,
    mark=at position 0.3 with {\arrow{latex}},
    mark=at position 0.93 with {\arrow{latex}}}]
\draw[white] (0,1) circle (0.1);
  \draw[postaction={decorate}, thick] (0,-0.8) -- (0,0.8);
  \shade[ball color=black] (0,0) circle (1ex);

\end{tikzpicture}
\quad \otimes \quad
\begin{tikzpicture}[equation, decoration={markings, mark=at position 0.7 with {\arrow{latex}}}]
\begin{scope}[xshift=0.7cm]    
    
\draw[white] (0,1) circle (0.2);
    
\draw (0.9,0.1) node {\small $L^{g^{-1}}$}; 
\draw (-0.7,0) node {\small $L^{g}$};     
    
\draw (0,0.8) node { \footnotesize $\dyad{h}{h}$}; 
\draw (0,-0.8) node {\footnotesize $\dyad{k}{k}$};         

\edgePEPO[0];   

\end{scope}

\end{tikzpicture}
\end{equation}

\noindent This tensor defines a PEPO, and as discussed in Section~\ref{sec:PEPS}, we can realize a local purification to obtain a PEPS with physical space $\mathcal{H}_e^2 = \mathcal{H}_e \otimes \mathcal{H}_e$, via the purification map $Q \mapsto \ket{Q} = Q \otimes \mathbbm{1} \ket{\Psi}$, where $\ket{\Psi} = \sum_{h} \ket{hh}$.
It will be convenient to also represent the virtual spaces $\mathcal{B}(\ell_2(G))$ as $\ell_2(G)\otimes \ell_2(G)$, using the same purification map.
With a minor abuse of notation, we will denote also by $\widetilde V_e$ the linear map from the virtual space $\mathcal{H}_{\partial e} = \ell_{2}(G)^{\otimes 8}$ to the physical space $\mathcal{H}_e^2$ given by
\begin{equation}\label{equa:slimtensor2}  
\widetilde{V}_{e}: \mathcal{H}_{\partial e} \longrightarrow \mathcal{H}_{e}^{2} \quad , \quad \quad \widetilde{V}_{e}  = \sum_{g,h,k \in G} \ket{hgk^{-1}} \ket{g} \,\, 
\begin{tikzpicture}[equation]
\draw[white] (0,1.2) circle (0.2);
\draw (1,0) node {\small $\bra{L^{\smash{g^{-1}}}}$}; 
\draw (-0.9,0) node {\small $\bra{L^{g}}$};   
\draw (0,0.8) node { \footnotesize $\bra{h \, h} $}; 
\draw (0,-0.8) node {\footnotesize $\bra{k \, k} $};         
\edgePEPO[0];   
\end{tikzpicture}
\end{equation}

The \emph{full} tensor $V_{e}$ is then constructed from $\widetilde{V}_{e}$ by adding the corresponding weights $\weightS$ and $\weightP$ on the boundary indices. Indeed, adding the weights to the representation \eqref{equa:slimtensor1} leads to the corresponding representation for $V_{e}$ simply replacing
\[ \dyad{h} \,\, \mapsto \,\, \weightS  \dyad{h}  \weightS \quad \quad , \quad \quad L^{g} \,\, \mapsto \,\, \weightP  L^{g}  \weightP\,, \]
whereas adding the weights to \eqref{equa:slimtensor2} leads to the same expression but replacing
\[ \bra{h \, h} \,\, \mapsto \,\, \bra{h \, h} (\weightS \otimes \weightS)  \quad , \quad \bra{L^{g}} \, \mapsto \,   \bra{L^{g}}  (\weightP \otimes \weightP)\,. \]
Using this last representation, we can relate
 \[ V_{e} = \widetilde{V}_{e} \, \mathcal{G}_{\partial e}\,,   \]
 where $\mathcal{G}_{\partial e}$ is a suitable tensor product of (positive and invertible) operators of the form $\weightP \otimes \weightP$ and $\weightS \otimes \weightS$.

We will use a simpler picture for this \emph{slim} tensor as well as for the \emph{full} edge-tensor:
\[
\begin{tikzpicture}[equation]    

\begin{scope}
\draw (-1.5,0.1) node {$\widetilde{V}_{e} \,\, =$}; 
\edgePEPO[0];   
\fill[black] (-0.1,-0.1) rectangle (0.1,0.1);
\end{scope}

\begin{scope}[xshift=5cm]
\draw (-1.5,0.1) node {$V_{e} \,\, =$}; 
\edgePEPO[1];    
\fill[black] (-0.1,-0.1) rectangle (0.1,0.1);
\end{scope}

\end{tikzpicture}  
\]

\begin{Rema}\label{Rema:orderContractionisImportant}
The tensor $V_e$ we constructed depends on the way we have ordered the four components coming from the plaquette and star terms. Choosing a different order would have given us a different PEPO tensor representing the same operator $e^{-\frac{\beta}{2} H^{\text{syst}}_{\Lambda}} $. This choice will be irrelevant for our purposes, as they all give rise to equivalent boundary states.
\end{Rema}

\subsection{Boundary states}

\subsubsection{Edge}

In the previous subsection we have described the \emph{slim}  and \emph{full} tensors of the PEPS associated to an edge, namely $\widetilde{V}_{e}$ and $V_{e}$, both depending on the prefixed orientation. Next we are going to construct the corresponding boundary states $\widetilde{\rho}_{\partial e}$ and $\rho_{\partial e}$ by contracting the physical indices. Let us first consider an edge with fixed orientation:

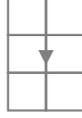
\[
\begin{tikzpicture}[equation, scale=0.5]

\begin{scope}
\foreach \x in {0,1} 
\foreach \y in {0,1,2}
\draw[gray, thick] (\x,\y) -- (\x+1,\y) -- (\x+1,\y+1) -- (\x,\y+1) -- cycle;
\end{scope}

\fill[gray, thick] (0.8,1.6) -- (1,1.2) -- (1.2,1.6);

\end{tikzpicture}
\]
\noindent The \emph{slim} boundary state is constructed as $\widetilde{\rho}_{\partial e} = \widetilde{V}_{e}^{\dagger} \widetilde{V}_{e}$, or equivalently, by considering $\widetilde{V}_{e} \otimes \widetilde{V}_{e}^{\dagger}$ and contracting the physical indices:
\[
\widetilde{\rho}_{\partial e} :=
\sum\limits_{\substack{g, h, k \in G\\ g', h', k' \in G}}  \mbox{\small $\big(\bra{h' g' k'^{-1}} \ket{hgk^{-1}} \bra{g'}\ket{g} \big)$} \,\,
\begin{tikzpicture}[equation]

\begin{scope}[xshift=1.2cm]    
\TransferOperator[1,0];
\draw (-0.8,-0.5) node {$\scriptstyle L^{ g}$};
\draw (-0.8,0.5) node {$\scriptstyle L^{ g'}$};
\end{scope}

 \draw (1.7,0) node {$\otimes$};  

\begin{scope}[xshift=2.2cm]
\begin{scope}[rotate=180]
\TransferOperator[1,0];
    \end{scope}
\draw (1,-0.5) node {$\scriptstyle L^{g^{-1}}$};
\draw (1,0.6) node {$\scriptstyle L^{g'^{-1}}$};
\end{scope}
 
 
 \draw (3.5,0) node {$\otimes$};   
 
\begin{scope}[xshift=4.3cm, yshift=0cm]    
\begin{scope}[rotate=180]
\TransferOperator[2,0];
\end{scope}    
\draw (0,0.8) node {$\scriptstyle \dyad{h'}{h'}$};
\draw (0, -0.8) node {$\scriptstyle \dyad{h}{h}$};
\end{scope}

 \draw (5.0,0) node {$\otimes$};

\begin{scope}[xshift=5.8cm, yshift=0cm]    
\TransferOperator[2,0];
\draw (0,0.8) node {$\scriptstyle \dyad{k'}{k'}$};
\draw (0, -0.8) node {$\scriptstyle \dyad{k}{k}$};
\end{scope} 

\end{tikzpicture} \,.
\]

\noindent Note that the scalar factor given in terms of the scalar product will be zero or one, the latter if and only if
\[ g = g' = (h^{-1} \, h')^{-1} \, g \, (k^{-1} \, k')\,, \]
so if we denote $b:=k^{-1} k'$ and $a:=h^{-1} h'$, then we can rewrite 
\[
\widetilde{\rho}_{\partial e} = 
\sum\limits_{\substack{g, a, b \, \in \,  G \colon \\  a=gbg^{-1}}} \,\,
  \begin{tikzpicture}[equation]
  
\begin{scope}[xshift=3cm]    
\TransferOperator[1,0];
\draw (-0.8,-0.5) node {$\scriptstyle L^{ g}$};
\draw (-0.8,0.5) node {$\scriptstyle L^{ g}$};
\end{scope}

 \draw (3.5,0) node {$\otimes$};  

\begin{scope}[xshift=4cm]
\begin{scope}[rotate=180]
\TransferOperator[1,0];
    \end{scope}
\draw (1,-0.5) node {$\scriptstyle L^{g^{-1}}$};
\draw (1,0.6) node {$\scriptstyle L^{g^{-1}}$};
\end{scope}
 
 
 \draw (5.5,0) node {$\otimes$};   
 
 \draw (6.3,-0.2) node {$\displaystyle \sum_{h \in G}$};    
 
 \begin{scope}[xshift=7.4cm, yshift=-0.2cm]    
 \begin{scope}[rotate=180]
\TransferOperator[2,0];
\end{scope} 
\draw (0,0.8) node {$\scriptstyle \dyad{ha}{ha}$};
\draw (0, -0.8) node {$\scriptstyle \dyad{h}{h}$};
\end{scope}

 \draw (8.3,0) node {$\otimes$};   
 
\draw (9.1,-0.2) node {$\displaystyle \sum_{k \in G}$};    
 
 \begin{scope}[xshift=10.2cm, yshift=0cm]    
\TransferOperator[2,0];
\draw (0,0.8) node {$\scriptstyle \dyad{kb}{kb}$};
\draw (0, -0.8) node {$\scriptstyle \dyad{k}{k}$};
\end{scope}
\end{tikzpicture} \,\,\,.
\]
Let us introduce some notation for $a,g \in G$
\[ \widetilde{\phi}_{a} := \sum_{h \in G}
\begin{tikzpicture}[equation]

\TransferOperator[2,0];
\draw (0,0.8) node {$\scriptstyle \dyad{ha}{ha}$};
\draw (0, -0.8) node {$\scriptstyle \dyad{h}{h}$};

\end{tikzpicture}
\quad \quad , \quad \quad  \widetilde{\psi}_{g} := 
\begin{tikzpicture}[equation]
\TransferOperator[1,0];
\draw (-0.8,-0.5) node {$\scriptstyle L^{ g}$};
\draw (-0.8,0.5) node {$\scriptstyle L^{ g}$};
\end{tikzpicture}\,\,\,.
\]
This allows us to rewrite the \emph{slim} boundary state in the simpler form:
\begin{equation}\label{equa:SlimBoundaryEdge1}
\widetilde{\rho}_{\partial e} \,\, = \sum_{\substack{ a,b,g \in G \\ a=g b g^{-1}}} \,\,
\begin{tikzpicture}[equation]

\begin{scope}[xshift=4cm,yshift=0.2cm]    
    
\edgefour[0,0];       
    
\draw (0.9,0) node {\small $\widetilde{\psi}_{g^{-1}}$}; 
\draw (-0.7,0) node {\small $\widetilde{\psi}_{g}$};     
    
\draw (0,0.8) node { \footnotesize $\widetilde{\phi}_{a}$}; 
\draw (0,-0.8) node {\footnotesize $\widetilde{\phi}_{b}$};         

\end{scope}

\end{tikzpicture}\,.
\end{equation}
\noindent It will also be convenient to introduce the weighted version of these tensors. For that, we introduce the analog of $\widetilde{\phi}_{a}$ and $\widetilde{\psi}_{g}$ when contracting with the weights of the PEPO representation of $e^{\frac{\beta}{2} A(s)}$ and $e^{\frac{\beta}{2} B(p)}$, namely
\[
\phi_{a} \,
:= \,\, \displaystyle \sum_{h \in G} \quad \left( \delta_{h,1} + \gamma_{\beta/2}\right)^{1/4} \,\,\left( \delta_{ha,1} + \gamma_{\beta/2}\right)^{1/4} \,\,
\begin{tikzpicture}[equation]

\TransferOperator[2,0];
\draw (0,0.8) node {$\scriptstyle \dyad{ha}{ha}$};
\draw (0, -0.8) node {$\scriptstyle \dyad{h}{h}$};

\end{tikzpicture}    
\]
and
\[
 \psi_{g} := \,\, \sum_{m,n \in \{ 0,1\}} \quad \left( n + \gamma_{\beta/2}\right)^{1/4} \,\,\left( m + \gamma_{\beta/2}\right)^{1/4} 
\begin{tikzpicture}[equation]

\TransferOperator[1,0];
\draw (0.5,-0.8) node {$\scriptstyle P_{m}L^{ g}P_{m}$};
\draw (0.5,0.8) node {$\scriptstyle P_{n}L^{ g}P_{n}$};

\end{tikzpicture}\,\,\,.
\]

\noindent Thus, analogously to the \emph{slim} case, we can represent the \emph{full} boundary state
\begin{equation}\label{equa:FullBoundaryEdge1}
\rho_{\partial e} \,\, := \,\, \displaystyle \sum_{\substack{ a,b,g \in G \\ a=g b g^{-1}}} \,\,    
\begin{tikzpicture}[equation]    
    
\edgefour[0,0];
\draw (0.9,0) node {\small $\psi_{g^{-1}}$}; 
\draw (-0.7,0) node {\small $\psi_{g}$};     
    
\draw (0,0.8) node { \footnotesize $\phi_{a}$}; 
\draw (0,-0.8) node {\footnotesize $\phi_{b}$};

\end{tikzpicture}\,\,\,.
\end{equation}

For the sake of applying the theory relating boundary states and the gap property of the parent Hamiltonian of a PEPS (see Section \ref{sec:PEPS}), we should look at the boundary states $\widetilde{\rho}_{e}$ and $\rho_{\partial e}$ as maps $\mathcal{H}_{\partial e} \longrightarrow \mathcal{H}_{\partial e}$. We could have taken the PEPS expressions for $\widetilde{V}_{e}$ and $V_{e}$ described in the previous section as maps $\mathcal{H}_{\partial e} \longrightarrow \mathcal{H}_{e}^{2}$, see \eqref{equa:slimtensor2}, and calculated $\widetilde{\rho}_{\partial e} = \widetilde{V}_{e}^{\dagger} \widetilde{V}_{e}$ and $\rho_{\partial e} = V^{\dagger}_{e}V_{e}$. This can be obtained also from \eqref{equa:SlimBoundaryEdge1} and \eqref{equa:FullBoundaryEdge1} by reinterpreting for $a,g \in G$ 
\[
\widetilde{\phi}_{a} = \sum_{h \in G} \ket{ha} \hspace{-0.7mm} \ket{ha} \bra{h} \hspace{-0.7mm} \bra{h} \quad \quad , \quad \quad \widetilde{\psi}_{g}= \ket{L^{g}} \bra{L^{g}}\,,
\]
and so
\begin{align*}
\phi_{a}  & = \, \sum_{h \in G} \left( \delta_{h,1} + \gamma_{\beta/2} \right)^{1/4} \left( \delta_{ha,1} + \gamma_{\beta/2} \right)^{1/4} \,\, \ket{ha} \hspace{-0.7mm} \ket{ha} \bra{h} \hspace{-0.7mm} \bra{h}\,,\\ 
\psi_{g} & = \sum_{n,m \in \{ 0,1\}} (n+ \gamma_{\beta/2})^{1/4} (m+ \gamma_{\beta/2})^{1/4} \ket{\smash{P_{n}L^{g}P_{n}}}\bra{\smash{P_{m}L^{g}P_{m}}}\,.
\end{align*}
We will use this notation for the rest of the paper. Recall that
\[   \phi_{a}= (\weightS \otimes \weightS) \, \widetilde{\phi}_{a} \, (\weightS \otimes \weightS) \quad , \quad \psi_{g}= (\weightP \otimes \weightP) \, \widetilde{\psi}_{g} \, (\weightP \otimes \weightP)\,. \]
Defining $\weight_{\partial e}$ as $\weightS \otimes \weightS \otimes \weightP \otimes \weightP$, up to a reordering of the tensor factors, we obtain that  
$\rho_{\partial e} = \weight_{\partial e} \, \widetilde{\rho}_{e} \, \weight_{\partial e}$.

\subsubsection{Plaquette}
\label{sec:PlaquetteBoundaryState}

Let us next describe the boundary state of a plaquette, constructed by placing the boundary state of each edge, as it was described in the previous subsection, and contracting indices accordingly:

\[
\rho_{p} \,\, = \,\,
\begin{tikzpicture}[equation, scale=1.5]
\plaquettefive[1,1](4,1,1,0);
\end{tikzpicture}
\]

\noindent For a more precise description, let us first label the edges and vertices of the plaquette $e_{1}, e_{2}, e_{3}, e_{4}$ and $v_{1}, v_{2}, v_{3}, v_{4}$ counterclockwise as in the next picture

\[
    \begin{tikzpicture}[equation, scale=1.2]
    \draw[very thick, gray] (0,0) rectangle (1,1);
    \draw (1.2,1.2) node {$v_{1}$}; 
    \draw (-0.2,1.2) node {$v_{2}$}; 
    \draw (-0.2,-0.2) node {$v_{3}$};
    \draw (1.2,-0.2) node {$v_{4}$};
    
    \draw (0.5,1.2) node {$e_{1}$}; 
    \draw (-0.2,0.5) node {$e_{2}$}; 
    \draw (0.5,-0.2) node {$e_{3}$};
    \draw (1.2,0.5) node {$e_{4}$};
    \end{tikzpicture}
\]

\noindent  At each vertex $v_{j}$, when contracting the indices of $\phi_{a}$ and $\phi_{a'}$ coming from the two incident edges we have
 \[
\begin{split}
\begin{tikzpicture}[equation]
\draw[thick, gray] (-0.5,0) -- (0.5,0);
\draw[thick, gray] (0,-0.5) -- (0,0.5);
\centerarc[red, ultra thick](0,0)(135:315:0.3);
\draw (0,-0.7) node {$\phi_{a'}$}; 
\draw (-0.7,0) node {$\phi_{a}$}; 
\end{tikzpicture} \,\,  = \,\, \sum_{h \in G} \, \sum_{h' \in G}\, & (\delta_{h,1} + \gamma_{\beta/2})^{1/4} \,  (\delta_{h',1} + \gamma_{\beta/2})^{1/4}  \, (\delta_{ha, 1} + \gamma_{\beta/2})^{1/4}  \cdot \\
&   (\delta_{h'a', 1} + \gamma_{\beta/2})^{1/4} \cdot \, \braket{ha}{h'a'}\, \braket{h}{h'} \, \cdot \, \ket{ha} \hspace{-0.7mm}  \ket{h'a'} \, \bra{h} \hspace{-0.7mm} \bra{h'}\,.
\end{split}
\]

\noindent Note that this contraction will be zero unless $a = a'$, and in that case only the summands with $h=h'$ will be nonzero. We will denote this tensor as $\phi^{(2)}_{a}$, being
\[  
\begin{tikzpicture}[equation]
\draw[thick, gray] (-0.5,0) -- (0.5,0);
\draw[thick, gray] (0,-0.5) -- (0,0.5);
\centerarc[red, ultra thick](0,0)(135:315:0.3);
\end{tikzpicture}\,\, \phi^{(2)}_{a} \,\, = \,\, \sum_{h \in G} \, (\delta_{h,1} + \gamma_{\beta/2})^{1/2} \, (\delta_{ha, 1} + \gamma_{\beta/2})^{1/2} \,\, \ket{ha} \hspace{-0.7mm} \ket{ha} \, \bra{h} \hspace{-0.7mm} \bra{h}  \,.
\]
Hence, we can expand

\[
\rho_{p} \,\, = \,\, \sum_{\substack{ g_{1}, g_{2}, g_{3}, g_{4} \\[0.5mm] a_{1}, a_{2}, a_{3}, a_{4} }} \,\,
    \begin{tikzpicture}[equation, scale=1.3]

    \plaquettefive[1,1](4,1,1,0);
    
    \draw (1.2,1.2) node {$\phi_{a_{1}}^{(2)}$}; 
    \draw (-0.2,1.2) node {$\phi_{a_{2}}^{(2)}$}; 
    \draw (-0.2,-0.2) node {$\phi_{a_{3}}^{(2)}$};
    \draw (1.2,-0.2) node {$\phi_{a_{4}}^{(2)}$};
    
    \draw (0.5,1.5) node {$\psi_{g_{1}}$}; 
    \draw (-0.5,0.5) node {$\psi_{g_{2}}$}; 
    \draw (0.5,-0.5) node {$\psi_{g_{3}^{-1}}$};
    \draw (1.6,0.5) node {$\psi_{g_{4}^{-1}}$};
    
    \end{tikzpicture} \,\,\,.
\]

\noindent  Remark that the elements $g_{j}$ and $a_{j}$ satisfy some \emph{compatibility conditions} for the corresponding summand to be nonzero, namely
\begin{equation}\label{eq:compatibility-condition}
a_{2} = g_{1}^{-1} a_{1} g_{1} \quad , \quad a_{3} = g_{2}^{-1} a_{2} g_{2} \quad , \quad a_{4} = g_{3} a_{3} g_{3}^{-1} \quad , \quad a_{1} = g_{4} a_{4} g_{4}^{-1}\,.
\end{equation}
The whole circle corresponds to the full contraction of the inner $\psi_{g}$'s placed on the edges, so it is a constant factor 
\[
\begin{split}
\begin{tikzpicture}[equation,scale=0.5]
\centerarc[blue,thick](0,0)(0:360:1);
\draw [blue,-{latex}] (1,0) -- (1,0.1);
\draw [blue,-{latex}] (0,1) -- (-0.1,1);
\draw [blue,-{latex}] (-1,0) -- (-1,-0.1);
\draw [blue,-{latex}] (0,-1) -- (0.1,-1);
\draw (0,1.7) node {$\psi_{g_{1}^{-1}}$}; 
\draw (-1.9,0) node {$\psi_{g_{2}^{-1}}$}; 
\draw (0,-1.7) node {$\psi_{g_{3}}$};
\draw (1.8,0) node {$\psi_{g_{4}}$};
\end{tikzpicture}
& = \,\, \sum_{\substack{m_{1}, \ldots m_{4} \in \{ 0,1\} \\ \substack{n_{1}, \ldots n_{4} \in \{ 0,1\} }}} \,\,  \prod_{j=1}^{4} (n_{j}+\gamma_{\beta/2})^{1/4}  \cdot (m_{j}+\gamma_{\beta/2})^{1/4}  \\
& \quad \quad \quad \quad \cdot \operatorname{Tr}(P_{n_{4}}L^{g_{4}}P_{n_{4}}  P_{n_{3}}L^{g_{3}}P_{n_{3}} 
P_{n_{2}}L^{g_{2}^{-1}}P_{n_{2}} P_{n_{1}}L^{g_{1}^{-1}}P_{n_{1}}  ) \\
& \quad \quad  \cdot \operatorname{Tr}(P_{m_{4}}L^{g_{4}}P_{m_{4}}
P_{m_{3}}L^{g_{3}}P_{m_{3}} 
P_{m_{2}}L^{g_{2}^{-1}}P_{m_{2}} 
P_{m_{1}}L^{g_{1}^{-1}}P_{m_{1}} )\,.
\end{split}
\]
Since $P_{k}P_{k'} = 0$ whenever $k \neq k'$ and $P_{k}L^{g} = L^{g}P_{k}$, we can simplify
\[
\begin{split}
\begin{tikzpicture}[equation,scale=0.5]
\centerarc[blue,thick](0,0)(0:360:1);
\draw [blue,-{latex}] (1,0) -- (1,0.1);
\draw [blue,-{latex}] (0,1) -- (-0.1,1);
\draw [blue,-{latex}] (-1,0) -- (-1,-0.1);
\draw [blue,-{latex}] (0,-1) -- (0.1,-1);
\draw (0,1.7) node {$\psi_{g_{1}^{-1}}$}; 
\draw (-1.9,0) node {$\psi_{g_{2}^{-1}}$}; 
\draw (0,-1.7) node {$\psi_{g_{3}}$};
\draw (1.8,0) node {$\psi_{g_{4}}$};
\end{tikzpicture}
& = \,\, \sum_{m, n \in \{ 0,1\}} \,\,  (n+\gamma_{\beta/2}) \cdot (m_{j}+\gamma_{\beta/2}) \cdot \\[-3mm]
& \quad \quad \cdot \operatorname{Tr}(P_{n}L^{g_{4}}L^{g_{3}}L^{g_{2}^{-1}}L^{g_{1}^{-1}})\cdot   \operatorname{Tr}(P_{m}L^{g_{4}}
L^{g_{3}}L^{g_{2}^{-1}}L^{g_{1}^{-1}} )\,.
\end{split}
\]
And the last expression can be simplified further using \eqref{equa:RegularCharDecompProjectors} 
\begin{equation}\label{equa:fullbluecontraction1}
\begin{split}
\begin{tikzpicture}[equation,scale=0.5]
\centerarc[blue,thick](0,0)(0:360:1);
\draw [blue,-{latex}] (1,0) -- (1,0.1);
\draw [blue,-{latex}] (0,1) -- (-0.1,1);
\draw [blue,-{latex}] (-1,0) -- (-1,-0.1);
\draw [blue,-{latex}] (0,-1) -- (0.1,-1);
    \draw (0,1.7) node {$\psi_{g_{1}^{-1}}$}; 
    \draw (-1.9,0) node {$\psi_{g_{2}^{-1}}$}; 
    \draw (0,-1.7) node {$\psi_{g_{3}}$};
    \draw (1.8,0) node {$\psi_{g_{4}}$};
\end{tikzpicture}
\,\, & = \,\, \left( 1 +  \gamma_{\beta/2} \, \chi^{reg}(g_{1}g_{2}g_{3}^{-1}g_{4}^{-1}) \right)^{2} \,.
\end{split}
\end{equation}

\noindent Note that, since $\chi^{reg}(g)/|G|$ is equal to zero or one, we have for every $g \in G$
\begin{equation}\label{equa:fullbluecontraction2} 
\left( 1+ \frac{e^{\beta/2} -1}{|G|} \, \chi^{reg}(g) \right)^{2} \,\,  = \,\, 1 + \frac{e^{\beta} - 1}{|G|} \chi^{reg}(g)\,.
\end{equation}

\noindent Therefore, we can take the scalar out and rewrite
\[
\rho_{p} \,\, = \,\, \sum_{\substack{ g_{1}, g_{2}, g_{3}, g_{4} \\[0.5mm] a_{1}, a_{2}, a_{3}, a_{4} }} \,\, \left( 1 +  \gamma_{\beta} \, \chi^{reg}(g_{1}g_{2}g_{3}^{-1}g_{4}^{-1}) \right) \,\,
    \begin{tikzpicture}[equation, scale=1.3]

    \plaquettefive[1,1](4,0,1,0);
    
    \draw (1.2,1.2) node {$\phi_{a_{1}}^{(2)}$}; 
    \draw (-0.2,1.2) node {$\phi_{a_{2}}^{(2)}$}; 
    \draw (-0.2,-0.2) node {$\phi_{a_{3}}^{(2)}$};
    \draw (1.2,-0.2) node {$\phi_{a_{4}}^{(2)}$};
    
    \draw (0.5,1.5) node {$\psi_{g_{1}}$}; 
    \draw (-0.5,0.5) node {$\psi_{g_{2}}$}; 
    \draw (0.5,-0.5) node {$\psi_{g_{3}^{-1}}$};
    \draw (1.6,0.5) node {$\psi_{g_{4}^{-1}}$};
    
    \end{tikzpicture}\,\,\,.
\]

\noindent Again, it should be noted that the sum expands over elements $g_{j}$ and $a_{j}$ satisfying the compatibility conditions~\eqref{eq:compatibility-condition}. They yield, in particular, that knowing $g_{j}$ for $j=1, 2,3,4$ and one of the $a_{j}$'s, we can determine the rest. As a consequence, we can rewrite 
\[
\rho_{p} \,\, = \,\,  \sum_{a \in G} \, \sum_{g_{1}, g_{2}, g_{3}, g_{4}} \,\, \left( 1 +  \gamma_{\beta} \chi^{reg}(g_{1}g_{2}g_{3}^{-1}g_{4}^{-1}) \right) \,\, 
    \begin{tikzpicture}[equation, scale=1.3]

    \plaquettefive[1,1](4,0,1,0);
    
        \draw (1.2,1.2) node {${\phi}_{a}^{(2)}$}; 
    \draw (-0.2,1.2) node {${\phi}_{\cdot}^{(2)}$}; 
    \draw (-0.2,-0.2) node {${\phi}_{\cdot}^{(2)}$}; 
    \draw (1.2,-0.2) node {${\phi}_{\cdot}^{(2)}$}; 
    
    \draw (0.5,1.5) node {${\psi}_{g_{1}}$}; 
    \draw (-0.5,0.5) node {${\psi}_{g_{2}}$}; 
    \draw (0.5,-0.5) node {${\psi}_{g_{3}^{-1}}$};
    \draw (1.6,0.5) node {${\psi}_{g_{4}^{-1}}$};
    
    \end{tikzpicture}\,\,\,,
\]

\noindent where the sum is taken over $a, g_{j}$ satisfying
\begin{equation}\label{equa:compatibilityCondition2} 
(g_{1} g_{2} g_{3}^{-1} g_{4}^{-1}) \, a \, (g_{1} g_{2} g_{3}^{-1} g_{4}^{-1})^{-1} = a \,. \end{equation}

\noindent It will also be useful to consider the \emph{slim} version $\widetilde{\rho}_{p}$ of $\rho_{p}$ obtained by ``removing'' the weights from the boundary virtual indices:
\[
\widetilde{\rho}_{p} \,\, = \,\, \sum_{a \in G} \, \sum_{g_{1}, g_{2}, g_{3}, g_{4}} \,  \left( 1 +  \gamma_{\beta} \, \chi^{reg}(g_{1}g_{2}g_{3}^{-1}g_{4}^{-1}) \right) \,
    \begin{tikzpicture}[equation, scale=1.3]
    
    \plaquettefive[1,1](4,0,1,0);
    
        \draw (1.2,1.2) node {$\widetilde{\phi}_{a}^{(2)}$}; 
    \draw (-0.2,1.2) node {$\widetilde{\phi}_{\cdot}^{(2)}$}; 
    \draw (-0.2,-0.2) node {$\widetilde{\phi}_{\cdot}^{(2)}$}; 
    \draw (1.2,-0.2) node {$\widetilde{\phi}_{\cdot}^{(2)}$}; 
    
    \draw (0.5,1.5) node {$\widetilde{\psi}_{g_{1}}$}; 
    \draw (-0.5,0.5) node {$\widetilde{\psi}_{g_{2}}$}; 
    \draw (0.5,-0.5) node {$\widetilde{\psi}_{g_{3}^{-1}}$};
    \draw (1.6,0.5) node {$\widetilde{\psi}_{g_{4}^{-1}}$};
    
    \end{tikzpicture}\,\,\,,
\]
where, once again, the sum is taken over elements satisfying the compatibility condition \eqref{equa:compatibilityCondition2}.


\subsubsection{Rectangular region}
\label{sec:PEPSrectangle}
 We aim at describing the boundary of the rectangular regions $\mathcal{R} \subset \Lambda \equiv \EE$. We are going to make the details in the case of a proper rectangle, since the cylinder case follows analogously with few adaptations. First we need to introduce some further notation regarding the edges and vertices that form $\mathcal{R}$:

\begin{footnotesize}
\begin{align*}
\EE_{\RR} & := \text{edges contained in $\RR$,}\\[2mm]
\EE_{\partial \RR} & := \text{edges with only one adjacent plaquette inside $\RR$,}\\[2mm]
\EE_{\mathring{\RR}} &:= \text{edges with both adjacent plaquettes inside $\RR$,}\\[2mm]
\VV_{\RR} & := \text{vertices contained in $\RR$,}\\[2mm]
\VV_{\partial \RR} & := \text{vertices with some but not all incident edges in $\RR$,}\\[2mm]
\VV_{\mathring{\RR}} &:= \text{vertices with all four incident edges in $\RR$,}\\[2mm]
n_{\RR} &:= \text{ number of plaquettes contained in $\RR$ }\,.
\end{align*}
\end{footnotesize}
To formally construct the transfer operator or boundary state $\rho_{\RR}$ on $\RR$, we must place at each edge $e$ contained in $\RR$ the transfer operator $\rho_{e}$  that was constructed in the previous subsection and contract the indices:

\[
\rho_{e} \,\, = \,\, V_{e}^{\dagger} V_{e} \,\, = \,\, \sum_{\substack{ a,b,g \in G \\ a=g b g^{-1}}}
\begin{tikzpicture}[equation]

\draw (0.9,0) node {\small $\psi_{g^{-1}}$}; 
\draw (-0.7,0) node {\small $\psi_{g}$};     
    
\draw (0,0.8) node { \footnotesize $\phi_{a}$}; 
\draw (0,-0.8) node {\footnotesize $\phi_{b}$};         
\edgefour[0,0];

\end{tikzpicture}\,\,\,.
\]

The philosophy is similar to the plaquette case, although more cumbersome to formalize. We are going to expand the expression for $\rho_{\partial \RR}$ 
\[
\rho_{\partial \RR} \,\, = \,\,
\begin{tikzpicture}[equation, scale=0.7]
\draw[white] (0,2.8) circle (0.1);
\plaquettefive[3,2](4,1,1,0);

\end{tikzpicture}
\,\,\,\, = \,\, \sum_{\substack{ \widehat{g}: \EE_{\RR} \rightarrow G \\[0.5mm] \widehat{a}:\VV_{\RR} \rightarrow G }} \,\,\,\,\,
\begin{tikzpicture}[equation, scale=0.7]
\draw (3.5,2.5) node {$\psi_{\, \widehat{g}}$}; 
\plaquettefive[3,2](1,1,1,0);

\end{tikzpicture}
\,\, \otimes \,\,\,\,
\begin{tikzpicture}[equation, scale=0.7]
\draw (3.5,2.5) node {$\phi_{\, \widehat{a}}$}; 
\plaquettefive[3,2](2,1,1,0);

\end{tikzpicture}    
\]

\noindent as a sum over maps
\[ 
\widehat{g}: \EE_{\RR} \longrightarrow G \quad , \quad \widehat{a}: \VV_{\RR} \longrightarrow G \,
\]
satisfying a certain compatibility condition. Let us explain the notation: at each edge $e \in \EE_{\RR}$, we set a value $g:=\widehat{g}(e)$ according to whether we have
\[
\begin{tikzpicture}[equation, thick,scale=0.7]

\def\R{0.8};

\begin{scope}[rotate=90]
\draw[thick, black!30!white]  (0,-1) -- (0,1) ;

\centerarc[blue,thick](-1,0)(-45:45:\R);
\draw [blue,-{latex}] (-1+\R,0) -- (-1+\R,0.1);

\centerarc[blue,thick](1,0)(135:225:\R);
\draw [blue,-{latex}] (1-\R,0) -- (1-\R,-0.1);

\draw (0,1.2) node {$e$};
\draw (0.8,0) node {$\psi_g$};
\draw (-0.8,0) node {$\psi_{g^{-1}}$};
\end{scope}
\end{tikzpicture}
\hspace{1cm}, \hspace{1cm}
\begin{tikzpicture}[equation, scale=0.7]

\def\R{0.8};

\draw[thick, black!30!white]  (0,-1) -- (0,1) ;

\centerarc[blue,thick](-1,0)(-45:45:\R);
\draw [blue,-{latex}] (-1+\R,0) -- (-1+\R,0.1);

\centerarc[blue,thick](1,0)(135:225:\R);
\draw [blue,-{latex}] (1-\R,0) -- (1-\R,-0.1);

\draw (0,1.2) node {$e$};
\draw (1,0) node {$\psi_{g^{-1}}$};
\draw (-0.8,0) node {$\psi_{g}$};

\end{tikzpicture}\,\,\,.
\]

At each vertex $v \in \VV_{\RR}$, the contraction of indices coming from incident edges and corresponding to say $\phi_{a}$ and $\phi_{a'}$ will be zero unless $a = a'$. Thus, a nonzero contraction will be determined by a unique $\widehat{a}(v) = a = a'$:

\[  
\begin{tikzpicture}[equation]
\draw[thick, gray] (-0.5,0) -- (0.5,0);
\draw[thick, gray] (0,-0.5) -- (0,0.5);
\centerarc[red, ultra thick](0,0)(135:315:0.3);
\end{tikzpicture}\,\, \phi^{(2)}_{a} \,\, = \,\, \sum_{h \in G} \, (\delta_{h,1} + \gamma_{\beta/2})^{1/2} \, (\delta_{ha, 1} + \gamma_{\beta/2})^{1/2} \,\, \ket{ha} \hspace{-0.7mm} \ket{ha} \, \bra{h} \hspace{-0.7mm} \bra{h}  \,,
\]

\[  
\begin{tikzpicture}[equation]
\draw[thick, gray] (-0.5,0) -- (0.5,0);
\draw[thick, gray] (0,-0.5) -- (0,0.5);
\centerarc[red, ultra thick](0,0)(45:315:0.3);
\end{tikzpicture}\,\, \phi^{(3)}_{a} \,\, = \,\, \sum_{h \in G} \, (\delta_{h,1} + \gamma_{\beta/2})^{3/4} \, (\delta_{ha, 1} + \gamma_{\beta/2})^{3/4} \,\, \ket{ha} \hspace{-0.7mm} \ket{ha} \, \bra{h} \hspace{-0.7mm} \bra{h}  \,,
\]

\begin{equation}\label{equa:fullredcontraction}
\begin{split}    
    \begin{tikzpicture}[equation]
    \draw (0,0.7) node {};
    \centerarc[red,ultra thick](0,0)(0:360:0.3);
    \draw[gray, thick] (0,-0.5) -- (0,0.5);
    \draw[gray, thick] (-0.5,0) -- (0.5,0);
    \end{tikzpicture}
    \,\, & = \sum_{h \in G} \left( \delta_{h,1} + \gamma_{\beta/2} \right) \, \left( \delta_{h a, 1} + \gamma_{\beta/2} \right) \, = \, \delta_{a,1} + 2 \gamma_{\beta/2} + |G| \, (\gamma_{\beta/2})^{2}\\
    & = \, \delta_{a,1} + 2 \left( \tfrac{e^{\beta/2}-1}{|G|} \right) + |G| \, \left( \tfrac{e^{\beta/2} - 1}{|G|}\right)^{2} \, = \,  \delta_{a,1} +  \frac{e^{\beta} - 1}{|G|} \, = \, \delta_{a,1}  +  \gamma_{\beta}\,.
\end{split}    
\end{equation}

\noindent Moreover, the sum in the aforementioned expression for $\rho_{\partial \RR}$ is taken over maps $\widehat{g}$ and $\widehat{a}$ satisfying the compatibility condition: for every edge $e \in \EE_{\RR}$

\[
\begin{tikzpicture}[equation, scale=0.7]
\def\r{0.5};

\begin{scope}[rotate=90]

\draw[thick, black!30!white]  (0,-1) -- (0,1) ;

\centerarc[red,thick](0,-1)(45:135:\r);
\centerarc[red,thick](0,1)(-135:-45:\r);

\draw (0,1.2) node {$e$};
\draw (0.3,0) node {$g$};
\draw (-0.5,0.5) node {$b$};
\draw (-0.5,-0.5) node {$a$};
\end{scope}
\end{tikzpicture}
\hspace{1cm} , \hspace{1cm}
\begin{tikzpicture}[equation, scale=0.7]

\def\r{0.5};

\draw[thick, black!30!white]  (0,-1) -- (0,1) ;

\centerarc[red,thick](0,-1)(45:135:\r);
\centerarc[red,thick](0,1)(-135:-45:\r);

\draw (0,1.2) node {$e$};
\draw (0.2,0) node {$g$};
\draw (-0.5,0.5) node {$a$};
\draw (-0.5,-0.5) node {$b$};

\end{tikzpicture}
\hspace{1cm}
\mbox{ such that } \,\, a=gbg^{-1}\,.
\]

\noindent The notation for the rectangles with $\psi_{\widehat{g}}$ and $\phi_{\widehat{a}}$ is the obvious.They specifically refer to the tensor product of operators $\phi_{a}^{(j)}$ and $\psi_{g}$ (respectively) that are present in the boundary. They do not encompass the entire circles (constants) that are observed within the plaquettes:
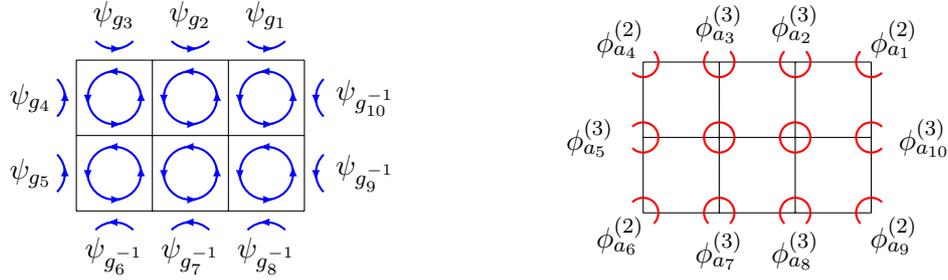
\[
\begin{tikzpicture}[equation]

\plaquettefive[3,2](1,1,1,0);
\foreach \x/\y [count=\i] in { 2.5/2.6, 1.5/2.6, 0.5/2.6, -0.6/1.5, -0.6/0.5} {
    \draw (\x,\y) node {$\psi_{g_{\i}}$}; 
}
\foreach \x/\y [count=\i from 6] in { 0.5/-0.6, 1.5/-0.6, 2.5/-0.6, 3.8/0.5, 3.8/1.5 } {
    \draw (\x,\y) node {$\psi_{g_{\i}^{-1}}$}; 
}
\end{tikzpicture}
\hspace{2cm}
\begin{tikzpicture}[equation]

\plaquettefive[3,2](2,1,1,0);
\foreach \x/\y/\d [count=\i] in { 3.3/2.3/2, 2/2.5/3, 1/2.5/3, -0.3/2.3/2, -0.7/1/3, -0.3/-0.3/2, 1/-0.5/3, 2/-0.5/3, 3.3/-0.3/2, 3.7/1/3 } {
    \draw (\x,\y) node {$\phi^{(\d)}_{a_{\i}}$}; 
}
\end{tikzpicture}    
\]

\noindent The compatibility condition yields an interesting consequence that we also pointed out for plaquettes in Section \ref{sec:PlaquetteBoundaryState}. For each pair $(\widehat{a}, \widehat{g})$ the map $\widehat{a}$ can be reconstructed knowing only its value at one vertex using $\widehat{g}$. In particular, if we fix a vertex, namely the lower right corner, we can rewrite the above expression of $\rho_{\partial R}$ as
\[
\rho_{\partial \RR} \,\, = \,\, \sum_{a \in G} \,\, \sum_{ \widehat{g}: \EE_{\RR} \rightarrow G } \quad
\begin{tikzpicture}[equation, scale=0.7]
\plaquettefive[3,2](1,1,0,0);
\draw (3.5,2.5) node {$\psi_{\, \widehat{g}}$}; 
\end{tikzpicture}
 \otimes \quad
\begin{tikzpicture}[equation, scale=0.7]
\plaquettefive[3,2](2,1,0,0);
\draw (3.5,2.5) node {$\phi_{\, \widehat{a}}$}; 
\draw (3.5,-0.2) node {$a$}; 
\end{tikzpicture}    
\]

\noindent Here $\widehat{a}$ is the only map compatible with $\widehat{g}$ and the choice $a$ in the prefixed vertex. In particular, all the elements $\widehat{a}(v) ,\, v \in \VV_{\RR}$ belong to the same conjugation class. This means that we have two possibilities: If $a =1$, resp. $a \neq 1$, then $\widehat{a}(v) = 1$, resp. $\widehat{a}(v) \neq 1 $, for every vertex of $\VV_{\RR}$, and thus by \eqref{equa:fullredcontraction}
\[
    \begin{tikzpicture}[equation, scale=0.7]
\plaquettefive[3,2](2,1,0,0)
\draw (3.5,2.5) node {$\phi_{\, \widehat{a}}$}; 
\draw (3.5,-0.2) node {$a$}; 
\end{tikzpicture}
= \,\, \left( \delta_{a,1} + \gamma_{\beta} \right)^{\# \VV_{\mathring{\RR}}} \,\, \cdot \quad
\begin{tikzpicture}[equation, scale=0.7]

\plaquettefive[3,2](2,0,0,0);
\draw (3.5,2.5) node {$\phi_{\, \widehat{a}}$}; 
\draw (3.5,-0.2) node {$a$}; 

    \end{tikzpicture}\,\,\,.
\]

\noindent where we have extracted the constant factor resulting from full contraction in the inner vertices of the rectangle. Analogously, we can argue with the $\psi_{\widehat{g}}$ factor, where now we get a constant factor for each inner plaquette by \eqref{equa:fullbluecontraction1} and \eqref{equa:fullbluecontraction2}:
\[
    \begin{tikzpicture}[equation, scale=0.7]

\plaquettefive[3,2](1,1,0,0);
\draw (3.5,2.5) node {$\psi_{\, \widehat{g}}$}; 

\end{tikzpicture}
= \,\, \displaystyle \prod_{ \substack{\text{$p$ inner}\\ \text{plaquette}} } \left( 1 + \gamma_{\beta} \, \chi^{reg}(\widehat{g}|_{p}) \, \right) \,\, \cdot
\begin{tikzpicture}[equation, scale=0.7]

\plaquettefive[3,2](1,0,0,0);
\draw (3.5,2.5) node {$\psi_{\, \widehat{g}}$};

\end{tikzpicture}\,\,\,,
\]

\noindent where we are denoting for each plaquette

\[
    \begin{tikzpicture}[baseline={([yshift=-.5ex]current bounding box.center)}]
    \draw (0,0) rectangle (1,1);
    \draw (0.5,0.5) node {$p$};
    \foreach \x/\y [count=\i] in {0.5/1.2, -0.3/0.5, 0.5/-0.2, 1.3/0.5}
     \draw (\x,\y) node {$e_{\i}$};
    \end{tikzpicture}
    \qquad 
    \chi^{reg}(\widehat{g}|_{p}) = \chi^{reg}(g_1 g_2 g_3^{-1} g_4^{-1}), \quad g_i = \widehat{g}(e_i).
\]

Note that the definition is independent of the enumeration of the edges as long as it is done counterclockwise ($\chi^{reg}$ is invariant under cyclic permutations) and respect the inverses on the lower and right edges.

We have then the following representation for the boundary state $\rho_{\RR}$ of the rectangular region $\RR$
\[
\rho_{\partial \RR} \,\, = \,\, \sum_{a \in G} \,\, \sum_{ \widehat{g}: \EE_{\RR} \rightarrow G } \,\, \left( \delta_{a,1} + \gamma_{\beta} \right)^{\abs{\VV_{\mathring{\RR}}}} \,\, c_{\beta}( \, \widehat{g} \, ) \quad
\begin{tikzpicture}[equation, scale=0.4]

\plaquettefive[3,2](1,0,0,0)
\draw (3.5,2.5) node {\scriptsize $\psi_{\, \widehat{g}}$}; 

\end{tikzpicture}
 \otimes \quad 
\begin{tikzpicture}[equation, scale=0.4]

\plaquettefive[3,2](2,0,0,0)

\draw (3.5,2.5) node {\scriptsize $\phi_{\, \widehat{a}}$}; 
\draw (3.5,-0.2) node {\scriptsize $a$}; 

\end{tikzpicture}    \,\,\,,
\]

\noindent where we are denoting

\begin{equation}\label{eq:c-beta-def}
c_{\beta}( \, \widehat{g} \, )  \, := \, \prod_{ \substack{\text{$p$ inner}\\ \text{plaquette}} } \left( 1 + \gamma_{\beta} \chi^{reg}(\widehat{g}|_{p}) \, \right)\,.
\end{equation}

We are going to deal with the \emph{slim} version of the boundary state resulting from removing the weights acting on the boundary indices:
\begin{equation}\label{equa:BoundaryProperRectangle1}
\widetilde{\rho}_{\partial \RR} \,\, = \,\, \sum_{a \in G} \,\, \sum_{ \widehat{g}: \EE_{\RR} \rightarrow G } \,\, \left( \delta_{a,1} + \gamma_{\beta} \right)^{\abs {\VV_{\mathring{\RR}}}}  c_{\beta}( \, \widehat{g} \, ) \quad 
\begin{tikzpicture}[equation, scale=0.4]

\plaquettefive[3,2](1,0,0,0)

\draw (3.5,2.5) node {\scriptsize $\widetilde{\psi}_{\, \widehat{g}}$}; 
\end{tikzpicture}
\otimes \quad
\begin{tikzpicture}[equation, scale=0.4]

\plaquettefive[3,2](2,0,0,0)

\draw (3.5,2.5) node {\scriptsize $\widetilde{\phi}_{\, \widehat{a}}$};
\draw (3.5,-0.2) node {\scriptsize $a$}; 

\end{tikzpicture}\,\,\,. 
\end{equation}
Defining $\weight_{\partial \RR}$  to be the suitable tensor product of  
$\weightS^{2}$, $\weightS^{3}$  and $\weightP$ 
corresponding to the boundary indices appearing in $\partial \RR$, the two versions of the boundary state are related by
\[ \rho_{\partial \RR} = \weight_{\partial \RR} \widetilde{\rho}_{\partial \RR} \weight_{\partial \RR} \,.\]

\subsubsection{Leading Term and Approximate factorization}

Finding a short explicit formula for the boundary states $\rho_{\partial \mathcal{R}}$ and $\widetilde{\rho}_{\partial \mathcal{R}}$ is not going to be feasible. But we will show that, after rearranging summands, there is a dominant term that we can explicitly describe in a short way. For that let us introduce
\[ 
\widetilde{\Delta}  :=  \frac{1}{|G|} \, \sum_{g \in G} \,\, \widetilde{\psi}_{g} \, = \, \frac{1}{|G|} \sum_{g \in G} \ket{L^{g}} \bra{L^{g}} \,\,
 \]
and consider the (product) operator 
\begin{equation}\label{equa:LeadingTermDefinition} 
\widetilde{\mathcal{S}}_{\partial \RR} \,\, := \,\,  
\begin{tikzpicture}[equation,scale=0.4]
    \plaquettefive[3,2](1,0,0,0);
    \draw (4.5,2.6) node {\scriptsize $\widetilde \Delta^{\otimes \EE_{\partial \RR}}$};
    \end{tikzpicture}
    \otimes \quad
    \begin{tikzpicture}[equation,scale=0.4]
    \plaquettefive[3,2](2,0,0,0);
    \draw (4,2.3) node {\scriptsize $\widetilde \phi_{\widehat{1}}$};
    \end{tikzpicture}\,.
\end{equation}
We will later show that this is indeed an orthogonal projection on $\mathcal{H}_{\partial \RR}$ made of local projections $\widetilde{\Delta}$ and $\widetilde{\phi}_{\widehat{1}}$ on $\ell_{2}(G) \otimes \ell_{2}(G)$. We next state the first main (and most involved) result of the section.

\begin{Theo}[Leading Term of the boundary]\label{Theo:leadingTerm}
Let us define the scalars
\[ \kappa_{\mathcal{R}}:= (1 + \gamma_{\beta})^{|\VV_{\mathring{\RR}}| + n_{\mathcal{R}}} \, |G|^{|\EE_{\RR}|} \quad , \quad \epsilon_{\RR}:= 3 \, |G|^{2} \, \left( \frac{\gamma_{\beta}}{1+ \gamma_{\beta}}\right)^{|\VV_{\mathring{\RR}}|}\,. \]
Then, we can decompose
\[ \widetilde{\rho}_{\partial \RR} = \kappa_{\RR} \, \left( \widetilde{\mathcal{S}}_{\partial \RR} + \widetilde{\mathcal{S}}^{rest}_{\partial \RR} \right) \]
for some observable $\widetilde{\mathcal{S}}^{rest}_{\partial \RR}$ with $\| \widetilde{\mathcal{S}}^{rest}_{\partial \RR}\| \leq \epsilon_{\RR}$. 
\end{Theo}


Before proving the theorem, let us discuss two useful consequences. Recall that the \emph{full} boundary state $\rho_{\partial \RR}$ and its \emph{slim} version $\widetilde{\rho}_{\partial \RR}$ are related via a transformation $\mathcal{G}_{\partial \RR}$ consisting of a tensor product of the (positive and invertible) weight-operators (see Section \ref{sec:PEPSrectangle}):
\[ \rho_{\partial \RR} = V_{\mathcal{R}}^{\dagger} V_{\mathcal{R}} = \mathcal{G}_{\partial \RR} \, \widetilde{V}_{\mathcal{R}}^{\dagger} \widetilde{V}_{\mathcal{R}} \, \mathcal{G}_{\partial \RR}= \mathcal{G}_{\partial \RR} \, \widetilde{\rho}_{\partial \RR} \, \mathcal{G}_{\partial \RR} \,\]

Denote by $J_{\partial \RR}$ and $\widetilde{J}_{\partial \RR}$ the orthogonal projections onto $(\ker \rho_{\partial \RR})^{\perp}$ and $(\ker \widetilde{\rho}_{\partial \RR})^{\perp}$, respectively.

\begin{Theo}[Approximate factorization of the boundary]\label{Coro:approx-factorization-bound}
Following the notation of the previous theorem, let us assume that $\epsilon_{\RR} < 1$. Then, the following assertions hold:
\begin{enumerate}
    \item[(i)] $J_{\partial \RR} = \widetilde{J}_{\partial \RR} = \widetilde{S}_{\partial \RR}$. In other words, $\widetilde{\mathcal{S}}_{\partial \RR}$ is the orthogonal projection onto the support of the \emph{slim} boundary state  and the \emph{full} boundary state. 
    \item[(ii)] The operator
    \[
    \sigma_{\partial \RR} := \kappa_{\RR}\, (\mathcal{G}_{\partial \RR} J_{\partial \RR}\, \mathcal{G}_{\partial \RR})  ,
\]
  satisfies  
    \[ 
    \left\| \rho_{\partial \RR}^{1/2} \sigma_{\partial \RR}^{-1} \rho_{\partial \RR}^{1/2} - J_{\partial R} \right\| < \epsilon_{\RR} \,\, , \,\, \left\| \rho_{\partial \RR}^{-1/2} \sigma_{\partial \RR} \rho_{\partial \RR}^{-1/2} - J_{\partial R} \right\|\, \le \frac{\epsilon_{\RR}}{1 - \epsilon_{\RR}}\,.
    \]
    where here the inverses are taken in the corresponding support.
\end{enumerate}
In particular, $J_{\partial \RR}$, $\widetilde{J}_{\partial \RR}$ and $\sigma_{\partial \RR}$ inherit the tensor product structure of $\widetilde{\mathcal{S}}_{\partial \RR}$, $\mathcal{G}_{\partial \RR}$.
\end{Theo}

In the rest of the subsection, we will develop the proofs of the above results. We have divided the whole argument into four parts. The first two parts are \emph{Tool Box 1} and \emph{Tool Box 2}, that contain some auxiliary results. Then, we will first prove Theorem \ref{Theo:leadingTerm}, and finally Theorem \ref{Coro:approx-factorization-bound}. 

\subsubsection*{Tool Box 1: Boundary projections}

We start with a few useful observations on the operators $\widetilde{\phi}$ and $\widetilde{\psi}$. Using the explicit formula $\widetilde{\phi}_{a} = \sum_{h \in G} \ket{ha}\hspace{-0.7mm}\ket{ha} \bra{h}\hspace{-0.7mm}\bra{h}$, it is easy to check that
\begin{itemize}
    \item[$\triangleright$] $\widetilde{\phi}_{a} \widetilde{\phi}_{a'} = \widetilde{\phi}_{a'a}$ for each $a,a' \in G$,
    \item[$\triangleright$] $\widetilde{\phi}_{1}$ is a projection.
\end{itemize}
For every $g,g' \in G$
\begin{equation}\label{equa:toolBox1Aux1}
 \braket{\smash{L^{g}}}{\smash{L^{g'}}} = \operatorname{Tr}(L^{g^{-1}} L^{g'}) = \operatorname{Tr}(L^{g^{-1}g'}) = \delta_{g,g'} |G| 
\end{equation}
This means that vectors $\frac{1}{\sqrt{|G|}} \ket{L^{g}}$ are orthonormal in $\ell_{2}[G] \otimes \ell_{2}[G]$, so
\begin{itemize}
    \item[$\triangleright$] $\frac{1}{|G|}\widetilde{\psi}_{g} = \frac{1}{|G|} \ket{L^{g}} \bra{L^{g}}$ is a one-dimensional orthogonal projection,
    \item[$\triangleright$] $\frac{1}{|G|}\widetilde{\psi}_{g}$ and $\frac{1}{|G|}\widetilde{\psi}_{g'}$ are mutually orthogonal if $g$ and $g'$ are distinct.
\end{itemize}

\noindent These facts will be used throughout the forthcoming results.

\begin{Prop}\label{Prop:BoundaryProjector1}
For each $\widehat{f}:\EE_{\partial \RR} \longrightarrow G$ let
\[ 
\mathcal{P}_{\partial \RR}(\widehat{f}):=\frac{1}{\abs{G}^{\abs{\EE_{\partial \RR}}}}  \,\,\,\,
    \begin{tikzpicture}[equation,scale=0.4]
    \plaquettefive[3,2](1,0,0,0)
    \draw (4,2.7) node {\scriptsize $\widetilde{\psi}_{\,     \widehat{f}}$};
    \end{tikzpicture}
    \otimes \,\,\,\, \begin{tikzpicture}[equation,scale=0.4]
    \plaquettefive[3,2](2,0,0,0);
    \draw (3.9,2.8) node {\scriptsize $\widetilde \phi_{\widehat{1}}$}; 
    \end{tikzpicture} \,.
\]    
Then, $\mathcal{P}_{\partial \RR}(\widehat{f})$ is an orthogonal projection on $\mathcal{H}_{\partial \RR}$.  Moreover if $\widehat{f}_{1}$ and $\widehat{f}_{2}$ are different, then $\mathcal{P}_{\partial \RR}(\widehat{f}_{1})$ and $\mathcal{P}_{\partial \RR}(\widehat{f}_{2})$ are mutually orthogonal, that is
\[ \mathcal{P}_{\partial \RR}(\widehat{f}_{1}) \, \mathcal{P}_{\partial \RR}(\widehat{f}_{2}) \, = \,  \mathcal{P}_{\partial \RR}(\widehat{f}_{2}) \, \mathcal{P}_{\partial \RR}(\widehat{f}_{1}) = 0 \]
\end{Prop}

\begin{proof}
The first statement is clear since $\mathcal{P}_{\partial \RR}(\widehat{f})$ is by definition a tensor product of projections $\widetilde{\phi}_{1}$ and $\frac{1}{|G|}\widetilde{\psi}_{g}$. Moreover, if $\widehat{f}_{1}$ and $\widehat{f}_{2}$ are different, then there is a boundary edge $e \in \EE_{\partial \RR}$ such that $\widehat{f}_{1}(e) \neq \widehat{f}_{2}(e)$. Thus, 
\[ 
\begin{tikzpicture}[equation,scale=0.4]
    \plaquettefive[3,2](1,0,0,0);
    \draw (4,2.6) node {\scriptsize $\widetilde \psi_{\widehat{f}_{1}}$};
    \end{tikzpicture} \,\, \circ \,\,\,\, 
    \begin{tikzpicture}[equation,scale=0.4]
    \plaquettefive[3,2](1,0,0,0);
    \draw (4,2.6) node {\scriptsize $\widetilde \psi_{\widehat{f}_{2}}$};
    \end{tikzpicture} \,\, = \,\, 0\,.
    \]
since $\widetilde{\psi}_{\widehat{f}_{1}(e)}$ and $\widetilde{\psi}_{\widehat{f}_{2}(e)}$ are mutually orthogonal by the above observations.
\end{proof}

\begin{Prop}\label{Prop:BoundaryProjector2}
For each $\widehat{f}:\EE_{\partial \RR} \longrightarrow G$ let us define
\[
\mathcal{Q}_{\partial \RR}(\widehat{f}) \,\, := \,\,  \frac{1}{\abs{G}^{\abs{\EE_{\partial \RR}}}} \,\,\,
\begin{tikzpicture}[equation,scale=0.4]
    \plaquettefive[3,2](1,0,0,0);
    \draw (4,2.6) node {\scriptsize $\widetilde \psi_{\widehat{f}}$};
    \end{tikzpicture}
    \otimes \quad \frac{1}{|G|} \,\sum_{a \in G} \,\,\,
    \begin{tikzpicture}[equation,scale=0.4]
    \plaquettefive[3,2](2,0,0,0);
    \draw (4,2.3) node {\scriptsize $\widetilde \phi_{\, \widehat{a}}$};
     \draw (3.7,0) node {\scriptsize $a$};
    \end{tikzpicture}
\]
Recall that for each $a \in G$  the element $\widehat{a}$ that appears on the right-hand side is the only choice compatible with $\widehat{f}$ and the fixed $a$. Then, $\mathcal{Q}_{\partial \RR}(\widehat{f})$ is an orthogonal projection. Moreover if $\widehat{f}_{1}$ and $\widehat{f}_{2}$ are different, then $\mathcal{Q}_{\partial \RR}(\widehat{f}_{1})$ and $\mathcal{Q}_{\partial \RR}(\widehat{f}_{2})$ are mutually orthogonal.
\end{Prop}

\begin{proof}
Let us first check that $\mathcal{Q}_{\partial \RR}(\widehat{f})$ is an orthogonal projection. For that, it is enough to check that each of the two (tensor product) factors is a projection. The first factor is indeed a tensor product of projections $\frac{1}{|G|} \widetilde{\psi}_{g}$. For the second factor, we just need to check that it is self-adjoint and idempotent. Since for every $a$, the adjoint of $\widetilde{\phi}_{a}$ is $\widetilde{\phi}_{a^{-1}}$, it holds
\[ \Big(\quad \begin{tikzpicture}[equation,scale=0.4]
    \plaquettefive[3,2](2,0,0,0);
    \draw (4,2.3) node {\scriptsize $\widetilde \phi_{\, \widehat{a}}$};
     \draw (3.7,0) node {\scriptsize $a$};
    \end{tikzpicture} \Big)^{\dagger} \,\, = \,\, \begin{tikzpicture}[equation,scale=0.4]
    \plaquettefive[3,2](2,0,0,0);
    \draw (4,3.3) node {\scriptsize $\widetilde \phi_{\,\widehat{a^{-1}}}$};
     \draw (3.8,0) node {\scriptsize $a^{-1}$};
    \end{tikzpicture} \]
and that the latter is again compatible with $\widehat{f}$, since at each edge the condition $a_{2} = g a_{1} g^{-1}$ is equivalent to $a_{2}^{-1} = g a_{1}^{-1} g^{-1}$. Thus, summing over all $a \in G$ we get self-adjointness, and so the first statement is proved. To see the second statement, note that 
\[  
\begin{tikzpicture}[equation,scale=0.4]
    \plaquettefive[3,2](2,0,0,0);
    \draw (4,2.3) node {\scriptsize $\widetilde \phi_{\, \widehat{a}}$};
     \draw (3.7,0) node {\scriptsize $a$};
    \end{tikzpicture} 
\, \circ  \,\,\,
    \begin{tikzpicture}[equation,scale=0.4]
    \plaquettefive[3,2](2,0,0,0);
    \draw (4,2.3) node {\scriptsize $\widetilde \phi_{\, \widehat{b}}$};
     \draw (3.7,0) node {\scriptsize $b$};
\end{tikzpicture} 
\,\,\, = \,\,\,
    \begin{tikzpicture}[equation,scale=0.4]
    \plaquettefive[3,2](2,0,0,0);
    \draw (4,2.3) node {\scriptsize $\widetilde \phi_{\, \widehat{ba}}$};
     \draw (3.7,0) node {\scriptsize $ba$};
\end{tikzpicture} 
\]
and that the resulting element $\widetilde{\phi}_{\widehat{ba}}$ is compatible with $g$, since $a_{2} = g a_{1} g^{-1}$ and $b_{2} = g b_{1} g^{-1}$ yield that $a_{2}b_{2} = g a_{1}b_{1} g^{-1}$. Hence, summing over $a,b \in G$ in the previous expression we get
\[  
\sum_{a \in G} \,\,\, \begin{tikzpicture}[equation,scale=0.4]
    \plaquettefive[3,2](2,0,0,0);
    \draw (4,2.3) node {\scriptsize $\widetilde \phi_{\widehat{a}}$};
     \draw (3.7,0) node {\scriptsize $a$};
    \end{tikzpicture} 
\, \circ  \,\,\,
   \sum_{b \in G} \,\,\, \begin{tikzpicture}[equation,scale=0.4]
    \plaquettefive[3,2](2,0,0,0);
    \draw (4,2.3) node {\scriptsize $\widetilde \phi_{\widehat{b}}$};
     \draw (3.7,0) node {\scriptsize $b$};
\end{tikzpicture} 
\,\,\, = \,\,\, |G| \, \sum_{c \in G} \,\,
    \begin{tikzpicture}[equation,scale=0.4]
    \plaquettefive[3,2](2,0,0,0);
    \draw (4,2.3) node {\scriptsize $\widetilde \phi_{\widehat{c}}$};
     \draw (3.7,0) node {\scriptsize $c$};
\end{tikzpicture} 
\]
This finishes the argument that $\mathcal{Q}_{\partial \RR}(\widehat{f})$ is an orthogonal projection. Finally, if $\widehat{f}_{1}$ and $\widehat{f}_{2}$ are different, then we can check that $\mathcal{Q}_{\partial \RR}(\widehat{f}_{1})$ and $\mathcal{Q}_{\partial \RR}(\widehat{f}_{2})$ are mutually orthogonal arguing as in the proof of Proposition \ref{Prop:BoundaryProjector1}. 
\end{proof}

\begin{Lemm}\label{Lemm:DeltaProjectorProperties}
 The operator 
 \[ \widetilde{\Delta} := \frac{1}{|G|} \sum_{g \in G}\widetilde{\psi}_{g} = \frac{1}{|G|} \sum_{g \in G}\ket{L^{g}} \bra{L^{g}} \] 
 is a projection on $\ell_{2}(G) \otimes \ell_{2}(G)$ satisfying
\[   \widetilde{\Delta} \widetilde{\psi}_{g} =  \widetilde{\psi}_{g} \widetilde{\Delta} = \widetilde{\psi}_{g} \,\,\, \text{for all } g \in G \quad \quad , \quad \quad (\weightP \otimes \weightP) \widetilde{\Delta} \, = \, \widetilde{\Delta} \, (\weightP \otimes \weightP)  \,. \]
\end{Lemm}

\begin{proof}
It is clear from the above observations on $\frac{1}{|G|}\widetilde{\psi}_{g}$ that $\widetilde{\Delta}$ is actually the projection onto the vector subspace generated by vectors of the form $\ket{L^{g}}$ and that $\widetilde{\Delta} \widetilde{\psi}_{g} =  \widetilde{\psi}_{g} \widetilde{\Delta} = \widetilde{\psi}_{g} \,\,\, \text{for all } g \in G$. To prove the last identity, let us apply \eqref{equa:decompLgIrrepProjectors} to decompose
\[
\begin{split}
\ket{L^{g}} \bra{L^{g}} = \ket{P_{1}} & \bra{P_{1}} + \ket{P_{1}} \bra{\smash{P_{0} L^{g} P_{0}}}  + \ket{\smash{P_{0} L^{g} P_{0}}} \bra{P_{1}} + \ket{\smash{P_{0} L^{g} P_{0}}} \bra{\smash{P_{0} L^{g} P_{0}}} \,.
\end{split}
\]

\noindent Since $P_{1} = \frac{1}{|G|}\sum_{g \in G}{L^{g}}$ (see~\eqref{eq:p1-projection} and the subsequent discussion), we get after summing over $g \in G$ in the previous expression
\[ \widetilde{\Delta} = \frac{1}{|G|} \sum_{g \in G} \ket{L^{g}} \bra{L^{g}} = \ket{P_{1}} \bra{P_{1}} + \frac{1}{|G|} \sum_{g \in G} \ket{\smash{P_{0} L^{g} P_{0}}} \bra{\smash{P_{0} L^{g} P_{0}}} \,.\]
Note that
\begin{align*}
(\weightP \otimes \weightP) \ket{P_{1}} & = \ket{\smash{\weightP P_{1} \weightP}} = (1 +\gamma_{\beta/2})^{1/4} \ket{P_{1}}\,, \\[2mm]
(\weightP \otimes \weightP) \ket{\smash{P_{0} L^{g} P_{0}}} & = \ket{\smash{\weightP P_{0} L^{g} P_{0} \weightP}} = (\gamma_{\beta/2})^{1/4} \, \ket{\smash{P_{0} L^{g} P_{0}}},
\end{align*}
i.e. $\ket{P_{1}}$ and $\ket{P_0L^{g}P_0}$ are both eigenvectors of $\weightP \otimes \weightP$.
As a consequence
\begin{equation*}
    \widetilde{\Delta} (\weightP \otimes \weightP) \, = \, (1+ \gamma_{\beta/2})^{1/4} \ket{P_{1}} \bra{P_{1}}  + (\gamma_{\beta/2})^{1/4} \, \frac{1}{|G|} \sum_{g\in G} \, \ket{\smash{P_{0}L^{g} P_{0}}} \bra{\smash{P_{0} L^{g} P_{0}}}
\end{equation*}
and so taking adjoints we immediately get that
\[ \widetilde{\Delta} (\weightP \otimes \weightP)  = (\weightP \otimes \weightP)  \widetilde{\Delta}\,.  \]
This concludes the proof.
\end{proof}

\begin{Prop}\label{Prop:BoundaryProjector3}
The operator
    \begin{equation*} 
\widetilde{\mathcal{S}}_{\partial \RR} \,\, := \,\,  
\begin{tikzpicture}[equation,scale=0.4]
    \plaquettefive[3,2](1,0,0,0);
    \draw (4.5,2.6) node {\scriptsize $\widetilde \Delta^{\otimes \EE_{\partial \RR}}$};
    \end{tikzpicture}
    \otimes \quad
    \begin{tikzpicture}[equation,scale=0.4]
    \plaquettefive[3,2](2,0,0,0);
    \draw (4,2.3) node {\scriptsize $\widetilde \phi_{\widehat{1}}$};
    \end{tikzpicture}
\end{equation*}
is a projection on $\mathcal{H}_{\partial \mathcal{R}}$  satisfying
\[  \widetilde{\rho}_{\partial \mathcal{R}} = \widetilde{\rho}_{\partial \mathcal{R}} \, \widetilde{\mathcal{S}}_{\partial \RR} = \widetilde{\mathcal{S}}_{\partial \RR}\, \widetilde{\rho}_{\partial \mathcal{R}} \quad \quad , \quad \quad \weight_{\partial \mathcal{R}} \, \widetilde{\mathcal{S}}_{\partial \RR} = \widetilde{\mathcal{S}}_{\partial \RR}\, \weight_{\partial \mathcal{R}} \]
\end{Prop}

\begin{proof}
The first statement is clear as it is a tensor product of projections $\widetilde{\Delta}$ and $\widetilde{\phi}_{\widehat{1}}$. Let us check that $\widetilde{\mathcal{S}}_{\partial \RR} \, \widetilde{\rho}_{\partial \mathcal{R}} \, = \, \widetilde{\rho}_{\partial \mathcal{R}}$. In view of the representation of  $\widetilde{\rho}_{\partial \RR}$ given in \eqref{equa:BoundaryProperRectangle1}, it is enough to prove that

\begin{equation}\label{equa:DominantTermBoundaryAux1}
\widetilde{\mathcal{S}}_{\partial \mathcal{R}} \circ \left(   \begin{tikzpicture}[equation, scale=0.4]

\plaquettefive[3,2](1,0,0,0)

\draw (3.5,2.5) node {\scriptsize $\widetilde{\psi}_{\, \widehat{g}}$}; 
\end{tikzpicture}
\otimes \quad
\begin{tikzpicture}[equation, scale=0.4]

\plaquettefive[3,2](2,0,0,0)

\draw (3.5,2.5) node {\scriptsize $\widetilde{\phi}_{\, \widehat{a}}$};

\end{tikzpicture}    \right) 
\,\, = \,\,
\begin{tikzpicture}[equation, scale=0.4]

\plaquettefive[3,2](1,0,0,0)

\draw (3.5,2.5) node {\scriptsize $\widetilde{\psi}_{\, \widehat{g}}$}; 
\end{tikzpicture}
\otimes \quad
\begin{tikzpicture}[equation, scale=0.4]

\plaquettefive[3,2](2,0,0,0)

\draw (3.5,2.5) node {\scriptsize $\widetilde{\phi}_{\, \widehat{a}}$};
\draw (3.5,-0.2) node {\scriptsize $a$}; 

\end{tikzpicture}   
  \end{equation}
for all possible choices of $\widehat{g}$ and $\widehat{a}$.  The composition on the left-hand side of \eqref{equa:DominantTermBoundaryAux1} is again a tensor product of operators of the form $\widetilde{\phi}_{1} \, \widetilde{\phi}_{a} = \widetilde{\phi}_{a}$ and $\widetilde{\Delta} \, \widetilde{\psi}_{g} = \widetilde{\psi}_{g}$, so the equality \eqref{equa:DominantTermBoundaryAux1} holds. Analogously, we can argue that $\widetilde{\rho}_{\partial \mathcal{R}} \, \widetilde{\mathcal{S}}_{\partial \mathcal{R}} = \widetilde{\rho}_{\partial \mathcal{R}}$.

Finally, to see that $\weight_{\partial \mathcal{R}} \, \widetilde{\mathcal{S}}_{\partial \RR} = \widetilde{\mathcal{S}}_{\partial \RR}\, \weight_{\partial \mathcal{R}}$, note that both $\weight_{\partial \mathcal{R}} $ and $\widetilde{\mathcal{S}}_{\partial \RR}$ have a \emph{compatible} tensor product structure, so that their product $\weight_{\partial \mathcal{R}} \, \widetilde{\mathcal{S}}_{\partial \RR}$ is again a tensor product of elements of the form
\[ (\weightP \otimes \weightP)  \widetilde{\Delta} \quad , \quad (\weightS^{3} \otimes \weightS^{3})  \widetilde{\phi}_{1} \quad , \quad (\weightS^{2} \otimes \weightS^{2})  \widetilde{\phi}_{1}\,, \]
and analogously for $\widetilde{\mathcal{S}}_{\partial \RR} \, \weight_{\partial \mathcal{R}}$, 
\[ \widetilde{\Delta}  (\weightP \otimes \weightP)  \quad , \quad \widetilde{\phi}_{1}  (\weightS^{3} \otimes \weightS^{3})  \quad , \quad  \widetilde{\phi}_{1}  (\weightS^{2} \otimes \weightS^{2})\,. \]
By Lemma \ref{Lemm:DeltaProjectorProperties} we know that $(\weightP \otimes \weightP)  \widetilde{\Delta} =  \widetilde{\Delta}  (\weightP \otimes \weightP) $. For the others, we only need to check that
\[ (\weightS \otimes \weightS) \, \widetilde{\phi}_{1} = \sum_{h \in G} (\delta_{h,1} + \gamma_{\beta/2})^{1/4} \ket{h}\hspace{-0.8mm}\ket{h} \bra{h}\hspace{-0.8mm}\bra{h} \, = \, \widetilde{\phi}_{1} \, (\weightS \otimes \weightS)\,.\]
This finishes the proof.
\end{proof}

\subsubsection*{Tool Box 2: Plaquette constants}
We need a couple of auxiliary results. First, let us introduce the notation
\[ \widetilde{\chi}^{reg}(g) = \chi^{reg}(g) - 1 = |G| \delta_{g,1} - 1\,, \quad g \in G\,. \]
If $P_{1}$ is the projection onto $V_{1}$, then
 \[ \widetilde{\chi}^{reg}(g) = \operatorname{Tr}(L^{g}(\mathbbm{1} - P_{1}))\,. \]

\begin{Lemm}\label{Lemm:GatheringPlaquettes}
 Let us fix $u,v \in G$ and complex numbers $a_{0}, b_{0}, a_{1}, b_{1} \in \mathbb{C}$. Then, 
 \[ 
\begin{split} 
  \sum_{g_{1}, \ldots, g_{m} \in G} \left( a_{0} + b_{0} \, \widetilde{\chi}^{reg}(u g_{1} \ldots g_{m})\right)   \,  & \left( a_{1} +  b_{1} \widetilde{\chi}^{reg}(g_{m}^{-1} \ldots g_{2}^{-1} g_{1}^{-1} v)\right) \\ 
 &\quad\quad\quad = |G|^{m} \, \left( a_{0}a_{1} + b_{0} b_{1} \widetilde{\chi}^{reg}(uv) \right)\,.  
\end{split} 
 \]
 \end{Lemm}

 \begin{proof}
 Note that the left-hand side of the equality can be rewritten as
 \[ 
 \begin{split} 
 \sum_{g \in G} \,\, \sum_{\substack{g_{1}, \ldots, g_{m} \in G \\ g_{1} \, \ldots \, g_{m} =g}} \,  & \left(a_{0} +  b_{0} \widetilde{\chi}^{reg}(ug)\right)  \,  \left(a_{1} + b_{1} \widetilde{\chi}^{reg}(g^{-1}v)\right) \\
 & \quad \quad = |G|^{m-1} \, \sum_{g \in G} \left(a_{0} + b_{0} \widetilde{\chi}^{reg}(ug)\right) \, \left(a_{1} + b_{1} \widetilde{\chi}^{reg}(g^{-1}v)\right)
 \end{split} 
 \]
 so we can restrict ourselves to the case $m=1$. First, let us expand 
 \begin{align*}
   &  \sum_{g \in G} \left(a_{0} + b_{0} \widetilde{\chi}^{reg}(ug)\right) \,  \left(a_{1} + b_{1} \widetilde{\chi}^{reg}(g^{-1}v)\right) \\[2mm]
     & \quad \quad \quad \quad  = \, a_{0} a_{1} |G| \, + \, a_{0} b_{1} \sum_{g \in G} \widetilde{\chi}^{reg}(g^{-1}v) \, + \, b_{0}  a_{1} \, \sum_{g \in G} \widetilde{\chi}^{reg}(ug)\\[2mm] 
     & \quad \quad \quad \quad \quad \quad \quad \quad   \quad \quad \quad + \, b_{0} b_{1} \sum_{g \in G} \widetilde{\chi}^{reg}(ug) \, \widetilde{\chi}^{reg}(g^{-1}v)\,.
 \end{align*}
The second and third summands are equal to zero, since
\[ \sum_{h \in G} \widetilde{\chi}^{reg}(h) = \sum_{h \in G} \left( |G| \, \delta_{h,1} -1 \right) = |G| - |G| = 0\,. \]
Moreover, 
 \begin{align*}
 \sum_{g \in G} \widetilde{\chi}^{reg}(ug) \, \widetilde{\chi}^{reg}(g^{-1}v) & = \sum_{g \in G} (|G| \, \delta_{ug,1} - 1) \, (|G| \delta_{g^{-1}v,1} - 1)\\[2mm]
 & = \sum_{g \in G} |G|^{2} \, \delta_{ug,1} \delta_{g^{-1}v, 1}  - |G| \, \sum_{g \in G} \big(\delta_{ug,1} + \delta_{g^{-1}v,1} \big) + |G|\\[2mm]
 & = |G|^{2} \delta_{uv, 1} \, - 2|G| + |G|\\[3mm]
 & = |G| \, \big( |G| \delta_{uv,1}  - 1\big) = |G| \, \big( \chi^{reg}(uv) - 1 \big)\,,
 \end{align*}
so the fourth summand in the aforementioned expansion is equal to $b_{0} b_{1} \widetilde{\chi}^{reg}(uv)$, leading to the desired statement.
\end{proof}

\noindent Next we need a result which help us to deal with the constants
\[ c_{\beta}( \, \widehat{g} \, )  \, := \, \prod_{ \substack{\text{$p$ inner}\\ \text{plaquette}} } \left( 1 + \gamma_{\beta} \, \chi^{reg}(\widehat{g}|_{p}) \, \right) \]

\noindent that appear in $\widetilde{\rho}_{\partial \RR}$, see \eqref{equa:BoundaryProperRectangle1}. Let us first extend the notation $\chi^{reg}(\widehat{g}|_{p})$ from plaquettes to the boundary of the rectangle $\RR$: let us enumerate its boundary edges $\EE_{\partial \RR}$ counterclockwise by fixing any initial edge $e_{1}, e_{2}, \ldots, e_{k}$. Given a map $\widehat{f} : \EE_{\partial \RR} \longrightarrow G$ associating $e_{j} \longmapsto g_{j}$ let us define
\[ 
\begin{tikzpicture}[equation, scale=0.7]

\draw[gray, thick] (0,0) grid (3,2);
\foreach \x/\y [count=\i] in { 2.5/2.4, 1.5/2.4, 0.5/2.4, -0.4/1.5, -0.4/0.5} {
    \draw (\x,\y) node {\small $g_{\i}$}; 
}
\foreach \x/\y [count=\i from 6] in { 0.5/-0.4, 1.5/-0.4, 2.5/-0.4, 3.6/0.5, 3.6/1.5 } {
    \draw (\x,\y) node {\small $g_{\i}$}; 
}
\end{tikzpicture} \quad \quad
 \chi^{reg}(\widehat{f}):= \chi^{reg}(g_{1}^{\sigma_{1}} g_{2}^{\sigma_{2}} \ldots)   \]
where  $\sigma_{j} = 1$ if $e_{j}$ is in the upper or left side of the rectangle, and $\sigma_{j} = -1$ otherwise. We will also use the notation \[ \widetilde{\chi}^{reg}(\widehat{f}) := \chi^{reg}(\widehat{f})-1\,. \]
As in the case of plaquettes, note that the definition is independent of the enumeration of the edges as long as it is done counterclockwise ($\chi^{reg}$ is invariant under cyclic permutations) and respects the inverses ($\sigma_{j}=-1$) on the lower and right edges.

\begin{Prop}\label{Prop:arranging_plaquettes}
For fixed $\widehat{f}: \EE_{\partial \RR} \longrightarrow G$ we have that 

\[ 
\sum_{\substack{\widehat{g}: \EE_{\RR} \longrightarrow G\\ \widehat{g}|_{\EE_{\partial \RR}}  =  \widehat{f} }} c_{\beta}(\widehat{g}) \,\, = \,\, |G|^{\abs{\EE_{\mathring{\RR}}}} \, \left( \left(1 + \gamma_\beta \right)^{n_{\RR}}  +  \gamma_\beta^{n_{\RR}} \,\, \widetilde{\chi}^{reg}(\widehat{f}) \right)\,. \]

\end{Prop}

\begin{proof}
If the region $\RR$ only consists of one plaquette, then the identity is trivial since there is only one summand $\widehat{g} = \widehat{f}$, so that
\[ c_{\beta}(\widehat{g}) \, = \, 1 + \gamma_{\beta} \chi^{reg}(\widehat{g}|_{p}) = \left( 1 + \gamma_{\beta}\right) + \gamma_{\beta} \,\, \widetilde{\chi}^{reg}(\widehat{g}|_{p})\,. \]

\noindent Let us denote to simplify notation
\[ a_{\beta} := 1 + \gamma_{\beta} \quad , \quad b_{\beta}:= \gamma_{\beta}\,. \]

\noindent For multiplaquette rectangular regions, the key is Lemma \ref{Lemm:GatheringPlaquettes}. Let us illustrate the case of a region $\RR$ consisting  of two (adjacent) plaquettes. 

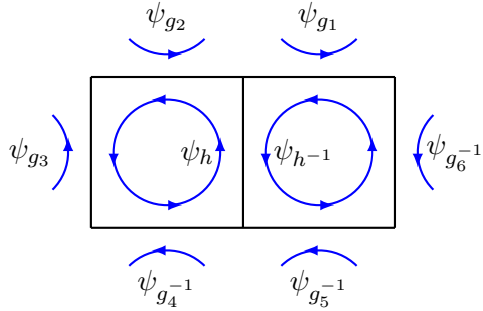
\[
\begin{tikzpicture}[equation, scale=2, thick]
     \plaquettefive[2,1](1,1,1,0);
    \foreach \x/\y [count=\i] in {1.5/1.4, 0.5/1.4, -0.4/0.5} {
        \draw (\x,\y) node {$\psi_{g_{\i}}$}; 
    }
    \foreach \x/\y [count=\i from 4] in {0.5/-0.4, 1.5/-0.4, 2.4/0.5} {
        \draw (\x,\y) node {$\psi_{g_{\i}^{-1}}$}; 
    }
    \draw (0.7,0.5) node {$\psi_h$};
    \draw (1.4, 0.5) node {$\psi_{h^{-1}}$};
\end{tikzpicture}
\]

\noindent Then, as a consequence of the aforementioned lemma
\[ 
\begin{split}
\sum_{h \in G} & \left(1 + \gamma_{\beta} \chi^{reg} (g_{2}g_{3}g_{4}^{-1}h^{-1})\right)  \left(1 + \gamma_{\beta} \chi^{reg}(hg_{5}^{-1}g_{6}^{-1}g_{1}) \right) = \\[2mm]
& \quad \quad  = \, \sum_{h \in G}  \left(a_{\beta} + b_{\beta} \, \widetilde{\chi}^{reg} (g_{2}g_{3}g_{4}^{-1}h^{-1})\right)  \left(a_{\beta} + b_{\beta} \, \widetilde{\chi}^{reg}(hg_{5}^{-1}g_{6}^{-1}g_{1}) \right)  \\[2mm]
& \quad \quad= \, |G| \, \Big( a_{\beta}^{2}  +  \, b_{\beta}^{2} \, \widetilde{\chi}^{reg} \, (g_{2}g_{3}g_{4}^{-1}g_{5}^{-1}g_{6}^{-1}g_{1}) \Big)\,.
\end{split}
\]
The procedure is now clear and can be formalized by induction. Let $\mathcal{R}$ be a rectangular region. We can obviously split $\mathcal{R}$ into two adjacent rectangles sharing one side $\mathcal{R}=AC$ as below. The induction hypothesis yields that the identity is true for $A$ and $C$. Let us set some notation for the boundary edges of $A$ and $C$:
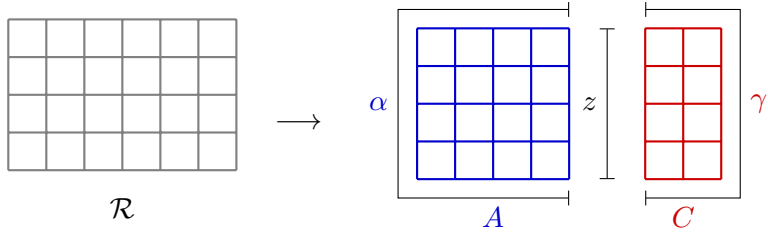
\[ \alpha:=\EE_{\partial A} \setminus \EE_{\partial C} \quad, \quad \gamma:= \EE_{\partial C} \setminus \EE_{\partial A} \quad , \quad z := \EE_{\partial A} \cap \EE_{\partial C}\,. \]
\[
    \begin{tikzpicture}[equation, scale=0.5, baseline={([yshift=-.5ex]current bounding box.center)}]
        \draw[gray, thick] (0,0) grid (6,4);
        \draw (3,-1) node {$\RR$};
    \end{tikzpicture}
    \quad \longrightarrow \quad 
    \begin{tikzpicture}[scale=0.5, baseline={([yshift=-.5ex]current bounding box.center)}]
    \draw[blue!80!black, thick] (0,0) grid (4,4);
    \draw[blue!80!black] (2,-1) node {$A$};
    \draw[red!80!black, thick] (6,0) grid (8,4);
    \draw[red!80!black] (7,-1) node {$C$};
    \draw[|-|] (4,-0.5) -- (-0.5, -0.5) --
        node[left, blue!80!black] {$\alpha$} (-0.5, 4.5) -- (4,4.5);
    \draw[|-|] (6, -0.5) -- (8.5, -0.5) --
        node[right, red!80!black] {$\gamma$} (8.5, 4.5) -- (6, 4.5);
    \draw[|-|] (5,0) -- 
        node[left] {$z$} (5,4);
    \end{tikzpicture}
\]

\noindent Fix a boundary function $\widehat{f} : \EE_{\partial \RR} \rightarrow G$ (note that $\EE_{\partial \RR} = \alpha \gamma$) let us consider the sum over all choices $\widehat{g}: \EE_{\RR} \longrightarrow G$ such that $\widehat{g}$ coincides with $\widehat{f}$ on the boundary edges of $\RR$
\begin{align}
  \sum_{\substack{\widehat{g}: \EE_{\RR} \rightarrow G \\ \widehat{g}|_{\partial \RR} = \widehat{f}}} c_{\beta}(\widehat{g}) \, & = \, \sum_{\substack{\widehat{g}: \EE_{\RR} \rightarrow G \\ \widehat{g}|_{\alpha \gamma} = \widehat{f}}} \, \prod_{\substack{p \subset \RR\\ \text{plaquette}}} \left( a_{\beta} + b_{\beta} \widetilde{\chi}^{reg}(\widehat{g}|_{p}) \right) = \notag \\[2mm] 
& = \sum_{\substack{\widehat{g}: \EE_{\RR} \rightarrow G \\ \widehat{g}|_{\alpha \gamma} = \widehat{f}}} \, \prod_{\substack{p \subset A\\ \text{plaquette}}} \left( a_{\beta} + b_{\beta} \widetilde{\chi}^{reg}(\widehat{g}|_{p}) \right) \, \prod_{\substack{p \subset C\\ \text{plaquette}}} \left( a_{\beta} + b_{\beta} \widetilde{\chi}^{reg}(\widehat{g}|_{p}) \right)\,.
\label{equa:PropChiBoundaryAux1} 
\end{align}
Next we want to apply the induction hypothesis on each subregion $A$ anc $C$. For that, we have to make the right expression appear. We can split the sum
\begin{equation}
 \sum_{\substack{\widehat{g}: \EE_{\RR} \rightarrow G \\[0.1mm] \widehat{g}|_{\alpha \gamma} = \widehat{f}}} \, = \, 
\sum_{\widehat{h}: z \rightarrow G }
\,\,
 \sum_{\substack{\widehat{g}: \EE_{\RR} \rightarrow G \\[0.1mm] \widehat{g}|_{\alpha \gamma} = \widehat{f} \\[0.1mm] \widehat{g}|_{z} = \widehat{h}}} \,\, = \,\, 
 \sum_{\widehat{h}: z \rightarrow G } 
 \,\,\, 
 \sum_{\substack{\widehat{g}: \EE_{A} \rightarrow G \\[0.1mm] \widehat{g}|_{\alpha} = \widehat{f}|_{\alpha} \\[0.1mm] \widehat{g}|_{z} = \widehat{h}}}  
\,\,\,
\sum_{\substack{\widehat{g}: \EE_{C} \rightarrow G \\[0.1mm] \widehat{g}|_{ \gamma} = \widehat{f}|_{\gamma} \\[0.1mm] \widehat{g}|_{z} = \widehat{h}}}  \,\,.
   \end{equation}

\noindent Hence, we can rewrite \eqref{equa:PropChiBoundaryAux1} as

\begin{align*}
 \sum_{\substack{\widehat{g}: \EE_{\RR} \rightarrow G \\[0.1mm] \widehat{g}|_{\alpha \gamma} = \widehat{f}}} \, &  \prod_{\substack{p \subset \RR\\ \text{plaquette}}} \left( a_{\beta} + b_{\beta} \widetilde{\chi}^{reg}(\widehat{g}|_{p}) \right) = \\[2mm] 
& = \sum_{\widehat{h}: z \rightarrow G } \Big( \sum_{\substack{\widehat{g}:\EE_{A} \rightarrow G \\[0.1mm] \widehat{g}|_{\alpha} = \widehat{f}|_{\alpha} \\[0.1mm] \widehat{g}|_{z}= \widehat{h}} } \,\,\prod_{\substack{p \subset A\\ \text{plaquette}}} \left( a_{\beta} + b_{\beta} \widetilde{\chi}^{reg}(\widehat{g}|_{p}) \right) \Big) 
\\
& \hspace{5cm} \cdot  \Big( \sum_{\substack{\widehat{g}:\EE_{C} \rightarrow G \\[0.1mm] \widehat{g}|_{\gamma} = \widehat{f}|_{\gamma} \\[0.1mm] \widehat{g}|_{z}= \widehat{h}} } \,\,\prod_{\substack{p \subset C\\ \text{plaquette}}} \left( a_{\beta} + b_{\beta} \widetilde{\chi}^{reg}(\widehat{g}|_{p}) \right) \Big)\\[2mm]
\end{align*}

\begin{center}
    \begin{tikzpicture}[scale=0.5]
    
    \plaquettefive[4,4](1,0,0,0)
   
    \foreach \y in {0.5,1.5,2.5,3.5}
    \centerarc[postaction={decorate},magenta!90!black,very thick](4.5,\y)(135:225:0.7*0.5);  
   
    \draw (2,-0.7) node {$A$};
    \draw (-1, 4) node {$\psi_{\widehat{f}}$};
    \draw (5.5,3) node[magenta!90!black] {$\psi_{\widehat{h}}$};
    \end{tikzpicture}
    \begin{tikzpicture}[scale=0.5]
    \plaquettefive[2,4](1,0,0,0);
    \foreach \y in {0.5,1.5,2.5,3.5}
    \centerarc[postaction={decorate},magenta!90!black,very thick](-0.5,\y)(-45:45:0.7*0.5);

    \draw (1,-0.7) node {$C$};
    \draw (3.3, 4) node {$\psi_{\widehat{f}}$};
    \end{tikzpicture}
\end{center}

\noindent Applying the induction hypothesis we have that the above expression can be rewritten as

\begin{align*}
 & \sum_{\widehat{h}: z \rightarrow G } |G|^{|\EE_{\mathring{A}}|} \left( a_{\beta}^{n_{A}} + b_{\beta}^{n_{A}} \, \widetilde{\chi}^{reg}( \widehat{g}|_{\alpha} \widehat{h} \, ) \right) |G|^{|\EE_{\mathring{C}}|}  \left( a_{\beta}^{n_{C}} + b_{\beta}^{n_{C}} \, \widetilde{\chi}^{reg}(\, \widehat{h} \, \widehat{g}|_{\gamma} ) \right) \\[2mm]
 & \quad \quad \quad= |G|^{|\EE_{\mathring{A}}| + |\EE_{\mathring{C}}|} 
 \,\, 
 \sum_{\widehat{h}: z \rightarrow G } 
 \left( a_{\beta}^{n_{A}} + b_{\beta}^{n_{A}} \,\, \widetilde{\chi}^{reg}(\widehat{g}|_{\alpha} \widehat{h} \, ) \right) 
 \,\,
 \left( a_{\beta}^{n_{C}} + b_{\beta}^{n_{C}} \,\, \widetilde{\chi}^{reg}( \, \widehat{h} \, \widehat{g}|_{\gamma} ) \right)\,,
 \end{align*}
 
\noindent where $n_{A}$ and $n_{C}$ are the number of inner plaquettes of $A$ and $C$ respectively; $\widehat{g}|_{\alpha} \widehat{h}$ is the map $\EE_{\partial A} \rightarrow G$ that coincides with $\widehat{g}$ on $\alpha$ and with $\widehat{h}$ on $z$; and  $\widehat{h} \, \widehat{g}|_{\gamma}$ is the map $\EE_{\partial C} \rightarrow G$ that coincides with $\widehat{g}$ on $\gamma$ and with $\widehat{h}$ on $z$. Finally we use Lemma \ref{Lemm:GatheringPlaquettes} to rewrite the above expression as
 
\begin{equation*}
 |G|^{|\EE_{\mathring{A}}| + |\EE_{\mathring{C}}| + |z|} \, \left( a_{\beta}^{n_{A} + n_{C}} + b_{\beta}^{n_{A} + n_{C}} \,\, \widetilde{\chi}^{reg}(\widehat{g}|_{\alpha} \widehat{g}|_{\gamma}) \right) \,.
\end{equation*}

\noindent Finally observe that the set $\EE_{\mathring{\RR}}$ of inner edges of $\RR$, is actually formed by the disjoint union of $\EE_{\mathring{A}}$, $\EE_{\mathring{C}}$ and $z$, whereas the number $n_{\RR}$ of inner plaquettes of $\RR$ is indeed the sum of $n_{A}$ and $n_{C}$, so that 

\[ 
 \sum_{\substack{\widehat{g}: \EE_{\RR} \rightarrow G \\ \widehat{g}|_{a \gamma} = \widehat{f}}} \,  \prod_{\substack{p \subset \RR\\ \text{plaquette}}} \left( a_{\beta} + b_{\beta} \,\, \widetilde{\chi}^{reg}(\widehat{g}|_{p}) \right) \,\, = \,\, |G|^{|\EE_{\mathring{\RR}}|} \, \left( a_{\beta}^{n_{\RR}} + b_{\beta}^{n_{\RR}} \,\, \widetilde{\chi}^{reg}(\widehat{f}) \right)\,. 
\]

\noindent This finishes the proof of the Proposition.

\end{proof}

\subsubsection*{Proof of the \emph{Leading Term of the boundary} Theorem \ref{Theo:leadingTerm}}

We are now ready to prove our main result about the boundary state of the thermofield double.
\noindent Our starting point is \eqref{equa:BoundaryProperRectangle1}. We can replace the sum over $\widehat{g}:\EE_{\RR} \rightarrow G$ with a double sum, the first one over $\widehat{f}: \EE_{\partial \RR} \rightarrow G$ (fixing first the values at the boundary), and the second one over $\widehat{g}:\EE_{\RR} \rightarrow G$ with $\widehat{g}|_{\EE_{\partial \RR}}=\widehat{f}$, so that
\begin{equation*}\label{equa:BoundaryProperRectangle2}
\widetilde{\rho}_{\partial \RR} \,\, = \,\, \sum_{a \in G} \,\,  \left( \delta_{a,1} + \gamma_{\beta} \right)^{|\VV_{\mathring{\RR}}|} \sum_{ \widehat{f}: \EE_{\partial \RR} \rightarrow G } \,\, c_{\beta}( \, \widehat{f} \, ) \quad \,\,
\begin{tikzpicture}[equation, scale=0.4]

\plaquettefive[3,2](1,0,0,0)

\draw (3.5,2.5) node {\scriptsize $\widetilde{\psi}_{\, \widehat{f}}$}; 
\end{tikzpicture}
\otimes \quad
\begin{tikzpicture}[equation, scale=0.4]

\plaquettefive[3,2](2,0,0,0)

\draw (3.5,2.5) node {\scriptsize $\widetilde{\phi}_{\, \widehat{a}}$};
\draw (3.5,-0.2) node {\scriptsize $a$}; 

\end{tikzpicture}    
\end{equation*}
where, from Proposition~\ref{Prop:arranging_plaquettes}, it follows that
\begin{equation}\label{eq:cf-definition}
c_{\beta}(\widehat{f}) \, := \, \displaystyle \sum_{ \substack{\widehat{g}: \EE_{\RR} \rightarrow G\\[1mm] \widehat{g}|_{\EE_{\partial \RR}}=\widehat{f} } } \,\, c_{\beta}( \, \widehat{g} \, ) \, \, = \, \, \abs{G}^{\abs{\EE_{\mathring{\RR}}}} \,\big[ (1+ \gamma_{\beta})^{n_{\RR}} + \gamma_{\beta}^{n_{\RR}} \, \widetilde{\chi}^{reg}(\widehat{f} \,) \big]\,. 
\end{equation}
Let us now split the sum over $a \in G$ into a first summand with $a = 1$ (which forces $\widehat{a}(v) = 1$ for every $v \in \VV_{\RR}$, as we argued in previous sections) and $a \neq 1$:
\begin{align*} 
\widetilde{\rho}_{\partial \RR} \,\, & = \,\, \left( 1 + \gamma_{\beta} \right)^{\abs{\VV_{\mathring{\RR}}}} \sum_{ \widehat{f}: \EE_{\partial \RR} \rightarrow G } \,\, c_{\beta}( \, \widehat{f} \, ) \quad 
\begin{tikzpicture}[equation, scale=0.4]
\plaquettefive[3,2](1,0,0,0)
\draw (3.5,2.5) node {\scriptsize $\widetilde{\psi}_{\, \widehat{f}}$}; 
\end{tikzpicture}
\otimes \quad
\begin{tikzpicture}[equation, scale=0.4]
\plaquettefive[3,2](2,0,0,0)
\draw (3.5,2.5) node {\scriptsize $\widetilde{\phi}_{\, \widehat{1}}$};
\end{tikzpicture}  +  \\
& \quad \quad  + \left( \gamma_{\beta} \right)^{\abs{\VV_{\mathring{\RR}}}} \,\,   \sum_{ \widehat{f}: \EE_{\partial \RR} \rightarrow G } \,\, c_{\beta}( \, \widehat{f} \, ) \quad 
\begin{tikzpicture}[equation, scale=0.4]
\plaquettefive[3,2](1,0,0,0)
\draw (3.5,2.5) node {\scriptsize $\widetilde{\psi}_{\, \widehat{f}}$}; 
\end{tikzpicture}
\otimes \quad \sum_{a \in G}  \,\,
\begin{tikzpicture}[equation, scale=0.4]
\plaquettefive[3,2](2,0,0,0)
\draw (3.5,2.5) node {\scriptsize $\widetilde{\phi}_{\, \widehat{a}}$};
\draw (3.5,-0.2) node {\scriptsize $a$}; 
\end{tikzpicture}\\  
&  \quad \quad -\left(\gamma_{\beta} \right)^{\abs{\VV_{\mathring{\RR}}}} \sum_{ \widehat{f}: \EE_{\partial \RR} \rightarrow G } \,\, c_{\beta}( \, \widehat{f} \, ) \quad 
\begin{tikzpicture}[equation, scale=0.4]
\plaquettefive[3,2](1,0,0,0)
\draw (3.5,2.5) node {\scriptsize $\widetilde{\psi}_{\, \widehat{f}}$}; 
\end{tikzpicture}
\otimes \quad
\begin{tikzpicture}[equation, scale=0.4]
\plaquettefive[3,2](2,0,0,0)
\draw (3.5,2.5) node {\scriptsize $\widetilde{\phi}_{\, \widehat{1}}$};
\end{tikzpicture}   
\end{align*}
In the first summand, we can moreover decompose
\begin{align*}
 \sum_{ \widehat{f}: \EE_{\partial \RR} \rightarrow G } \,\, c_{\beta}( \, \widehat{f} \, ) \quad 
\begin{tikzpicture}[equation, scale=0.4]
\plaquettefive[3,2](1,0,0,0)
\draw (3.5,2.5) node {\scriptsize $\widetilde{\psi}_{\, \widehat{f}}$}; 
\end{tikzpicture} \,\, & = \,\,  |G|^{\abs{\EE_{\mathring\RR}}} \,\, (1+\gamma_{\beta})^{n_{\RR}} \sum_{\widehat{f}:\EE_{\partial \RR} \longrightarrow G} \quad 
\begin{tikzpicture}[equation, scale=0.4]
\plaquettefive[3,2](1,0,0,0)
\draw (3.5,2.5) node {\scriptsize $\widetilde{\psi}_{\, \widehat{f}}$}; 
\end{tikzpicture}\\
& 
\, + \,
 |G|^{\abs{\EE_{\mathring\RR}}} \, \gamma_{\beta}^{n_{\RR}} \, \sum_{\widehat{f}:\EE_{\partial \RR} \longrightarrow G}  \widetilde{\chi}^{reg}(\widehat{f}) \quad 
\begin{tikzpicture}[equation, scale=0.4]
\plaquettefive[3,2](1,0,0,0)
\draw (3.5,2.5) node {\scriptsize $\widetilde{\psi}_{\, \widehat{f}}$}; 
\end{tikzpicture}
\end{align*}
Thus, combining both expressions
\begin{align*} 
\widetilde{\rho}_{\partial \RR} \,\, & = \,\, \left( 1 + \gamma_{\beta} \right)^{\abs{\VV_{\mathring{\RR}}}} \left( 1 + \gamma_{\beta} \right)^{n_{\RR}} |G|^{\abs{\EE_{\mathring \RR}}} \sum_{ \widehat{f}: \EE_{\partial \RR} \rightarrow G } \,\, 
\begin{tikzpicture}[equation, scale=0.4]
\plaquettefive[3,2](1,0,0,0)
\draw (3.5,2.5) node {\scriptsize $\widetilde{\psi}_{\, \widehat{f}}$}; 
\end{tikzpicture}
\otimes \quad
\begin{tikzpicture}[equation, scale=0.4]
\plaquettefive[3,2](2,0,0,0)
\draw (3.5,2.5) node {\scriptsize $\widetilde{\phi}_{\, \widehat{1}}$};
\end{tikzpicture}   \\
& \quad + \,\, \left( 1 + \gamma_{\beta} \right)^{\abs{ \VV_{\mathring{\RR}}}} \left(\gamma_{\beta} \right)^{n_{\RR}} |G|^{\abs{\EE_{\mathring \RR}}} \sum_{ \widehat{f}: \EE_{\partial \RR} \rightarrow G } \,\, \widetilde{\chi}^{reg}(\widehat{f}) \quad 
\begin{tikzpicture}[equation, scale=0.4]
\plaquettefive[3,2](1,0,0,0)
\draw (3.5,2.5) node {\scriptsize $\widetilde{\psi}_{\, \widehat{f}}$}; 
\end{tikzpicture}
\otimes \quad
\begin{tikzpicture}[equation, scale=0.4]
\plaquettefive[3,2](2,0,0,0)
\draw (3.5,2.5) node {\scriptsize $\widetilde{\phi}_{\, \widehat{1}}$};
\end{tikzpicture}  \\
& \quad \quad  + \left( \gamma_{\beta} \right)^{\abs{\VV_{\mathring{\RR}}}} \,\,   \sum_{ \widehat{f}: \EE_{\partial \RR} \rightarrow G } \,\, c_{\beta}( \, \widehat{f} \, ) \quad 
\begin{tikzpicture}[equation, scale=0.4]
\plaquettefive[3,2](1,0,0,0)
\draw (3.5,2.5) node {\scriptsize $\widetilde{\psi}_{\, \widehat{f}}$}; 
\end{tikzpicture}
\otimes \quad \sum_{a \in G}  \,\,
\begin{tikzpicture}[equation, scale=0.4]
\plaquettefive[3,2](2,0,0,0)
\draw (3.5,2.5) node {\scriptsize $\widetilde{\phi}_{\, \widehat{a}}$};
\draw (3.5,-0.2) node {\scriptsize $a$}; 
\end{tikzpicture}\\  
&  \quad \quad -\left(\gamma_{\beta} \right)^{\abs{\VV_{\mathring{\RR}}}} \sum_{ \widehat{f}: \EE_{\partial \RR} \rightarrow G } \,\, c_{\beta}( \, \widehat{f} \, ) \quad 
\begin{tikzpicture}[equation, scale=0.4]
\plaquettefive[3,2](1,0,0,0)
\draw (3.5,2.5) node {\scriptsize $\widetilde{\psi}_{\, \widehat{f}}$}; 
\end{tikzpicture}
\otimes \quad
\begin{tikzpicture}[equation, scale=0.4]
\plaquettefive[3,2](2,0,0,0)
\draw (3.5,2.5) node {\scriptsize $\widetilde{\phi}_{\, \widehat{1}}$};
\end{tikzpicture}   
\end{align*}

\noindent Dividing the above expression by 
\[ \kappa_{\RR} = \abs{G}^{\abs{\EE_{\RR}}} \, (1+ \gamma_{\beta})^{\abs{\VV_{\mathring \RR}}+n_{\RR}} \, = \, \abs{G}^{\abs{\EE_{\mathring \RR}} + \abs{\EE_{\partial \RR}}} \, (1+ \gamma_{\beta})^{\abs{\VV_{\mathring \RR}}+n_{\RR}} \,, \]
we obtain
\begin{equation}\label{equa:BoundaryFourTermsDecompositionAux} 
\frac{1}{\kappa_{\RR}} \widetilde{\rho}_{\partial \RR}  =  \mathcal{S}_{1} + \left(\frac{\gamma_{\beta}}{1+\gamma_{\beta}}\right)^{n_{\RR}} \, \mathcal{S}_{2} + \, \left(\frac{\gamma_{\beta}}{1+\gamma_{\beta}}\right)^{|\VV_{\mathring\RR}|}  \, \mathcal{S}_{3} - \, \left(\frac{\gamma_{\beta}}{1+\gamma_{\beta}}\right)^{|\VV_{\mathring\RR}|}  \, \mathcal{S}_{4} 
\end{equation}
where
\begin{align*}
    \SSS_{1} & \,\,  \,\, = \,\, \sum_{\widehat{f}: \EE_{\partial \RR} \longrightarrow G} \,\,\,\frac{1}{\abs{G}^{\abs{\EE_{\partial \RR}}}} \,\,\,\,
    \begin{tikzpicture}[equation,scale=0.4]
    \plaquettefive[3,2](1,0,0,0);
    \draw (4.2,2.6) node {\scriptsize $\widetilde{\psi}_{\widehat f}$};
    \end{tikzpicture}
    \otimes \quad
    \begin{tikzpicture}[equation,scale=0.4]
    \plaquettefive[3,2](2,0,0,0);
    \draw (4,2.3) node {\scriptsize $\widetilde \phi_{\widehat{1}}$};
    \end{tikzpicture} \,\,, \\
    \SSS_{2} & \,\, = \, \,\,
    \sum_{ \widehat{f}: \EE_{\partial \RR} \rightarrow G } \widetilde{\chi}^{reg}(\widehat f)\,\,\,\, \frac{1}{\abs{G}^{\abs{\EE_{\partial \RR}}}} \,\,\,\,
    \begin{tikzpicture}[equation,scale=0.4]
    \plaquettefive[3,2](1,0,0,0)
    \draw (4,2.6) node {\scriptsize $\widetilde{\psi}_{\,     \widehat{f}}$};
    \end{tikzpicture}
    \otimes \quad
    \begin{tikzpicture}[equation,scale=0.4]
    \plaquettefive[3,2](2,0,0,0);
    \draw (4.1,2.8) node {\scriptsize $\widetilde \phi_{\widehat{1}}$};
    \end{tikzpicture} \,\,, \\
    \SSS_{3} & \,\, = 
    \,\,
    \sum_{ \widehat{f}: \EE_{\partial \RR} \rightarrow G } \,\, \frac{c_{\beta}( \, \widehat{f} \, )}{(1+\gamma_{\beta})^{n_{\mathcal{R}}}\, \abs{G}^{\abs{\EE_{\mathring \RR}}}}\,\, \frac{1}{\abs{G}^{\abs{\EE_{\partial \RR}}}}  \,\,\,\,
    \begin{tikzpicture}[equation,scale=0.4]
    \plaquettefive[3,2](1,0,0,0)
    \draw (4,2.7) node {\scriptsize $\widetilde{\psi}_{\,     \widehat{f}}$};
    \end{tikzpicture}
    \otimes  \sum_{a \in G} \,\, \begin{tikzpicture}[equation,scale=0.4]
    \plaquettefive[3,2](2,0,0,0);
    \draw (4.1,2.8) node {\scriptsize $\widetilde \phi_{\widehat{a}}$};
    \draw (3.6,0.2) node {\scriptsize $a$}; 
    \end{tikzpicture} \,\,, \\
    \SSS_{4} & \,\, = \,\,\,\,
    \sum_{ \widehat{f}: \EE_{\partial \RR} \rightarrow G } \,\, \frac{c_{\beta}( \, \widehat{f} \, )}{(1+\gamma_{\beta})^{n_{\mathcal{R}}}\, \abs{G}^{\abs{\EE_{\mathring \RR}}}}\,\, \frac{1}{\abs{G}^{\abs{\EE_{\partial \RR}}}}  \,\,\,\,
    \begin{tikzpicture}[equation,scale=0.4]
    \plaquettefive[3,2](1,0,0,0)
    \draw (4,2.7) node {\scriptsize $\widetilde{\psi}_{\,     \widehat{f}}$};
    \end{tikzpicture}
    \otimes \,\,\,\, \begin{tikzpicture}[equation,scale=0.4]
    \plaquettefive[3,2](2,0,0,0);
    \draw (4.1,2.8) node {\scriptsize $\widetilde \phi_{\widehat{1}}$}; 
    \end{tikzpicture} \,\,.\\
\end{align*}

\noindent Note that the first summand actually corresponds to
\[
\SSS_{1} \, =  \,\,\,
\begin{tikzpicture}[equation,scale=0.4]
\plaquettefive[3,2](1,0,0,0)
\draw (4.5,2.6) node {\scriptsize $\widetilde{\Delta}^{\otimes \EE_{\partial \RR}}$};
\end{tikzpicture}
 \otimes \,\, \begin{tikzpicture}[equation,scale=0.4]
    \plaquettefive[3,2](2,0,0,0);
    \draw (4.1,2.8) node {\scriptsize $\widetilde \phi_{\widehat{1}}$}; 
    \end{tikzpicture} 
\, = \, \widetilde{\mathcal{S}}_{\partial \RR} \,.
\]
To estimate the norm of $\mathcal{S}_{2}$ and $\mathcal{S}_{4}$ we are going to use that the operators
\[ 
\frac{1}{\abs{G}^{\abs{\EE_{\partial \RR}}}}  \,\,\,\,
    \begin{tikzpicture}[equation,scale=0.4]
    \plaquettefive[3,2](1,0,0,0)
    \draw (4,2.7) node {\scriptsize $\widetilde{\psi}_{\,     \widehat{f}}$};
    \end{tikzpicture}
    \otimes \,\,\,\, \begin{tikzpicture}[equation,scale=0.4]
    \plaquettefive[3,2](2,0,0,0);
    \draw (4.1,2.8) node {\scriptsize $\widetilde \phi_{\widehat{1}}$}; 
    \end{tikzpicture} 
    \quad , \quad \widehat{f}:\EE_{\partial \RR} \longrightarrow G\,.
\]
are (orthogonal) projections and mutually orthogonal, see Proposition \ref{Prop:BoundaryProjector1}. Then, we can estimate
\[ \| \SSS_{2} \| \,\, \leq \,\, \sup_{\widehat{f}: \EE_{\partial \RR} \longrightarrow G} \,\,\, |\widetilde{\chi}^{\, reg}(\widehat{f} \,)| \leq |G|\,, \]
and by recalling that Proposition~\ref{Prop:arranging_plaquettes} implies the formula for $c_\beta(\widehat{f})$ given in \eqref{eq:cf-definition}, we obtain
\begin{align*}
\| \SSS_{4}\|  \,\, \leq \,\, \sup_{\widehat{f}: \EE_{\RR} \longrightarrow G}  \,\,\, \frac{c_{\beta}( \, \widehat{f} \, )}{(1+\gamma_{\beta})^{n_{\RR}}\, \abs{G}^{\abs{\EE_{\mathring \RR}}}} \leq 1 + \left( \frac{\gamma_{\beta}}{1+\gamma_{\beta}}\right)^{n_{\RR}} \, \abs{G}\,. 
\end{align*}
To estimate the norm of $\mathcal{S}_{3}$ we are going to use that the operators
\[ 
\frac{1}{\abs{G}^{\abs{\EE_{\partial \RR}}}}  \,\,\,\,
    \begin{tikzpicture}[equation,scale=0.4]
    \plaquettefive[3,2](1,0,0,0)
    \draw (4,2.7) node {\scriptsize $\widetilde{\psi}_{\,     \widehat{g}}$};
    \end{tikzpicture}
    \otimes \,\, \frac{1}{|G|} \, \sum_{a \in G} \,\, \begin{tikzpicture}[equation,scale=0.4]
    \plaquettefive[3,2](2,0,0,0);
    \draw (4.1,2.8) node {\scriptsize $\widetilde \phi_{\widehat{a}}$};
    \draw (3.6,0.2) node {\scriptsize $a$}; 
    \end{tikzpicture}
    \quad , \quad \widehat{f}:\EE_{\partial \RR} \longrightarrow G
\]
are (orthogonal) projections and mutually orthogonal, see Proposition \ref{Prop:BoundaryProjector2}. Then again by using \eqref{eq:cf-definition},
\[ \| \SSS_{3}\| = \sup_{\widehat{f}: \EE_{\partial \RR} \longrightarrow G} \,\, \frac{|G| \, c_{\beta}(\widehat{f})}{(1+\gamma_{\beta})^{n_{\RR}} \abs{G}^{\abs{\EE_{\mathring \RR}}}} \, \leq \, |G|\, \left[1 + \left( \frac{\gamma_{\beta}}{1+ \gamma_{\beta}}\right)^{n_{\beta}}\abs{G}\right] \, . \]

\noindent Combining all these estimates, we can reformulate \eqref{equa:BoundaryFourTermsDecompositionAux} as
\[ \widetilde{\rho}_{\partial \RR} \,\, = \,\, \kappa_{\RR} \,\, \big( \, \widetilde{\SSS}_{\partial \RR} \, + \, \widetilde{\SSS}_{\partial \RR}^{rest} \, \big)\,,  \]
where the observable $\widetilde{\SSS}_{\partial \RR}^{rest}$ satisfies
\begin{align*} 
\| \widetilde{\mathcal S} _{\partial \RR}^{rest} \| \,\, & \leq \,\, \Big(\frac{\gamma_{\beta}}{1 + \gamma_{\beta}}\Big)^{n_{\RR}} \, |G| \, +
 \Big(\frac{\gamma_{\beta}}{1 +\gamma_{\beta}}\Big)^{|\VV_{\mathring{R}}|} \, \, (|G|+1) \,\, \Big[ \, 1 +   \Big(\frac{\gamma_{\beta}}{1 + \gamma_{\beta}}\Big)^{n_{\RR}} \, |G|\, \Big]\, \\[2mm]
 & \leq  \Big(\frac{\gamma_{\beta}}{1 +\gamma_{\beta}}\Big)^{|\VV_{\mathring{R}}|} \, \left(|G| + \left( 1 + |G|  \right)^{2} \right) \\[2mm]
 & \leq \, 3  \, |G|^{2} \,  \Big(\frac{\gamma_{\beta}}{1 +\gamma_{\beta}}\Big)^{|\VV_{\mathring{R}}|} =: \epsilon_{\RR}\,,
 \end{align*}
having used in the second inequality that $ |\mathcal{V}_{\mathring \RR}| \leq n_{\RR}$.

 \subsubsection*{Proof of the \emph{Approximate Factorization of the boundary} Theorem \ref{Coro:approx-factorization-bound}}
 
Let us assume that $\epsilon_{\RR} <1$. By Proposition \ref{Prop:BoundaryProjector3}, we have that  $\widetilde{S}_{\partial \RR}$ is a projection satisfying  $\widetilde{\rho}_{\partial \RR} = \widetilde{\rho}_{\partial \RR}\widetilde{S}_{\partial \RR}= \widetilde{S}_{\partial \RR} \widetilde{\rho}_{\partial \RR}$. Moreover, by Theorem \ref{Theo:leadingTerm}
\begin{equation}\label{equa:approximationAux1}
 \left\| \kappa_{\RR}^{-1} \, \widetilde{\rho}_{\partial \RR} - \widetilde{S}_{\partial \RR} \right\| \, = \, \| \widetilde{S}_{\partial \RR}^{rest} \| \leq \epsilon_{\RR} < 1 \,.  
 \end{equation}
These three properties yield that $\widetilde{S}_{\partial \RR} = \widetilde{J}_{\partial \RR}$, as a consequence of the next general observation:\\

\emph{Let $T$ be a self-adjoint operator on $\mathcal{H} \equiv \mathbb{C}^{d}$, and $\Pi$ an orthogonal projection onto a  subspace $W \subset \mathcal{H}$ such that $T \, \Pi = \Pi \, T = T$ and $\| T- \Pi\| < 1$. Then, $\Pi$ is the orthogonal projection onto $(\ker{T})^{\perp}$}. Indeed, since $T = \Pi T \Pi$ we have that $(\ker{T})^{\perp}$ is contained in $W$. If they were different subspaces, then there would be a state $\ket{u} \in W$ with $T\ket{u} = 0$. But then, $1 = \braket{u}{u} = \braket{u}{(\Pi-T)u} \leq \| T - \Pi\| < 1$.\\

\noindent Next, observe that by Proposition \ref{Prop:BoundaryProjector3}
\[ [\mathcal{G}_{\partial \RR}, \widetilde{J}_{\partial \RR}] \, = \, [\mathcal{G}_{\partial \RR}, \widetilde{S}_{\partial \RR} ] \, = \, 0\,.  \]
This yields that $\rho_{\partial \RR}$ and $\widetilde{\rho}_{\partial \RR}$ have the same support, that is $J_{\partial \RR} = \widetilde{J}_{\partial \RR}$, since they are related via $\rho_{\partial \RR} = \mathcal{G}_{\partial \RR} \, \widetilde{\rho}_{\partial \RR} \, \mathcal{G}_{\partial \RR}$ where $\mathcal{G}_{\partial \RR}$ is invertible. Moreover, the operators $\sigma_{\partial \RR} = \kappa_{\RR} \, \mathcal{G}_{\partial \RR} \, \widetilde{J}_{\RR} \, \mathcal{G}_{\partial \RR}$ and $\widetilde{\sigma}_{\partial \RR} = \kappa_{\RR}\, \widetilde{J}_{\RR}$ will also have the same support as $\rho_{\partial \RR}$. We can thus argue as in the proof of Proposition \ref{prop:gauge-invariance} to get  
\begin{align*}
\| \rho_{\partial \RR}^{1/2} \, \sigma_{\partial \RR}^{-1} \, \rho^{1/2}_{\partial \RR} - J_{\partial \RR} \| & = \| \widetilde{\rho}_{\partial \RR}^{1/2} \, (\kappa_{\RR} \widetilde{J}_{\partial \RR})^{-1} \, \widetilde{\rho}^{1/2}_{\partial \RR} - \widetilde{J}_{\partial \RR} \|\\ 
& = \| \kappa_{\RR}^{-1} \, \widetilde{\rho}_{\partial \RR}^{1/2} \widetilde{J}_{\partial \RR} \widetilde{\rho}_{\partial \RR}^{1/2} - \widetilde{J}_{\partial \RR} \| \\
& = \| \kappa_{\RR}^{-1} \, \widetilde{\rho}_{\partial \RR} - \widetilde{J}_{\partial \RR} \|\, \leq \, \epsilon_{\RR}
\end{align*}
where in the last line we have used again  \eqref{equa:approximationAux1}. Finally, since $\sigma_{\partial \RR}$ and $\rho_{\partial \RR}$ have the same support, we can use the identity
\[ \rho_{\partial \RR}^{-1/2} \, \sigma_{\partial \RR} \, \rho^{-1/2}_{\partial \RR} = \big( \rho_{\partial \RR}^{1/2} \, \sigma_{\partial \RR}^{-1} \, \rho^{1/2}_{\partial \RR}  \big)^{-1} = J_{\partial \RR} + \sum_{m=1}^{\infty} \,\big( J_{\partial \RR} -  \rho_{\partial \RR}^{1/2} \, \sigma_{\partial \RR}^{-1} \, \rho^{1/2}_{\partial \RR}\big)^{m} \]
to deduce that
\[ \| \rho_{\partial \RR}^{-1/2} \, \sigma_{\partial \RR} \, \rho^{-1/2}_{\partial \RR} - J_{\partial \RR}\| \, \leq \, \sum_{m=1}^{\infty} \, \epsilon_{\RR}^{m} = \frac{\epsilon_{\RR}}{1-\epsilon_{\RR}}\,.  \]

\subsection{Approximate factorization of the ground state projections}

In this section, we study for a given rectangular region $\RR$ split into three suitable regions $A, B, C$ where $B$ shields $A$ from $C$ and such that $AB, BC$ are again rectangular regions, how to estimate
\[ \| P_{ABC} - P_{AB} P_{BC}\| \]
in terms of the size of $B$. This problem is related to the approximate factorization property of the boundary states that we studied in the previous sections, (see Section \ref{sec:non-injective-peps} for details on this relation and the main results in the context of general PEPS). Our main tools will be Theorems \ref{theo:non-injective-approx-fact} and \ref{Coro:approx-factorization-bound}.
 
\begin{Coro}\label{coro:approx-fact-rectangles}
Let us consider a rectangular region with open boundary conditions

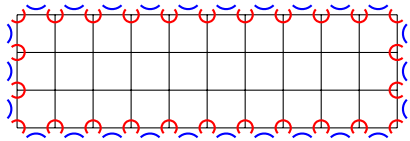
\[
\begin{tikzpicture}[scale=0.5]
\plaquettefive[10,3](4,0,0,0);
\end{tikzpicture}
\]
split into three subregions $A,B,C$ as in the next picture
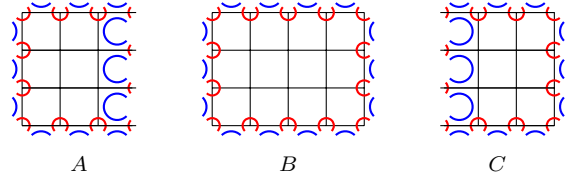
\[
\begin{tikzpicture}[equation, scale=0.5]


\begin{scope}
 \draw (1.5,-1) node {\scriptsize $A$};

\draw[step=1.0,black,thin] (0,0) grid (2.9,3);


\foreach \y in {0.5,1.5,2.5}{
    \centerarc[blue, thick](-0.5,\y)(-45:45:0.35);
}


\foreach \x in {0.5,1.5,2.5}{
    \centerarc[blue, thick](\x,3.5)(225:315:0.35);
    \centerarc[blue, thick](\x,-0.5)(45:135:0.35);
}

\foreach \y in {0.5,1.5,2.5}{
    \centerarc[blue, thick](2.5,\y)(45:315:0.35);
}

\foreach \y in {1,2}{
    \centerarc[red, thick](0,\y)(-135:135:0.2);
}

\centerarc[red, thick](0,0)(-45:135:0.2);

\centerarc[red, thick](0,3)(-135:45:0.2);

\foreach \x in {1,2}{
    \centerarc[red, thick](\x,0)(-45:225:0.2);
    \centerarc[red, thick](\x,3)(-225:45:0.2);
}

\foreach \y in {0,1,2,3}{
    \centerarc[red, thick](3,\y)(135:225:0.2);
}


\draw[thin] (0,0) -- (3,0);
\draw[thin] (0,3) -- (3,3);
\draw[thin] (0,1) -- (3,1);
\draw[thin] (0,2) -- (3,2);
\end{scope}


\begin{scope}[xshift=5cm]
\plaquettefive[4,3](4,0,0,0);
 \draw (2,-1) node {\scriptsize $B$};
\end{scope}


\begin{scope}[xshift=11cm]
 \draw (1.5,-1) node {\scriptsize $C$};

\draw[step=1.0,black,thin] (0.1,0) grid (3,3);


\foreach \y in {0.5,1.5,2.5}{
    \centerarc[blue, thick](0.5,\y)(-135:135:0.35);
}


\foreach \x in {0.5,1.5,2.5}{
    \centerarc[blue, thick](\x,3.5)(225:315:0.35);
    \centerarc[blue, thick](\x,-0.5)(45:135:0.35);
}

\foreach \y in {0.5,1.5,2.5}{
    \centerarc[blue, thick](3.5,\y)(135:225:0.35);
}

\foreach \y in {0,1,2,3}{
    \centerarc[red, thick](0,\y)(-45:45:0.2);
}

\centerarc[red, thick](3,0)(45:225:0.2);

\centerarc[red, thick](3,3)(135:315:0.2);

\foreach \x in {1,2}{
    \centerarc[red, thick](\x,0)(-45:225:0.2);
    \centerarc[red, thick](\x,3)(-225:45:0.2);
}

\foreach \y in {1,2}{
    \centerarc[red, thick](3,\y)(45:315:0.2);
}


\draw (0,0) -- (3,0);
\draw (0,3) -- (3,3);
\draw (0,1) -- (3,1);
\draw (0,2) -- (3,2);
\end{scope}

\end{tikzpicture}
\]

\noindent  so that $AB, BC, B$ are again rectangular regions. If $B$ has $M$ plaquettes per row and $N$ plaquettes per column with
\begin{equation}\label{coro:approx-fact-rectangles-hypothesis1} 
\epsilon_{B}:= 3 \, |G|^{2} \, \left( \frac{\gamma_{\beta}}{1+\gamma_{\beta}} \right)^{(M-1)(N-1)}  \, < \, \frac{1}{2}\,, 
\end{equation}
then the orthogonal projections $P_{\RR}$ onto $\operatorname{Im}(V_{\RR})$ satisfy
\[ \| P_{AB} P_{BC} - P_{ABC}\| \leq 16  \epsilon_{B} \,. \]
\end{Coro}

\begin{proof}
In order to apply Theorem \ref{theo:non-injective-approx-fact}, we arrange the virtual indices of $ABC$, $AB$, $BC$ and $B$ into four sets $a,c,\alpha, \gamma$ as in the next picture: 

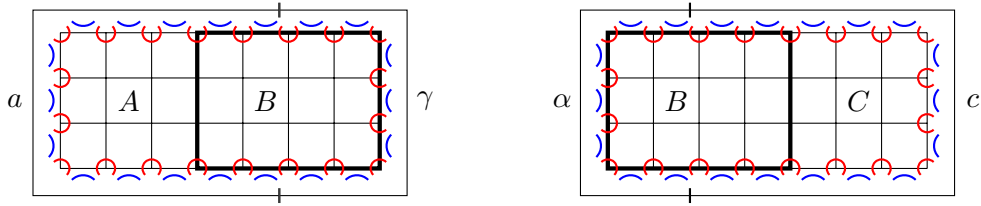
\begin{center}
\begin{tikzpicture}[scale=0.6]

\begin{scope}
\draw[black, ultra thick] (3,0) rectangle (7,3);
\plaquettefive[7,3](4,0,0,0);
\draw [|-|] (4.8, 3.5) -- (-0.6,3.5) -- (-0.6,-0.6) -- (4.8, -0.6);
\draw [|-|] (4.8, 3.5) -- (7.6,3.5) -- (7.6,-0.6) -- (4.8, -0.6);
\draw (1.5, 1.5) node {$A$};
\draw (4.5, 1.5) node {$B$};
\draw (-1, 1.5) node {$a$};
\draw (8, 1.5) node {$\gamma$};
\end{scope}

\begin{scope}[xshift=12cm]
\draw[black, ultra thick] (0,0) rectangle (4,3);
\plaquettefive[7,3](4,0,0,0);
\draw [|-|] (1.8, 3.5) -- (-0.6,3.5) -- (-0.6,-0.6) -- (1.8, -0.6);
\draw [|-|] (1.8, 3.5) -- (7.6,3.5) -- (7.6,-0.6) -- (1.8, -0.6);
\draw (1.5, 1.5) node {$B$};
\draw (5.5, 1.5) node {$C$};
\draw (-1, 1.5) node {$\alpha$};
\draw (8, 1.5) node {$c$};
\end{scope}

\end{tikzpicture}
\end{center}

\noindent so that
\[ \partial ABC = ac \quad , \quad \partial AB = a \gamma \quad , \quad \partial BC = \alpha c \quad , \quad \partial B = \alpha \gamma\,.  \]

\noindent The hypothesis $\epsilon_{B} < 1$ ensures that the hypothesis $\epsilon_{\RR}<1$ in Theorem \ref{Coro:approx-factorization-bound} is satisfied for $\RR \in \{ B, AB, BC, ABC \}$, and so $J_{\partial \RR}$ and $\sigma_{\partial \RR}$ have a simple tensor product structure. This allows us to factorize
\begin{align*}
J_{\partial ABC} = J_{a}  \otimes J_{c}    
\quad & , \quad 
J_{\partial B} = J_{\alpha}  \otimes J_{\gamma} \,,    \\[2mm]
\sigma_{\partial AB} = J_{a}  \otimes J_{\gamma}    
\quad &, \quad 
\sigma_{\partial BC} = J_{\alpha}  \otimes J_{c}\,.
\end{align*}
where $J_{a}$, $J_{c}$, $J_{\alpha}$ and $J_{\gamma}$ are projections. The local structure of $\widetilde{\mathcal{S}}_{\partial \RR}$ and $\mathcal{G}_{\partial \RR}$ allows us to decompose both operators as a tensor product of operators acting on $a,c,\alpha, \gamma$. Hence, we can define
\begin{align*}
   & \sigma_{a} := \frac{\kappa_{AB}}{ \sqrt{\kappa_{B}}} \, \mathcal{G}_{a} \widetilde{\mathcal{S}}_{a} \mathcal{G}_{a} \quad , \quad  \sigma_{c} := \frac{\kappa_{BC}}{\sqrt{\kappa_{B}}} \, \mathcal{G}_{c} \, \widetilde{\mathcal{S}}_{c} \, \mathcal{G}_{c} \,,\\[2mm]
   & \sigma_{\alpha} := \sqrt{\kappa_{B}} \, \mathcal{G}_{\alpha} \, \widetilde{\mathcal{S}}_{\alpha} \, \mathcal{G}_{\alpha} \quad , \quad  \sigma_{\gamma} := \sqrt{\kappa_{B}} \, \mathcal{G}_{\gamma} \, \widetilde{\mathcal{S}}_{\gamma} \, \mathcal{G}_{\gamma}\,.
\end{align*}
Then, we can easily verify that $\kappa_{ABC} \kappa_{B} = \kappa_{AB} \kappa_{BC}$, and so 
\begin{align*}
\sigma_{\partial ABC} = \sigma_{a}  \otimes \sigma_{c}    
\quad & , \quad 
\sigma_{\partial B} = \sigma_{\alpha}  \otimes \sigma_{\gamma} \,,    \\[2mm]
\sigma_{\partial AB} = \sigma_{a}  \otimes \sigma_{\gamma}    
\quad &, \quad 
\sigma_{\partial BC} = \sigma_{\alpha}  \otimes \sigma_{c}\,.
\end{align*}

\noindent From Theorem~\ref{Coro:approx-factorization-bound}, it follows that for each $\RR \in \{ ABC, AB, BC, B\}$
\[ \| \rho_{\partial \RR}^{1/2} \sigma_{\partial \RR}^{-1} \rho_{\partial \RR}^{1/2} - J_{\partial \RR} \| < \epsilon_{B}\,. \]
and
\[ \| \rho_{\partial \RR}^{-1/2} \, \sigma_{\partial \RR} \, \rho_{\partial \RR}^{-1/2} - J_{\partial \RR} \| \, \leq \, \frac{\epsilon_{B}}{1-\epsilon_{B}} \leq 2 \epsilon_{B}\,.\]
Thus, applying Theorem \ref{theo:non-injective-approx-fact} we conclude the result.
\end{proof}

Analogously, we can prove results for the torus and cylinders.

\begin{Coro}\label{coro:approx-fact-cylinder}
Let us consider a cylinder with open boundary conditions and split it into four sections $A,B,C,B'$ so that $B'AB$ and $BCB'$ are two overlapping rectangles whose intersection is formed by two disjoint rectangles $B$ and $B'$, as in the next picture:

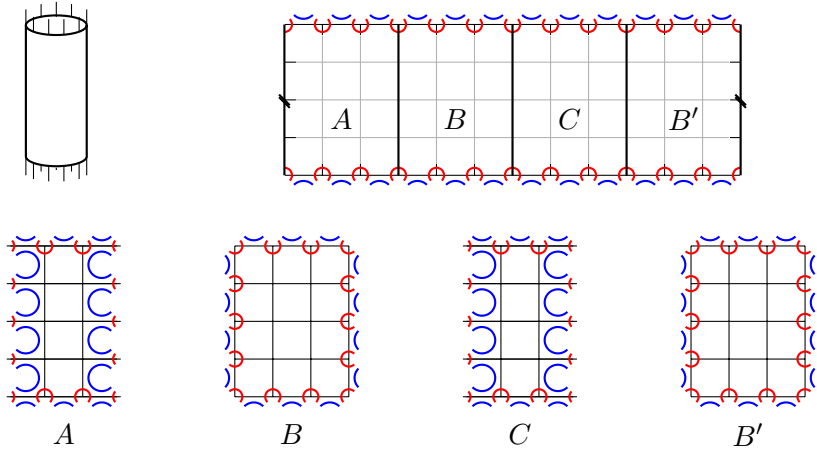
\[
\begin{tikzpicture}[equation, scale=1]

\begin{scope}[xshift=0.5cm]
\node [cylinder, draw, minimum height=2cm, minimum width=0.8cm, fill=white, rotate=90, thick, black, fill=white] {};


\draw[black] (-0.4,1) -- (-0.4,1.2);
\draw[black] (-0.2,1.09) -- (-0.2,1.27);
\draw[black] (0,1.1) -- (0,1.3);
\draw[black] (0.2,1.1) -- (0.2,1.27);
\draw[black] (0.4,1) -- (0.4,1.2);

\draw[black] (-0.3,0.9) -- (-0.3,1.1);
\draw[black] (-0.1,0.87) -- (-0.1,1.07);
\draw[black] (0.1,0.87) -- (0.1,1.07);
\draw[black] (0.3,0.9) -- (0.3,1.1);


\draw[black] (-0.4,-0.7) -- (-0.4,-1);
\draw[black] (-0.2,-0.87) -- (-0.2,-0.95);
\draw[black] (0,-0.9) -- (0,-0.93);
\draw[black] (0.2,-0.87) -- (0.2,-0.95);
\draw[black] (0.4,-0.7) -- (0.4,-1);

\draw[black] (-0.3,-0.85) -- (-0.3,-1.05);
\draw[black] (-0.1,-0.88) -- (-0.1,-1.08);
\draw[black] (0.1,-0.88) -- (0.1,-1.08);
\draw[black] (0.3,-0.85) -- (0.3,-1.05);

\end{scope}


\begin{scope}[xshift=3cm, yshift=-1cm, scale=0.5]

\begin{scope}
\plaquettefive[14,4](4,0,0,0);
\end{scope}

\fill[white] (-0.5,-1) rectangle (1,5);
\fill[white] (13,-1) rectangle (14.5,5);

\draw[thick, black, postaction=torus vertical] (1,0) -- (1,4); 
\draw[thick, black, postaction=torus vertical] (13,0) -- (13,4); 


\fill[white] (1.3,0.3) rectangle (12.7,3.7);

\draw[white!70!black] (1.3,0.3) grid (12.7,3.7);
\draw[black, thick] (4,0) -- (4,4);
\draw[black, thick] (7,0) -- (7,4);
\draw[black, thick] (10,0) -- (10,4);

\draw (2.5,1.5) node{$A$};
\draw (5.5,1.5) node{$B$};
\draw (8.5,1.5) node{$C$};
\draw (11.5,1.5) node{$B'$};
\end{scope}

\end{tikzpicture}
\]

\[
\begin{tikzpicture}[equation, scale=0.5, decoration={markings, mark=at position 0.65 with {\arrow[thick, scale=1.5]{>}}}]


\cylinderopen[3,4];
\draw (1.5,-1) node{$A$};


\begin{scope}[xshift=6cm]
\plaquettefive[3,4](4,0,0,0);
\draw (1.5,-1) node{$B$};
\end{scope}


\begin{scope}[xshift=12cm]
\cylinderopen[3,4];
\draw (1.5,-1) node{$C$};
\end{scope}


\begin{scope}[xshift=18cm]
\plaquettefive[3,4](4,0,0,0);
\draw (1.5,-1) node{$B'$};
\end{scope}

\end{tikzpicture} 
\]

\noindent Let us assume that  $B$ and $B'$ have at least $M$ plaquettes per row and $N$ plaquettes per column with
\[ \epsilon_{BB'}:= 3 \, |G|^{2} \, \left( \frac{\gamma_{\beta}}{1+\gamma_{\beta}} \right)^{(M-1)(N-1)}  \, \leq \, \frac{1}{2}\,, \]
then the orthogonal projections $P_{\RR}$ onto $\operatorname{Im}(V_{\RR})$ satisfy
\[ \| P_{B'AB} P_{BCB'} - P_{ABCB'}\| \leq 48 \epsilon_{BB'} \,. \]
\end{Coro}
\begin{proof}
We are going to apply Theorem \ref{theo:non-injective-approx-fact} for the three regions $A$, $BB'$ and $C$. Firstly, we need to find a suitable arrangement of the virtual indices.  First we split the boundary of $BB'$ into two four parts $\gamma, \gamma',\alpha, \alpha'$ as in
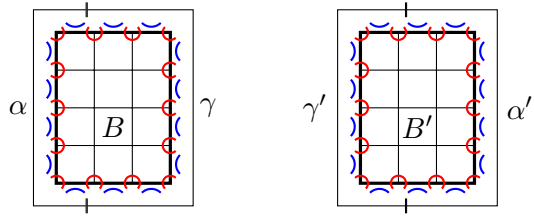
\[
\begin{tikzpicture}[equation, scale=0.5]

\def\sep{0.6};

\begin{scope}
\draw[very thick] (0,0) rectangle (3,4);
\plaquettefive[3,4](4,0,0,0);
\draw (1.5,1.5) node{$B$};

\draw[|-|] (0.8, 0-\sep) -- (0-\sep, 0-\sep) -- (0-\sep, 4+\sep) -- (0.8, 4+\sep);

\draw[|-|] (0.8, 0-\sep) -- (3+\sep, 0-\sep) -- (3+\sep, 4+\sep) -- (0.8, 4+\sep);

\draw (-1,2) node{$\alpha$};

\draw (4,2) node{$\gamma$};
\end{scope}

\begin{scope}[xshift=8cm]
\draw[very thick] (0,0) rectangle (3,4);
\plaquettefive[3,4](4,0,0,0);
\draw (1.5,1.5) node{$B'$};

\draw[|-|] (1.2, 0-\sep) -- (0-\sep, 0-\sep) -- (0-\sep, 4+\sep) -- (1.2, 4+\sep);

\draw[|-|] (1.2, 0-\sep) -- (3+\sep, 0-\sep) -- (3+\sep, 4+\sep) -- (1.2, 4+\sep);

\draw (-1.2,2) node{$\gamma'$};

\draw (4.2,2) node{$\alpha'$};
\end{scope}

\end{tikzpicture}
\]
and define regions $a$ and $c$ according to
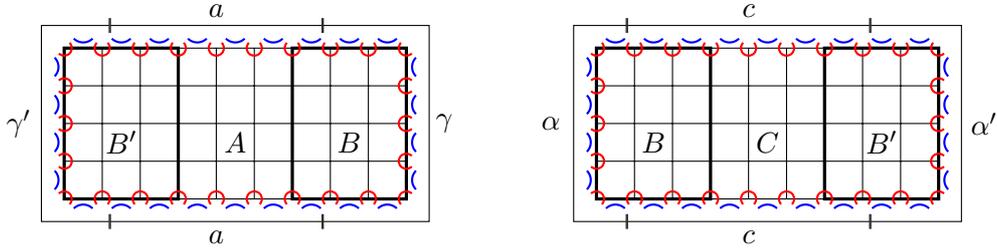
\[
\begin{tikzpicture}[equation, scale=0.5]

\def\sep{0.6};

\begin{scope}
\draw[very thick] (0,0) rectangle (3,4);
\draw[very thick] (6,0) rectangle (9,4);
\plaquettefive[9,4](4,0,0,0);
\draw (1.5,1.5) node{$B'$};
\draw (4.5,1.5) node{$A$};
\draw (7.5,1.5) node{$B$};
\draw[|-|] (1.2, 0-\sep) -- (0-\sep, 0-\sep) -- (0-\sep, 4+\sep) -- (1.2, 4+\sep);
\draw[|-|] (1.2, 0-\sep) -- (6.8, 0-\sep);
\draw[|-|] (1.2, 4+\sep) -- (6.8, 4+\sep);
\draw[|-|] (6.8, 0-\sep) -- (9+\sep, 0-\sep) -- (9+\sep, 4+\sep) -- (6.8, 4+\sep);
\draw (-1.2,2) node{$\gamma'$};
\draw (10,2) node{$\gamma$};
\draw (4,5) node{$a$};
\draw (4,-1) node{$a$};
\end{scope}

\begin{scope}[xshift=14cm]
\draw[very thick] (0,0) rectangle (3,4);
\draw[very thick] (6,0) rectangle (9,4);
\plaquettefive[9,4](4,0,0,0);
\draw (1.5,1.5) node{$B$};
\draw (4.5,1.5) node{$C$};
\draw (7.5,1.5) node{$B'$};
\draw[|-|] (0.8, 0-\sep) -- (0-\sep, 0-\sep) -- (0-\sep, 4+\sep) -- (0.8, 4+\sep);
\draw[|-|] (0.8, 0-\sep) -- (7.2, 0-\sep);
\draw[|-|] (0.8, 4+\sep) -- (7.2, 4+\sep);
\draw[|-|] (7.2, 0-\sep) -- (9+\sep, 0-\sep) -- (9+\sep, 4+\sep) -- (7.2, 4+\sep);
\draw (-1.2,2) node{$\alpha$};
\draw (10.2,2) node{$\alpha'$};
\draw (4,5) node{$c$};
\draw (4,-1) node{$c$};
\end{scope}

\end{tikzpicture}
\]
so that
\begin{align*} 
& \partial B'AB=\gamma' a \gamma \quad , \quad \partial BCB' = \alpha c \alpha' \quad , \quad \partial ABCB' = ac\,,\\
& \partial B= \alpha \gamma \quad , \quad \partial B' = \gamma' \alpha'\,.
\end{align*}

\noindent The local structure of $\widetilde{S}_{\partial \RR}$ and $\mathcal{G}_{\partial \RR}$ allows us to write them as a tensor product of operators acting on the defined boundary segments, and consider
\begin{align*}
   \sigma_{a} &:= \frac{\kappa_{B'AB}}{ \sqrt{\kappa_{B}\kappa_{B'}}} \mathcal{G}^{\dagger}_{a} \widetilde{\mathcal{S}}_{a} \mathcal{G}_{a}, & 
   \sigma_{c} &:= \frac{\kappa_{BCB'}}{\sqrt{\kappa_{B}\kappa_{B'}}} \,. \mathcal{G}^{\dagger}_{c} \widetilde{\mathcal{S}}_{c} \mathcal{G}_{c}\,,\\[2mm]
    \sigma_{\alpha} &:= \sqrt{\kappa_{B}\kappa_{B'}} \mathcal{G}^{\dagger}_{\alpha}\,, \widetilde{\mathcal{S}}_{\alpha} \mathcal{G}_{\alpha} &
    \sigma_{\alpha'} &:= \mathcal{G}^{\dagger}_{\alpha'} \widetilde{\mathcal{S}}_{\alpha'} \mathcal{G}_{\alpha'}\,, \\[2mm]
   \sigma_{\gamma} &:= \sqrt{\kappa_{B}\kappa_{B'}} \mathcal{G}^{\dagger}_{\gamma}\,, \widetilde{\mathcal{S}}_{\gamma} \mathcal{G}_{\gamma} &
   \sigma_{\gamma'} &:= \mathcal{G}^{\dagger}_{\gamma'} \widetilde{\mathcal{S}}_{\gamma'} \mathcal{G}_{\gamma'}\,.
\end{align*}
Then, using that $\kappa_{ABC} \kappa_{B} \kappa_{B'} = \kappa_{B'AB} \kappa_{BCB'}$ we can easily check that
\begin{align*}
    & \sigma_{\partial B'AB} = \sigma_{\gamma'} \otimes \sigma_{a} \otimes \sigma_{\gamma} \quad , \quad \sigma_{\partial BCB'} = \sigma_{\alpha'} \otimes \sigma_{c} \otimes \sigma_{\alpha'} \quad , \quad \sigma_{\partial ABCB'} = \sigma_{a} \otimes \sigma_{c}\,,\\[2mm]
    & \sigma_{\partial B} = \sigma_{\alpha} \otimes \sigma_{\gamma} \quad , \quad \sigma_{\partial B'} = \sigma_{\gamma'} \otimes \sigma_{\alpha'}\,.
\end{align*} 
By Theorem~\ref{Coro:approx-factorization-bound}, we have for each $\RR \in \{ B, B', B'AB, BCB', ABCB' \}$
\begin{equation*}\label{equa:cylinderApproxFactorization1} 
\| \rho_{\partial \RR}^{1/2} \sigma_{\partial \RR}^{-1} \rho_{\partial \RR}^{1/2} - J_{\partial \RR} \| \,\, , \,\, \| \rho_{\partial \RR}^{-1/2} \sigma_{\partial \RR} \rho_{\partial \RR}^{-1/2} - J_{\partial \RR} \| \, \leq \, 2 \epsilon_{BB'} \,.
\end{equation*}
In particular, considering the joint region $BB'$ we get
\begin{align*}
\| \rho_{\partial BB'}^{-1/2} \sigma_{\partial BB'} \rho_{\partial BB'}^{-1/2} - J_{\partial BB'}\| & \, \leq \, \| (\rho_{\partial B}^{-1/2} \sigma_{\partial B} \rho_{\partial B}^{-1/2}) (\rho_{\partial B'}^{-1/2} \sigma_{\partial B'} \rho_{\partial B'}^{-1/2}) - J_{\partial B} J_{\partial B'}\|\\[2mm]
& \, \leq \, \| \rho_{\partial B}^{-1/2} \sigma_{\partial B}  \rho_{\partial B}^{-1/2} - J_{\partial B} \| \cdot \|\rho_{\partial B'}^{-1/2} \sigma_{\partial B'} \rho_{\partial B'}^{-1/2} \|\\[2mm]
& \quad \quad + \, \| \rho_{\partial B'}^{-1/2} \sigma_{\partial B'} \rho_{\partial B'}^{-1/2} - J_{\partial B'} \| \cdot \|J_{\partial B} \| \\[2mm]
& \, \leq \,  2\epsilon_{BB'} (1+2\epsilon_{BB'}) + 2\epsilon_{BB'}\\[2mm]
& \, \leq \,   6 \epsilon_{BB'}\,.
\end{align*}
Applying Theorem \ref{theo:non-injective-approx-fact}, we conclude the result.
\end{proof}

\begin{Coro}\label{coro:approx-fact-torus}
We consider a decomposition of the torus into four regions $A,B,C,B'$ as below. We have then two overlapping cylinders $B'AB$ and $BCB'$, whose intersection is formed by two disjoint cylinders $B$ and $B'$.
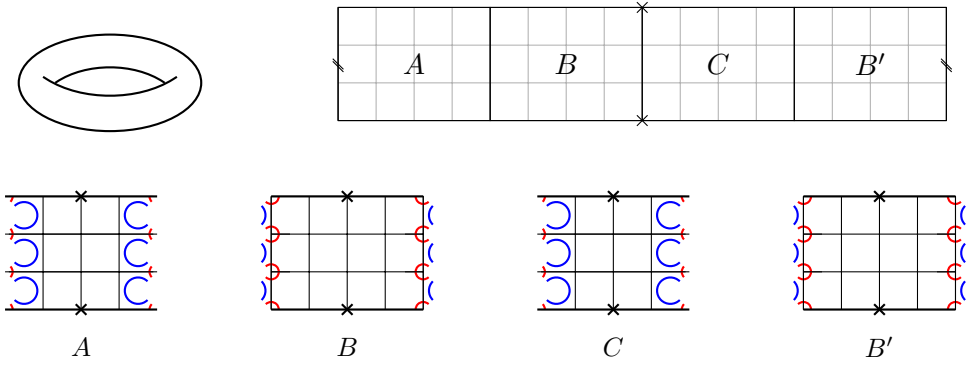
\[
\begin{tikzpicture}[equation, scale=1]

\begin{scope}[scale=0.8, thick]

\path[rounded corners=24pt, blue] (-.9,0)--(0,.6)--(.9,0) (-.9,0)--(0,-.56)--    (.9,0);
\draw[rounded corners=28pt, black] (-1.1,.1)--(0,-.6)--(1.1,.1);
\draw[rounded corners=24pt, black] (-.9,0)--(0,.6)--(.9,0);
\draw[black] (0,0) ellipse (1.5 and 0.8);
\end{scope}   
    
    
\begin{scope}[xshift=3cm, scale=0.5, yshift=-1cm]

\draw[step=1,white!70!black,thin] (0,0) grid (16,3);

\draw[black] (0,0) rectangle (4,3);
\draw (2,1.5) node{$A$};
\draw[black] (4,0) rectangle (8,3);
\draw (6,1.5) node{$B$};
\draw[black] (8,0) rectangle (12,3);
\draw (10,1.5) node{$C$};
\draw[black] (12,0) rectangle (16,3);
\draw (14,1.5) node{$B'$};

\draw[black, postaction=torus horizontal] (0,0) -- (16,0); 
\draw[black, postaction=torus horizontal] (0,3) -- (16, 3);
\draw[black, postaction=torus vertical] (0,0) -- (0, 3);
\draw[black, postaction=torus vertical] (16,0) -- (16, 3);

\end{scope}

\end{tikzpicture} 
\]


\[
\begin{tikzpicture}[equation, scale=0.5]


\begin{scope}

 \draw (2,-1) node {\small $A$};

\draw[step=1.0,black,thin] (0.1,0) grid (3.9,3);


\foreach \y in {0.5,1.5,2.5}{
    \centerarc[blue, thick](0.5,\y)(-135:135:0.35);
}

\foreach \y in {0.5,1.5,2.5}{
    \centerarc[blue, thick](3.5,\y)(45:315:0.35);
}

\foreach \y in {0,1,2}{
    \centerarc[red, thick](0,\y)(0:45:0.2);
}

\foreach \y in {1,2,3}{
    \centerarc[red, thick](0,\y)(-45:0:0.2);
}

\foreach \y in {0,1,2}{
    \centerarc[red, thick](4,\y)(135:180:0.2);
}

\foreach \y in {1,2,3}{
    \centerarc[red, thick](4,\y)(180:225:0.2);
}


\draw[thick, postaction=torus horizontal] (0,0) -- (4,0);
\draw[thick, postaction=torus horizontal] (0,3) -- (4,3);
\draw (0,1) -- (4,1);
\draw (0,2) -- (4,2);
\end{scope}



\begin{scope}[xshift=7cm]

    \draw (2,-1) node {\small $B$};

\plaquettefive[4,3](4,0,0,0);
\fill[white] (0,3) rectangle (4,3.5);
\fill[white] (0,0) rectangle (4,-0.5);
\fill[white] (0.5,0) rectangle (3.5,3);
\draw[step=1.0,black,thin] (0,0) grid (4,3);
\draw[thick, postaction=torus horizontal] (0,0) -- (4,0);
\draw[thick, postaction=torus horizontal] (0,3) -- (4,3);

\end{scope}


\begin{scope}[xshift=14cm]

\draw (2,-1) node {\small $C$};

\draw[step=1.0,black,thin] (0.1,0) grid (3.9,3);


\foreach \y in {0.5,1.5,2.5}{
    \centerarc[blue, thick](0.5,\y)(-135:135:0.35);
}

\foreach \y in {0.5,1.5,2.5}{
    \centerarc[blue, thick](3.5,\y)(45:315:0.35);
}

\foreach \y in {0,1,2}{
    \centerarc[red, thick](0,\y)(0:45:0.2);
}

\foreach \y in {1,2,3}{
    \centerarc[red, thick](0,\y)(-45:0:0.2);
}

\foreach \y in {0,1,2}{
    \centerarc[red, thick](4,\y)(135:180:0.2);
}

\foreach \y in {1,2,3}{
    \centerarc[red, thick](4,\y)(180:225:0.2);
}


\draw[thick, postaction=torus horizontal] (0,0) -- (4,0);
\draw[thick, postaction=torus horizontal] (0,3) -- (4,3);
\draw (0,1) -- (4,1);
\draw (0,2) -- (4,2);
\end{scope}


\begin{scope}[xshift=21cm]

    \draw (2,-1) node {\small $B'$};

\plaquettefive[4,3](4,0,0,0);
\fill[white] (0,3) rectangle (4,3.5);
\fill[white] (0,0) rectangle (4,-0.5);
\fill[white] (0.5,0) rectangle (3.5,3);
\draw[step=1.0,black,thin] (0,0) grid (4,3);
\draw[thick, postaction=torus horizontal] (0,0) -- (4,0);
\draw[thick, postaction=torus horizontal] (0,3) -- (4,3);

\end{scope}

\end{tikzpicture}
\]

\noindent Let us assume that  $B$ and $B'$ have at least $M$ plaquettes per row and $N$ plaquettes per column with
\[ \epsilon_{BB'}:= 3 \, |G|^{2} \, \left( \frac{\gamma_{\beta}}{1+\gamma_{\beta}} \right)^{(M-1)(N-1)}  \, \leq \, \frac{1}{2}\,, \]
then the orthogonal projections $P_{\RR}$ onto $\operatorname{Im}(V_{\RR})$ satisfy
\[ \| P_{B'AB} P_{BCB'} - P_{ABCB'}\| \leq 48 \epsilon_{BB'}\, . \]
\end{Coro}

\begin{proof}
We are going to apply Theorem \ref{theo:non-injective-approx-fact} for the three regions $A$, $BB'$ and $C$. Firstly, we need to find a suitable arrangement of the virtual indices. The boundary of $B$ (resp. $B'$) can be split into two regions: the part connecting with $A$, denoted by $\alpha$ (resp. $\alpha'$); and the part connecting with $C$, denoted by $\gamma$ (resp. $\gamma'$): 
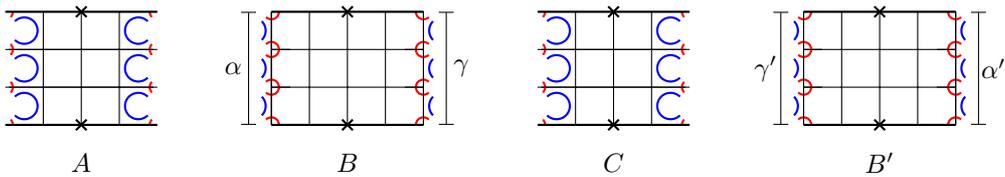
\[
\begin{tikzpicture}[equation, scale=0.5]


\begin{scope}

 \draw (2,-1) node {\small $A$};

\draw[step=1.0,black,thin] (0.1,0) grid (3.9,3);


\foreach \y in {0.5,1.5,2.5}{
    \centerarc[blue, thick](0.5,\y)(-135:135:0.35);
}

\foreach \y in {0.5,1.5,2.5}{
    \centerarc[blue, thick](3.5,\y)(45:315:0.35);
}

\foreach \y in {0,1,2}{
    \centerarc[red, thick](0,\y)(0:45:0.2);
}

\foreach \y in {1,2,3}{
    \centerarc[red, thick](0,\y)(-45:0:0.2);
}

\foreach \y in {0,1,2}{
    \centerarc[red, thick](4,\y)(135:180:0.2);
}

\foreach \y in {1,2,3}{
    \centerarc[red, thick](4,\y)(180:225:0.2);
}


\draw[thick, postaction=torus horizontal] (0,0) -- (4,0);
\draw[thick, postaction=torus horizontal] (0,3) -- (4,3);
\draw (0,1) -- (4,1);
\draw (0,2) -- (4,2);
\end{scope}



\begin{scope}[xshift=7cm]

    \draw (2,-1) node {\small $B$};

\plaquettefive[4,3](4,0,0,0);
\fill[white] (0,3) rectangle (4,3.5);
\fill[white] (0,0) rectangle (4,-0.5);
\fill[white] (0.5,0) rectangle (3.5,3);
\draw[step=1.0,black,thin] (0,0) grid (4,3);
\draw[thick, postaction=torus horizontal] (0,0) -- (4,0);
\draw[thick, postaction=torus horizontal] (0,3) -- (4,3);

\draw[thin, black, |-|] (-0.6,0) -- (-0.6,3); 
\draw[thin, black, |-|] (4.6,0) -- (4.6,3); 

\draw (-1,1.5) node {\small$\alpha$};
\draw (5,1.5) node {\small $\gamma$};

\end{scope}


\begin{scope}[xshift=14cm]

\draw (2,-1) node {\small $C$};

\draw[step=1.0,black,thin] (0.1,0) grid (3.9,3);


\foreach \y in {0.5,1.5,2.5}{
    \centerarc[blue, thick](0.5,\y)(-135:135:0.35);
}

\foreach \y in {0.5,1.5,2.5}{
    \centerarc[blue, thick](3.5,\y)(45:315:0.35);
}

\foreach \y in {0,1,2}{
    \centerarc[red, thick](0,\y)(0:45:0.2);
}

\foreach \y in {1,2,3}{
    \centerarc[red, thick](0,\y)(-45:0:0.2);
}

\foreach \y in {0,1,2}{
    \centerarc[red, thick](4,\y)(135:180:0.2);
}

\foreach \y in {1,2,3}{
    \centerarc[red, thick](4,\y)(180:225:0.2);
}


\draw[thick, postaction=torus horizontal] (0,0) -- (4,0);
\draw[thick, postaction=torus horizontal] (0,3) -- (4,3);
\draw (0,1) -- (4,1);
\draw (0,2) -- (4,2);
\end{scope}


\begin{scope}[xshift=21cm]

    \draw (2,-1) node {\small $B'$};

\plaquettefive[4,3](4,0,0,0);
\fill[white] (0,3) rectangle (4,3.5);
\fill[white] (0,0) rectangle (4,-0.5);
\fill[white] (0.5,0) rectangle (3.5,3);
\draw[step=1.0,black,thin] (0,0) grid (4,3);
\draw[thick, postaction=torus horizontal] (0,0) -- (4,0);
\draw[thick, postaction=torus horizontal] (0,3) -- (4,3);

\draw[thin, black, |-|] (-0.6,0) -- (-0.6,3); 
\draw[thin, black, |-|] (4.6,0) -- (4.6,3); 

\draw (-1,1.5) node {\small $\gamma'$};
\draw (5,1.5) node {\small $\alpha'$};

\end{scope}

\end{tikzpicture}
\]

\noindent The local structure of $\widetilde{S}_{\partial \RR}$ and $\mathcal{G}_{\partial \RR}$ allows us to write them as a tensor product of operators acting on the defined boundary segments, and define
\begin{align*}
\sigma_{\alpha} := \sqrt{\kappa_B \kappa_{B'}} \, \mathcal{G}_{\alpha}^{\dagger} \widetilde{\mathcal{S}}_{\alpha} \mathcal{G}_{\alpha}   \quad & , \quad \sigma_{\gamma} :=  \sqrt{\kappa_B \kappa_{B'}} \,   \mathcal{G}_{\alpha}^{\dagger} \widetilde{\mathcal{S}}_{\alpha} \mathcal{G}_{\alpha} \\[2mm]
\sigma_{\alpha'} :=  \frac{\kappa_{BCB'}}{\sqrt{\kappa_B \kappa_{B'}}}    \, \mathcal{G}_{\alpha'}^{\dagger} \widetilde{\mathcal{S}}_{\alpha'} \mathcal{G}_{\alpha'}   \quad & , \quad \sigma_{\gamma'} := \frac{\kappa_{B'AB}}{\sqrt{\kappa_B \kappa_{B'}}}  \, \mathcal{G}_{\gamma'}^{\dagger} \widetilde{\mathcal{S}}_{\gamma'} \mathcal{G}_{\gamma'}
\end{align*}
Using that $\kappa_{ABCB'} \kappa_{B} \kappa_{B'} = \kappa_{B'AB} \, \kappa_{BCB'}$, we can easily verify 
\begin{align*}
\sigma_{\partial B} = \sigma_{\alpha} \otimes \sigma_{\gamma} \quad & , \quad \sigma_{\partial B'} = \sigma_{\alpha'} \otimes \sigma_{\gamma'}  \\[2mm]
\sigma_{\partial B'AB} = \sigma_{\gamma'} \otimes \sigma_{\gamma}  \quad & , \quad \sigma_{\partial BCB'} = \sigma_{\alpha} \otimes \sigma_{\alpha'} \,.
\end{align*}
By Theorem \ref{Coro:approx-factorization-bound}, we have  for any cylinder $\mathcal{R} \in \{ B'AB, BCB', B, B' \}$
\[ 
\| \rho_{\partial \RR}^{1/2} \sigma_{\partial \RR}^{-1} \rho_{\partial \RR}^{1/2} - J_{\partial \RR} \| \,\, , \,\, \| \rho_{\partial \RR}^{-1/2} \sigma_{\partial \RR} \rho_{\partial \RR}^{1/2} - J_{\partial \RR} \| \, \leq \, 2 \epsilon_{BB'} \,.
\]
Reasoning as in the previous corollary, we have for the region joint $BB'$ 
\begin{align*}
\| \rho_{\partial BB'}^{1/2} \sigma_{\partial BB'}^{-1} \rho_{\partial BB'}^{1/2} - J_{\partial BB'}\| \, \leq \,   6 \,  \epsilon_{BB'}\,.
\end{align*}
Applying Theorem \ref{theo:non-injective-approx-fact}, we conclude the result.
\end{proof}

\subsection{Parent Hamiltonian of the thermofield double}
\label{sec:parent-hamiltonian}

The PEPS description of $\ket*{\rho_{\beta}^{1/2}}$ is given in terms of a family of tensors whose contraction defines linear maps
\[ V_{\RR}: \mathcal{H}_{\partial \RR} \longrightarrow \mathcal{H}^2_{\RR} \]
where $\RR$ runs over all rectangular regions $\RR \in \mathcal{F}_{N}$.Recall that we denote by $P_{\RR}$ be the orthogonal projection onto $\operatorname{Im}(V_{\RR})$.
\begin{Coro}
The family of orthogonal projectors $(P_{\RR})_{\RR}$ satisfies the Martingale Condition with the decay function $\delta:(0, \infty) \longrightarrow \mathbb{R}$ given by
\[ \delta(\ell) = \min{\big\{ 1, 144 \abs{G}^2 \qty(\frac{\gamma_\beta}{1+\gamma_\beta})^{ \ell-1} \,\, \big\}} 
\quad \text{where} \quad \gamma_{\beta} := \frac{e^{\beta} - 1}{|G|}\,. 
\]
\end{Coro}

\begin{proof}
We have to check conditions $(i)$-$(iii)$ from Definition \ref{Defi:MartingaleCondition} for the given function $\delta(\ell)$. To see $(i)$, notice first that $\| P_{ABC} - P_{AB}P_{BC}\| \leq 1$ always holds by Lemma \ref{Lemm:MartingaleProjector}. So let us assume that $B$ contains at least $\ell$ plaquettes along the splitting direction for a value $\ell$ satisfying $\delta(\ell) < 1$, which necessarily means that 
\begin{equation}\label{equa:MCaux1} 
\delta(\ell) = 144 \abs{G}^2 \qty(\frac{\gamma_\beta}{1+\gamma_\beta})^{ \ell-1} < 1\,. 
\end{equation}
Using the notation of Corollary \ref{coro:approx-fact-rectangles}, and taking into account that $B$ contains at least two plaquettes along the non-splitting direction, we deduce that
\begin{equation}\label{equa:MCaux2} 
\epsilon_{B} < 3 |G|^{2} \qty(\frac{\gamma_\beta}{1+\gamma_\beta})^{ \ell-1} = \frac{\delta(\ell)}{48} < \frac{1}{48}\,.
\end{equation}
Thus, the aforementioned corollary can be applied to obtain the desired estimate
\[ \| P_{AB}P_{BC} - P_{ABC}\| \leq 8 \epsilon_{B}(1+\epsilon_{B}) < 16 \epsilon_{B} \leq \delta(\ell)\,.  \]
The proof of conditions $(ii)$ and $(iii)$ is analogous, using respectively Corollaries \ref{coro:approx-fact-cylinder} and \ref{coro:approx-fact-torus}. For instance, to see $(ii)$, we can again assume that $B$ and $B'$ have both at least $\ell$ plaquettes along the splitting direction where $\ell$ satisfies \eqref{equa:MCaux1}. This yields, using the notation of Corollary \ref{coro:approx-fact-cylinder}, that \eqref{equa:MCaux2} holds with $\epsilon_{BB'}$ instead of $\epsilon_{B}$. And thus, applying the same corollary we get
\[ \| P_{B'AB}P_{BCB'} - P_{ABCB'}\|  < 48 \epsilon_{B} \leq \delta(\ell)\,.  \]
\end{proof}

Next, we aim at constructing a parent Hamiltonian of this PEPS following the guidelines of Section \ref{subsec:parentHam}, i.e. by verifying the conditions of Proposition~\ref{Prop:parentHamiltonianProperty}. 

\begin{Lemm}\label{rema:n-beta}
For each $\beta>0$, let us fix a natural number
\[n(\beta) \geq 4\left( 1+(1+\gamma_{\beta}) \log(288\abs{G}^2)\right)\,.\]
Then, for every $\ell \geq n(\beta)$ it holds that 
\[ \delta(\ell/4) \leq \left( 1/2\right)^{\frac{\ell}{n(\beta)}} \le 1/2\,.\] 
\end{Lemm}

\begin{proof}
Using that the map $x \mapsto (1-1/x)^x$ is upper bounded by $1/e$ for $x>1$, we can estimate from above
\begin{equation*} 
\qty( \frac{\gamma_{\beta}}{1+\gamma_{\beta}})^{\frac{n(\beta)}{4}-1} \, = \, \left( 1-\frac{1}{1+\gamma_{\beta}}\right)^{(1+\gamma_{\beta})\frac{\frac{n(\beta)}{4}-1}{1+\gamma_{\beta}}} \, \leq \, e^{-\frac{\frac{n(\beta)}{4}-1}{1+\gamma_{\beta}}} \, \leq \, \frac{1}{288 |G|^{2}} \,.  
\end{equation*}
Hence, for every $\ell \geq n(\beta)$
\begin{equation*}
 \qty(\frac{\gamma_\beta}{1+\gamma_\beta})^{ \frac{\ell}{4}-1} \leq \qty(\frac{\gamma_\beta}{1+\gamma_\beta})^{ (\frac{n(\beta)}{4}-1) \frac{\ell}{n(\beta)}}  \leq \left(\frac{1}{288 |G|^{2}}\right)^{\frac{\ell}{n(\beta)}} \leq \frac{1}{144\abs{G}^{2}} \left( \frac{1}{2}\right)^{\frac{\ell}{n(\beta)}} \,.
\end{equation*}
 The last inequality immediately yields the result.
\end{proof}
Next, consider the family of proper rectangles $\XXX_{n,N} = \mathcal{F}_{n,N}^{rect}$ having at most $n \in \mathbb{N}$ plaquettes per row and per column
\[ 
\begin{tikzpicture}[equation, scale=0.4]
\begin{scope}[xshift=30cm]
\draw[step=1.0,gray,thick] (0,0) grid (5,3);
\draw[<->] (0,-0.5) -- (5,-0.5);
\draw[<->] (-0.5,0) -- (-0.5,3);
\draw (2.5,-1) node {$a$};
\draw (-1,1.5) node {$b$};
\end{scope}
\end{tikzpicture} \hspace{2cm} 
1 \leq a,b \leq n\,, \]
and define the local Hamiltonian 
\[ H_{\EE_N} = \sum_{X \in \XXX_{n,N}}{P_{X}^{\perp}}  \,, \]
where $P_{X}^{\perp} := \mathbbm{1} - P_{X}$, that is, the orthogonal projection onto $\operatorname{Im}(V_{X})^{\perp}$. Note that the range of interaction of the Hamiltonian $H_{\EE_N}$ depends on the parameter $n$. 


If we choose $n=n(\beta)$ from Lemma \ref{rema:n-beta}, we obtain as a consequence of Proposition~\ref{Prop:parentHamiltonianProperty}, that $H_{\mathcal{E}_N}$ is a parent Hamiltonian for $\ket*{\rho_\beta^{1/2}}$.
\begin{Coro}[Uniqueness of the ground state]
For every rectangular region $\RR \in \mathcal{F}_{N}$ containing at least $n(\beta)$ plaquettes per row and per column we have that the associated Hamiltonian  
\[ H_{\RR} = \sum_{X \in \XXX , X \subset \RR} P_{X}^{\perp} \quad \quad \mbox{satisfies} \quad \quad \ker(H_{\RR}) = \operatorname{Im}(P_{\RR})\,. \]
In other words, $P_{\RR}$ is the orthogonal projector onto the ground state space of $H_{\RR}$. In particular, $\ket*{\rho_{\beta}^{1/2}}$ is the unique ground state of $H_{\EE_{N}}$.
\end{Coro}

Let us now turn to the spectral gap properties of this Hamiltonian. Combining Theorems~\ref{Theo:spectralGapOpenBoundary} and~\ref{Theo:RecursiveGapEstimate}, we can estimate the spectral gap of the parent Hamiltonian uniformly in the system size.

\begin{Coro}[Spectral gap]\label{coro:spectral-gap-parent-hamiltonian}
There is a positive constant $\mathcal{K}>0$ independent of $N$ and $\beta$ such that
\[ \gap(\mathcal{F}_{N,N}^{torus}) \, \geq \, \mathcal{K}\,. \]
\end{Coro}

\begin{proof}
Let us denote  $\delta_{k}:= \delta(\lfloor \tfrac{n(\beta)}{4}(\sqrt{9/8}\,)^{k}\rfloor\,)$ and $s_k:=\lfloor (\sqrt{4/3}\,)^{k} \rfloor$ for each integer $k \geq 0$. Observe that the choice of $n(\beta)$ in Lemma \ref{rema:n-beta} and the fact that $\log{(288)} > 5$ ensure that $N \geq n(\beta) >16$, and that for each $k \geq 0$
\begin{equation} \label{eq:delta-k-bound}
\delta_{k} \leq (1/2)^{(\sqrt{9/8})^{k}} < 1/2\,.
\end{equation}
We can apply Theorems \ref{Theo:fromPeriodictoOpenBoundary} and \ref{Theo:spectralGapOpenBoundary} to estimate from below
\[ \gap(\mathcal{F}_{N,N}^{torus}) \geq \frac{1}{16}\gap(\mathcal{F}_{N,N}^{rect}) \, \geq \, \frac{1}{16} \left[\, \prod_{k=0}^{\infty} \frac{1-\delta_{k}}{1+\frac{1}{s_{k}}}\right] \gap(\mathcal{F}^{rect}_{N,n})  \,. \]
Observe that for regions $X \in \mathcal{F}^{rect}_{N,n}$ we have that $\gap(H_{X}) \geq 1$ since $H_{X} \geq P_{X}^{\perp}$, as $P_{X}^{\perp}$ itself is a local interaction, and $P_{X}$ is the projector onto the ground space of $H_{X}$ by the previous corollary. Therefore, we can lower bound $\gap(\mathcal{F}^{rect}_{N,n}) \geq 1$, and so
\[ \gap(\mathcal{F}_{N,N}^{torus}) \geq \frac{1}{16} \left[\, \prod_{k=0}^{\infty} \frac{1-\delta_{k}}{1+\frac{1}{s_{k}}}\right]\,. \]
The infinite product in the r.h.s. is positive and can be bounded independently of $N$ and $\beta$, since
\[ 
\prod_{k=0}^{\infty} \frac{1-\delta_{k}}{1+\frac{1}{s_{k}}} \geq \exp\left[ - 2\sum_{k}\delta_{k} -\sum_{k} \frac{1}{s_{k}} \right] .
\]
Both of the two series are summable and their value is upper bounded by a constant independent of $N$ and $\beta$, since $s_k$ does not depend on either and $\delta_k$ can be estimated as in \eqref{eq:delta-k-bound}.
\end{proof}

\section{Davies generators for Quantum Double Models}
\label{sec:davies-generators}

\subsection{Davies generators}
We will now recall the construction of the generator of a semigroup of quantum channels which describes a weak-coupling limit of the joint evolution of the system with a local thermal bath, known as the Davies generator \cite{Davies}.
This construction applies to any commuting local Hamiltonian, but for simplicity of notation we will only consider the same setup of the previous sections, i.e. the qudits $\mathbb{C}^{d}$ live on the edges, and not the vertices, of a lattice $\Lambda = (\mathcal V, \mathcal E)$. 

\begin{Defi}
\label{def:detailed-balance}
Let $\rho_{\beta}$ denote the Gibbs state associated to $H_\Lambda^{\text{syst}}$ at inverse temperature $\beta$. Then
\[ \langle A, B\rangle_{\beta} := \operatorname{Tr}(\rho_{\beta} A^{\dagger}B) \qc A, B \in \mathcal{B}_{\Lambda}, \]
defines a scalar product on $\mathcal{B}_{\Lambda}$, called the \emph{Liouville} or \emph{GNS} scalar product.

An operator $\mathcal{T}:\mathcal{B}_{\Lambda} \to \mathcal{B}_{\Lambda}$ satisfies \emph{detailed balance} if it is self-adjoint with respect to the GNS scalar product.
\end{Defi}

When the system is in contact with a thermal bath at inverse temperature $\beta$, the joint Hamiltonian of the system+bath is given by 
\begin{equation}
H_{\lambda} = H_\Lambda^{\text{syst}} \otimes \mathbbm{1}^{\text{bath}} + \mathbbm{1}^{\text{\text{syst}}} \otimes H^{\text{bath}} + \lambda H_I
\end{equation}
where $H_\Lambda^{\text{syst}}$ is the Hamiltonian of the system,  $H^{\text{bath}}$ is the Hamiltonian of the bath, $H_I$ is the coupling term between system and bath, and $\lambda \ge 0$ is the coupling strength.
We will assume that the coupling interaction is \emph{local}, in the sense that
\begin{equation}
    H_I = \sum_{e\in \EE} \sum_{\alpha \in \mathcal{A}} S_{e,\alpha} \otimes B_{e,\alpha},
\end{equation}
where for each edge $e$, $\{S_{e,\alpha}\}_{\alpha \in \mathcal{A}}$ is a finite family of self-adjoint operators on $\mathcal{H}_e$, while $B_{e,\alpha}$ are self-adjoint operators on the bath. 

We will assume that the bath satisfies the appropriate conditions that guarantee (see \cite{Davies}), in the weak-coupling limit, that the reduced dynamics  of the system in the Heisenberg picture is described by a Quantum Markov semigroup $T_{t} = e^{t \mathcal G}:\mathcal{B}_\Lambda \longrightarrow \mathcal{B}_\Lambda$ whose generator takes the following form:
\begin{equation}\label{eq:davies-generator}
    \mathcal G = i\delta_{H_\Lambda^{\text{syst}}} + \mathcal{L},
\end{equation}
where $\delta_{H_\Lambda^{\text{syst}}}(Q) = \comm{H_\Lambda^{\text{syst}}}{Q}$ is a derivation (generating a unitary evolution), while 
the dissipative term $\mathcal L$ has the specific form
\begin{align}\label{eq:davies-generator-dissipation}
\mathcal L &= \sum_{e\in \mathcal{E}} \mathcal L_e  \quad , \quad  \mathcal L_e = \, \sum_{\alpha, \omega} \widehat{g}_{e,\alpha}(\omega) \, \mathcal{D}_{e,\alpha,\omega},\\ \notag \mathcal{D}_{e,\alpha,\omega}(Q) &= \frac{1}{2}\qty( S_{e,\alpha}^{\dagger}(\omega) [Q, S_{e,\alpha}(\omega)] +  [S_{e,\alpha}^{\dagger}(\omega), Q] \, S_{e,\alpha}(\omega) ).
\end{align}

The variable $\omega$ runs over the finite set of Bohr frequencies of $H_\Lambda^{\text{syst}}$ (the differences between energy levels), while  $\widehat{g}_{e,\alpha}(\omega)$ are positive transition rates which depend on the autocorrelation function of the bath. In particular, if the bath is assumed to be at thermal equilibrium at inverse temperature $\beta$, then they satisfy $\widehat{g}_{e,\alpha}(-\omega) = e^{- \beta \omega} \, \widehat{g}_{e,\alpha}(\omega)$.
The jump operators  $S_{e,\alpha}(\omega)$ are the Fourier components of $S_{e, \alpha}$ evolving under $H_\Lambda^{\text{syst}}$, namely
\[ e^{itH_\Lambda^{\text{syst}}} S_{e,\alpha} e^{-itH_\Lambda^{\text{syst}}} = \sum_{\omega} S_{e,\alpha}(\omega) \, e^{-i \omega t}\,, \,\, t \in \mathbb{R}\,. \]
From this definition, it follows that $S_{\alpha}^{\dagger}(\omega) = S_{\alpha}(-\omega)$.

We can then rewrite the sum in the definition of $\mathcal L_e$ only over $\omega \geq 0$:
\[ \mathcal{L}_e = \sum_{\alpha} \qty[
    \widehat{g}_{e, \alpha}(0) \, \mathcal{D}_{e,\alpha,0} + 
\sum_{\omega > 0} \qty(\, \widehat{g}_{e, \alpha}(\omega) 
\, \mathcal{D}_{e,\alpha,\omega} + \widehat{g}_{e, \alpha}(-\omega) \, \mathcal{D}_{e,\alpha,-\omega} \, )
].
\]
We now denote for each $e, \alpha$ and $\omega \geq 0$
\[ 
\mathcal{L}_{e, \alpha, \omega} = \begin{cases}
    \widehat{g}_{e, \alpha}(0) \, \mathcal{D}_{e,\alpha,0} & \text{if $\omega = 0$},\\
    \widehat{g}_{e, \alpha}(\omega) \, \mathcal{D}_{e,\alpha,\omega} + \widehat{g}_{e, \alpha}(-\omega) \, \mathcal{D}_{e,\alpha,-\omega} & \text{otherwise}.
\end{cases}
\]

The properties of the Davies generator which we will need are summarized in the following proposition.
\begin{Prop}[\cite{Davies}]\label{prop:davies-generator}
\leavevmode
\begin{enumerate}
    \item $\mathcal L_{e}$ commutes with $\delta_{H_\Lambda^{\text{syst}}}$ for each $e$.
    \item $\mathcal L_{e,\alpha,\omega}$ satisfies detailed balance for each $e, \alpha$ and $\omega$.
    \item $- \mathcal L_{e,\alpha,\omega}$ is positive semidefinite w.r.t. the GNS scalar product. In fact
    \begin{equation}\label{equa:positive1} 
-\langle A, \mathcal{L}_{e, \alpha, \omega}(A)\rangle_{\beta}  =  \widehat{g}_{e, \alpha}(\omega)    \, \norm{ [A, S_{e, \alpha}(\omega)] }^2_{\beta} + \, \widehat{g}_{e, \alpha}(-\omega)  \, \norm*{ [A, S_{e, \alpha}^{\dagger}(\omega)]}^2_{\beta} \ge 0,
\end{equation}
    \item The Gibbs state  $\rho_{\beta}$ is an invariant state for $\mathcal L_{e,\alpha,\omega}$, in the sense that 
    \[
     \operatorname{Tr}(\rho_{\beta}e^{t {\mathcal L_{e,\alpha,\omega}}}(Q)) = \operatorname{Tr}(\rho_{\beta} Q) \quad \text{ for every  $Q$ and every $t \ge 0$.}
    \]
    \end{enumerate}
    \end{Prop}

Our main result will only apply under the following conditions on the thermal bath: the first will guarantee that the generator $\mathcal{L}$ is translation invariant, while the second will imply its ergodicity.

\begin{Assumption}\label{assumption-ti}
The coupling operators $S_{e,\alpha} $  and the transition rates $\widehat{g}_{e,\alpha}(\omega)$ are \emph{translation invariant}, in the sense that 
\begin{enumerate}
    \item for a fixed $\alpha \in \mathcal{A}$, $S_{e,\alpha}$ are translates of each other when we vary $e$;
    \item $\widehat{g}_{e,\alpha}(\omega)$ does not depend on $e$.
\end{enumerate}
In this case, we will write $\widehat{g}_{\alpha}(\omega)$ for the common value of $\widehat{g}_{e,\alpha}(\omega)$.
\end{Assumption}
\begin{Assumption}\label{assumption-primitivity}
The following two conditions hold:
\begin{enumerate}
\item $\widehat{g}_{e,\alpha}(\omega) > 0$ for every choice of $e\in \EE$, $\alpha$ and $\omega$ such that $S_{e,\alpha}(\omega)  \neq 0$.
\item  For every $e \in \EE$
\begin{equation}
\{ S_{e,\alpha}\}'_{\alpha \in \mathcal{A}} = \mathbb{C}\mathbbm{1}_{\mathcal{H}_e},
\end{equation}
where $\{\}'$ denotes the commutant.
\end{enumerate}
\end{Assumption}

\begin{Prop}\label{prop:davies-kernel}
Suppose that $\widehat{g}_{e,\alpha}(\omega) > 0$ for every choice of $e\in \EE$, $\alpha$ and $\omega$ such that $S_{e,\alpha}(\omega)  \neq 0$.
Then for each $X\subseteq \EE$, we have that
    \begin{equation}\label{eq:ker-davies}
    \ker\qty(\sum_{e\in X} \mathcal{L}_e )= \qty{ S_{e,\alpha}(\omega) \mid e\in X, \forall \alpha, \omega }'.
    \end{equation}
 In particular, if Assumption \ref{assumption-primitivity} holds, then
   \[ 
    \ker\qty(\sum_{e\in X} \mathcal{L}_e ) \subset \mathbbm{1}_{\mathcal{H}_X} \otimes \mathcal{B}(\mathcal{H}_{\EE\setminus X}),
   \]
   which implies that $\mathcal{L} = \sum_{e \in \EE} \mathcal{L}_e$ has a unique invariant state.
\end{Prop}
\begin{proof}
The inclusion of the commutant of the jump operators in the kernel of the generator can be verified directly. The converse follows from the fact that $\rho_\beta$ is a full rank invariant state for $\sum_{e}\mathcal{L}_e$. This implies that $\ker(\sum_{e} \mathcal{L}_e)$ is a *-algebra \cite{Frigerio}. If $A \in \ker(\sum_{e} \mathcal{L}_e)$, then also $A^\dag$ and $A^\dag A$ are in $\ker(\sum_{e} \mathcal{L}_e)$. Then 
one can see that
\begin{align*}
0 &= \sum_{e\in X} \mathcal{L}_e(A^\dag A) - A^\dag \mathcal{L}_e(A) - \mathcal{L}_e(A)^\dag A 
\\ &= \sum_{e\in X,\alpha,\omega} \widehat{g}_{e,\alpha}(\omega)
\Big(
S_{e,\alpha}(\omega)^\dag A^\dag A S_{e,\alpha}(\omega)
- S_{e,\alpha}(\omega)^\dag A^\dag S_{e,\alpha}(\omega)A \\
& \phantom{= \sum_{e\in X,\alpha,\omega}} - A^\dag S_{e,\alpha}(\omega)^\dag  A S_{e,\alpha}(\omega)
+ A^\dag S_{e,\alpha}(\omega)^\dag  S_{e,\alpha}(\omega) A \Big)
\\ &= \sum_{e\in X,\alpha,\omega} \widehat{g}_{e,\alpha}(\omega) {\comm{A}{S_{e,\alpha}(\omega)}}^\dag \comm{A}{S_{e,\alpha}(\omega)}\,,
\end{align*}
which is only possible when $A$ commutes with all the jump operators $S_{e,\alpha}(\omega)$.
The second part follows from the fact that $S_{e,\alpha} = \sum_{\omega} S_{e,\alpha}(\omega)$ and therefore
\[ \qty{ S_{e,\alpha}(\omega) \mid e\in X, \forall \alpha, \omega }' \subset  \{ S_{e,\alpha} \mid e\in X, \forall \alpha \}' , \]
and, if Assumption \ref{assumption-primitivity} holds, then the r.h.s. of the last equation is simply $\mathbbm{1}_{\mathcal{H}_X} \otimes \mathcal{B}(\mathcal{H}_{\EE\setminus X})$, which reduces to $\mathbb{C} \mathbbm{1}_{\mathcal{H}_{\EE}}$ when $X = \EE$.
\end{proof}

\begin{Coro}
If Assumption~\ref{assumption-primitivity} holds, then $\mathcal{L}$ is primitive, in which case
\[ \gap(\mathcal{G}) = \gap(\mathcal{L}). \]
\end{Coro}

In what follows, we will always assume that the choice of the coupling terms $S_{e,\alpha}$ guarantees that the Davies generator is primitive.

\subsection{Davies generators for the Quantum Double Models}
We can now give an explicit description of the Davies generators for the Quantum Double Models.

\begin{Prop}\label{prop:davies-quantum-double}
If $H_\Lambda^{\text{syst}}$ is the Hamiltonian of a Quantum Double model, then for each $e\in \EE$, $\mathcal L_e$ only acts non-trivially on the two plaquettes and the two stars containing $e$ (see Figure~\ref{fig:support-local-lindbladian}).
Moreover, $S_{e,\alpha}(\omega)$ is zero unless $\omega\in\Omega := \{-4,-3,\dots, 3, 4\}$.
\end{Prop}
\begin{figure}[ht]
\begin{center}  \begin{tikzpicture}   
  \draw[step=1.0,gray,thin] (-1.5,-0.5) grid (2.5,2.25);
  \draw[ultra thick] (0,0) -- (1, 0) -- (1,1) -- (0,1) -- (0,0);
   \draw[ultra thick] (2,1) -- (1,1) -- (1,2) -- (0,2) -- (0,1) -- (-1,1);
    \draw[ultra thick, red] (0,1) -- (1,1);

    \foreach \c in {0.5, 1, 1.5}{
    \foreach \x in {-1, -0.5, 0}
        \shade[ball color=black] ($ (\x + \c, \c -\x -0.5) $) circle (0.8ex);
    }
  \node at (0.5, 0.5) {$p_1$};
  \node at (0.5, 1.5) {$p_2$};
  \node at (-0.2,1.2) {$v_1$};
 \node at (1.2,0.8) {$v_2$};
\end{tikzpicture}
\end{center}
\caption{\label{fig:support-local-lindbladian}
Region supporting $\mathcal L_e$ (the edge $e = (v_1, v_2)$ is marked in red).}
\end{figure}
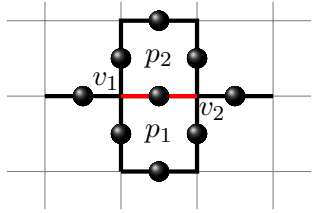
\begin{proof}
Since the local terms of $H_\Lambda^{\text{syst}}$ are commuting, for fixed $e \in \EE$ and $\alpha$ and for every $t \in \mathbb{R}$ we have that
\begin{equation}\label{equa:FourierSfrequencies}
\begin{split}
e^{it H_\Lambda^{\text{syst}}} & S_{e, \alpha} e^{-it H_\Lambda^{\text{syst}}}\\ 
& = e^{it(A(v_{1}) + A(v_2) + B(p_1) + B(p_2))} \, S_{e, \alpha} \, e^{-it(A(v_{1}) + A(v_2) + B(p_1) + B(p_2))} \,, 
\end{split}
\end{equation}
where $v_{1}, v_{2}$ are the vertices of $e$ and $p_{1}, p_{2}$ are the two plaquettes containing $e$.
The local terms $A(v)$ and $B(p)$ are both projections. For a projection $\Pi$ with orthogonal complement $\Pi^{\perp}:=\mathbbm{1} - \Pi$, it holds that
\[ e^{it \Pi} =  \Pi^{\perp} + e^{it} \Pi\,, \]
and so
\[ e^{it \Pi} Q e^{-it \Pi} = \Pi Q \Pi + \Pi^{\perp}Q\Pi^{\perp} + e^{-it} \Pi^{\perp}Q\Pi + e^{it} \Pi Q \Pi^{\perp}\,. \]
Consequently, we can rewrite \eqref{equa:FourierSfrequencies} as a sum
\[ e^{itH_\Lambda^{\text{syst}}} \, S_{e, \alpha} \, e^{-itH_\Lambda^{\text{syst}}} \, = \, \sum_{\omega=-4}^{4} \, S_{e, \alpha}(\omega) \, e^{it \omega}\,. \]
where each $S_{e, \alpha}(\omega)$ has also support contained in $\partial v_1 \cup \partial v_2 \cup p_{1} \cup p_{2}$ .
\end{proof}

\subsection{Davies generator as a local Hamiltonian}
We have seen in the previous sections that $- \mathcal L = -\sum_{e\in \EE}\mathcal{L}_e$ is local, self-adjoint, and positive with respect to the GNS scalar product, with a unique element in its kernel corresponding to the Gibbs state $\rho_\beta$. We will now describe how to convert it into a frustration free local Hamiltonian whose unique ground state is the thermofield double of $\rho_\beta$ (a local purification of the Gibbs state), defined as
\begin{equation}\label{eq:thermofield-double}
\ket*{\rho_\beta^{1/2}} = \frac{1}{Z_\beta^{1/2}} \sum_{\lambda \in \sigma(H_\Lambda^{\text{syst}})} e^{-\frac{\beta}{2}} \ket{\lambda}\otimes \ket{\lambda} = \qty(\rho_\beta^{1/2} \otimes \mathbbm{1}) \ket{\Psi},
\end{equation}
where $\ket{\Psi}$ is a maximally entangled state on $\mathcal{H}_\Lambda^2 = \mathcal{H}_{\Lambda} \otimes \mathcal{H}_{\Lambda} $.
We will find convenient to work with a ``vectorized'' representation of the GNS scalar product, which we introduce in the next proposition.
\begin{Prop}\label{prop:gns-purification}
Let $\ket{\Psi}$ be a maximally entangled state on $\mathcal{H}_\Lambda^2 = \mathcal{H}_{\Lambda} \otimes \mathcal{H}_{\Lambda} $, and denote 
\[\iota(A) = A\rho_\beta^{1/2},\quad A \in \mathcal{B}_{\Lambda}.\]
Then
\begin{equation}
    \mathcal{B}_\Lambda \ni Q \longrightarrow \ket*{\iota(Q)} = \qty( Q \rho_\beta^{1/2} \otimes \mathbbm{1}) \ket{\Psi} \in \mathcal{H}_\Lambda^2
\end{equation}
is an isometry between $(\mathcal{B}_{\Lambda}, \norm{\cdot}_\beta)$ and $\mathcal{H}_\Lambda^2$ (equipped with the natural Hilbert tensor scalar product).
\end{Prop}
\begin{proof}
The map is manifestly linear, it is a bijection since $\rho_\beta$ is full rank, and a simple calculation shows that it preserves the scalar product:
\[ \braket{\iota(A)}{\iota(B)}
    = \expval{\qty(\rho_\beta^{1/2} A^\dag B \rho_\beta^{1/2}) \otimes \mathbbm{1}}{\Omega} = \tr(\rho_\beta A^\dag B) = 
\langle A, B\rangle_{\beta}  \] 
for every $A,B \in \mathcal{B}_\Lambda$
\end{proof}
Note that $(\mathcal{H}_{\Lambda}^2, \pi, \ket*{\rho_\beta^{1/2}})$, with $\pi(Q) = Q\otimes \mathbbm{1}$, is a GNS triple for the pair $(\mathcal{B}_\Lambda,\rho_\beta)$.

This isometry allows us to define a local Hamiltonian representing the dissipative part of the Davies generator $\mathcal{L}$:

\begin{Prop}\label{prop:davies-local-gap}
Let $\widetilde{H}$ be the operator on $\mathcal H^2_\Lambda$ defined by
\[ \widetilde{H} = \sum_{e \in \mathcal{E}} \widetilde{H}_{e} \quad \text{where} \quad  \widetilde{H}_{e}\ket{\iota(Q)} = -\ket{\iota(\mathcal{L}_e (Q)}\qc \forall Q \in \mathcal{B}_{\Lambda}.\]
Then, $\widetilde{H}$ is self-adjoint and positive semidefinite.
Moreover, if Assumption~\ref{assumption-primitivity} is satisfied, the following statements hold:
\begin{enumerate}
\item If $X \subset \mathcal E$ is a finite subset, then
\begin{align}
    \ker\qty( \sum_{e\in X} \widetilde{H}_e) &= 
    \left\{  \ket*{\iota(Q)}  \mid Q \in \ker\qty(\sum_{e \in X} \mathcal{L}_e) \right\} \\
    &\subseteq \left\{ \qty( (Q  \otimes \mathbbm{1}_{X}) \rho_\beta^{1/2} ) \otimes \mathbbm{1} \ket{\Psi} \mid Q \in \mathcal{B}_{\EE\setminus X} \right\} \notag
\end{align}
In particular, the thermofield double $\ket*{\rho_\beta^{1/2}}= \ket{\iota(\mathbbm{1})}$ is the unique ground state of $\widetilde H$.
\item  $\gap{\mathcal{L}} = \gap{\widetilde{H}}$.
\end{enumerate}
\end{Prop}
\noindent The proof is an immediate consequence of Proposition~\ref{prop:gns-purification} and Proposition~\ref{prop:davies-kernel}.
\begin{Rema}\label{Rema:supportHtilde}
Note that the support of $\widetilde{H}_{e}$ will in general be larger than the support of $\mathcal{L}_e$, see Figure \ref{fig:SupportLtildee}.

\begin{figure}[ht]
\begin{center}  \begin{tikzpicture}   
  \draw[step=1.0,gray,thin] (-2.5,-1.5) grid (3.5,3.25);
    \draw[ultra thick] (-2,1) -- (3,1);
    \draw[ultra thick] (-1,2) -- (2,2) -- (2,0) -- (-1,0) -- cycle;
     \draw[ultra thick] (0,3) -- (1,3) -- (1,-1) -- (0,-1) -- cycle;
    \draw[ultra thick, red] (0,1) -- (1,1);
    
    \foreach \c in {0, 0.5, ..., 2}{
    \foreach \x in {-1.5, -1, ..., 0.5}
        \shade[ball color=black] ($ (\x + \c, \c -\x -0.5) $) circle (0.8ex);
    }
  
\end{tikzpicture}
\end{center}
  \caption{Region supporting $\widetilde{H}_{e}$ (the edge $e$ is marked in red).}
  \label{fig:SupportLtildee}
\end{figure}
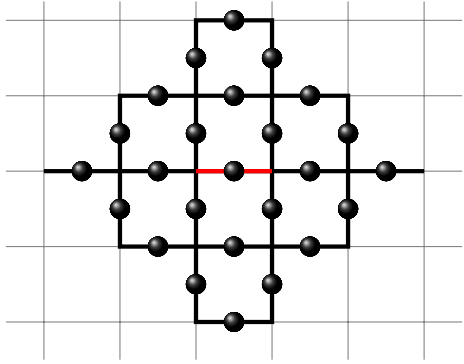
\end{Rema}

In general, $\widetilde H_e$ will not be a projection. In order to simplify the comparison with the Parent Hamiltonian, we will first lower bound $\sum_{e\in X} \widetilde H_e$ by a projection.

\begin{Prop}\label{Prop:lowerBoundLocalDaviesHamiltonian}
For each $X \subset \EE$, let $\Pi_X$ denote the projection on the subspace
\begin{equation}
\left\{ \qty( (Q  \otimes \mathbbm{1}_{X}) \rho_\beta^{1/2} ) \otimes \mathbbm{1} \ket{\Psi} \mid Q \in \mathcal{B}_{\EE\setminus X} \right\}.
\end{equation}
Then, under Assummptions~\ref{assumption-ti} and \ref{assumption-primitivity}, there exist positive constants $C_1$, $C_2$ independent of the system size and of $\beta$ such that for every $X \subset \EE$
    \begin{equation}\label{eq:local-gap-davies}
    \sum_{e\in X} \widetilde H_e \ge\frac{C_2}{\abs{\Omega}}  \frac{\widehat{g}_{\min}}{\abs{X}}e^{-C_1 \beta \abs{R_X}} \,  \Pi_X^\perp \qc \widehat{g}_{\min} = \min_\alpha \min_{\omega \in \Omega} \widehat{g}_{\alpha}(\omega),
\end{equation}
where $R_X$ denotes the region supporting $\sum_{e \in X} \widetilde H_e$, and $\abs{X}$ and $\abs{R_X}$ denote the number of edges in $X$ and $R_X$, respectively.
\end{Prop}
\begin{proof}
Let us denote by
\[
    \delta_{e,\alpha, \omega}(A) = \comm{A}{ S_{e,\alpha}(\omega)} \qc \delta_{e,\alpha}(A) = \comm{A}{ S_{e,\alpha}} = \sum_\omega \delta_{e,\alpha,\omega}(A).
\]
We can then rewrite \eqref{equa:positive1} as 
\begin{multline*}
\braket*{\phi}{\widetilde{H}_e \phi}
 =-\braket{\iota^{-1} \phi}{\mathcal{L}_{e}(\iota^{-1} \phi)}_{\beta} 
= \sum_{\alpha,\, \omega } \widehat{g}_{\alpha}(\omega)  \braket{\delta_{e,\alpha, \omega}\iota^{-1}(\phi)}{\delta_{e,\alpha, \omega}\iota^{-1}(\phi)}_{\beta} \\
= \sum_{\alpha,\, \omega } \widehat{g}_{\alpha}(\omega)  \norm{\delta_{e,\alpha, \omega}\iota^{-1}(\phi)}^2_{\beta},
\end{multline*}
which implies
\begin{equation}\label{eq:proof-local-gap-1}
\sum_{e \in X} \braket*{\phi}{\widetilde{H}_e \phi} \ge \widehat{g}_{\min} \sum_{e \in X}  \sum_{\alpha,\, \omega } \norm{\delta_{e,\alpha, \omega}(Q)}^2_{\beta} \ge 
\frac{\widehat{g}_{\min}}{\abs{\Omega}} \sum_{e \in X}  \sum_{\alpha} \norm{\delta_{e,\alpha}(Q)}^2_{\beta}\,, 
\end{equation}
where we denoted $Q = \iota^{-1}(\phi) \in \mathcal{B}_{\Lambda}$, or equivalently $ \ket\phi = (Q\rho_{\beta}^{1/2} \otimes \mathbbm{1}) \ket{\Psi}$.

Let $R_X$ be the region supporting $\sum_{e \in X} \widetilde H_e$. We define a \emph{localized} version of the norm $\norm{\cdot}_{\beta}$, by
\[
    \norm{Q}_{R_X}^2 = \tr[ e^{-{\beta} H_{R_X}^{\text{syst}}} Q^* Q ] \qc Q \in \mathcal{B}_{\Lambda}.
\]
Note that 
\[
    \norm{Q}_\beta^2 = \frac{1}{Z_\beta} \norm{Q e^{-\frac{\beta}{2}(H_\Lambda^{\text{syst}} - H_{R_X}^{\text{syst}})} }^2_{R_X}.
\]
In other words, denoting $Q' = \frac{1}{Z_\beta^{1/2}}Q e^{-\frac{\beta}{2}(H_\EE^{\text{syst}} - H_{R_X}^{\text{syst}})}$,
we have that $\norm{Q}_\beta= \norm{Q'}_{R_X}$.

Moreover, since $ e^{-\beta \norm*{H_{R_X}^{\text{syst}}}} \mathbbm{1} \le e^{-\beta H_{R_X}^{\text{syst}}} \le \mathbbm{1}$ where  $\norm*{H_{R_X}^{\text{syst}}}$ denotes the largest eigenvalue of $H_{R_X}^{\text{syst}}$, it follows that
\begin{equation}
    e^{-\frac{\beta}{2} \norm*{H_{R_X}^{\text{syst}}}} \norm{Q'}_{HS} \le \norm{Q'}_{R_X} \le \norm{Q'}_{HS},
\end{equation}
where $\norm{\cdot}_{HS}$ denotes the standard Hilbert-Schmidt scalar product on $\mathcal{B}_\Lambda$.
Now let us consider the following quantity
\[
    |||Q'|||_X := \qty( \sum_{e\in X} \sum_{\alpha} \norm{\delta_{e,\alpha}(Q')}_{HS}^2)^{1/2},
\] 
so that we can lower bound the r.h.s. of \eqref{eq:proof-local-gap-1}, obtaining
\begin{equation}\label{eq:proof-local-gap-2}
\sum_{e \in X}  \sum_{\alpha} \norm{\delta_{e,\alpha}(Q)}^2_{\beta} 
= \sum_{e \in X}  \sum_{\alpha} \norm{\delta_{e,\alpha}(Q')}^2_{R_X}
\ge e^{-\beta \norm*{H_{R_X}^{\text{syst}}}}  |||Q'|||_X^2.
\end{equation}

We can see that $|||\cdot|||_X$ is a seminorm, and $|||Q'|||_X =0$ if and only if 
\[ Q' \in \{ S_{e,\alpha} \mid e\in X, \forall \alpha \}'= \mathcal{B}_{\EE \setminus X}.\] Moreover, for each $Q_0 \in \mathcal{B}_{\EE\setminus X}$, it holds that
\begin{align*} 
    |||Q' - Q_0 \otimes \mathbbm{1}_X|||_X^2 & = \sum_{e\in X, \alpha} || \delta_{e,\alpha}(Q') - \delta_{e,\alpha}(Q_0 \otimes \mathbbm{1}_X)||_{HS}^2\\ 
    & = \sum_{e \in X,\alpha} || \delta_{e,\alpha}(Q')||_{HS}^2 = |||Q'|||_X^2.
\end{align*}
Therefore we can pass to the quotient space $\mathcal{B}_{\EE}/\mathcal{B}_{\EE\setminus X}$, and compare the norm induced by $|||\cdot|||_X$ with the norm induced by the Hilbert-Schmidt scalar product. 
As we show in Lemma \ref{lemm:almostcommutelocally}, there exists a constant $C_2>0$, independent of system size and $\beta$, such that for every $X$
\begin{equation}\label{eq:proof-local-gap-3}
    |||Q'|||_X^2 \ge \frac{C_2}{\abs{X}} \inf_{Q_0 \in \mathcal{B}_{\EE \setminus X}} \norm{Q'-Q_0\otimes \mathbbm{1}_X}_{HS}^2
    \ge \frac{C_2}{\abs{X}} \inf_{Q_0 \in \mathcal{B}_{\EE \setminus X}} \norm{Q'-Q_0\otimes \mathbbm{1}_X}_{R_X}^2.
\end{equation}

\noindent Let us now consider $\Pi_X^\perp$. We observe that, once again setting
$\ket\phi = (Q\rho_{\beta}^{1/2} \otimes \mathbbm{1}) \ket{\Psi}$,
\begin{align*}
    \braket*{\phi}{\Pi_X^\perp\phi} = \norm*{\Pi_X^\perp \ket{\phi}}^2 &= 
    \inf_{Q_0 \in \mathcal{B}_{\EE\setminus X}} \norm{Q - Q_0 \otimes \mathbbm{1}_{X}}_\beta^2
    \\[2mm] &= \inf_{Q_0 \in \mathcal{B}_{\EE \setminus X}} \frac{1}{Z_\beta} \norm{ (Q - Q_0 \otimes \mathbbm{1}_{X})  e^{-\frac{\beta}{2}(H_\EE^{\text{syst}} - H_{R_X}^{\text{syst}})}}_{R_X}^2
    \\[2mm] &= \inf_{Q_0 \in \mathcal{B}_{\EE \setminus X}} \norm{ \frac{1}{Z_\beta} Q e^{-\frac{\beta}{2}(H_\EE^{\text{syst}} - H_{R_X}^{\text{syst}})} - Q_0 \otimes \mathbbm{1}_{X}}_{R_X}^2 
    \\[2mm] &= \inf_{Q_0 \in \mathcal{B}_{\EE \setminus X}} \norm{Q'-Q_0\otimes \mathbbm{1}_X}_{R_X}^2.
\end{align*}
Putting this last expression together with bounds \eqref{eq:proof-local-gap-1}, \eqref{eq:proof-local-gap-2} and \eqref{eq:proof-local-gap-3} concludes the proof, by observing that $\norm*{H_{R_X}^{\text{syst}}} \le C_1 \abs{R_X}$ for a constant $C_1$ independent of $X$ and of the system size.
\end{proof}

In order to prove the bound on the $||| \cdot |||_X$ seminorm needed to complete the proof of the last proposition, we first show an intermediate result about a related quantity.

\begin{Lemm}\label{lemm:equivalentSeminorms}
Let $(S_{\alpha})_{\alpha \in \mathcal{A}}$ and $(V_{ \gamma})_{\gamma \in \mathcal{C}}$ be two families of operators in $\mathcal{B}(\mathbb{C}^{d})$ such that their commutators satisfy $\{ S_{\alpha}\}_{\alpha}' = \{ V_{\gamma}\}_{\gamma}' = \mathbb{C}\mathbbm{1}_{d}$. For each $e \in \Lambda$ let $S_{e, \alpha}$ and $V_{e, \gamma}$ be the associated elements in the canonical inclusion $\mathcal{B}_{e} \hookrightarrow \mathcal{B}_{\Lambda}$, and define the following seminorms on $\mathcal{B}_{\Lambda}$
\begin{align*} 
|||Q|||_{e} = \left(\sum_{\alpha \in \mathcal{A}} \| [S_{e, \alpha}, Q] \|_{HS}^{2}\right)^{1/2} \quad , \quad |||Q|||_{e}' = \left(\sum_{\gamma \in \mathcal{C}} \| [V_{e, \gamma}, Q] \|_{HS}^{2}\right)^{1/2} \,.
\end{align*}
Then, there is a constant $K>0$ independent of the system size and of the edge $e$ such that 
\begin{equation}\label{equa:equivalentSeminorms}
K |||Q|||_{e} \leq {|||Q|||'}_{e} \leq \frac{1}{K} |||Q|||_{e} 
\end{equation}
for every $Q \in \mathcal{B}_{\Lambda}$.
\end{Lemm}

\begin{proof}
We can then rephrase the condition $\{ S_{\alpha}\}_{\alpha}' = \{ V_{\gamma}\}_{\gamma}' = \mathbb{C}\mathbbm{1}_{d}$ as the fact that the seminorms $|||\cdot |||_{e}$ and $|||\cdot |||_{e}'$ restricted to $\mathcal{B}_{e} \equiv \mathcal{B}(\mathbb{C}^{d})$ define actual norms on the quotient space $\mathcal{B}_{e}/\mathbb{C}\mathbbm{1} \equiv \mathcal{B}(\mathbb{C}^{d})/\mathbb{C} \mathbbm{1}$, whose dimension is $d^{2}-1$. Hence, both quotient norms have to be equivalent. But note that  $|||Q|||_{e} = ||| Q+\mathbb{C} \mathbbm{1}|||_{e}$ and $|||Q|||_{e}' = ||| Q+\mathbb{C} \mathbbm{1}|||_{e}'$ for every $Q \in \mathcal{B}_{e}$. Thus, there exists a constant $K>0$ depending on both of the families and on $d$, but independent of the system size and $e$, such that for every $Q \in \mathcal{B}_{e}$
\begin{equation}\label{equa:equivalentSeminormsAux1}
K |||Q|||_{e} \leq {|||Q|||'}_{e} \leq \frac{1}{K} |||Q|||_{e} \,.
\end{equation}
To extend the latter inequalities to every $Q \in \mathcal{B}_{\Lambda}$, let us fix such a $Q$ with $HS$-norm equal to one, and consider its Schmidt decomposition $Q = \sum_{k=1}^{d^{2}} s_{k} Q_{k} \otimes Q_{k}'$ w.r.t.\ the HS scalar product, where $\{ Q_{k}\}_{k=1}^{d^2} \subset \mathcal{B}_{e}$ and $\{Q_{k}'\}_{k=1}^{d^2} \subset \mathcal{B}_{\Lambda \setminus \{ e\}}$ are orthonormal sets and $\sum_{k}s_{k}^{2}=1$. Then
\[ \| [S_{e, \alpha}, Q] \|_{HS}^{2} = \| \sum_{k=1}^{d^{2}} s_{k} [S_{e, \alpha}, Q_{k}]  \otimes Q_{k}'  \|_{HS}^{2} = \sum_{k=1}^{d^{2}} s_{k}^{2} \| [S_{e, \alpha}, Q_{k}]\|_{HS}^{2} \, \| Q_{k}'\|_{HS}^{2}\,, \]
so that
\begin{align*} 
||| Q|||_{e}^{2} = \sum_{\alpha \in \mathcal{A}} \| [S_{e, \alpha}, Q] \|_{HS}^{2} & = \sum_{k=1}^{d^{2}} s^2_{k}\left(\sum_{\alpha \in \mathcal{A}} \| [S_{e, \alpha}, Q_{k}]\|_{HS}^{2}\right) \| Q_{k}'\|_{HS}^{2}\\ 
& = \sum_{k=1}^{d^{2}}s^2_{k} ||| Q_{k}|||_{e}^{2}\,.   
\end{align*}
An analogous equality holds for $||| \cdot |||_{e}'$, replacing the $S_{e, \alpha}$ with $V_{e, \gamma}$. Thus, applying \eqref{equa:equivalentSeminormsAux1} we conclude that \eqref{equa:equivalentSeminorms} holds.
\end{proof}

\begin{Lemm}\label{lemm:almostcommutelocally}
There exists a constant $C_2>0$, depending on the local dimension $d$ and the bath operators $\{ S_{\alpha}\}_{\alpha}$, such that for every $X \subset \EE$ and every $Q \in \mathcal{B}_{\Lambda}$
\begin{equation}\label{equa:almostcommutelocally}
|||Q|||_X^2 := \sum_{e\in X} \sum_{\alpha} \norm{\delta_{e,\alpha}(Q)}_{HS}^2\ge \frac{C_{2}}{|X|} \inf_{Q_0 \in \mathcal{B}_{\EE \setminus X}} \norm{Q-Q_0\otimes \mathbbm{1}_X}_{HS}^2 
\end{equation}
\end{Lemm}
\begin{proof}
The idea is similar to \cite{Nachtergaele2013} where it is shown that if an observable almost commutes with every observable supported on a region $X$, then it is close to an observable supported in $\Lambda \setminus X$. 

Let us start with the following observation: there exists a family of unitaries $\{ U_{j}\}_{j \in [d^{2}]}$, where $[d^{2}]:=\{ 0,1,\ldots, d^{2}-1\}$, that forms an orthogonal basis of $\mathcal{B}(\mathbb{C}^{d})$, and such that 
\begin{equation}\label{equa:almostcommutelocallyAux0} 
\frac{1}{d^{2}} \sum_{j \in [d^{2}]} U_{j} Q U_{j}^{\dagger} = \tr(Q) \frac{1}{d} \mathbbm{1}_{d}\, 
\end{equation}
for every $Q \in \mathcal{B}(\mathbb{C}^{d})$. Indeed, one can define $\theta := e^{2\pi i/d}$ and consider the Sylvester operators (also known as Weyl-Heisenberg matrices) $\Sigma_{1} = \sum_{k=0}^{d-2}\ketbra{k}{k+1} + \ketbra{d-1}{0}$ and $\Sigma_{3} = \sum_{k=0}^{d-1} \theta^{k} \ketbra{k}{k}$, which satisfy $\Sigma_{3} \Sigma_{1} = \theta \Sigma_{1} \Sigma_{3}$. Then, the family of unitaries
\[ \{ U_{k,l}:= \Sigma_{1}^{k} \Sigma_{3}^{l} = \theta^{kl} \Sigma_{3}^{l} \Sigma_{1}^{k} \colon 0 \leq k,l \leq d-1\}\] 
satisfies the above conditions. One can easily demonstrate that they are orthogonal with respect to the Hilbert-Schmidt scalar product, and so they form a basis. To check \eqref{equa:almostcommutelocallyAux0},  it is sufficient to check its validity when considering $Q$ as the elements of the basis, which can be easily verified. This completes the proof of the observation. Notice also that the commutant of this family satisfies $\{ U_{j}\}_{j \in [d^{2}]}'=\mathbb{C} \mathbbm{1}_{d}$ since it spans the whole algebra.

For each site $e \in \Lambda$, let us identify the local Hilbert space $\mathcal{H}_{e} \equiv \mathbb{C}^{d}$  and take the family $\{U_{e,k}\}_{k \in [d^{2}]}$ in $\mathcal{B}_{e} \hookrightarrow \mathcal{B}_{\Lambda}$ given in the previous observation. The families $\{ S_{e, \alpha}\}_{\alpha}$ and $\{U_{e,k}\}_{k \in [d^{2}]}$ satisfy the conditions of Lemma \ref{lemm:equivalentSeminorms} by the assumptions on the bath operators $S_{e, \alpha}$ and the construction of $U_{e, k}$, so there exists a constant $K>0$ independent of the system size and of the edge $e$, such that the norms given in the aforementioned lemma satisfy
\begin{equation}\label{equa:almostcommutelocallyAux0,5}
K |||Q|||_{e} \leq {|||Q|||'}_{e} \leq \frac{1}{K} |||Q|||_{e} 
\end{equation}
for every $Q \in \mathcal{B}_{\Lambda}$.

Let us now consider a finite subset $X$ of $\Lambda$. For each $\mathbf{k}: X \to [d^{2}]$ define the unitary $U_{X,\mathbf{k}}:= \prod_{e \in X} U_{e, \mathbf{k}(e)}$ supported on $X$. Then, as a consequence of \eqref{equa:almostcommutelocallyAux0} we have that
\begin{equation}
\frac{1}{d^{2|X|}} \sum_{\mathbf{k}: X \to [d^{2}]} U_{X,\mathbf{k}} Q U_{X,\mathbf{k}}^{\dagger} = \tr_{X}(Q) \otimes \tfrac{1}{d^{|X|}} \mathbbm{1}_{X} \,,
\end{equation}
and therefore
\[ Q - \tr_{X}(Q) \otimes \tfrac{1}{d^{|X|}} \mathbbm{1}_{X}  = \frac{1}{d^{2 |X|}}\sum_{\mathbf{k}:X \to [d^{2}]}  [Q, U_{X,\mathbf{k}}] U_{X,\mathbf{k}}^{\dagger}\,.\]
We can then upper bound the norm of this difference, using the fact that the HS norm is unitarily invariant,
\begin{equation}\label{equa:almostcommutelocallyAux1}
\begin{split}
\left\| Q - \tr_{X}(Q) \otimes \tfrac{1}{d^{|X|}} \mathbbm{1}_{X}  \right\|_{HS}^{2} & \leq  \frac{1}{d^{4|X|}} \left\|\sum_{\mathbf{k}:X \to [d^{2}]}  [Q, U_{X,\mathbf{k}}] U_{X,\mathbf{k}}^{\dagger}  \right\|_{HS}^{2}\\ 
& \leq \frac{1}{d^{2|X|}} \sum_{\mathbf{k}:X \to [d^{2}]}  \|[Q, U_{X,\mathbf{k}}] \|_{HS}^{2}\,.
\end{split}
\end{equation}
Next, we will make use of the fact that for every family of unitary operators $B_{1}, \ldots, B_{m}$ it holds that
\[ \|[Q, B_{1} \ldots B_{m}]\|_{HS} = \| \sum_{j=1}^{m} B_{1} \ldots B_{j-1}[Q, B_{j}] B_{j+1} \ldots B_{m} \|_{HS} \leq \sum_{j=1}^{m} \| [Q, B_{j}]\|_{HS}\,. \]
Using the product expression defining $U_{X,\mathbf{k}}$, we can apply the previous estimation to bound from above \eqref{equa:almostcommutelocallyAux1} as follows
\begin{align*}
\left\| Q - \tr_{X}(Q) \otimes \tfrac{1}{d^{|X|}} \mathbbm{1}_{X}  \right\|_{HS}^{2} & \leq \frac{1}{d^{2|X|}} \sum_{\mathbf{k}:X \to [d^{2}]}  \left( \sum_{e \in X} \| [Q,U_{e, \mathbf k(e)}] \|_{HS}\right)^{2}\\
& \leq \frac{|X|}{d^{2|X|}} \sum_{\mathbf{k}:X \to [d^{2}]} \sum_{e \in X} \| [Q, U_{e,\mathbf{k}(e)}] \|_{HS}^{2} \\
& = \frac{|X|}{d^{2|X|}} \sum_{e \in X} d^{2|X|-2}\sum_{k \in [d^{2}]} \|[Q, U_{e,k}] \|_{HS}^{2} \\
& \leq \frac{|X|}{ d^{2}K^2} \sum_{e \in X}   \sum_{\alpha \in \mathcal{A}} \|[Q, S_{e,\alpha}] \|_{HS}^{2} \\
& = \frac{|X|}{ d^{2}K^2} |||Q|||_{X}^{2}\,.
\end{align*}
where in the last inequality we have applied \eqref{equa:almostcommutelocallyAux0,5}. Finally, note that the infimum on the right hand-side of \eqref{equa:almostcommutelocally} is upper bounded by the first term in the previous expression, since $Q_{0} = \tr_{X}(Q) \otimes \tfrac{1}{d^{|X|}}\mathbbm{1}_{X} \in \mathcal{B}_{\Lambda \setminus X}$, so taking $C_{2}=d^{2}K^2$, we conclude with the desired inequality.

\end{proof}

\begin{Rema}\label{rema:translation-invariance}
One could weaken Assumption~\ref{assumption-ti}, and not require that the coupling constants $\widehat{g}_{e,\alpha}$ 
are independent of the edge $e$. In this case, one can still prove Proposition~\ref{Prop:lowerBoundLocalDaviesHamiltonian}, replacing $\widehat{g}_{\min}$ by 
\[
\widehat{g}_X = \min_{e\in X} \min_{\alpha, \omega} \widehat g_{e,\alpha}(\omega).
\]
If one is nonetheless able to find a lower bound to $\widehat{g}_{X}$ uniform in $X$, then one could use that bound in place of $\widehat{g}_{\min}$ to recover a uniform result as in the translation invariant case.
Similarly, if the jump operators $S_{e,\alpha}$ depend on the location $e$, then the constants $C_1$ and $C_2$ in Proposition~\ref{Prop:lowerBoundLocalDaviesHamiltonian} are not independent on $X$ anymore, but the result could be recovered if one is able to obtain uniform estimates on them. For $C_2$, one can take the supremum over the constants $K$ from Lemma~\ref{lemm:equivalentSeminorms}.
\end{Rema}

\subsection{Parent Hamiltonian vs Davies generator}
We now have have two Hamiltonians, one coming from the Davies generator and the other from the parent Hamiltonian construction (see Section \ref{sec:parent-hamiltonian})
\[ \sum_{e \in \EE} \widetilde{H}_{e} \quad \text{and} \quad H_{\EE} = \sum_{X \in \XXX}P_{X}^{\perp}\, \]
both having the same (unique) ground state $\ket*{\rho_\beta^{1/2}}$. Recall that $\chi = \chi_{n,N} = \mathcal{F}_{n,N}^{rect}$ where $n=n(\beta)$ is chosen as in Lemma \ref{rema:n-beta}. Since we have computed the gap of the parent Hamiltonian in Section~\ref{sec:parent-hamiltonian}, we now want to show that we can use that estimate to bound the gap of the Davies generator.
\begin{Prop}\label{Prop:davies-to-parent-hamiltonian}
If the Davies generator for the Quantum Double Model with group $G$ satisfies Assumptions~\ref{assumption-ti}-\ref{assumption-primitivity}, then its spectral gap $\mathcal L$ is bounded by
    \begin{equation}
        \operatorname{gap}(\mathcal L) \ge C\, \hat{g}_{\min} \frac{e^{-c \beta\, n(\beta)^2}} {n(\beta)^{4}} \, \operatorname{gap}(H_{\EE}),
    \end{equation}
    where $H_\EE$ is the parent Hamiltonian of the thermofield double state $\ket*{\rho_\beta^{1/2}}$ with parameter $n(\beta)$, $c$ and $C$ are positive constants independent of $\beta$, $G$ and the system size, and $\widehat{g}_{\min} = \min_{\alpha, \omega} \widehat{g}_{\alpha}(\omega)$.
\end{Prop}

 By combining Proposition~\ref{Prop:davies-to-parent-hamiltonian} with Corollary~\ref{coro:spectral-gap-parent-hamiltonian}, Theorem \ref{Theo:MainResultIntroduction} is then obtained as a corollary,: $\gap(H_{\EE})$ is a constant independent on $N$ and $\beta$, and $n(\beta)$ scales as $C'e^{\beta}$
for some constant $C'$ depending on $\abs{G}$ (see Lemma~\ref{rema:n-beta}). Therefore we can always find constants $\lambda$ and $c'$, independent of system size and $\beta$, such that
$\gap(\mathcal{L}) \ge \hat{g}_{\min} e^{-c' \, e^{\beta}} \lambda$.

\begin{proof}
We begin by counting how many rectangles $X \in \XXX$ contain a given edge $e$: we want to find an upper bound $m(\XXX)$ independent of the edge $e$. 


To estimate $m(\XXX)$, recall that each $X \in \XXX$ has dimension $a \times b$ with $2 \leq a,b \leq n(\beta)$, and so $X$ contains $O(n(\beta)^2)$ edges. For a fixed choice of $a$ and $b$, a given edge is therefore contained in at most $O(n(\beta)^2)$ rectangles of size $a\times b$. Since there are $O(n(\beta)^2)$ possible choices of $a$ and $b$ within the allowed range, we can roughly estimate 
\[ m(\XXX) \leq O(n(\beta)^{4})\,. \]
Then this implies that
\[ \sum_{e \in \EE} \widetilde{H}_{e}  \, \geq \, \frac{1}{m(\XXX)} \, \sum_{X \in \XXX}\left( \sum_{e \in X} \widetilde{H}_{e}  \right) \, \ge \, \frac{1}{m(\XXX)} \frac{C_2}{\abs{\Omega}} \widehat{g}_{\min} \sum_{X \in \XXX}  \frac{1}{\abs{X}}e^{-C_1 \beta \abs{R_X}} \Pi_{X}^{\perp}\, \]
where we used Proposition~\ref{prop:davies-local-gap}. As we just discussed, $\abs{X}$ is $O(n(\beta)^2)$, while $R_X$ is contained in a rectangle with sides of length at most $n(\beta)+4$, which also contains $O(n(\beta)^2)$ edges.  We can therefore find positive constants $C$ and $c$ such that
\[
\sum_{e \in \EE} \widetilde{H}_{e}  \, \geq \, C \, \widehat{g}_{\min} \frac{e^{-c \beta n(\beta)^2}}{n(\beta)^4}    \sum_{X \in \XXX} \Pi_{X}^{\perp}
\]

Next, we note that $P_{X} \geq \Pi_X$ for every rectangular region $X \subset \EE$, or equivalently
\begin{equation*} 
\ker{\Pi_X^{\perp}} \, = \, \{  \,\,    (A  \otimes e^{-\frac{\beta}{2} H_X}) e^{-\frac{\beta}{2}(H_\EE - H_X)} \,\, \colon \,\, A \in \mathcal{B}_{\EE \setminus X} \,\, \} \,  \subseteq  \,\, \operatorname{Im}(V_{X})\,.
\end{equation*}
Indeed, using the PEPS decomposition
\[
\begin{tikzpicture}[equation, yscale=0.6]

\draw (0.4,-0.5) -- (0.4,1.5);
\draw (1.6,-0.5) -- (1.6,1.5);

\filldraw[fill=gray!40!white, draw=black] (-0.3,0) rectangle (2.3,1);

\draw (1,0.5) node {$e^{-\frac{\beta}{2} H_{\EE}}$};

\draw (0.4,-1) node {$\mathcal{H}_{\EE \setminus X}$};
\draw (1.6,-1) node {$\mathcal{H}_{X}$};

\end{tikzpicture}
\quad = \quad
\begin{tikzpicture}[equation, yscale=0.6]

\draw (0.5,-0.5) -- (0.5,1.5);
\draw (2,-0.5) -- (2,1.5);

\filldraw[fill=blue!40!white, draw=black] (1.5,0) rectangle (2.5,1);

\draw (2,0.5) node {$V_{X}$};

\filldraw[fill=blue!40!white, draw=black] (0,0) rectangle (1,1);

\draw (0.5,0.5) node {$V_{\EE \setminus X}$};

\draw[thin, decorate, decoration={snake, segment length=1mm, amplitude=0.5mm}] (1,0.5) -- (1.5,0.5);

\draw (0.5,-1) node {$\mathcal{H}_{\EE \setminus X}$};
\draw (2,-1) node {$\mathcal{H}_{X}$};

\end{tikzpicture}
\]

\noindent we can compare

\[
\ker \Pi_{X}^{\perp} \, = \, \Bigg\{ \quad \begin{tikzpicture}[equation,yscale=0.6, baseline={([yshift=0.2cm]current bounding box.center)}]

\draw (0.4,-0.5) -- (0.4,3);
\draw (1.6,-0.5) -- (1.6,3);

\filldraw[fill=white!40!white, draw=black] (-0.3,1.5) rectangle (0.9,2.5);

\filldraw[fill=gray!40!white, draw=black] (-0.3,0) rectangle (2.3,1);

\draw (0.3,2) node {$A$};
\draw (1,0.5) node {$e^{-\frac{\beta}{2} H_{\EE} }$};

\draw (0.4,-1) node {$\mathcal{H}_{\EE \setminus X}$};
\draw (1.6,-1) node {$\mathcal{H}_{X}$};

\end{tikzpicture} 
\quad \colon A \in \mathcal{B}_{\EE \setminus X} \Bigg\} 
\]

\noindent and

\[
\operatorname{Im}(V_{X})  \, = \,\Bigg\{ \quad
\begin{tikzpicture}[equation, baseline={([yshift=0.2cm]current bounding box.center)}]

\begin{scope}[yscale=0.6, yshift=-9cm]

\draw (0.4,-0.5) -- (0.4,3);
\draw (1.9,-0.5) -- (1.9,3);

\filldraw[fill=white!40!white, draw=black] (-0.3,0) rectangle (0.9,2.5);

\filldraw[fill=blue!40!white, draw=black] (1.4,0) rectangle (2.4,1);

\draw (1.9,0.5) node {$V_{X}$};

\draw[thin, decorate, decoration={snake, segment length=1mm, amplitude=0.5mm}] (0.9,0.5) -- (1.4,0.5);

\draw (0.3,1.25) node {$A'$};

\draw (0.4,-1) node {$\mathcal{H}_{\EE \setminus X}$};
\draw (1.9,-1) node {$\mathcal{H}_{X}$};

\end{scope}

\end{tikzpicture}
\quad \colon A' \in \mathcal{B}_{\EE \setminus X} \otimes \mathcal{H}_{\partial (\EE \setminus X)} \Bigg\}
\]
This implies that $\sum_{X\in\XXX} \Pi_X^\perp \ge \sum_{X\in\XXX} P_X^\perp = H_{\EE}$, and chaining all the lower bounds we obtain the claimed result.
\end{proof}


\end{document}